\documentclass[11pt,reqno]{article}
\usepackage{geometry}
\usepackage{amsmath}
\usepackage{amssymb}
\usepackage{amsfonts}
\usepackage{amsthm}
\usepackage{graphicx}
\usepackage{subcaption}
\usepackage{tabularx}
\usepackage{booktabs}
\usepackage{enumitem}
\usepackage{wasysym}
\usepackage{hyperref}
\usepackage{algorithm}
\usepackage{tikz}
\usetikzlibrary{arrows.meta}
\usepackage{pgfplots}
\pgfplotsset{compat=1.18}

\numberwithin{equation}{section}

\theoremstyle{plain}
\newtheorem{theorem}{Theorem}[section]
\theoremstyle{definition}
\newtheorem{definition}[theorem]{Definition}

\newtheorem{remark}[theorem]{Remark}
\newtheorem{example}[theorem]{Example}
\newtheorem{proposition}[theorem]{Proposition}

\newtheorem{lemma}[theorem]{Lemma}
\newtheorem{corollary}[theorem]{Corollary}

\newcommand{\N}{\mathbb{N}}
\newcommand{\Z}{\mathbb{Z}}
\newcommand{\Q}{\mathbb{Q}}
\newcommand{\R}{\mathbb{R}}
\newcommand{\JD}{J\!D}
\newcommand{\JDN}{J\!D\!N}
\newcommand{\JDE}{J\!D\!E}
\newcommand{\TT}{\mathrm{TT}}
\newcommand{\TAI}{\mathrm{TAI}}
\newcommand{\UTone}{\mathrm{UT1}}
\newcommand{\UTC}{\mathrm{UTC}}
\newcommand{\DeltaT}{\Delta\mathrm{T}}
\newcommand{\ru}{P}
\newcommand{\rv}{Q}
\newcommand{\ix}{I}
\newcommand{\moon}{\mathrm{moon}}
\renewcommand{\sun}{\mathrm{sun}}
\newcommand{\sgang}{d}
\newcommand{\md}{t_0}
\newcommand{\ms}{\mu}

\newcommand{\eqmoon}{\texttt{moon\_equ}}
\newcommand{\eqsun}{\texttt{sun\_equ}}
\newcommand{\moonTab}{\texttt{moon\_tab}}
\newcommand{\sunTab}{\texttt{sun\_tab}}
\newcommand{\trueDate}{\texttt{true\_date}}
\newcommand{\meanDate}{\texttt{mean\_date}}

\newcommand{\Tab}{\texttt{Tab}}
\newcommand{\san}{r} 
\newcommand{\den}{\operatorname{den}}

\newcommand{\lcm}{\operatorname{lcm}}

\newcommand{\reformspec}[8]{%
\begin{center} \fbox{%
\begin{minipage}{0.95\textwidth}\small 
\begin{tabular}{@{}p{0.22\textwidth}p{0.73\textwidth}@{}} 
\emph{Goal} & #1\\[4pt] 
\emph{Precision target} & #2\\[4pt] 
\emph{Month module} & #3\\[4pt] 
\emph{Day module} & #4\\[4pt] 
\emph{Civil-day trigger} & #5\\[4pt] 
\emph{Time package} & #6\\[4pt] 
\emph{Numerical contract} & #7\\[4pt] 
\emph{Intended use} & #8 
\end{tabular} 
\end{minipage}} \end{center} 
}

\title{Possible Reforms of the Tibetan Lunisolar Calendar}
\author{Tsogtgerel Gantumur\\[0.5ex]
\small McGill University, Montr\'{e}al, QC, Canada\\
\small National University of Mongolia, Ulaanbaatar, Mongolia\\
\small Mongolian Academy of Sciences, Institute of Mathematics and Digital Technology\\[0.5ex]
\small \texttt{gantumur.tsogtgerel@mcgill.ca}}

\date{Vernal Equinox, 2026}

\begin{document}

\maketitle

\begin{abstract}
The family of Tibetan lunisolar calendars, which formalized a principle found in the
Kālacakra Tantra, operates on a common arithmetic axiom (\(67\) lunar months \(=65\)
solar months) that gives the tradition its rigid and predictable structure but also produces
an observable seasonal drift. The present study deconstructs the Tibetan calendar through a 
progressive analytical sequence: it first presents the calendar as an explicit computational 
procedure for leap months and lunar-day numbering, then isolates its structural core of 
incidence rules and mean-motion models. This separation clarifies which features are 
structurally forced and which are tradition-dependent, allowing the calendar's true 
inaccuracies to be rigorously decomposed into distinct sources: internal arithmetic drift, 
long-term seasonal misalignment of the sidereal framework, and anomaly-phase defects. 
Crucially, alongside these inaccuracies, an exhaustive computational analysis of the system also reveals a remarkable historical robustness: the specific discrete arithmetic of the traditional day rules renders boundary tie-cases operationally absent on historical timescales, while a combination of large internal temporal buffers and the inherent multi-hour inaccuracy of the classical lunar model historically buffered the calendar against moderate geographic variation.

On this basis, the paper develops a stratified reform space rather than a single replacement
proposal.  The resulting standards range from conservative rational repairs that preserve a
strongly traditional arithmetic character to increasingly explicit astronomical and numerical
reconstructions, culminating in fully dynamical calendar models based on true solar and
lunar motion.  Throughout, the guiding question is how far astronomical correction can be
carried without discarding the specifically Tibetan calendrical identity embodied in the
structural rules for month and day labeling.

A further theme of the paper is that calendric reform is not only a question of formulas and
constants, but also of numerical semantics and reproducibility.  The proposed standards are
therefore formulated not merely as abstract models but as executable and comparable
specifications, suitable for implementation, validation, and long-term transmission across
different computational environments.
\end{abstract}

\tableofcontents

\section{Introduction}
\label{s:intro}

The Tibetan calendar, originating from the Indian Kālacakra Tantra (translated into Tibetan c.\ 11th century), represents not a single entity but a family of closely related lunisolar systems integral to Tibetan culture and its sphere of influence \cite{schuh,henning,janson}.
Standardized in Tibet by the 13th century, its various traditions—most notably the Phugpa and Tsurphu schools—share core principles but differ in computational details.
A central calendrical anchor across Tibetan and Himalayan Buddhist communities is the New Year, Losar, often described as the most important celebration in the Tibetan calendar and observed, with regional variations, in Tibet as well as in Bhutan, Nepal, and India.
Variants of this calendrical tradition remain significant today in both religious and civil life.
For example, the Tögsbuyant (New Genden) version is used in Mongolia, where the New Year festival Tsagaan Sar is a major national public holiday;
Mongolia also pegs certain state observances to the traditional lunisolar calendar, notably National Pride Day, placed on the first day of the first winter month by the lunar calendar.
Similarly, related traditions are used by Tibetan Buddhist communities in Russia—for instance among Buryats and Tuvinians (Tögsbuyant) and Kalmyks (Phugpa)—and in Bhutan, where an official calendar tradition is used for dating major religious observances and public holidays.
This breadth of religious practice, cultural tradition, and (in some settings) official civil use, together with regional variation in the underlying rules, makes it essential to understand both the shared structure and the points where traditions genuinely diverge, and to ask in a disciplined way which aspects of that structure are candidates for preservation under reform.

\begin{figure}[ht]
\centering
\includegraphics[width=.7\textwidth]{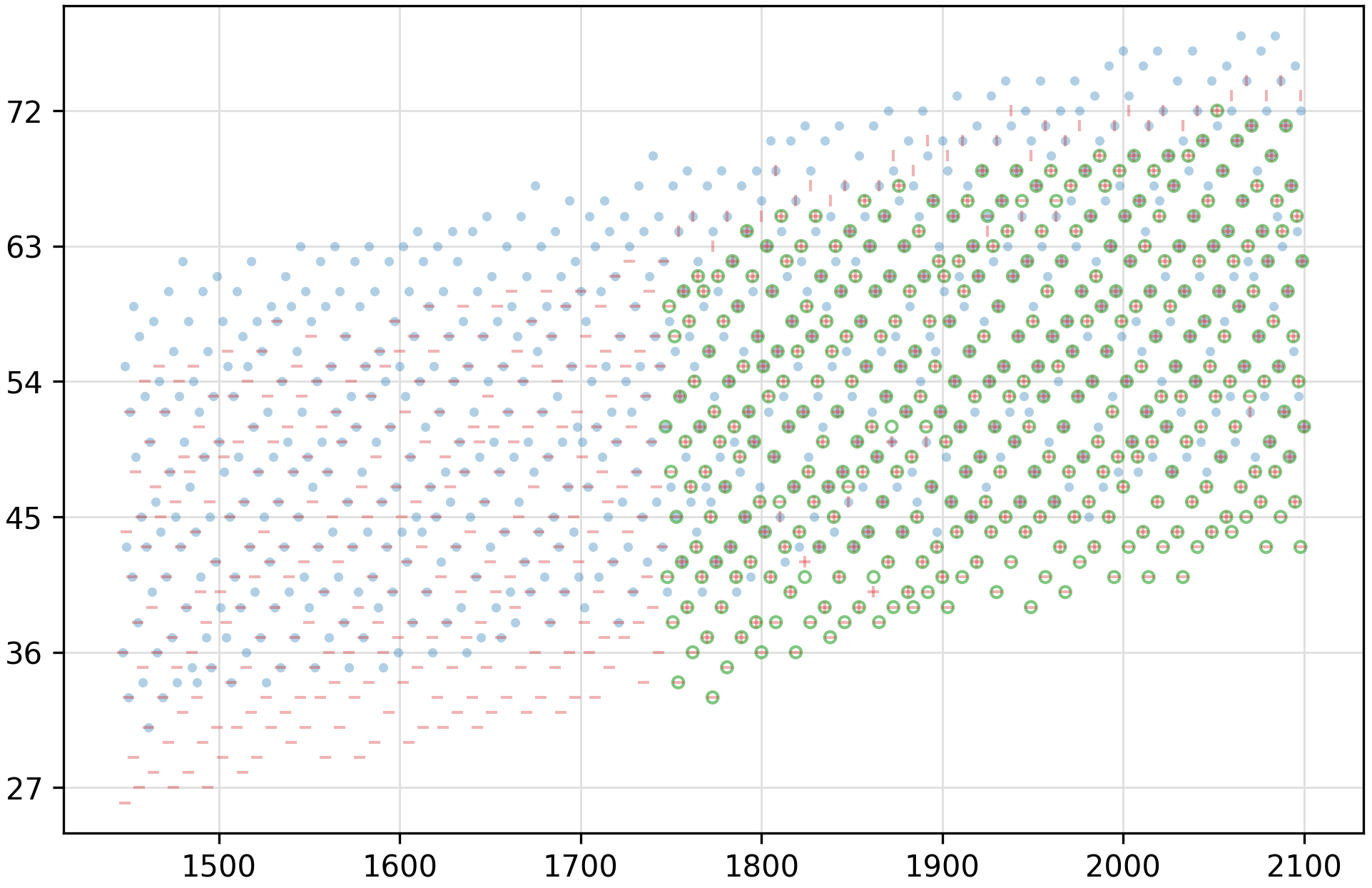}
\caption{Losar / Tsagaan Sar date (days since winter solstice) versus Gregorian year.
Markers: $\bullet$ Phugpa; $-$ Tsurphu; $|$ Bhutan; $\circ$ Mongol.}
\label{fig:losar}
\end{figure}

The central issue motivating this paper is the inherent seasonal drift shared by this family of calendars.
As lunisolar systems, they reconcile the incommensurable cycles of the synodic month and the tropical year.
Across the major variants (Phugpa, Tsurphu, Bhutanese, and Mongolian), this reconciliation is enforced by
a fixed arithmetic axiom \cite{schuh,henning,janson}:
\emph{67 mean lunar months are taken to equal 65 mean solar months.}
The axiom is elegant and computationally rigid, but it is not astronomically exact; the resulting small mismatch
accumulates over centuries and shifts calendrical dates progressively later with respect to the seasons.
This drift is directly visible in the gradual postponement of New Year, as seen in the shifting dates of Losar and Tsagaan~Sar.
The first recorded occurrence of Losar in March within the computed range appears in 1843 for Phugpa, in 1911 for Bhutan, and in 2025 for the Tsurphu and Mongol traditions.
Figure~\ref{fig:losar} illustrates both the shared long-term trend and the small, systematic offsets between traditions that arise from differences in intercalation, epoch choice, and day-counting conventions. 
These mechanisms are unpacked in \S\ref{s:analysis}-\S\ref{s:inaccuracies}.
The computations generating Figure~\ref{fig:losar} are fully reproducible; indeed, tracking these long-term divergences serves as a prime example of how implementation-level diagnostics can illuminate the structural differences between traditions.

Our strategy is to separate \emph{what is structurally forced} from \emph{what is tradition-dependent}.
We first present each calendar as an explicit arithmetic procedure, suitable for implementation and cross-checking.
We then reinterpret these procedures as incidence relations and inverse mean-motion models, isolating a small set of discrete
and continuous parameters that control the leap-month pattern, day numbering, and long-term drift.
This two-level description makes comparisons transparent, provides a common language for assessing where (and how)
changes to the arithmetic or the underlying astronomical model propagate through the system, and also makes it possible
to express both traditional calendars and reform proposals as executable, reproducible standards rather than as prose recipes alone.

The result is not a single reform proposal but a stratified reform space.
Some proposals preserve a strongly traditional arithmetic character while correcting only the most consequential defects;
others move toward increasingly explicit astronomical and numerical standardization.
A central theme of the paper is that these choices need not be framed as a binary opposition between
``traditional'' and ``astronomical'' calendars: the Tibetan calendar admits a layered decomposition, and reform can act on different layers with different degrees of commitment.

Given the interdisciplinary nature of this study, readers may wish to navigate directly to sections matching their primary objectives:
\begin{itemize}
    \item {Implementers and software developers} seeking the raw algorithms to compute traditional dates should focus on the operational recipes in Section~\ref{s:tib_principles} and the data tables in the Appendices.
    \item {Astronomers and historians of science} interested in the kinematic structure, implicit geography, and error budget of the historical calendar will find the core analytical deconstruction in Sections~\ref{s:analysis} and \ref{s:inaccuracies}.
    \item {Policymakers, clergy, and civil authorities} looking for the actual reform proposals can turn directly to Section~\ref{s:reforms}, which catalogs the specific modernization standards and their intended uses.
\end{itemize}

To keep the analysis modular, the remainder of the paper unfolds in a progressive sequence: from operational deconstruction and structural analysis to diagnostic repair and, ultimately, formal standard-setting.
Section~\ref{s:tib_principles} presents the Tibetan calendar strictly in its operational form: it establishes the foundational chronological elements, mean motions, and the leap-month rule (\S\ref{ss:leap_rule}), and details the evaluation of true celestial longitudes alongside the explicit day-counting algorithm (\S\ref{ss:day_calculation}).
Section~\ref{s:analysis} then opens the box and isolates the small number of structural inputs that make these discrete recipes
so rigid and computable: \S\ref{subsec:dual} formulates the abstract incidence mechanisms (containment versus inheritance)
and their duality; \S\ref{ss:mean_sun_models} derives the leap-month cycle from a linear mean-Sun model with \(12\)
definition points and explains the intercalation-index shortcut; and \S\ref{ss:day_models} reconstructs the day algorithm as
an exact inverse formulation of a first-anomaly kinematic model. 
This extensive subsection 
derives the continuous underlying constants, evaluates them against modern astronomical 
benchmarks, and culminates in an exhaustive tie-case search proving that boundary ambiguities 
are effectively eliminated by the rigid arithmetic structure of the traditional rules. 
Building on this, \S\ref{ss:traditions} analyzes the calendar's epoch constants, 
rigorously demonstrating how a combination of massive internal temporal buffers and the 
inherent multi-hour inaccuracy of the classical lunar model historically insulated the 
system from geographic variance, despite its implicit anchor near Lhasa.

Having isolated this mathematical skeleton, Section~\ref{s:inaccuracies} systematically deconstructs the calendar's
astronomical error budget and lays the technical groundwork for modernization. Rather than merely cataloging defects, 
this section formulates explicit solutions: it quantifies the internal arithmetic drift and precessional misalignment, 
derives highly accurate rational intercalation schemes (such as a proposed 334-year cycle), and specifies the precise 
physical corrections---ranging from secondary lunar inequalities and anomaly-phase alignment to spherical sunrise geometry 
and the equation of time---required to achieve multi-century precision.

Drawing on these modular components, Section~\ref{s:reforms} develops a ladder of reform standards ranging from
conservative rational repairs to fully astronomical realizations, together with low-commitment alternatives and questions of enactment.
The software and reproducibility dimension of this program is addressed explicitly in \S\ref{ss:reference_impl}, where
the proposed standards are related to a reference implementation and validation framework.
Finally, Section~\ref{s:conclusion} summarizes the structural findings, the resulting reform space, and the broader methodological lesson that a modern calendric standard must now be understood as a combination of formulas, constants, conventions, numerical semantics, and executable validation tools.

The appendices centralize the concrete infrastructure underlying this analysis, ensuring that extensive reference constants, parameters, algorithms, and theoretical derivations do not distract from the main narrative flow. While many of these values are introduced conceptually in the main body, collecting them here provides an easily searchable, unified reference. This material includes traditional tables, modern astronomical formulas, low-commitment arithmetic baselines, and strictly reproducible numerical functions. Notably, Appendix~\ref{app:reference_modules} acts as a comprehensive technical specification library, cleanly recording the exact rational constants and reusable computational modules that formally define the proposed reform tiers.

\section{Tibetan Calendrical Rules and Guiding Principles}
\label{s:tib_principles}

The Tibetan calendrical traditions implement a deterministic arithmetic pipeline: starting from a small set of mean-motion constants and an epoch, they decide which lunations are regular or intercalary and then assign month and day labels (with the familiar repeated/skipped days) to the resulting sequence of boundary times.  For the purposes of later “reform” questions, it is helpful to keep two layers conceptually separate: the underlying celestial model (how one produces mean Sun/Moon phases and anomaly corrections) versus the incidence conventions that turn those phases into month and day labels.  In this section we therefore present the month and day rules in a mostly operational form, while the accompanying “theoretical basis” discussions indicate which parts of the computation are forced by simple model structures and which are merely conventions—pointing toward what might plausibly be regarded as the calendar’s structural “essence” and what might be changed without altering that essence.

\subsection{The leap month rule: The common arithmetic foundation}
\label{ss:leap_rule}

The cornerstone of the structure shared by the Phugpa, Tsurphu, and other traditional calendars is the arithmetic rule governing the insertion of leap months (intercalation). This rule is a formalization of an approximation found in the foundational Kālacakra Tantra. While the original Tantra, as a practical \textit{karana} text, treated the relationship as a useful but inexact guide requiring observational correction, the later \textit{siddhānta} traditions elevated it to a foundational axiom \cite{schuh,henning,janson}:
\begin{quote}
\centering
\textbf{67 mean lunar months = 65 mean solar months.}
\end{quote}
This relation is treated as definitionally exact within the traditional logic and directly implies a celestial model based on {\em mean} solar and lunar motions. 
The primary calendrical consequence is the mandated frequency of leap months: exactly 2 must occur for every 65 regular months, resulting in a predictable 65-year cycle containing 24 leap months.

\subsubsection{Computational algorithm}
\label{sss:comp_algo}

In practice, the placement of leap months in all major Tibetan-derived
calendars is governed by a single, purely arithmetic procedure.
A \emph{number of solar months} $M^{*}$ is counted from a fixed epoch
$(Y_{0},M_{0})$.  For a given month $M$ in year $Y$, one defines
\begin{equation}
  M^{*} = 12\,(Y-Y_{0}) + (M-M_{0}).
  \label{eq:Mstar}
\end{equation}
Here we use the standard convention (see e.g. \cite{janson}): $Y$ is the \emph{lunar-year label} (the year that begins at Losar/Tsagaan~Sar), and $M\in\{1,\dots,12\}$ is the lunar month number within that lunar year. 
Thus a lunar year $Y$ may overlap two civil (solar/Gregorian) years; for example, its final months can fall in early $Y\!+\!1$ in the civil calendar.
Now from $M^*$, an \emph{intercalation index} $\ix$ is computed as\footnote{Janson \cite{janson} denotes the intercalation index by $ix$.}
\begin{equation}
  \ix = (67 M^{*} + \beta^{*}) \bmod 65
     = (2 M^{*} + \beta^{*}) \bmod 65,
  \label{eq:ix}
\end{equation}
where $\beta^{*}$ is an integer constant determined by the chosen epoch
and calendrical tradition \cite{janson}.
A leap month is then inserted whenever $\ix$ falls within a small tradition-specific set of ``trigger'' values,
meaning that the calendar contains two consecutive lunar months carrying the same month label
(either $M$ or $M-1$, depending on the tradition; see Remark~\ref{rem:leap-naming}).

\begin{table}[h]
\centering
\begin{tabular}{llll}
\toprule
Tradition & $(Y_0,M_0)$ & $\beta^\ast$ & Trigger set\\
\midrule
Phugpa (E1987) & $(1987,3)$ & $0$  & $\ix\in\{48,49\}$\\
Tsurphu (E1732) & $(1732,3)$ & $59$ & $\ix\in\{0,1\}$\\
Tsurphu (E1852) & $(1852,3)$ & $14$ & $\ix\in\{0,1\}$\\
Bhutan (E1754) & $(1754,3)$ & $2$ & $\ix\in\{59,60\}$\\
Mongol (E1747) & $(1747,3)$ & $10$ & $\ix\in\{46,47\}$\\
\bottomrule
\end{tabular}
\caption{Leap-month parameters for the intercalation-index rule.}
\label{tab:ix_rules}
\end{table}

Although the astronomical constants (such as the mean solar and lunar motions)
are shared across several traditions, different choices of epoch
$(Y_{0},M_{0})$, constants $\beta^{*}$, and threshold values for $\ix$
lead to distinct but perfectly predictable leap-month sequences.
Representative choices for the four principal traditions discussed
in this paper are summarized in Table~\ref{tab:ix_rules}.

\begin{example}
\label{eg:mongol-month-leap}
Using the Mongol parameters from Table~\ref{tab:ix_rules},
for $(Y,M)=(2026,1)$ we get
\[
M^*=12(2026-1747)+(1-3)=3346,
\]
and
\[
\ix=(2M^*+\beta^*)\bmod 65=(2\cdot 3346+10)\bmod 65=6702\bmod 65=7.
\]
Since $7\notin\{46,47\}$, this month is regular.
\end{example}

\begin{example}
Using the Phugpa parameters, for $(Y,M)=(2024,6)$ we compute
\[
M^*=12(2024-1987)+(6-3)=447,
\]
and
\[
\ix=(2M^*+\beta^*)\bmod 65=(2\cdot 447)\bmod 65=894\bmod 65=49.
\]
Since $49\in\{48,49\}$, the intercalation rule triggers: the month number $M=6$ occurs twice in succession (one regular and one leap, cf. Remark~\ref{rem:leap-naming}).
\end{example}

\begin{figure}[ht]
\centering
\includegraphics[width=\textwidth]{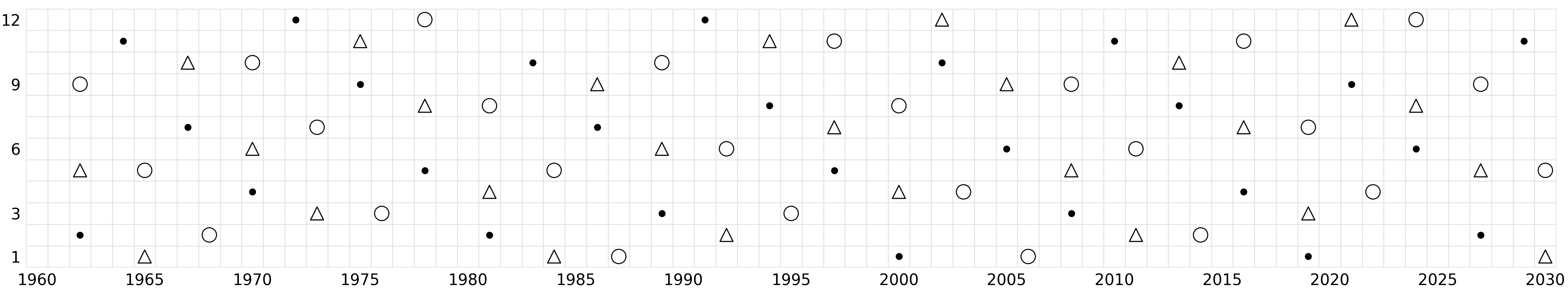}
\caption{Leap-month placements (1960–2030). Each square represents a month $(Y,M)$. 
Filled dots: Phugpa; open circles: Tsurphu/Mongol; triangles: Bhutan.}
\label{fig:leapmonth_grid}
\end{figure}

\begin{remark}
Fix the normalization $M_0=3$ as in Table~\ref{tab:ix_rules} and \cite{janson}, and write
\[
M^*_{Y_0}(Y,M)=12(Y-Y_0)+(M-3),
\qquad
\ix_{Y_0,\beta^*}(Y,M)\equiv 2M^*_{Y_0}(Y,M)+\beta^* \pmod{65}.
\]
A direct computation gives the \emph{epoch-translation formula}
\begin{equation}
\ix_{Y_1,\beta_1^*}(Y,M)-\ix_{Y_2,\beta_2^*}(Y,M)
\equiv 24\,(Y_2-Y_1)+(\beta_1^*-\beta_2^*) \pmod{65},
\label{eq:ix_translation}
\end{equation}
which is independent of $(Y,M)$.

\begin{enumerate}[label=(\alph*), leftmargin=2.2em, itemsep=4pt, topsep=2pt]
\item 
We show that the Mongol and Tsurphu calendars have identical leap months.
Applying \eqref{eq:ix_translation} with $(Y_1,\beta_1^*)=(1747,10)$ and
$(Y_2,\beta_2^*)=(1852,14)$ yields
\[
\ix_{1747,10}(Y,M)-\ix_{1852,14}(Y,M)\equiv 24\cdot(1852-1747)+(10-14)
\equiv 46 \pmod{65}.
\]
Hence $\ix_{1747,10}\equiv \ix_{1852,14}+46 \pmod{65}$, so
\[
\ix_{1852,14}(Y,M)\in\{0,1\}\quad\Longleftrightarrow\quad
\ix_{1747,10}(Y,M)\in\{46,47\}.
\]
Therefore the Tsurphu leap rule ($\ix\in\{0,1\}$) and the Mongol leap rule
($\ix\in\{46,47\}$) select the same leap months as pairs $(Y,M)$.
See Figure~\ref{fig:leapmonth_grid}.
\item Now let us look at the two Tsurphu epochs.
Applying \eqref{eq:ix_translation} with $(Y_1,\beta_1^*)=(1732,59)$ and
$(Y_2,\beta_2^*)=(1852,14)$ yields
\[
\ix_{1732,59}(Y,M)-\ix_{1852,14}(Y,M)\equiv 24\cdot(1852-1732)+(59-14)
\equiv 0 \pmod{65}.
\]
Thus $\ix_{1732,59}(Y,M)=\ix_{1852,14}(Y,M)$ for all $(Y,M)$, and the two published
Tsurphu epochs give identical leap-month placements.
\end{enumerate}
\end{remark}

\begin{remark}\label{rem:leap-naming}
Assume that the intercalation index triggers for the labeled month $(Y,M)$.
Then one extra lunation is inserted at this point, so two consecutive months appear where normally there would be one.
Traditions differ only in how the \emph{labels} of these two months are assigned.
In the standard Tibetan convention used here (Phugpa, Tsurphu, Mongol), it is the label $M$ that repeats:
\[
\cdots,\ (Y,M\!-\!1),\ (Y,M)_{\mathrm{leap}},\ (Y,M)_{\mathrm{reg}},\ (Y,M\!+\!1),\ \cdots .
\]
In the Bhutanese convention the repeated label is attached to the preceding month:
\[
\cdots,\ (Y,M\!-\!1)_{\mathrm{reg}},\ (Y,M\!-\!1)_{\mathrm{leap}},\ (Y,M),\ (Y,M\!+\!1),\ \cdots .
\]
Equivalently, one may keep the slogan ``repeat the current label'' by shifting the parameter $\beta^*$.
Since $ix\equiv 2M^*+\beta^*\pmod{65}$, replacing $\beta^*$ by $\beta^*+2$ shifts the trigger condition back by one
month. 
Alternatively, one may keep $\beta^*$ fixed and shift only the trigger residues.
For Bhutan, this means using $\beta^*=4$ (instead of $2$) and keeping the trigger set $T=\{59,60\}$, 
or keeping $\beta^*=2$ and replacing $T$ by $T-2=\{57,58\}$.
With this reparametrization, the leap month is taken to be the \emph{later} of the two consecutive months carrying the repeated label.
Schuh~\cite{schuh} notes that some published Tsurphu almanacs insert the leap month one month later than the convention described above (e.g.\ in 1964/65 and 1970/71).
\end{remark}

\begin{remark}
\label{rem:true-month-index}
It is often necessary to convert a labeled month $(Y,M)$ into the corresponding \emph{true-month index} $n$, the running count of lunations from a chosen epoch month.
This provides a single, unambiguous month counter across leap-month repetitions.

Since the intercalation rule is purely arithmetic and $65$-periodic, this conversion can be written in closed form.
Recall that a leap month occurs precisely when $\ix\equiv 2M^*+\beta^* \pmod{65}$ falls in the tradition-dependent trigger set
$T=\{\tau,\tau+1\}\subset\Z/65\Z$.
Let $\gamma\in\{0,1,\dots,64\}$ be the least nonnegative residue satisfying
\[
\gamma\equiv -\tau \pmod{65}.
\]
Thus $(Y,M)$ is a trigger label precisely when
\[
2M^*+\beta^*+\gamma\equiv 0 \text{ or } 1 \pmod{65}.
\]
Define
\[
n_+(Y,M)
= \Big\lfloor \frac{67M^*+\beta^*+\gamma}{65}\Big\rfloor
= M^* + \Big\lfloor \frac{2M^*+\beta^*+\gamma}{65}\Big\rfloor.
\]
If $(Y,M)$ is a \emph{non-trigger} label, then the month $(Y,M)$ occurs once and its true-month index is
$n(Y,M)=n_+(Y,M)$.
If $(Y,M)$ is a \emph{trigger} label, then two consecutive lunations carry the same label in the standard Tibetan
convention used here (Remark~\ref{rem:leap-naming}); the later copy has index $n_+(Y,M)$ and the earlier
(intercalary) copy has index $n_-(Y,M)=n_+(Y,M)-1$.

For Phugpa, $T=\{48,49\}$ so $\gamma=17$, and this agrees with Janson's formula \cite[(5.10)]{janson}
after translating normalizations.  For Tsurphu, $T=\{0,1\}$ and thus $\gamma=0$, so in particular
the epoch month $(Y_0,3)$ has $n=0$, in agreement with the traditional Tsurphu rounding rule.
In the Bhutanese convention the repeated label is the \emph{preceding} month.
If we adopt the reparametrization of Remark~\ref{rem:leap-naming} by keeping $\beta^*$ fixed
and replacing the trigger set $\{59,60\}$ by $\{57,58\}$,
then the \emph{repeated} label is again the trigger label and the same conversion applies verbatim:
a trigger label $(Y,M)$ corresponds to two consecutive lunations with indices
$n_-(Y,M)=n_+(Y,M)-1$ and $n_+(Y,M)$, but now the \emph{later} copy is designated as the leap month
(and the earlier as the regular month).
\end{remark}

\begin{example}
\label{eg:true-month-index}
Using the Mongol parameters from Table~\ref{tab:ix_rules}, for $(Y,M)=(2026,1)$ we have
\[
M^* \;=\; 12(2026-1747)+(1-3)\;=\;3346.
\]
For Mongol, $\beta^*=10$ and $T=\{46,47\}$, hence $\tau=46$ and $\gamma=65-\tau=19$.
First check that $(2026,1)$ is a non-trigger label:
\[
\ix \equiv 2M^*+\beta^* \equiv 2\cdot 3346 + 10 \equiv 6702 \equiv 7 \pmod{65}\notin\{46,47\}.
\]
Therefore $n(2026,1)=n_+(2026,1)$, and by Remark~\ref{rem:true-month-index},
\[
n(2026,1)
= 3346 + \Bigl\lfloor \frac{2\cdot 3346 + 10 + 19}{65}\Bigr\rfloor
= 3346 + \Bigl\lfloor \frac{6721}{65}\Bigr\rfloor
= 3346 + 103
= 3449.
\]
\end{example}

\subsubsection{Conceptual basis: \textit{Sgang} and leap months}
\label{ss:sgang}

While the intercalation index gives an efficient computational rule, its logic comes from a more basic
mean-Sun criterion.  In the traditional description, lunar months are regulated by the passage of the
\emph{mean} Sun through a fixed set of twelve reference longitudes on the ecliptic, called \emph{definition points}
(\emph{sgang}, short for \emph{dbugs sgang}) \cite{henning,schuh,janson}.
These markers are closely related\footnote{Schuh discusses the Chinese \emph{qi} terminology used in Tibetan sources: the ``centers'' (\emph{zhongqi})
correspond to Tibetan \emph{sgang}, while the related ``knots'' (\emph{jieqi}) correspond to Tibetan \emph{dbugs}
\cite{schuh}.  For a modern overview of the Chinese month rule in terms of principal terms, see \cite{aslaksen}.}
 (historically and logically) to the Chinese system of principal solar terms,
though Tibetan traditions typically use longitudes that are shifted relative to standard sign boundaries. 
For example,
in the Phugpa system the first month is associated to the mean Sun passing roughly $8^\circ$ Aquarius, then $8^\circ$ Pisces, and so on, cf. Remark~\ref{r:sgang-trad} and \cite{henning}.

A lunation is assigned the month number $M$ if the mean Sun crosses the $M$-th definition point during that lunation.
If no definition point is crossed, the lunation is intercalary.  
This is the same intercalary lunation in all traditions; they differ only in how it is \emph{named}:
in the Bhutanese convention it is labeled by the preceding month, whereas in the non-Bhutanese traditions it is labeled by the following month.
The arithmetic index $\ix$ is a computational shortcut:
it packages this crossing test into a purely periodic congruence, so one does not need to perform the geometric
comparison month by month \cite{janson}.  A more detailed derivation of $\ix$ from the mean-Sun model is given in
\S\ref{ss:mean_sun_models}.

As a concrete example, Figure~\ref{fig:month-diag} shows the Phugpa leap month of 2024, with the solar longitude understood throughout as the \emph{calendar's own computed mean solar longitude}. 
The lunation before the shaded interval contains the calendar's crossing of the fifth definition point, so it is labeled month~5. 
During the shaded lunation, the computed mean Sun moves from just after \(\mathrm{Sg}\,5\) to just before \(\mathrm{Sg}\,6\), without crossing any definition point; it is therefore intercalary. 
In the non-Bhutanese convention, such an intercalary lunation takes the label of the following month, so the shaded lunation is the leap month~6, and the next lunation, which contains the calendar's crossing of \(\mathrm{Sg}\,6\), is the regular month~6.

\begin{figure}[ht]
\centering
\includegraphics[width=.7\textwidth]{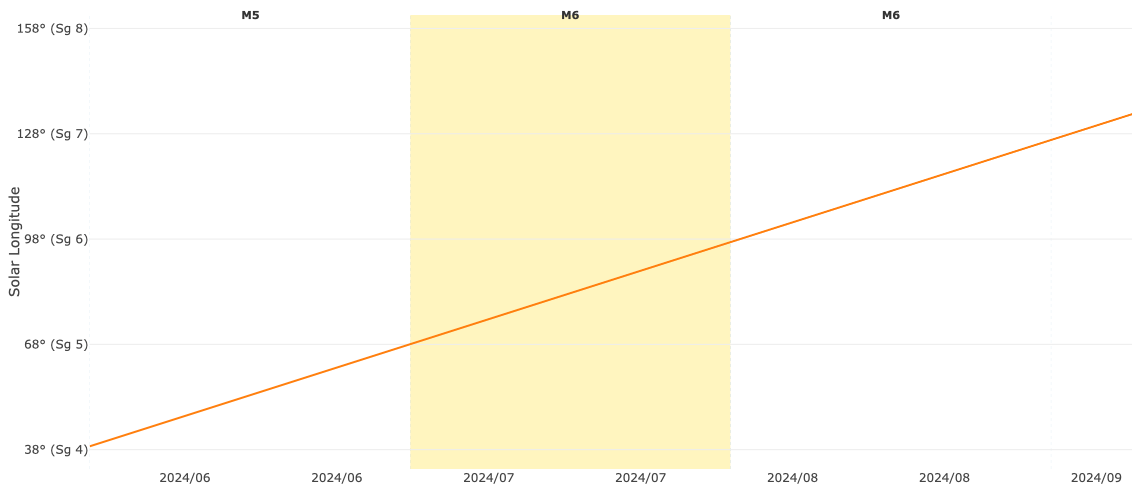}
\caption{Phugpa leap month of 2024.}
\label{fig:month-diag}
\end{figure}

This structure can be contrasted with the neighboring Indian and Chinese systems, which use different logical rules and celestial models to regulate their months:
\begin{itemize}
\item \emph{Indian system (inheritance):} Months are regulated by 12 solar divisions given by the zodiac
(r\=a\'si). A lunation is named by the r\=a\'si of the Sun at a distinguished lunar phase
(new moon in the am\=anta system, full moon in the p\=urnim\=anta system). Modern Indian
practice uses the true Sun and Moon. The scheme is typically sidereal (star-fixed), hence it
drifts against the seasons due to precession \cite{dershowitz,sewell}.

\item \emph{Chinese system (containment):} Months are regulated by 12 principal solar terms (zh\=ongq\`i),
anchored to the solstices and equinoxes (tropical, seasonally fixed). A lunation is named by
the principal term it contains; a lunation containing no principal term is intercalary \cite{aslaksen,doggett}.

\item \emph{Tibetan system (containment):} Months are regulated by the 12 definition points (\emph{sgang}),
and the month rule is expressed using the \emph{mean} Sun rather than the true Sun \cite{henning,schuh,janson}.
Using the mean Sun makes the leap-month pattern perfectly regular and therefore reducible to a purely arithmetic rule (the index $\ix$);
schemes based on the true Sun typically lead to an irregular intercalation pattern.
\end{itemize}

\begin{figure}[ht]
\centering
\begin{tikzpicture}[x=1cm,y=1cm, font=\small]

\tikzset{
  bgbox/.style={draw=black!55, fill=black!6, line width=0.45pt},
  baseline/.style={draw=black!35, line width=0.6pt},
  tick/.style={draw=black!75, line width=0.8pt},
  post/.style={draw=black!85, line width=1.2pt},
  lab/.style={text=black!85},
}

\def\W{6.6}
\def\Gap{1.2}

\begin{scope}[shift={(0,0)}]

  \draw[bgbox] (0.00,1.55) rectangle (1.80,1.95);
  \draw[bgbox] (1.80,1.55) rectangle (3.60,1.95);
  \draw[bgbox] (3.60,1.55) rectangle (5.00,1.95);
  \draw[bgbox] (5.00,1.55) rectangle (\W,1.95);

  \node[lab] at (0.90,1.75) {$A$};
  \node[lab] at (2.70,1.75) {$B$};
  \node[lab] at (4.30,1.75) {$C$};
  \node[lab] at (5.80,1.75) {$D$};

  \draw[baseline] (0.00,0.90) -- (\W,0.90);

  \foreach \t in {0.30,1.90,3.50,5.10}{
    \draw[post] (\t,0.8)--(\t,1.0);
  }

  \node[lab] at (0.30,0.6) {$A$};
  \node[lab] at (1.90,0.6) {$B$};
  \node[lab] at (3.50,0.6) {$B$};
  \node[lab] at (5.10,0.6) {$D$};

\end{scope}

\begin{scope}[shift={(\W+\Gap,0)}]

  \draw[baseline] (0.20,1.55) -- (\W,1.55);

  \foreach \p/\L in {1.5/A, 3.50/B, 4.90/C, 6.5/D}{
    \draw[post] (\p,1.45)--(\p,1.65);
    \node[lab] at (\p,1.82) {$\L$};
  }

  \draw[bgbox] (0.20,0.55) rectangle (1.80,0.95);
  \draw[bgbox] (1.80,0.55) rectangle (3.40,0.95);
  \draw[bgbox] (3.40,0.55) rectangle (5.00,0.95);
  \draw[bgbox] (5.00,0.55) rectangle (6.60,0.95);

  \node[lab] at ({0.5*(0.2+1.8)},0.75) {$A$};
  \node[lab, font=\itshape] at ({0.5*(1.8+3.4)},0.75) {};
  \node[lab] at ({0.5*(3.4+5.0)},0.75) {$B,C$};
  \node[lab] at ({0.5*(5.0+6.6)},0.75) {$D$};

\end{scope}

\end{tikzpicture}

\caption{Inheritance (left) and containment (right): posts inherit interval labels vs.\ intervals are labeled by the post(s) they contain (empty or multiple $\Rightarrow$ anomalous labeling).}
\label{f:cont-inh}
\end{figure}
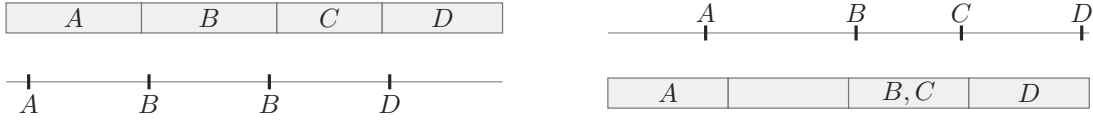

\noindent
The labels \emph{inheritance} and \emph{containment} in the bullets above refer to two complementary ways of attaching
names to lunations, illustrated in Figure~\ref{f:cont-inh}.  In an \emph{inheritance} rule, labels live on background
intervals (e.g.\ solar divisions) and are inherited by foreground posts (e.g.\ new/full moons) according to which
interval contains the post.  In a \emph{containment} rule, labels live on background posts (e.g.\ principal terms) and
are assigned to foreground intervals according to which labeled posts they contain.  
We formalize these notions in \S\ref{subsec:dual} using a simple incidence framework.

A useful difference between the two schemes is where the \emph{conventional choice} enters. In an \emph{inheritance} rule, skipped labels are unambiguous: if no foreground post falls in a given background interval, then that interval's label is simply absent from the month sequence. By contrast, when two consecutive foreground posts fall in the same background interval, they inherit the same label, so one must decide which of the two is to be designated as the ``extra'' or ``leap'' occurrence. In Indian lunisolar practice this convention is standard: the first of the two equal labels is the \emph{adhika} month and the second is the ordinary one. For a \emph{containment} rule, the situation is complementary. 
If a foreground interval contains no labeled background post, then the interval is unambiguously the ``extra'' one; but one must decide which label it receives, namely whether it is named from the preceding or the following post. 
If a foreground interval contains two labeled posts, then it is again clear that one label must be skipped, but a convention is needed to determine which of the two available labels is retained by that interval and which one becomes the skipped label. 
Thus an inheritance rule needs a convention only for repeated labels, whereas a containment rule may require conventions for both repeated and skipped labels. 
In the traditional Tibetan month rule only the first issue arises, because the use of the mean Sun makes skipped months impossible.
One therefore needs only the convention that distinguishes the Bhutanese and non-Bhutanese naming of the intercalary month.

\subsection{Day calculation and numbering}
\label{ss:day_calculation}

Day numbering exhibits the same kinds of irregularities as month naming---notably repeated and skipped
labels---but on a finer time scale.  We describe the Tibetan day rule in the same informal, operational
style as \S\ref{ss:leap_rule}, deferring the abstract incidence framework to \S\ref{subsec:dual} and the
celestial-model interpretation to \S\ref{ss:day_models}.

The daily structure is governed by computations that model the Moon's motion relative to the Sun.
The basic theoretical unit is the \emph{lunar day} (\emph{tshes-zhag}; also called \emph{tithi}),
defined in terms of elongation (the angular separation between the Sun and Moon): one lunar day is
the time needed for the elongation to increase by $12^\circ$ ($=360^\circ/30$). Because the Moon's
apparent speed varies, a lunar day has variable length (about $21.5$ to $25.7$ hours). The
calendrical problem is to translate this continuous sequence of lunar-day boundaries into the
discrete sequence of civil days (\emph{nyin-zhag}), counted from dawn to dawn.

In the Tibetan calendar, the naming rule is \cite{schuh,henning,janson}:
\begin{quote}\centering
\textbf{A civil day is assigned the number of the lunar day that is current at its beginning (dawn).}
\end{quote}
Figure~\ref{fig:tithi-to-civil} shows this inheritance rule schematically.  If a lunar day is long enough to cover two dawns,
then two consecutive civil days inherit the same label (a repeated day).  If a lunar day begins and ends between two dawns,
then no civil day begins during it and its label never appears (a skipped day).

\begin{figure}[ht]
\centering
\begin{tikzpicture}[x=1cm,y=1cm, font=\small]

\tikzset{
  bgbox/.style={draw=black!55, fill=black!6, line width=0.45pt},
  baseline/.style={draw=black!35, line width=0.6pt},
  tick/.style={draw=black!75, line width=0.8pt},
  post/.style={draw=black!85, line width=1.2pt},
  lab/.style={text=black!85},
}

\def\W{12}
\def\Gap{1.2}

\begin{scope}[shift={(0,0)}]

  \draw[bgbox] (0.00,1.55) rectangle (2.00,1.95);
  \draw[bgbox] (2.00,1.55) rectangle (4.05,1.95);
  \draw[bgbox] (4.05,1.55) rectangle (6.15,1.95);
  \draw[bgbox] (6.15,1.55) rectangle (8.15,1.95);
  \draw[bgbox] (8.15,1.55) rectangle (10.02,1.95);
  \draw[bgbox] (10.02,1.55) rectangle (\W,1.95);

  \node[lab] at (1.00,1.75) {$30$};
  \node[lab] at (3.025,1.75) {$1$};
  \node[lab] at (5.10,1.75) {$2$};
  \node[lab] at (7.15,1.75) {$3$};
  \node[lab] at (9.10,1.75) {$4$};
  \node[lab] at (11.06,1.75) {$5$};

  \draw[baseline] (0.00,0.90) -- (\W,0.90);

  \foreach \t in {0.10,2.10,4.10,6.10,8.10,10.10}{
    \draw[post] (\t,0.8)--(\t,1.0);
  }

  \node[lab] at (0.28,0.6) {$30$};
  \node[lab] at (2.20,0.6) {$1$};
  \node[lab] at (4.20,0.6) {$2$};
  \node[lab] at (6.20,0.6) {$2$};
  \node[lab] at (8.20,0.6) {$3$};
  \node[lab] at (10.20,0.6) {$5$};

\end{scope}

\end{tikzpicture}

\caption{Inheritance for days: the foreground civil days (lower row, dawn-to-dawn) inherit the label of the background lunar day
interval (upper row) that contains their start (dawn).}
\label{fig:tithi-to-civil}
\end{figure}
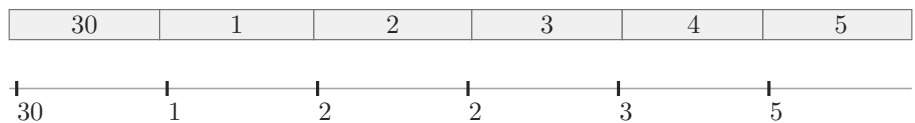

A concrete example is shown in Figure~\ref{fig:day-diag}.  The horizontal axis shows the sequence of Gregorian civil days, while the vertical axis shows the lunar phase, measured in units of lunar days).  The orange line is the continuously varying lunar phase \emph{as computed by the calendar itself}, not necessarily the true astronomical phase.
In the displayed Bhutanese example, the successive civil-day labels are \(12,14,15,15,16\).  Thus lunar day \(13\) is skipped: the orange phase line rises from level \(12\) to level \(14\) between two successive dawns, so no civil day begins while lunar day \(13\) is current.  By contrast, lunar day \(15\) is repeated: after the computed phase reaches level \(15\), it remains between \(15\) and \(16\) across two successive dawns, so two consecutive civil days inherit the same label \(15\).  The sequence \(12,14,15,15,16\) therefore illustrates directly how skipped and repeated dates arise from the interaction between the calendar's computed lunar phase and the dawn-based civil-day rule.

\begin{figure}[ht]
\centering
\includegraphics[width=.9\textwidth]{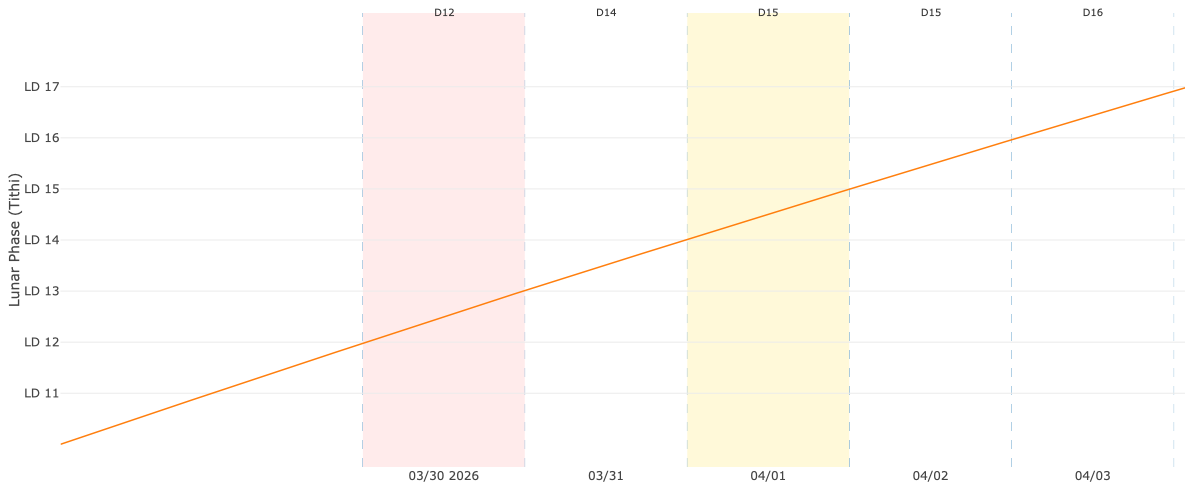}
\caption{Bhutanese repeated and skipped days around 1~April~2026.}
\label{fig:day-diag}
\end{figure}

\subsubsection{Computational algorithm}
\label{sss:day_algo}

To make the dawn-based naming rule operational, one computes the boundary times of the successive
lunar days within a given lunar month and then compares those times with the sequence of dawn instants.
We denote by
\[
t = \trueDate(d,n)
\]
the value of a \emph{local Julian day count} for the instant at which lunar day $d$ ends in the lunar
month with true-month index $n$ (equivalently, the instant at which lunar day $d+1$ begins).  In the
traditional convention \cite[Remark~6]{janson}, integer values of this day count occur at local (mean)
daybreak, 
rather than at noon UT as in the standard astronomical Julian Date.
Here $n$ is the running count of lunar months from the chosen
epoch; given a labeled month $(Y,M)$, the corresponding value of $n$ can be obtained as in Remark~\ref{rem:true-month-index}.

The traditional algorithm approximates this boundary time by starting from a uniform-motion prediction
(mean date) and then applying tabulated corrections that account for the principal
(non-uniform) inequalities in the motions of the Moon and the Sun.

The calculation proceeds as follows:
\begin{enumerate}
\item \emph{Calculate the mean date:} A linear function gives the average end time of a lunar day:
    $$\md = \meanDate(d,n) = m_0 + n \cdot m_1 + d \cdot m_2$$
    Here, $n$ is the true month count from an epoch, and $d$ is the lunar day number (1-30). 
    Moreover $m_2=\frac{11135}{11312}$ is the mean lunar-day length, and $m_1=30m_2$ is the mean synodic month length in JD units.
    The constant $m_0$ is an epoch offset and is tradition- and epoch-dependent; we collect a complete list for the principal traditions in Appendix~\ref{app:constants}. 

\item \emph{Calculate the anomalies.}
The correction terms depend on the orbital phases of the Moon and the Sun, encoded by their mean anomalies, measured in turns (so $1$ is a full revolution) and reduced modulo $1$.
For the Moon we use
\[
A_{\moon}(d,n) = a_0 + n\,a_1 + d\,a_2 ,
\]
where the constants are listed in Appendix~\ref{app:constants}. 
We have $a_2\approx 1/28$, so the lunar anomaly makes about one turn in $28$ lunar days and $a_1\approx30a_2\pmod{1}$.
For the Sun we compute the mean solar longitude and then shift it by a quarter-turn:
\[
\mu(d,n) = s_0 + n\,s_1 + d\,s_2,
\qquad
A_{\sun}(d,n) = \mu(d,n)-\frac14,
\]
again with constants in Appendix~\ref{app:constants}. 
We have $s_1\approx 0.08$, which is about one solar turn in $12.4$ lunar months, and $s_2\approx s_1/30$.

\item \emph{Find the corrections (equations):} The anomalies are used as inputs to lookup tables ($\moonTab$ and $\sunTab$) that approximate sine functions. 
    These tables yield the corrections, known as the equation of the moon ($\eqmoon$) and the equation of the sun ($\eqsun$). 
    The explicit numerical lookup tables (including symmetry/periodicity conventions) are shown in
Figure~\ref{fig:equ-tables} and listed in Appendix~\ref{a-ss:tabs}.

\item \emph{Calculate the true date:} The final end time of the lunar day is:
    $$t = \md + \eqmoon/60 - \eqsun/60$$
    The divisors of 60 are a consequence of the traditional mixed-radix units used.
\end{enumerate}
The integer part of $t$ gives the Julian day number of the calendar day on which the lunar day ends. The sequence of these integer values determines the skipped and repeated days.
Equivalently, the civil-day labels are determined by how the dawn instants interleave with the lunar day boundary times
$\trueDate(d,n)$.

As a concrete application, to locate Losar/Tsagaan~Sar for year $Y$ one first computes the month index $n(Y,1)$  as in Remark~\ref{rem:true-month-index} and Example~\ref{eg:true-month-index}, 
then determines the new-moon boundary ending the preceding month via
$t=\trueDate\!\left(30,\,n(Y,1)-1\right)$, and finally takes the first dawn after this boundary as the first day of month~1.

\begin{example}[Mongol: Tsagaan Sar 2026]
\label{eg:tsagaan-sar-2026}
For $(Y,M)=(2026,1)$, Example~\ref{eg:true-month-index} gives $n=3449$. 
Set $(d,n-1)=(30,3448)$. The mean end time is
\[
\md=\meanDate(30,3448)\approx 2461088.9152,
\]
and the true end time is
\[
t=\trueDate(30,3448)\approx 2461089.2121.
\]
Thus the net anomaly correction is
\[
\Delta t = t-\md\approx 0.2969\ \text{days}\ (\approx 7\mathrm{h}\,7\mathrm{m}),
\]
and in this example it \emph{does} change the integer part:
$\lfloor\md\rfloor=2461088$ but $\lfloor t\rfloor=2461089$.
Therefore
\[
JD(\text{Tsagaan Sar 2026})
=1+\lfloor t\rfloor
=2461090,
\]
which corresponds to \emph{18 February 2026}.
\end{example}

\begin{example}[Phugpa: Losar 2027]
\label{eg:phugpa-losar-2027}
Using the Phugpa parameters $(Y_0,M_0)=(1987,3)$, $\beta^*=0$, trigger set
$T=\{48,49\}$, so $\gamma=65-48=17$, for $(Y,M)=(2027,1)$ we compute
\[
M^*=12(2027-1987)+(1-3)=478,
\qquad
\ix\equiv 2M^*+\beta^*\equiv 2\cdot 478\equiv 46\pmod{65},
\]
so $(2027,1)$ is a non-trigger label and
\[
n
=n_+(2027,1)
=\Big\lfloor \frac{67\cdot 478+0+17}{65}\Big\rfloor
=\Big\lfloor \frac{32026+17}{65}\Big\rfloor
=492.
\]
Hence the preceding lunation end is evaluated at $(d,n-1)=(30,491)$.
The mean end time is
\[
\md=\meanDate(30,491)\approx 2461443.2397,
\]
and the true end time is
\[
t=\trueDate(30,491)\approx 2461443.4053,
\]
with
$\Delta t = t-\md\approx 0.16560\ \text{days}\ (\approx 3\mathrm{h}\,58\mathrm{m})$.
Therefore
\[
JD(\mathrm{Losar\ 2027})=1+\lfloor t\rfloor=2461444,
\]
which corresponds to \emph{7 February 2027}.
\end{example}

\subsubsection{Conceptual basis: Lunar and civil days}
\label{sss:day_theory_comparison}

The algorithm of \S\ref{sss:day_algo} is the celestial model's procedure for computing the
boundary times of the theoretical lunar day.  In conceptual terms, let
$\lambda_{\moon}(t)$ and $\lambda_{\sun}(t)$ denote the model ecliptic longitudes of the Moon and Sun.
A lunar day boundary is defined by the elongation condition
\[
\lambda_{\moon}(t)-\lambda_{\sun}(t)=12^\circ\cdot d \quad (\mathrm{mod}\ 360^\circ),
\]
so $t=\trueDate(d,n)$ is the (modeled) time at which the elongation reaches the next multiple
of $12^\circ$.  In modern ``forward'' astronomy one computes $\lambda_{\moon}(t)$ and $\lambda_{\sun}(t)$
as functions of $t$ and then solves this boundary equation by root-finding; by contrast, the
traditional siddh\=anta scheme is best viewed as an \emph{inverse closed-form approximation} that
maps the discrete label $(d,n)$ directly to an approximate boundary time.
We return to the underlying mean-motion and anomaly models in \S\ref{ss:day_models}.

\medskip
\noindent
The calendrical numbering rule then compares these lunar day boundaries with the sequence of dawn
instants: a civil day (measured from dawn to dawn) is assigned the number of the lunar day that is
current at its beginning \cite{janson}.  The phenomena of skipped and repeated days are the
direct, deterministic consequences of applying this rule to variable-length lunar days.
A \emph{skipped day} (\emph{chad}) occurs when a lunar day begins and ends between two successive dawns,
so no civil day begins while that lunar day is current and its number is omitted from the date
sequence.  Conversely, a \emph{repeated day} (\emph{lhag}) occurs when a lunar day spans two dawns, so two
consecutive civil days begin during the same lunar day and the same date number is used twice.

\medskip
\noindent
This daily structure can be compared with neighboring systems:
\begin{itemize}
\item \emph{Indian system:} The Tibetan day-numbering rule is inherited from Indian calendrical
science: both are based on {\em tithi} (lunar days) and assign the civil day the lunar day number current at sunrise \cite{dershowitz,sewell}.
The main difference lies in the celestial model used to compute lunar day boundaries: modern Indian
practice uses true (astronomical) positions, whereas the traditional Tibetan system applies the
same rule to the output of its mean-motion scheme.

\item \emph{Chinese system:} The Chinese calendar uses an unbroken day count within each lunation:
the day of the new moon is day~1, and subsequent days are numbered $2,3,\dots$ until the next new
moon, yielding months of 29 or 30 days \cite{aslaksen,doggett}.  There are therefore no skipped or repeated day numbers
within a month.  This prioritizes a simple sequential count, at the cost of the direct
phase-based interpretation encoded by the lunar day system.
\end{itemize}

\noindent
Taken together with the month structure in \S\ref{ss:leap_rule}, this illustrates a deliberate
division of labor in the Tibetan calendar: for the micro-cycle of days, where ritual timing is
tied closely to lunar phase, it retains the Indian {\em tithi} (lunar day) convention; for the macro-cycle of
months, it uses solar markers to regulate intercalation and long-term seasonal drift.

\section{Structural Analysis and Celestial Models}
\label{s:analysis}

Section~\ref{s:tib_principles} treats the Tibetan calendar as a black-box algorithm.
In this section we open that box: we extract the small number of structural ingredients
that make the algorithm rigid and computable, and we connect them to minimal celestial-model
assumptions in a way that can be reused for reforms.
The organization mirrors the operational pipeline.
We first isolate an incidence calculus (\S\ref{subsec:dual}) that makes precise the two
label-transfer conventions (inheritance vs.\ containment) and fixes the endpoint bookkeeping.
We then derive the leap-month cycle from a linear mean-Sun model with definition points
(\S\ref{ss:mean_sun_models}), including the intercalation-index shortcut and the recovery of
tradition-dependent phases (\S\ref{ss:defpoints}).
Next we reinterpret the day computation as an explicit inverse for a first-anomaly elongation
model (\S\ref{ss:day_models}), and we use the rational structure of that inverse to analyze
periodicity and exclude tie-breaking on historical time scales (\S\ref{sss:ties}).
Finally, we apply the extracted parameterization to compare traditions and motivate the
geographical-adaptation viewpoint (\S\ref{ss:traditions}), separating what is structurally forced
from what is epoch- and locale-dependent.

\subsection{Containment and inheritance as dual incidence rules}
\label{subsec:dual}

We start by isolating a common abstract structure behind the month and day rules: both rules assign discrete labels
(month names, or lunar day/day labels) by comparing a \emph{foreground} sequence (lunations or civil days) with a
\emph{background} sequence (definition points or boundary times) through an ordered incidence relation.
There are two natural conventions—labels can be transferred from background \emph{points} to foreground \emph{intervals}
(containment), or from background \emph{intervals} to foreground \emph{points} (inheritance)—and the key observation is
that these two conventions are dual to one another once the incidence relation is reversed.

\subsubsection{Definition and examples}

We formulate the day/month naming rules in a common abstract language.
Let $B=\{b_k\}_{k\in\mathbb Z}$ and $F=\{\tau_f\}_{f\in\mathbb Z}$ be strictly
increasing sequences in $\mathbb R$ (``background posts'' and ``foreground posts'').
They induce the background intervals
\begin{equation}
  I_k := (b_{k-1},\,b_k], \qquad k\in\mathbb Z,
  \label{eq:Ik}
\end{equation}
and the foreground intervals
\begin{equation}
  J_f := (\tau_{f-1},\,\tau_f], \qquad f\in\mathbb Z.
  \label{eq:Jf}
\end{equation}
We emphasize that the intended applications are:
\begin{itemize}
\item \emph{Days:} $F$ are sunrises, $B$ are lunar day-end events (or integer elongation crossings).
\item \emph{Months:} $F$ are mean new moons, $B$ are solar-term/sgang boundary crossings.
\end{itemize}

\noindent
Define the two (a priori different) integer-valued incidence functions
\begin{equation}
  N_f := \#\{k:\ b_k\in J_f\},
  \qquad
  \tilde N_k := \#\{f:\ \tau_f\in I_k\}.
  \label{eq:NfNk}
\end{equation}
Thus $N_f$ counts how many background posts occur inside the $f$-th foreground interval,
while $\tilde N_k$ counts how many foreground posts occur inside the $k$-th background interval.

We now define the two dual naming maps.

\begin{definition}\label{def:inherit}
The \emph{inheritance} label of the $f$-th foreground post is the index of the
background interval containing it:
\begin{equation}
  L_{\mathrm{inh}}(f) := k \quad \text{where }\ \tau_f\in I_k.
  \label{eq:Linh}
\end{equation}
Equivalently, $L_{\mathrm{inh}}(f)$ is characterized by $b_{k-1}<\tau_f\le b_k$.
\end{definition}

\begin{definition}\label{def:contain}
The \emph{containment} label(s) of the $f$-th foreground interval are the indices
of the background posts it contains:
\begin{equation}
  L_{\mathrm{con}}(f) := \{\,k:\ b_k\in J_f\,\}.
  \label{eq:Lcon}
\end{equation}
In general $L_{\mathrm{con}}(f)$ can be empty or contain multiple indices.
When a single-valued label is desired (as in concrete calendrical rules), we set
\[
\ell_{\mathrm{con}}(f):=
\begin{cases}
\min L_{\mathrm{con}}(f), & L_{\mathrm{con}}(f)\neq\varnothing,\\
\ell_{\mathrm{con}}(f-1), & L_{\mathrm{con}}(f)=\varnothing,
\end{cases}
\]
with an initial value fixed at the chosen epoch.
\end{definition}

\begin{example}[Tibetan days as inheritance]\label{ex:tibdays}
In the Tibetan day-counting rule, the foreground posts $\tau_f$ are civil-day boundaries (sunrises),
while the background posts $b_k$ are the ends of lunar days, given by
\[
b_k=\trueDate(k),
\]
the Julian Day of the instant at which lunar day $k$ ends (equivalently, lunar day $k+1$ begins).
Thus lunar day $k$ occupies the interval $(b_{k-1},b_k]$, and the civil day
$H_f=[\tau_{f},\tau_{f+1})$
is assigned the number of the lunar day (tithi) that is \emph{current at its beginning},
i.e.\ the unique $k$ such that
\[
\tau_{f}\in (b_{k-1},b_k].
\]
Formally, the label is first assigned to the sunrise post by inheritance,
$L_{\mathrm{inh}}(f)=k$ where $\tau_f\in(b_{k-1},b_k]$, and the civil day $H_f$ is then displayed with this label.
Equivalently, this is the smallest index $k$ for which $b_k\in H_f$; if several lunar-day ends
$b_k,b_{k+1},\dots$ occur during the same civil day, the civil date is the first of these labels.

The incidence count $N_f$ therefore controls the familiar irregularities of Tibetan dates: if $N_f=0$ then no lunar day ends
during the civil day and the date number repeats, whereas if $N_f\ge2$ then two or more lunar days
end during the civil day and one or more date numbers are skipped. The use of right-closed
intervals $(\cdot,\cdot]$ encodes the traditional convention that a boundary event is assigned to
the day in which it \emph{ends}; in particular, since new moon is the end of lunar day~30, the civil day containing the new
moon is labelled ``30'' unless a skip occurs, and the new month begins on the following
civil day.
\end{example}

\begin{example}[Indian months as inheritance]\label{ex:indmonths}
In many Indian lunisolar calendars, a lunar month is named by the zodiacal sign
occupied by the Sun at a distinguished lunar phase (new moon in the \emph{am\=anta}
system, full moon in the \emph{p\=urnim\=anta} system), cf.\ \cite{dershowitz,sewell}.
In the present framework, take the foreground posts $\tau_f$ to be the sequence of
(new or full) moons, and the background posts $b_k$ to be the solar sign boundaries
(integer multiples of $30^\circ$ in solar longitude).  The month name is then given
by inheritance: the $f$-th lunation inherits the index of the solar interval
containing $\tau_f$, i.e.\ by $L_{\mathrm{inh}}(f)$.
The incidence count $N_f$ governs the familiar irregularities of month names:
if $N_f=0$ then two consecutive lunations inherit the same sign-label (an intercalary
or \emph{adhika m\=asa}), while if $N_f\ge 2$ then the inherited label jumps, producing
a skipped month-name (\emph{k\=saya m\=asa}).
\end{example}

\begin{example}[Chinese months as containment]\label{ex:chmonths}
In the traditional Chinese lunisolar calendar, month naming is governed by
containment rather than inheritance \cite{aslaksen,doggett}.  Let the foreground
intervals $J_f$ be lunations (from one astronomical new moon to the next), and let
the background posts $b_k$ be the \emph{principal (major) solar terms}
(\emph{zh\=ongq\`i}), i.e.\ the instants when the Sun's ecliptic longitude reaches
multiples of $30^\circ$.  The containment set $L_{\mathrm{con}}(f)$ records which
principal terms occur during the $f$-th lunation.

In the generic situation, $L_{\mathrm{con}}(f)$ has cardinality $0$ or $1$:
if it is empty, the lunation contains no principal term and is designated as an
intercalary (leap) month, taking the same month number as the preceding month; if
it is nonempty, the month is named by the unique principal term it contains.
In rare ``extreme'' configurations (possible with the \emph{true} Sun), a lunation
may contain \emph{two} principal terms; in that case one must specify an additional
selection convention, and such cases are closely tied to the well-known exceptional
behavior of leap-month placement (e.g.\ the 2033 anomaly discussed by \cite{aslaksen}).
\end{example}

\begin{example}[Tibetan months as containment]\label{ex:tibmonths}
In Tibetan siddh\=anta-style traditions, the conceptual basis of month naming is
also containment, but formulated using the \emph{mean} Sun and a mean-motion
lunation model \cite{janson,henning,schuh}.  Let the foreground intervals $J_f$ be
\emph{mean} lunations (consecutive mean new moons, i.e.\ mean Sun--Moon conjunctions),
so that the lunation sequence itself is generated by the mean motions of \emph{both}
luminaries.  Let the background posts $b_k$ be the twelve Tibetan definition points
(\emph{sgang}), i.e.\ fixed ecliptic longitudes for the mean Sun (typically shifted
relative to the standard sign boundaries, e.g.\ by $8^\circ$ in Phugpa).
The month label is then assigned by containment: a lunation is labeled by the (unique)
definition point crossed by the mean Sun during that lunation; if no definition point
is crossed, the lunation is intercalary.  Because the rule is expressed in mean motions,
the resulting intercalation pattern is perfectly periodic and can be packaged as a
purely arithmetic congruence test (the intercalation index).
\end{example}

\subsubsection{Mathematical properties}

In both inheritance and containment schemes, the behavior of the naming rule is largely
controlled by the incidence counts $N_f$. The following lemma makes this explicit.

\begin{lemma}\label{l:skips}
For each $f\in\mathbb Z$,
\begin{equation}
  L_{\mathrm{inh}}(f)-L_{\mathrm{inh}}(f-1) \;=\; N_f.
  \label{eq:jump=Nf}
\end{equation}
Consequently:
\begin{itemize}
\item $N_f=0$ if and only if inheritance repeats a label at step $f$
      (i.e.\ $L_{\mathrm{inh}}(f)=L_{\mathrm{inh}}(f-1)$);
\item $N_f=1$ if and only if inheritance advances by one;
\item $N_f\ge 2$ if and only if inheritance skips labels (a jump of size $\ge 2$).
\end{itemize}
Moreover, $|L_{\mathrm{con}}(f)|=N_f$, so the same $N_f$ governs multiplicity
($N_f\ge2$) or emptiness ($N_f=0$) in the containment rule.
\end{lemma}

\begin{proof}
Let $k=L_{\mathrm{inh}}(f-1)$, i.e.\ $\tau_{f-1}\in(b_{k-1},b_k]$.
As $\tau$ increases from $\tau_{f-1}$ to $\tau_f$, the inheritance index increases by one
each time a background post $b_\ell$ is crossed.
The number of such crossings is exactly $\#\{\ell:\ b_\ell\in(\tau_{f-1},\tau_f]\}=N_f$,
hence \eqref{eq:jump=Nf}.
The statements about repeats/skips follow immediately, and
$|L_{\mathrm{con}}(f)|=N_f$ is tautological from \eqref{eq:Lcon}.
\end{proof}

\begin{proposition}[Duality]\label{prop:duality}
If one swaps the roles of background and foreground (i.e.\ exchanges $B$ and $F$),
then the inheritance and containment constructions are interchanged.
More precisely, define the swapped sequences
\[
\widetilde b_f := \tau_f,\qquad \widetilde\tau_k := b_k,
\]
and build the swapped intervals $\widetilde I_f=(\widetilde b_{f-1},\widetilde b_f]$ and
$\widetilde J_k=(\widetilde\tau_{k-1},\widetilde\tau_k]$. Then
\[
\widetilde L_{\mathrm{con}}(k)=\{\,f:\ \tau_f\in I_k\,\}
\qquad\text{and}\qquad
\widetilde N_k=|\widetilde L_{\mathrm{con}}(k)|=\tilde N_k,
\]
while the swapped inheritance labels satisfy
\[
\widetilde L_{\mathrm{inh}}(k)=f\quad\text{where}\quad b_k\in J_f,
\]
so that inheritance and containment exchange roles under swapping.
\end{proposition}

\begin{proof}
This is a direct unpacking of Definitions~\ref{def:inherit}--\ref{def:contain} after
interchanging posts and intervals.  The key observation is that the incidence relation
``a post lies in an interval'' is symmetric under swapping the two families of posts.
\end{proof}

\begin{remark}\label{rem:foreground-background}
The distinction between inheritance and containment depends not only on the incidence relation
itself, but also on what is chosen as foreground and background.
In Tibetan day numbering, the foreground posts are sunrises and the background units are lunar days;
the displayed date on a civil day is determined by the inheritance label $L_{\mathrm{inh}}(f)$
(the lunar-day interval containing the sunrise), equivalently by the first lunar-day end that occurs
during that civil day, with repeats/skips governed by the incidence count $N_f$.

For month naming one may instead insist that the labels be solar, since months are meant to track
seasonal divisions, while new moons are merely lunar markers.  Both the Tibetan and Chinese month
rules are then naturally stated as containment: a lunation is regular precisely when it contains a
designated solar marker, and it is intercalary when it contains none.  Proposition~\ref{prop:duality}
amounts to the observation that the same containment data can also be read in the reverse direction:
one may ask, for a given solar marker, which lunation contains it.  This is not a new calendrical
convention but simply the inverse lookup of the same incidence relation, and it clarifies that
``containment'' and ``inheritance'' are two equivalent presentations once one decides which family is
being indexed.
\end{remark}

\subsection{Mean-Sun model and leap-month arithmetic}
\label{ss:mean_sun_models}

We now apply the incidence viewpoint of \S\ref{subsec:dual} to the leap-month rule from \S\ref{ss:leap_rule}.
The traditional \emph{sgang} principle is a containment statement: a lunation is regular precisely when the mean Sun
crosses one of $12$ fixed definition points during that lunation.
Encoding this crossing logic arithmetically leads to a floor-function counter with rational slope, which immediately
explains both periodicity and the familiar modular “intercalation index’’ shortcut.

\subsubsection{The model set up}

Let $t_n$ denote the successive mean new-moon instants of the mean-motion lunation model, and let
$n\in\Z$ index these instants.  Write $\mu(n)$ for the mean solar longitude at $t_n$, measured in
revolutions (so the physical longitude is $\mu(n)\bmod 1\in\R/\Z$).  A \emph{sign} is one twelfth
of a revolution, i.e.\ $1/12$ (equivalently $30^\circ$).

We adopt a boundary convention that will be used throughout this subsection.  Astronomical usage
often regards lunations as left-closed intervals $[t_n,t_{n+1})$ between consecutive new moons, so
the boundary instant $t_n$ belongs to the following lunation.  For the Tibetan month rule, which is
a containment rule for solar definition points, it is more natural to use the right-closed
convention: a boundary crossing is assigned to the lunation in which it ends.  Accordingly, we
treat lunation $n$ as the right-closed interval $(t_n,t_{n+1}]$.
This matches the global convention of right-closed incidence intervals in \S\ref{subsec:dual}.

We assume a linear mean-Sun law on the chosen lift,
\begin{equation}\label{eq:ms_linear}
  \mu(n)=s_0+n s_1 \qquad (s_0\in\R,\ s_1>0),
\end{equation}
with $s_1$ the mean solar advance per lunation and $s_0=\mu(0)$ the mean solar longitude at the
model new moon indexed by $n=0$ (true-month count $0$).  The labeled epoch $(Y_0,M_0)$ used in
\S\ref{ss:leap_rule} may correspond to a different lunation index $n_0$ (cf.\ \cite[Remark~5]{janson}),
but this affects only the choice of origin and not the arithmetic below.  The decisive structural
assumption is that the mean Sun advances by a rational fraction of a sign per lunation:
\begin{equation}\label{eq:uv_slope}
  12s_1=\frac{\ru}{\rv}\qquad\text{in lowest terms, with }0<\ru<\rv.
\end{equation}
For all principal Tibetan traditions one has $12s_1=65/67$, so $(\ru,\rv)=(65,67)$.

Fix an offset $\sgang_0\in\R/\Z$ and define the 12 definition points (sgang) by%
\footnote{Janson \cite{janson} denotes these definition points by $p_M$.}
\begin{equation}\label{eq:defpoints}
  \sgang_M = \sgang_0+\frac{M}{12}\pmod{1},\qquad M\in\{1,2,\dots,12\}.
\end{equation}
The sgang principle is then a containment statement in the right-closed lunation: lunation $n$ is
{\em regular} precisely when the mean Sun crosses some $\sgang_M$ at a time $t\in(t_n,t_{n+1}]$, 
and it is {\em intercalary} (leap) when no such crossing occurs.

To encode this without geometry, define the integer-valued counter
\begin{equation}\label{eq:An_def}
  A_n = \Bigl\lfloor 12\bigl(\mu(n)-\sgang_0\bigr)\Bigr\rfloor\in\Z.
\end{equation}
Intuitively, $A_n$ counts how many definition points have been passed by the left endpoint $t_n$.
With the right-closed convention, a crossing exactly at $t_{n+1}$ is counted with lunation $n$,
while a crossing at $t_n$ is counted with the preceding lunation.

\begin{lemma}
\label{l:mean_sun_cycle}
Assume \eqref{eq:ms_linear}--\eqref{eq:uv_slope} and define $A_n$ by \eqref{eq:An_def}. Then:
\begin{enumerate}[label=\textnormal{(\alph*)}, leftmargin=2.2em, itemsep=2pt]
\item
For every $n$,
\begin{equation}\label{eq:An_jump_01}
  A_{n+1}-A_n\in\{0,1\}.
\end{equation}
Moreover, in the right-closed lunation convention we are using, lunation $n$ is regular iff $A_{n+1}-A_n=1$, and intercalary iff $A_{n+1}-A_n=0$.

\item
Writing
\begin{equation}\label{eq:alpha_def}
  \alpha = 12(s_0-\sgang_0)\in\R,
\end{equation}
one has the explicit form
\begin{equation}\label{eq:An_floor}
  A_n=\Bigl\lfloor \alpha+\frac{n\ru}{\rv}\Bigr\rfloor.
\end{equation}

\item
For every $n$,
\begin{equation}\label{eq:An_v_step}
  A_{n+\rv}-A_n=\ru.
\end{equation}
Equivalently, among the $\rv$ increments $\{A_{n+1}-A_n,\dots,A_{n+\rv}-A_{n+\rv-1}\}$
there are exactly $\ru$ ones and exactly $\rv-\ru$ zeros.
In particular, the sgang rule forces exactly $\ell=\rv-\ru$ leap months per $\rv$ lunations.

\item
Fix an epoch lunation $n_0$ corresponding to the chosen \emph{nominal} epoch $(Y_0,M_0)$ of
\S\ref{ss:leap_rule}, and define
$M^*(n)=A_n-A_{n_0}\in\Z$.  Then $M^*(n)$ is constant across an intercalary lunation and increases
by $1$ across a regular lunation.  Moreover, once we fix the representative of $\sgang_0$ in its
class modulo $1/12$ by requiring that the epoch month labels match those of \S\ref{ss:leap_rule},
the month counter $M^*(n)$ agrees exactly with the “solar-month count from the epoch” used in
\S\ref{ss:leap_rule}.  (Depending on the phase $\beta^*$, the lunation with $M^*=0$ need not be the
one with ``true month count'' $0$ in other conventions; shifting $n_0$ or $\sgang_0\mapsto \sgang_0+k/12$
changes $M^*$ by an additive constant only.)
\end{enumerate}
\end{lemma}

\begin{proof}
(a) Set $x_n=12(\ms(n)-\sgang_0)$.  Then $x_{n+1}-x_n=12s_1=\ru/\rv\in(0,1)$, so the floor
$A_n=\lfloor x_n\rfloor$ can increase only by $0$ or $1$, proving \eqref{eq:An_jump_01}.
A definition-point crossing during lunation $n$ (with the right-closed convention) means that
$x(t)$ hits an integer at some time in the interval $(t_n,t_{n+1}]$, equivalently that there exists
an integer $m$ with
\[
x_n < m \le x_{n+1}.
\]
This holds if and only if $\lfloor x_{n+1}\rfloor-\lfloor x_n\rfloor=1$, i.e.\ $A_{n+1}-A_n=1$.
If no such integer exists, then $A_{n+1}-A_n=0$.  Note that a crossing exactly at the right endpoint
($x_{n+1}\in\Z$) is counted for lunation $n$ by the $\le$ above, whereas a crossing at the left
endpoint ($x_n\in\Z$) is not counted for lunation $n$, matching the right-closed convention.

\noindent
(b) Insert \eqref{eq:ms_linear} into \eqref{eq:An_def} and use \eqref{eq:alpha_def} to obtain
\eqref{eq:An_floor}.

\noindent
(c) From \eqref{eq:An_floor},
\[
A_{n+\rv}-A_n
=\Bigl\lfloor \alpha+(n+\rv)\frac{\ru}{\rv}\Bigr\rfloor-\Bigl\lfloor \alpha+n\frac{\ru}{\rv}\Bigr\rfloor
=\Bigl\lfloor \alpha+n\frac{\ru}{\rv}+\ru\Bigr\rfloor-\Bigl\lfloor \alpha+n\frac{\ru}{\rv}\Bigr\rfloor
=\ru,
\]
which is \eqref{eq:An_v_step}.  Summing the $\rv$ increments (each $0$ or $1$) then forces exactly
$\ru$ ones and $\rv-\ru$ zeros.

\noindent
(d) This is immediate from the congruence
\(
\rv A_n\equiv \rv M^*(n)+\rv A_{n_0}\pmod\ru.
\)
\end{proof}

\begin{example}\label{rem:leap-month-rules-drift}
These results apply to \emph{any} rational choice $12s_1=\ru/\rv$ and therefore
describe a whole family of \emph{possible} intercalation schemes.  This is useful for discussing reforms or
for comparing traditions, but it is important to stress that \emph{all} principal Tibetan systems use the
slope $12s_1=65/67$ (equivalently $s_1=65/804$); they are \emph{not} Metonic.

For a general slope $12s_1=\ru/\rv$, the leap pattern in lunation index is periodic with period $\rv$ and contains
exactly $\rv-\ru$ leap months per period.  The \emph{year-aligned} pattern (month numbers modulo $12$) repeats after
\[
N=\frac{12\rv}{\gcd(\ru,12)}\ \text{lunations},
\qquad\text{i.e. } K=\frac{\ru}{\gcd(\ru,12)}
\ \text{mean-solar years in the model.}
\]
A convenient benchmark for \emph{seasonal drift} relative to the modern tropical year is to compare the time
spanned by $N$ synodic months with the time spanned by $K$ tropical years:
\[
\Delta = N\,(\text{synodic month})-K\,(\text{tropical year}) .
\]
Here the sign indicates whether the calendar drifts late ($\Delta>0$) or early ($\Delta<0$) in season.

Two classical alternatives (not Tibetan) are often mentioned in the broader calendrical literature.
Using modern mean values (tropical year $\approx365.2421897$~d, synodic month $\approx29.53058885$~d) gives the following.
\begin{itemize}
\item \emph{Tibetan rule.}  Here $(\ru,\rv)=(65,67)$, so the leap pattern has period $67$ lunations with
$\rv-\ru=2$ leap months per period.  Since $\gcd(65,12)=1$, the year-aligned pattern closes only after $m=12$ periods,
i.e.\ $N=804$ lunations $=65$ model years.  One finds
\[
\Delta \approx +1.85\ \text{days per }65\ \text{years}
\quad\Longrightarrow\quad
\text{about }+2.85\ \text{days per century}.
\]

\item \emph{Rule $168/163$.}  Here we have $(\ru,\rv)=(163,168)$.
It forces $\rv-\ru=5$ leap months per $\rv=168$ lunations.  Since $\gcd(163,12)=1$, the year-aligned cycle closes only after
$2016$ lunations $=163$ years.  One finds
\[
\Delta \approx -0.810\ \text{days per }163\ \text{years}
\quad\Longrightarrow\quad
\text{about }-0.497\ \text{days per century}.
\]

\item \emph{Metonic cycle.}  This corresponds to $(\ru,\rv)=(228,235)$.  
Thus the leap pattern has period
$\rv=235$ lunations with $\rv-\ru=7$ leap months, and because $\gcd(228,12)=12$ the year-aligned cycle already closes
after $235$ lunations $=19$ years.  One finds
\[
\Delta \approx +0.0868\ \text{days per }19\ \text{years}
\quad\Longrightarrow\quad
\text{about }+0.457\ \text{days per century}.
\]
\end{itemize}
In all cases, the slope $\ru/\rv$ fixes the \emph{number} of leap months per period (rigidity), while the phase
$\alpha=12(s_0-\sgang_0)$ controls their \emph{placement} within the period.
\end{example}

\subsubsection{Intercalation index and inverse month map}

Thus the \emph{sgang} rule is a pure carry/no-carry phenomenon in a floor formula, cf. Lemma~\ref{l:mean_sun_cycle}. 
We next rewrite this carry
test as a single congruence condition involving the intercalation index used operationally in \S\ref{ss:leap_rule}, and then give an explicit inverse map
from the solar-month counter \(M^*\) to the right-end lunation index \(n_+(M^*)\), generalizing Remark~\ref{rem:true-month-index}.

\begin{proposition}\label{p:leap-index}
Assume \eqref{eq:ms_linear}--\eqref{eq:uv_slope}, so \(12s_1=\ru/\rv\) in lowest terms, and define
\(A_n\) and \(M^*(n)=A_n-A_{n_0}\) as in Lemma~\ref{l:mean_sun_cycle}.  Write \(\ell=\rv-\ru\).

\begin{enumerate}[label=\textnormal{(\alph*)}, leftmargin=2.2em, itemsep=2pt]
\item
Let
\(\gamma=\bigl\lfloor\rv\alpha\bigr\rfloor+1\in\mathbb Z\), \(\gamma^* \equiv \ell-\gamma \pmod \ru\),
and
\(T=\{0,1,\dots,\ell-1\}\subset\mathbb Z/\ru\mathbb Z\).
Then lunation \(n\) is intercalary (leap) iff
\begin{equation}\label{eq:ix_uv}
  (\rv A_n+\gamma^*)\bmod \ru \in T .
\end{equation}
Equivalently, with \(\gamma^*_0 \equiv (\rv A_{n_0}+\gamma^*)\pmod \ru\), lunation \(n\) is intercalary iff
\begin{equation}\label{eq:ix_uv_Mstar}
  (\rv M^*(n)+\gamma^*_0)\bmod \ru \in T .
\end{equation}

\item
Let \(\theta\) be the fractional part of the epoch phase,
\[
  \theta=\alpha+\frac{n_0\ru}{\rv}-A_{n_0}\in[0,1),
\]
so that \(M^*(n_0+k)=\lfloor \theta+k\ru/\rv\rfloor\) for \(k\in\mathbb Z\).
For \(M^*\in\mathbb Z\), define the \emph{right-end} (later) lunation index realizing \(M^*\) by
\[
  n_+(M^*)=\max\{\,n\in\mathbb Z:\ M^*(n)=M^*\,\}.
\]
Then
\begin{equation}\label{eq:nplus_general_ceiling_merged}
  n_+(M^*)
  =
  n_0+\Bigl\lceil \frac{\rv(M^*+1-\theta)}{\ru}\Bigr\rceil-1.
\end{equation}

\item
Assume in addition that \(\theta=r/\rv\) for some \(r\in\{0,1,\dots,\rv-1\}\), and define the
\emph{lifted epoch phase constant}
\begin{equation}\label{eq:gamma0hat_def}
  \widehat{\gamma}_0 := \rv-1-r \in \{0,1,\dots,\rv-1\}.
\end{equation}
Then \eqref{eq:nplus_general_ceiling_merged} reduces to the floor form
\begin{equation}\label{eq:nplus_general_floor_merged}
  n_+(M^*) = n_0+\Bigl\lfloor \frac{\rv M^*+\widehat{\gamma}_0}{\ru}\Bigr\rfloor .
\end{equation}
Moreover, the modular phase constant in \eqref{eq:ix_uv_Mstar} satisfies
\begin{equation}\label{eq:gamma0_beta_relation}
  \gamma^*_0 \equiv \widehat{\gamma}_0 \pmod{\ru},
\end{equation}
and we have the inverse formula
\begin{equation}\label{eq:inverse-month}
  M^*(n_0+k)
  =
  \Bigl\lfloor \frac{\ru k-\widehat{\gamma}_0-1}{\rv}\Bigr\rfloor + 1
  =
  \Bigl\lceil \frac{\ru k-\widehat{\gamma}_0}{\rv} \Bigr\rceil .
\end{equation}
\end{enumerate}
\end{proposition}

\begin{proof}
\noindent{(a)}
Let $t=\rv\alpha$ so that \eqref{eq:An_floor} is equivalent to $A_n=\bigl\lfloor (t+n\ru)/\rv\bigr\rfloor$.
Since
\[
  \rv A_n \le t+n\ru < \rv A_n+\rv ,
\]
we see that the integer $E_n=\rv A_n-n\ru$ always lies in the set $\{\gamma-\rv,\dots,\gamma-1\}$,
where we recall $\gamma=\lfloor t\rfloor+1$.  
Writing $t+n\ru=\rv A_n+r_n$ with $r_n=t-E_n\in[0,\rv)$,
we have
\[
A_{n+1}-A_n=\Bigl\lfloor \frac{r_n+\ru}{\rv}\Bigr\rfloor,
\]
so $n$ is intercalary iff $r_n+\ru<\rv$, i.e.\ iff \(r_n<\ell\).
Since \(E_n\in\mathbb Z\), the inequality \(t-E_n<\ell\) is equivalent to
\(
E_n\in\{\gamma-\ell,\gamma-\ell+1,\dots,\gamma-1\}.
\)
Reducing modulo \(\ru\) and using \(E_n\equiv \rv A_n\pmod \ru\) yields \eqref{eq:ix_uv}, and shifting by \(n_0\)
gives \eqref{eq:ix_uv_Mstar}.

\smallskip\noindent{(b)}
Reindex $n=n_0+k$ so that $M^*(n_0+k)=\lfloor \theta+k\ru/\rv\rfloor$.
The condition $M^*(n_0+k)=M^*$ is equivalent to
\[
  M^*\le \theta + \frac{k\ru}{\rv} < M^*+1,
\]
hence
\[
  \frac{\rv(M^*-\theta)}{\ru}\le k < \frac{\rv(M^*+1-\theta)}{\ru}.
\]
The maximal integer $k$ in this interval is
$k=\left\lceil \frac{\rv(M^*+1-\theta)}{\ru}\right\rceil-1$,
which gives \eqref{eq:nplus_general_ceiling_merged}.

\smallskip\noindent{(c)}
If \(\theta=r/\rv\), then \(X=\rv(M^*+1-\theta)=\rv(M^*+1)-r\in\mathbb Z\), hence
\[
\left\lceil \frac{X}{\ru}\right\rceil-1
=
\left\lfloor \frac{X-1}{\ru}\right\rfloor,
\]
which gives
\[
n_+(M^*)=n_0+\Bigl\lfloor \frac{\rv M^*+\rv-1-r}{\ru}\Bigr\rfloor
=n_0+\Bigl\lfloor \frac{\rv M^*+\widehat{\gamma}_0}{\ru}\Bigr\rfloor.
\]
This proves \eqref{eq:nplus_general_floor_merged}.
Since \(\widehat{\gamma}_0=\rv-1-r\), the congruence
\eqref{eq:gamma0_beta_relation} is immediate from
\[
\gamma^*_0 \equiv \rv-1-r \pmod{\ru},
\]
which is exactly the relation obtained from \eqref{eq:ix_uv_Mstar}.
Finally, substituting \(\theta=r/\rv\), equivalently \(\widehat{\gamma}_0=\rv-1-r\), into
\(M^*(n_0+k)=\lfloor \theta+k\ru/\rv\rfloor\) yields \eqref{eq:inverse-month}.
\end{proof}

\begin{remark}
Let us define the intercalation index
\[
  \ix(M^*) \equiv \ell M^*+\gamma^*_0 \equiv \rv M^*+\gamma^*_0 \pmod{\ru}.
\]
Then in light of \eqref{eq:ix_uv_Mstar},
the label \(M^*\) occurs \emph{twice} iff
\[
  \ix(M^*)\in\{0,1,\dots,\ell-1\}.
\]
In that trigger case, the consecutive lunations carrying label \(M^*\) have indices
\[
  n_-(M^*)=n_+(M^*)-1,
  \qquad
  n_+(M^*),
\]
and the convention-dependent choice of which copy is \emph{called} “leap’’ is handled separately, cf.\ Remark~\ref{rem:leap-naming}.
\end{remark}

\begin{example}\label{ex:explicit_rules_Mstar}
Proposition~\ref{p:leap-index} gives the following leap tests for various existing and hypothetical calendars, with the epoch-dependent constants $\gamma^*_0$.
\begin{itemize}
\item Tibetan: $(2M^*+\gamma^*_0)\bmod 65\in\{0,1\}$.
\item Rule $168/163$: $(5M^*+\gamma^*_0)\bmod 163\in\{0,1,2,3,4\}$.
\item Metonic: $(7M^*+\gamma^*_0)\bmod 228\in\{0,1,\dots,6\}$.
\end{itemize}
\end{example}

\begin{example}\label{ex:metonic_phugpa_gamma0}
We illustrate how the phase constant \(\gamma_0^*\) is computed in a concrete hybrid design,
where $(\ru,\rv)=(228,235)$, the definition points are shifted by about \(+8^\circ\) relative to the cardinal \(30^\circ\)-grid,
and the traditional Phugpa epoch phase \(s_0=0\) at E1987, cf. Table~\ref{tab:epoch-constants}.  
To stay in the purely arithmetic regime of Proposition~\ref{p:leap-index}(c), it is convenient to choose
\(\sgang_0\) with denominator \(12\rv\); for instance we take
\[
\sgang_0=\frac{63}{12\rv}=\frac{63}{2820}\ \text{turns}
\qquad\Bigl(\text{i.e. }360^\circ\sgang_0\approx 8.04^\circ\Bigr) .
\]
Then
\[
\alpha=12(s_0-\sgang_0)=-12\cdot\frac{63}{2820}=-\frac{63}{235},
\qquad
\rv\alpha=-63\in\Z,
\]
so we are exactly in the arithmetic phase situation of Proposition~\ref{p:leap-index}(c).
Let the epoch lunation be \(n_0=0\), corresponding to \(M^*=0\) at \((Y_0,M_0)=(1987,3)\).
Then
\[
A_{n_0}=A_0=\Bigl\lfloor \alpha\Bigr\rfloor=\Bigl\lfloor-\frac{63}{235}\Bigr\rfloor=-1.
\]
From Proposition~\ref{p:leap-index}(a), \(\gamma=\lfloor \rv\alpha\rfloor+1=-62\), hence
\[
\gamma^*\equiv \ell-\gamma \equiv 7-(-62)\equiv 69 \pmod{228},
\]
and therefore
\[
\gamma_0^*\equiv (\rv A_{n_0}+\gamma^*) \equiv 235(-1)+69 \equiv -235+69 \equiv 62 \pmod{228}.
\]
Thus the explicit Phugpa-like Metonic leap test is
\[
(7M^*+62)\bmod 228 \in \{0,1,\dots,6\}.
\]
Furthermore, we have
\[
n_+(M^*)
=\Bigl\lfloor\frac{235M^*+62}{228}\Bigr\rfloor
=M^* + \Bigl\lfloor\frac{7M^*+62}{228}\Bigr\rfloor,
\]
and in the trigger case the earlier copy is \(n_-(M^*)=n_+(M^*)-1\).
\end{example}

\subsubsection{Definition points of the principal traditions}
\label{ss:defpoints}

Having reduced the sgang intercalation rule to a congruence test, we can also run the logic
backward.  For the Tibetan slope $12s_1=65/67$, the leap-test constants $(\beta^*,\{\tau,\tau+1\})$
determine a residue class $\gamma$ and hence constrain the definition-point phase $\sgang_0$ to a
short interval of length $1/804$, taken modulo $1/12$.  The interval is naturally expressed relative
to the mean Sun evaluated at the chosen epoch lunation $n_0$, i.e.\ in terms of $\mu(n_0)$; this
makes the reconstruction insensitive to ``epoch bookkeeping'' choices in which the published epoch
month label does not coincide with the lunation indexed by $n=0$.

\begin{proposition}
\label{prop:sgang-restore}
Assume the Tibetan slope \(12s_1=65/67\), fix an epoch lunation \(n_0\), and set \(M^*(n)=A_n-A_{n_0}\) as in
Lemma~\ref{l:mean_sun_cycle}.  Suppose a tradition specifies its leap lunations by the test
\begin{equation}\label{eq:trad_rule}
  (2M^*(n)+\beta^*)\bmod 65 \in \{\tau,\tau+1\},
\end{equation}
with \(\beta^*\in\Z\) and a chosen consecutive target pair \(\{\tau,\tau+1\}\subset \Z/65\Z\).
Let \(\gamma\in\{1,2,\dots,65\}\) be the unique representative satisfying
\begin{equation}\label{eq:gamma_from_beta}
  \gamma \equiv 2+\tau-\beta^* \pmod{65}.
\end{equation}
Then every \(\sgang_0\) producing the leap rule \eqref{eq:trad_rule} satisfies
\begin{equation}\label{eq:p0_interval}
\sgang_0 \in \Bigl(\mu(n_0)-\frac{\gamma}{804},\ \mu(n_0)-\frac{\gamma-1}{804}\Bigr]
\pmod{\frac1{12}}.
\end{equation}
Conversely, any \(\sgang_0\) satisfying \eqref{eq:p0_interval} yields \eqref{eq:trad_rule}.
\end{proposition}

\begin{proof}
Write \(\alpha=12(s_0-\sgang_0)\), so \(\alpha\) is only defined modulo \(1\) (equivalently, \(\sgang_0\) is only
defined modulo \(1/12\)).
For \((\ru,\rv)=(65,67)\), Proposition~\ref{p:leap-index}\,(a) gives a leap test in the form
\begin{equation}\label{eq:leap_test_M}
  (67M^*(n)+\gamma_0^*)\bmod 65 \in \{0,1\},
\end{equation}
with
\[
  \gamma=\lfloor 67\alpha\rfloor+1,
  \qquad
  \gamma_0^*\equiv (67A_{n_0}+2-\gamma) \pmod{65}.
\]
Note that shifting \(\alpha\mapsto \alpha+k\) by an integer \(k\) (equivalently,
\(\sgang_0\mapsto \sgang_0-k/12\)) does not change the calendar:
it adds \(k\) to every \(A_n\), hence leaves \(M^*(n)=A_n-A_{n_0}\) unchanged.  Under this shift,
\(\gamma\) changes by \(67k\) and \(A_{n_0}\) changes by \(k\), so the combination
\(\gamma_0^*\equiv 67A_{n_0}+2-\gamma\) in \eqref{eq:leap_test_M} is invariant.  Therefore we may choose the
representative of \(\alpha\) modulo \(1\) so that \(A_{n_0}=0\).  
In this normalization we have
\[
  0=A_{n_0}=\Bigl\lfloor 12(\mu(n_0)-\sgang_0)\Bigr\rfloor
  =\lfloor \alpha_{0}\rfloor,
  \qquad
  \alpha_{0}:=12(\mu(n_0)-\sgang_0)=\alpha+n_0\frac{65}{67},
\]
so $\alpha_{0}\in[0,1)$.
Moreover \(\gamma_0^*\equiv2-\gamma\),
and the leap test becomes
\[
  (2M^*(n)+2-\gamma)\bmod 65 \in \{0,1\} ,
\]
where we have also taken into account \(67\equiv 2\pmod{65}\).
Comparing this with the normalized tradition rule \((2M^*(n)+\beta^*-\tau)\bmod 65\in\{0,1\}\) gives
\(\beta^*-\tau\equiv 2-\gamma\pmod{65}\), hence \(\gamma\equiv 2+\tau-\beta^*\pmod{65}\).
With $\gamma\in\{1,\dots,65\}$ chosen as in \eqref{eq:gamma_from_beta}, the relation
$\gamma=\lfloor 67\alpha_{0}\rfloor+1$ is equivalent to
\[
  \alpha_{0}\in\Bigl[\frac{\gamma-1}{67},\frac{\gamma}{67}\Bigr).
\]
Finally, $\alpha_{0}=12(\mu(n_0)-\sgang_0)$ gives \eqref{eq:p0_interval}.
\end{proof}

\begin{remark}
\label{r:sgang-trad}
Working modulo $1/12$, Proposition~\ref{prop:sgang-restore} yields the following allowed intervals
(of length $1/804$) for the phase parameter $\sgang_0$ of the principal traditions, using the epoch
data in Appendix~\ref{app:constants}.  We record the corresponding epoch lunation index $n_0$, since
the reconstruction is expressed in terms of $\mu(n_0)$.

\begin{itemize}
\item \emph{Bhutan:} Here $n_0=0$ and $\mu(n_0)\equiv s_0\equiv \frac{1}{67}$.
The leap test uses $\beta^*=2$ with target $\{57,58\}$, giving
\[
  \sgang_0 \in \Bigl(\frac{22}{804},\frac{23}{804}\Bigr] \pmod{\frac1{12}}
  \qquad (\text{i.e. } 9.85^\circ \text{ to } 10.29^\circ).
\]
Janson \cite[C.6]{janson} states that the value $\sgang_0=10^\circ$ is in fact given in the original text by Lhawang Lodro, 
even though Janson's own calculation somehow yields $\sgang_0=9^\circ$.

\item \emph{Mongol:} Here $n_0=0$ and $\mu(n_0)\equiv s_0\equiv \frac{397}{402}$.
The leap test uses $\beta^*=10$ with target $\{46,47\}$, giving
\[
  \sgang_0 \in \Bigl(\frac{19}{804},\frac{20}{804}\Bigr] \pmod{\frac1{12}}
  \qquad (\text{i.e. } 8.51^\circ \text{ to } 8.96^\circ).
\]
Janson \cite{janson} mentions the value $\sgang_0=8\frac23^\circ$ as a possible choice.

\item \emph{Tsurphu:} Table~\ref{tab:epoch-constants} lists two common epoch choices, and in both cases
we have $n_0=0$ (so $\mu(n_0)\equiv s_0$).
For E1732 ($s_0\equiv -\frac{5983}{108540}$, $\beta^*=59$, target $\{0,1\}$) one obtains
\[
  \sgang_0 \in \Bigl(\frac{991}{54270},\frac{2117}{108540}\Bigr] \pmod{\frac1{12}}
  \qquad (\text{i.e. } 6.57^\circ \text{ to } 7.02^\circ).
\]
For E1852 ($s_0\equiv \frac{23}{27135}$, $\beta^*=14$, target $\{0,1\}$) one obtains the same
interval modulo $1/12$, as expected since these epoch choices give identical leap-month placements.
According to \cite{janson}, the value $\sgang_0=7^\circ$ was found by Henning as a good candidate.

\item \emph{Phugpa:} Two commonly used epochs illustrate the general formula and the ``epoch bookkeeping''
subtlety.  For E1987 we have $n_0=0$ and $\mu(n_0)\equiv s_0\equiv 0$; the leap test uses $\beta^*=0$
with target $\{48,49\}$, giving
\[
  \sgang_0 \in \Bigl(\frac{17}{804},\frac{18}{804}\Bigr] \pmod{\frac1{12}}
  \qquad (\text{i.e. } 7.61^\circ \text{ to } 8.06^\circ).
\]
This is consistent with the value $\sgang_0=8^\circ$ given in \cite{janson}.
For E1927 the published epoch month is $M_0=3$, but the corresponding epoch lunation index is $n_0=1$.
Appendix~\ref{app:constants} gives $s_0\equiv \frac{749}{804}$ at $n=0$, hence with $s_1=\frac{65}{804}$ we have
\[
\mu(n_0)=\mu(1)\equiv s_0+s_1 \equiv \frac{10}{804}\pmod{1}.
\]
The leap test uses $\beta^*=55$ with the same target $\{48,49\}$ (so $\tau=48$), hence
$\gamma\equiv 2+\tau-\beta^*\equiv 60\pmod{65}$ and
\[
\sgang_0 \in \Bigl(\mu(n_0)-\frac{60}{804},\ \mu(n_0)-\frac{59}{804}\Bigr]
=
\Bigl(-\frac{50}{804},-\frac{49}{804}\Bigr]
=
\Bigl(\frac{17}{804},\frac{18}{804}\Bigr]\pmod{\frac1{12}}.
\]
Thus the recovered Phugpa interval agrees with the E1987 computation, as it must.
\end{itemize}
\end{remark}

\begin{remark}\label{r:d0-s0}
While Proposition~\ref{prop:sgang-restore} establishes $d_0$ as the theoretical control for seasonal alignment, historical practice suggests that calendar designers often pragmatically tuned the epoch mean sun $s_0$ directly. This adjustment likely served to counteract accumulated precessional drift or to align the civil year with specific observational norms.

We can quantify this tuning by comparing the actual epoch values $s_0$ (Table~\ref{tab:epoch-constants-806}) against the theoretical target intervals derived in Remark~\ref{r:sgang-trad}. Let $\delta \approx \text{mid}(d_0) - s_0 \pmod{1/12}$ represent the ``phase lag'' of the chosen epoch sun relative to the ideal target. Converting this lag into days ($365.25 \times \delta$), we obtain:
\begin{itemize}
    \item {Tsurphu:} $\delta \approx 0.0006 \implies \approx 0.2$ days delay.
    \item {Mongol:} $\delta \approx 0.0007 \implies \approx 0.3$ days delay.
    \item {Bhutanese:} $\delta \approx 0.0081 \implies \approx 3.0$ days delay.
    \item {Phugpa:} $\delta \approx 0.0168 \implies \approx 6.1$ days delay.
\end{itemize}
A near-zero $\delta$ (as seen in Tsurphu and Mongol) implies the epoch sun was set to the earliest permissible phase, resulting in the earliest possible Losar dates. A larger positive $\delta$ effectively delays the solar phase relative to the calendar, causing the New Year to fall later in the season.

This parameter tuning correlates perfectly with the macroscopic behavior observed in Figure~\ref{fig:losar}. 
The Tsurphu and Mongol traditions, having virtually identical and minimal lags ($\approx 0.2$--$0.3$ days), cluster together in the earliest seasonal band. The Bhutanese tradition occupies a distinct band shifted approximately 3 days later, matching the calculated difference. Finally, the Phugpa tradition appears in the latest band, consistent with its substantial initial design lag of $\approx 6$ days relative to the Tsurphu baseline.
\end{remark}

\subsection{First-anomaly models and the inverse day map}
\label{ss:day_models}

We now turn to the second operational rule from Section~\ref{s:tib_principles}, namely the day-counting map of
\S\ref{ss:day_calculation}: given a discrete label (lunation index $n$ and lunar day index $d$), one outputs an
approximate boundary time via a linear formula plus two table lookups.
In this section we explain why this is a natural design from the
viewpoint of a first-anomaly celestial model.

Throughout we work in \emph{turns} (revolutions) and treat longitudes as real-valued phases.
Reducing a phase modulo $1$ recovers the corresponding angle in $\R/\Z$ (with $1$~turn $=360^\circ$).
The main point is conceptual: in a modern ``forward'' model, lunar day boundaries are defined implicitly by an
elongation equation, whereas the siddh\=anta-style computation is built as an \emph{inverse map} that
approximates the solution time directly from the discrete label, essentially via one first-order correction step (which could be called a Picard step, cf. Appendix~\ref{app:picard}).
We then compare the Tibetan constants and tables with the corresponding modern mean quantities.

\subsubsection{Modern first-anomaly models}\label{sss:modern_first_anomaly}

Let $\lambda_{\moon}(t)$ and $\lambda_{\sun}(t)$ denote (real-valued) geocentric ecliptic longitudes at
time $t$, measured in days.  The corresponding physical angles are $\lambda_{\moon}(t)\bmod 1$ and
$\lambda_{\sun}(t)\bmod 1$, but we keep the real phases as primary quantities.
Define the (real) elongation phase
\[
  E(t)=\lambda_{\moon}(t)-\lambda_{\sun}(t)\in\R.
\]
A \emph{lunar day boundary} is then specified by an integer boundary index $k\in\Z$ via
\begin{equation}\label{eq:tithi_root_turns}
  E(t)=\frac{k}{30},
  \qquad k\in\Z,
\end{equation}
since one lunar day corresponds to $12^\circ=1/30$~turn of elongation.
Reducing \eqref{eq:tithi_root_turns} modulo $1$ recovers the usual congruence
$E(t)\equiv d/30\pmod 1$ with $d\equiv k\pmod{30}$.
In a modern forward computation one specifies the functions $\lambda_{\moon}(t)$ and $\lambda_{\sun}(t)$, and then
solves \eqref{eq:tithi_root_turns} for each $k$ (equivalently, for each $d$ together with the chosen branch).

A minimal ``first-anomaly'' (single-inequality) model writes each longitude as a uniform mean motion plus
a small periodic correction (equation of center):
\begin{equation}\label{eq:first_anom_forward}
  \lambda_i(t)=\bar\lambda_i(t)+\varepsilon_i\sin\bigl(2\pi A_i(t)\bigr)+O(\varepsilon_i^2),
  \qquad i\in\{\moon,\sun\}.
\end{equation}
Here $\bar\lambda_i(t)=\lambda_{i,0}+\omega_i t$ is the linear (mean-motion) part and
$A_i(t)=A_{i,0}+\Omega_i t$ is a linear anomaly phase, with $|\varepsilon_i|\ll 1$.
We regard $\lambda_{i,0},A_{i,0},\varepsilon_i$ as angles and $\omega_i,\Omega_i$ as angular velocities.
In the formulas we use turns, so $\sin(2\pi A_i(t))$ is understood with $A_i(t)$ measured in turns.
For readability, Table~\ref{tab:j2000-firstanom} lists degree representatives at the epoch J2000.0,
with $t$ measured in days from $\mathrm{JD}=2451545.0$ (Terrestrial Time, TT). 
Note that $\sin(2\pi A_i)$ should be understood as $\sin((\pi/180)A_i)$ when $A_i$ is measured in degrees.
The amplitude $\varepsilon_i$ is the dominant (single-harmonic) equation-of-center magnitude in this
first-anomaly truncation.

\begin{table}[ht]
\begin{center}
\renewcommand{\arraystretch}{1.15}
\begin{tabular}{lrrrrr}
\hline
$i$ &
$\bar\lambda_{i,0}$ (deg) & $\omega_i$ (deg/day) &
$A_{i,0}$ (deg) & $\Omega_i$ (deg/day) &
$\varepsilon_i$ (deg)\\
\hline
$\sun$ &
$280.46645$ & $0.9856473602$ &
$357.5291092$ & $0.9856002800$ &
$\approx 1.915$\\
$\moon$ &
$218.3164477$ & $13.1763965268$ &
$134.9633964$ & $13.0649929509$ &
$\approx 6.29$\\
\hline
\end{tabular}
\caption{J2000.0 degree representatives for the parameters in \eqref{eq:first_anom_forward}.}
\label{tab:j2000-firstanom}
\end{center}
\end{table}

With these forward first-anomaly models in hand, we now derive a corresponding \emph{inverse}
approximation for lunar day-boundary times.  For conceptual clarity we allow a real boundary level
$x\in\R$ and define the boundary-time map $t=t(x)$ implicitly by
\begin{equation}\label{eq:tithi_root_lift}
  E(t)=x,\qquad x\in\R .
\end{equation}
For the calendar one takes $x=k/30$ with $k\in\Z$; writing $d\equiv k\pmod{30}$ recovers the usual
lunar day index, and the additional label (lunation index $n$) selects the intended \emph{branch} of the
inverse map.

Under \eqref{eq:first_anom_forward} the elongation has the split
\begin{equation}\label{eq:E_split}
  E(t)=\bar E(t)+\varepsilon_{\moon}\sin\bigl(2\pi A_{\moon}(t)\bigr)
        -\varepsilon_{\sun}\sin\bigl(2\pi A_{\sun}(t)\bigr) + O(\varepsilon^2),
\end{equation}
where $\bar E(t)=\bar\lambda_{\moon}(t)-\bar\lambda_{\sun}(t)$ is linear:
\begin{equation}\label{eq:mean-elong}
  \bar E(t)=\bar E(0) + \omega t,
  \qquad \omega = \omega_{\moon}-\omega_{\sun}>0.
\end{equation}
Notice that $\omega=\omega_{\moon}-\omega_{\sun}$ is the \emph{mean elongation rate}.
In particular $1/\omega$ is the mean synodic month.
Thus the \emph{mean} inverse is explicit: for each $x\in\R$ set
\begin{equation}\label{eq:mean_root}
  t_0(x)=\frac{x-\bar E(0)}{\omega}
  \qquad\Longleftrightarrow\qquad
  \bar E\bigl(t_0(x)\bigr)=x .
\end{equation}
We now correct $t_0(x)$ by one first-order step (cf. Appendix~\ref{app:picard}).

\begin{lemma}\label{lem:one_step_inverse}
Assume \eqref{eq:E_split} with $|\varepsilon_{\moon}|,|\varepsilon_{\sun}|\ll 1$, and suppose that $E:\R\to\R$ is strictly increasing and hence invertible.
Then
\begin{equation}\label{eq:one_step_time}
  t(x)=t_0(x)
  -\frac{\varepsilon_{\moon}}{\omega} \sin\bigl(2\pi A_{\moon}(t_0(x))\bigr)
        +\frac{\varepsilon_{\sun}}{\omega} \sin\bigl(2\pi A_{\sun}(t_0(x))\bigr)        
  +O(\varepsilon^2),
\end{equation}
where $\varepsilon=\max\{|\varepsilon_{\moon}|,|\varepsilon_{\sun}|\}$.
\end{lemma}

\begin{proof}
Write $t=t_0(x)+\delta$.  Since $\bar E(t_0(x))=x$, the boundary equation $E(t)=x$ becomes
\[
0=E(t_0+\delta)-x=\omega\delta
+\varepsilon_{\moon}\sin\bigl(2\pi A_{\moon}(t_0+\delta)\bigr)
-\varepsilon_{\sun}\sin\bigl(2\pi A_{\sun}(t_0+\delta)\bigr)
+O(\varepsilon^2) .
\]
Because $t(x)=E^{-1}(x)$ is a smooth perturbation of the mean inverse, one has $\delta=O(\varepsilon)$.
Since each $A_i(t)$ is linear, $A_i(t_0+\delta)=A_i(t_0)+O(\delta)$, and therefore
\[
\sin\bigl(2\pi A_i(t_0+\delta)\bigr)=\sin\bigl(2\pi A_i(t_0)\bigr)+O(\delta)
=\sin\bigl(2\pi A_i(t_0)\bigr)+O(\varepsilon).
\]
Substituting this into the previous display gives
\[
\omega\delta
+\varepsilon_{\moon}\sin\bigl(2\pi A_{\moon}(t_0)\bigr)
-\varepsilon_{\sun}\sin\bigl(2\pi A_{\sun}(t_0)\bigr)
=O(\varepsilon^2),
\]
hence \eqref{eq:one_step_time}.
\end{proof}

\begin{remark}\label{rem:monotone_and_inverse}
For the minimal first-anomaly model \eqref{eq:first_anom_forward}, with the J2000.0 parameters from Table~\ref{tab:j2000-firstanom} yields the uniform estimate
\[
  E'(t)\ \ge\ 12.19075 - 1.46723 \ \approx\ 10.72\qquad\text{deg/day},
\]
hence $E$ is strictly increasing on $\R$ with a wide margin and therefore globally invertible.

Moreover, inserting the explicit mean inverse \eqref{eq:mean_root} into \eqref{eq:one_step_time}
and rewriting $A_i(t_0(x))$ as an affine function of $x$ yields the following convenient closed
inverse template:
\begin{equation}\label{eq:explicit_inverse_constants}
  t(x)\approx \widehat m_0+\widehat m_1 x
   -\widehat b_{\moon}\sin\bigl(2\pi(\widehat a_0+\widehat a_1 x)\bigr)
   +\widehat b_{\sun}\sin\bigl(2\pi(\widehat\san_0+\widehat\san_1 x)\bigr),
\end{equation}
where we recall $\omega=\omega_{\moon}-\omega_{\sun}>0$, and
\[
  \widehat m_1=\frac{1}{\omega},
  \qquad
  \widehat m_0=-\frac{\bar E(0)}{\omega},
  \qquad
  \widehat b_i=\frac{\varepsilon_i}{\omega},
  \qquad
  \widehat a_1=\frac{\Omega_{\moon}}{\omega},
  \qquad
  \widehat \san_1=\frac{\Omega_{\sun}}{\omega},
\]
\[
  \widehat a_0=A_{\moon,0}-\widehat a_1\,\bar E(0),
  \qquad
  \widehat \san_0=A_{\sun,0}-\widehat s_1\,\bar E(0).
\]
Here $\widehat a_0,\widehat a_1,\widehat \san_0,\widehat \san_1$ are understood as phases in turns (i.e.\ modulo $1$).
Table~\ref{tab:j2000-inverse-constants} lists numerical values for these derived parameters at J2000.0.
\end{remark}

\begin{table}[ht]
\begin{center}
\renewcommand{\arraystretch}{1.15}
\begin{tabular}{lrl}
\hline
constant & value & unit/comment \\
\hline
$\omega=\omega_{\moon}-\omega_{\sun}$ & $12.1907491666$ & deg/day \\
$\widehat m_1=1/\omega$ & $29.5305887304$ & days/turn \\
$\widehat m_0=-\bar E(0)/\omega$ & $5.0981282153$ & days \\
$\widehat b_{\moon}=\varepsilon_{\moon}/\omega$ & $0.5159650087$ & days \\
$\widehat b_{\sun}=\varepsilon_{\sun}/\omega$ & $0.1570863262$ & days \\
$\widehat a_1=\Omega_{\moon}/\omega$ & $1.0717137044$ & dimensionless \\
$\widehat \san_1=\Omega_{\sun}/\omega$ & $0.0808482126$ & dimensionless \\
$\widehat a_0$ & $201.5704056^\circ\ (0.5599177933)$ & deg (turns) \\
$\widehat \san_0$ & $2.5538258^\circ\ (0.0070939605)$ & deg (turns) \\
\hline
\end{tabular}
\caption{Derived constants in the explicit inverse approximation \eqref{eq:explicit_inverse_constants},
computed from Table~\ref{tab:j2000-firstanom} (J2000.0, $\mathrm{JD}=2451545.0$ TT).}
\label{tab:j2000-inverse-constants}
\end{center}
\end{table}

The inverse model \eqref{eq:explicit_inverse_constants} is the structural template behind the Tibetan day algorithm:
\emph{(i)} compute an explicit mean solution $t_0$ from $(d,n)$, \emph{(ii)} evaluate two periodic correction terms
at linear phases, and \emph{(iii)} add them with fixed coefficients.

\subsubsection{The Tibetan scheme as an explicit inverse}\label{sss:tibetan_inverse_scheme}

In the Tibetan computation the input is a lunation index $n$ and a lunar day index $d\in\{0,1,\dots,30\}$,
and the output is an approximate boundary time in days.  The mean solution is a linear map
\begin{equation}\label{eq:mean_date_model_new}
  t_0=\texttt{mean\_date}(d,n)=m_0+n m_1+d m_2,
\end{equation}
with $m_1$ the mean synodic month length (in days) and $m_2=m_1/30$ the mean lunar day length.
The anomalies (linear phases) are likewise computed linearly from $(d,n)$:
\begin{equation}\label{eq:anom_linear_new}
  A_{\moon}(d,n)=a_0+n a_1+d a_2,\qquad
  A_{\sun}(d,n)= -\textstyle\frac14 + s_0 + n s_1 + d s_2,
\end{equation}
where $A_{\moon},A_{\sun}$ are taken modulo $1$ (turns), and $s_2=s_1/30$.
The correction terms are not literal trigonometric functions but \emph{odd periodic lookup tables} with sine-like symmetries,
evaluated by linear interpolation between integer arguments (Figure~\ref{fig:equ-tables}):
\begin{equation}\label{eq:tables_as_functions_new}
\begin{split}
  \texttt{moon\_equ}(d,n)&=\texttt{moon\_tab}\bigl(28 A_{\moon}(d,n)\bigr),\\
  \texttt{sun\_equ}(d,n)&=\texttt{sun\_tab}\bigl(12 A_{\sun}(d,n)\bigr).
\end{split}
\end{equation}
Finally the ``true'' boundary time is defined by
\begin{equation}\label{eq:true_date_model_new}
  t_1=\texttt{true\_date}(d,n)
  = t_0+\frac{1}{60}\,\texttt{moon\_equ}(d,n)-\frac{1}{60}\,\texttt{sun\_equ}(d,n).
\end{equation}

\begin{figure}[ht]
\centering

\begin{minipage}{0.48\textwidth}
\centering
\begin{tikzpicture}
\begin{axis}[
  width=\textwidth,
  grid=both,
  ymin=-26, ymax=26,
  xmin=0, xmax=28,
]
\addplot[
  domain=0:28,
  samples=400,
  dotted,
  thick,
  color=black!35,
  mark=none
] {25*sin(deg(2*pi*x/28))};

\addplot[mark=none] coordinates {
(0,0) (1,5) (2,10) (3,15) (4,19) (5,22) (6,24) (7,25)
(8,24) (9,22) (10,19) (11,15) (12,10) (13,5) (14,0)
(15,-5) (16,-10) (17,-15) (18,-19) (19,-22) (20,-24) (21,-25)
(22,-24) (23,-22) (24,-19) (25,-15) (26,-10) (27,-5) (28,0)
};
\end{axis}
\end{tikzpicture}
\end{minipage}\hfill
\begin{minipage}{0.48\textwidth}
\centering
\begin{tikzpicture}
\begin{axis}[
  width=\textwidth,
  grid=both,
  ymin=-12, ymax=12,
  xmin=0, xmax=12,
]
\addplot[
  domain=0:12,
  samples=400,
  dotted,
  thick,
  color=black!35,
  mark=none
] {11*sin(deg(2*pi*x/12))};

\addplot[mark=none] coordinates {
(0,0) (1,6) (2,10) (3,11) (4,10) (5,6) (6,0)
(7,-6) (8,-10) (9,-11) (10,-10) (11,-6) (12,0)
};
\end{axis}
\end{tikzpicture}
\end{minipage}

\caption{Lookup-table corrections in \eqref{eq:tables_as_functions_new}.
Left: lunar table $\moonTab(u)$ with $u=28A_{\rm moon}$; right: solar table $\sunTab(v)$ with $v=12A_{\rm sun}$.
Solid curves show the piecewise-linear interpolation used in the traditional computation; dotted curves show the reference sines
$25\sin(2\pi u/28)$ and $11\sin(2\pi v/12)$ for comparison.
The maximum deviation of the piecewise-linear surrogate from the corresponding sine is about $0.85$ (3.4\%) for the lunar table and
$0.50$ (4.5\%) for the solar table.}
\label{fig:equ-tables}
\end{figure}
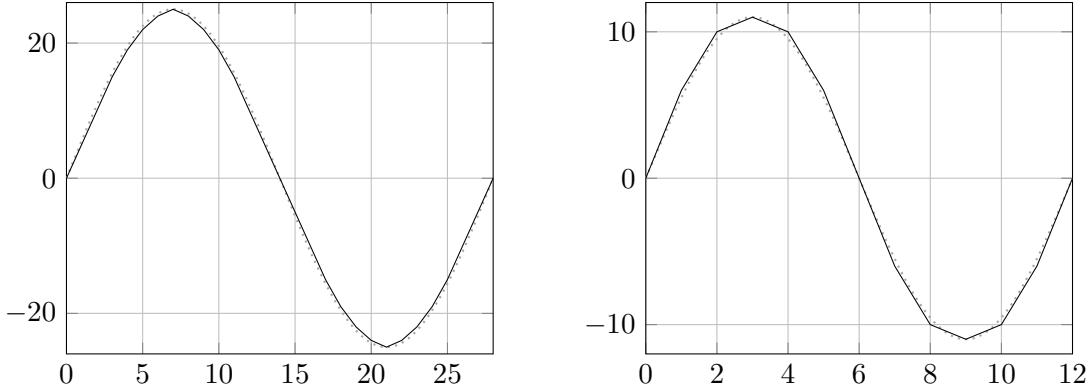

Comparing \eqref{eq:true_date_model_new} with the one-step inverse \eqref{eq:one_step_time}, we may read the
Tibetan design as follows.
\begin{itemize}
\item The linear map \eqref{eq:mean_date_model_new} provides an explicit \emph{mean} boundary time $t_0$ (a closed-form initial guess).
\item The phases \eqref{eq:anom_linear_new} provide the anomaly phases $A_{\rm moon}(t_0)$ and $A_{\rm sun}(t_0)$ (taken modulo one turn).
\item The tables \eqref{eq:tables_as_functions_new} supply explicit piecewise-linear surrogates for the sine corrections in \eqref{eq:one_step_time},
via linear interpolation between integer arguments.
\item The fixed coefficient $1/60$ converts table-units to the time correction in the traditional mixed-radix convention
(equivalently, it encodes an effective inverse scale factor together with the chosen units).
\end{itemize}
Thus the Tibetan day computation is naturally an \emph{inverse approximation} (implemented by a single Picard step)
to the lunar day boundary equation \eqref{eq:tithi_root_turns}:
it maps the discrete label $(d,n)$ directly to an approximate boundary time.

\begin{remark}\label{rem:why_28_12}
The factors $28$ and $12$ in the table evaluations (e.g.\ $u=28A_{\moon}$ and $v=12A_{\sun}$)
should not be read as ``periods in days''.  They are \emph{scaling choices} for the anomaly phases,
chosen so that the periodic ``equations'' can be tabulated on a convenient integer grid and then
evaluated by linear interpolation.

For the Moon, write the (first) anomaly phase in the standard linear form
\[
A_{\moon}(d,n)=a_0+n a_1+d a_2,
\qquad u(d,n)=28A_{\moon}(d,n).
\]
If $a_2=1/28$, then $u(d+1,n)-u(d,n)=1$, so as $d$ advances the table argument is sampled at
\emph{unit-spaced} points.  The initial value $u(0,n)$ need not be an integer, so one does not
``hit the knots'' of the table in general; rather, the fractional part of $u(d,n)$ is constant
within a fixed lunation, and each successive day evaluates the piecewise-linear table at the
same relative position inside the next unit interval.

Across a lunation boundary, however, the indices change by $(d,n)\mapsto(d-30,n+1)$, so the net
increment of the anomaly phase is $\Delta A_{\moon}=a_1+30a_2$.  A natural alternative normalization is
\[
a_2=\frac{1+a_1}{30},
\]
which enforces $\Delta A_{\moon}=1+a_1$, i.e.\ one full anomaly cycle plus a small excess per synodic month,
and thus yields a ``smoother'' evolution of the anomaly phase across month boundaries.  This choice trades the
exact unit-step property $28a_2=1$ for smoother month-to-month phase bookkeeping; in practice, both
normalizations lead to extremely close boundary times, with differences coming only from how the
piecewise-linear proxy is sampled/interpolated.

For the Sun, the scale $12$ plays an analogous role: the solar anomaly phase varies much more slowly,
the solar equation has smaller amplitude, and a coarser tabulation (with the standard sine-like symmetries)
is adequate.
\end{remark}

\begin{remark}\label{rem:numerical_comparisons_constants}
Comparing the traditional parameters against modern benchmarks (see Appendix~\ref{app:modern-periods}), the mean lunation $m_1$ drifts by only $\approx -0.16\,\mathrm{s}$ per month relative to $S_{\mathrm{syn}}$. However, the implied solar year $Y_{\mathrm{model}} = m_1/s_1 \approx 365.27\,\mathrm{d}$ accumulates a significant seasonal error of $\approx 41$ minutes per year ($\approx 2.85$ days per century). The derived anomalistic month deviates from $S_{\mathrm{anom}}$ by $\approx 3.6\,\mathrm{s}$. 
Finally, the equation of center corrections in \eqref{eq:true_date_model_new} have amplitudes of roughly $10$~hours (Moon) and $4.4$ hours (Sun), consistent with the expected magnitude of first-order anomalies.
\end{remark}

\subsubsection{Periodicity and tie cases in rational inverse schemes}
\label{sss:ties}

All boundary times produced by the traditional day-counting rules are obtained by combining affine
phases in $(d,n)$ with exact table lookups (piecewise-affine interpolation on rational knots).
Consequently, every boundary time $t(d,n)$ is a \emph{rational number}.  This raises the logical
possibility of a \emph{tie}: a computed boundary may land \emph{exactly} on a civil-day boundary.
Such a tie would force the calendar rules to invoke a tie-break convention.
In the inheritance formulation this tie-break is exactly the rule for labeling a sunrise that coincides with a lunar-day end ($\tau_f=b_k$):
under our right-closed lunar-day intervals $(b_{k-1},b_k]$, the sunrise post inherits the label $k$, hence the following civil day carries $k$.

We prove that tie cases do not occur on human time scales for any of the principal traditions:
Tsurphu has no ties at all, and the remaining principal traditions admit ties only at intervals of
$L=23{,}873{,}976$ lunar months (about $1.9$ million years).

The key structural point is that tie questions are \emph{finite}.  Writing the true boundary time in
the standard “inverse-scheme” form
\[
  t(d,n)\;=\;M_d(n)\;+\;C_d(n),
\]
where $M_d(n)$ is the affine (identity) contribution and $C_d(n)$ is the correction built from
periodic table terms, Appendix~\ref{app:congruence} shows that for each fixed $d$ the fractional parts of $M_d$ and $C_d$
repeat with explicit periods in $n$.  In particular, for the principal siddh\=anta motions one has
\[
  P_{\rm md}=\den(m_1)=5656,\qquad
  P_{\rm sun}=\den(s_1)=804,\qquad
  P_{\rm moon}=\den(a_1)=3528,
\]
and
\[
  P_{\rm corr}=\lcm(P_{\rm sun},P_{\rm moon})=236376,
\]
so that
\[
  M_d(n+P_{\rm md})\equiv M_d(n)\pmod 1,
  \qquad
  C_d(n+P_{\rm corr})\equiv C_d(n)\pmod 1.
\]
Hence the tie condition
\[
  t(d,n)\in\Z
  \qquad\Longleftrightarrow\qquad
  M_d(n)+C_d(n)\equiv 0\pmod 1
\]
reduces to a finite congruence problem on the product of the two period rings.  Let
\[
  L=\lcm(P_{\rm md},P_{\rm corr})=23873976,\qquad g=\gcd(P_{\rm md},P_{\rm corr})=56.
\]
Then for each fixed $d$ the set of ties is a (possibly empty) union of residue classes $n\bmod L$.

We implement this finite reduction with an exact meet-in-the-middle algorithm.
Fix a tradition and fix $d$.  Precompute the two periodic arrays
\[
  \mathcal M_d[s]=\{M_d(s)\}\in\Q/\Z,\qquad s=0,\dots,P_{\rm md}-1,
\]
\[
  \mathcal C_d[r]=\{C_d(r)\}\in\Q/\Z,\qquad r=0,\dots,P_{\rm corr}-1.
\]
We build a hash map for the \emph{larger} array $\mathcal C_d[r]$, so that the $P_{\rm md}$ queries
coming from $\mathcal M_d[s]$ can be answered in $O(1)$ average time.  To incorporate the Chinese
remainder compatibility condition, we split this hash map by the residue class $r\bmod g$.  
We then scan $s=0,\dots,P_{\rm md}-1$ and, for each $s$, look up
the target value $-\mathcal M_d[s]$ in the \emph{bucket indexed by $s\bmod g$}.  Each match returns some
$r$ with
\[
  \mathcal M_d[s]+\mathcal C_d[r]\equiv 0 \pmod 1
  \qquad\text{and}\qquad
  r\equiv s \pmod g,
\]
together with the congruences
\[
  n\equiv r \pmod{P_{\rm corr}},
  \qquad
  n\equiv s \pmod{P_{\rm md}}.
\]
By the Chinese remainder theorem, every compatible pair $(r,s)$ lifts to a unique residue class
$n\equiv n_0\pmod L$, and the resulting set of $n_0$ is exactly
the set of tie classes for that $d$.  The computation is performed in exact rational arithmetic (no
floating point), so the reported tie classes are mathematically rigorous.

The periodicity reduction and the two congruence prefilters used below are proved in
Appendix~\ref{app:congruence}: the finiteness/periodicity mechanism and the general ``single-table''
localisation principle are in Proposition~\ref{prop:periodicity_prime_obstruction}, the principal
period computation is in Remark~\ref{rem:princ-period}, and the $101$-- and $67$--local prefilters are
in Remarks~\ref{rem:local-101} and~\ref{rem:local-67}, with the conversion to explicit residue classes
via Proposition~\ref{prop:p-local-affine}.

The two appendix filters can be applied as prefilters, further shrinking the practical workload.
The $101$-filter (Remark~\ref{rem:local-101}) forces $n$ into a single residue class modulo $101$ for
each fixed $d$, and the $67$-filter (Remark~\ref{rem:local-67}) forces $n$ into a single residue class
modulo $67$.  Together these restrict $n$ to one class modulo $6767=67\cdot 101$, reducing the
effective search over one full $L$-period to only $L/6767=3528$ candidates per $d$.
In practice, this means the hash over $\mathcal C_d[r]$ is built only on the admissible residue class
$r\equiv \nu_d\pmod{67}$, and the scan over $s$ is restricted to $s\equiv n_d\pmod{101}$.

In all computations we allow $d$ to range over $0,1,\dots,30$, where $d=0$ denotes the month boundary
(Janson uses $d=0$ for “beginning of month” values), and where $(n,30)$ is equivalent to $(n+1,0)$ at
the next month boundary.  Running the above algorithm for the four principal traditions yields the
tie residue classes shown in Table~\ref{tab:tie-fast}.  The Tsurphu tradition produces \emph{no} ties.
By contrast, Phugpa, Mongolian, and Bhutan each produce exactly six tie residue classes, occurring at
the same set of day-indices $d\in\{0,4,10,12,20,24\}$ but at different residues $n\bmod L$.

Finally, note that each residue class $n\equiv n_0\pmod L$ corresponds to an arithmetic progression of
tie events separated by $L$ lunar months (about $1.9$ million years).  For the principal traditions
exhibiting ties, the nearest representatives to the standard eighteenth-century epochs occur tens of
thousands of years away (and the remaining representatives are hundreds of thousands to millions of
years away), so tie-breaking is absent on historical time scales even when ties exist mathematically.

\begin{table}[ht]
\centering
\renewcommand{\arraystretch}{1.12}
\begin{tabular}{lcp{0.54\linewidth}}
\hline
Tradition & epoch (JD) & tie classes $(d,\;n\bmod L)$\\
\hline
Phugpa & $2359237$ &
$(0,16267085)$;\;
$(4,3674149)$;\;
$(10,12833960)$;\;
$(12,7092386)$;\;
$(20,16221971)$;\;
$(24,1064342)$\\
Tsurphu & $2353745$ & none\\
Mongolian & $2359237$ &
$(0,889286)$;\;
$(4,12170326)$;\;
$(10,21330137)$;\;
$(12,15588563)$;\;
$(20,844172)$;\;
$(24,9560519)$\\
Bhutan & $2361807$ &
$(0,18255228)$;\;
$(4,5662292)$;\;
$(10,14822103)$;\;
$(12,9080529)$;\;
$(20,18210114)$;\;
$(24,3052485)$\\
\hline
\end{tabular}
\caption{Tie residue classes computed by the fast meet-in-the-middle + CRT algorithm.  
Each listed pair $(d,r)$
means that every month index $n\equiv r\pmod L$ yields a tie $t(d,n)\in\Z$.}
\label{tab:tie-fast}
\end{table}

\subsection{Variations among the principal traditions}
\label{ss:traditions}

The various modern traditions---such as Phugpa, Tsurphu, Bhutan, and the later Mongol reform---represent a series of distinct mathematical retrofits applied to the ancient Indian \emph{Kālacakra} framework. By disaggregating the calendar's output against modern ephemerides, we can reverse-engineer the specific historical, textual, and geographic imperatives that guided these reforms. This section systematically dissects these variations, moving from the foundational lunar engine and its temporal drift, through the anomaly variance and sidereal solar lag, ultimately culminating in a quantitative evaluation of the geographical adaptation hypothesis.

\subsubsection{Distribution of new moon temporal offsets}
\label{sec:offset_distribution}

To quantify the behavior of the Tibetan traditions, we calculated the precise instant of the true astronomical new moon using the modern JPL DE422 ephemeris \cite{DE422JPL} ($t_{\text{DE422}}$) and compared it with the time predicted by the respective Tibetan tradition ($t_{\text{Tib}}$). The temporal offset is defined simply as $\Delta t = t_{\text{Tib}} - t_{\text{DE422}}$. 

Figure \ref{fig:offset_histograms} presents the distribution of these offsets for the four principal traditions. To capture the state of each system near its historical era of active use, the histograms cover specific 200-year windows: 1750–1950 for the later Mongol reform, and 1450–1650 for the older Phugpa, Tsurphu, and Bhutanese traditions.

\begin{figure}[htpb]
    \centering
    \begin{subfigure}[b]{0.4\textwidth}
        \centering
        \includegraphics[width=\textwidth]{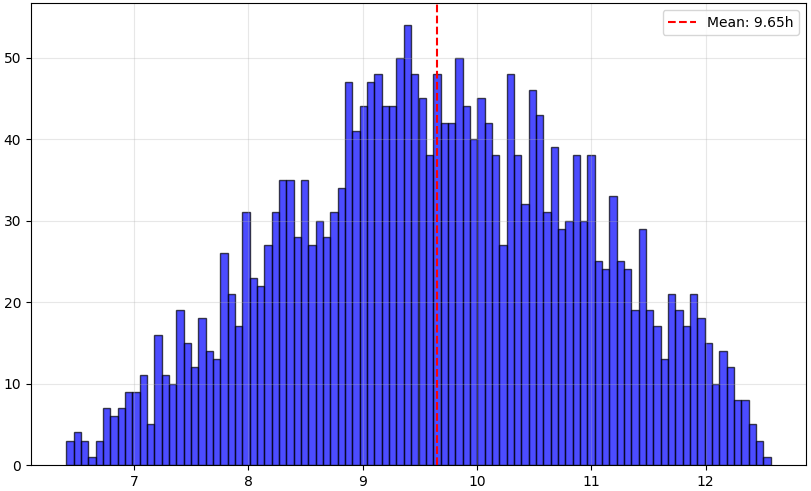} 
        \caption{Phugpa (1450--1650)}
        \label{fig:hist_phugpa}
    \end{subfigure}
    \hfill
    \begin{subfigure}[b]{0.4\textwidth}
        \centering
        \includegraphics[width=\textwidth]{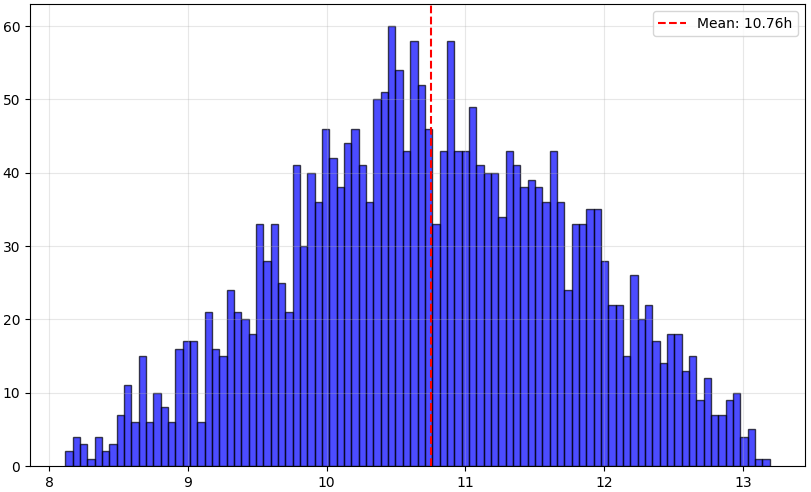}
        \caption{Tsurphu (1450--1650)}
        \label{fig:hist_tsurphu}
    \end{subfigure}
    \begin{subfigure}[b]{0.4\textwidth}
        \centering
        \includegraphics[width=\textwidth]{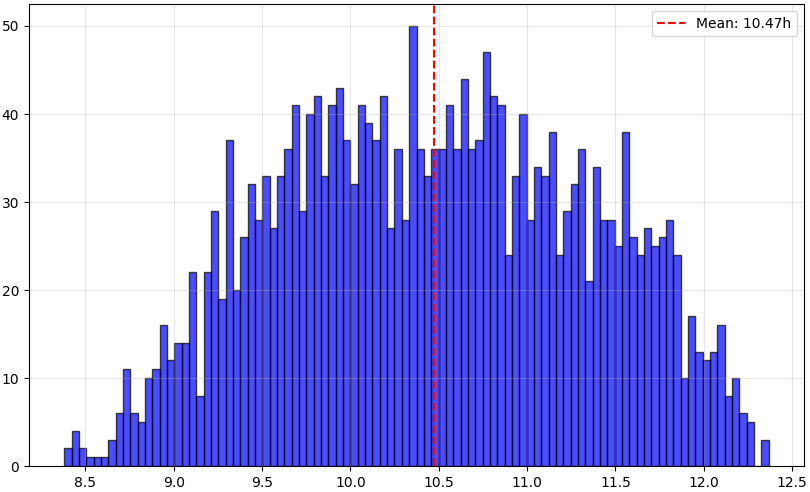}
        \caption{Bhutan (1450--1650)}
        \label{fig:hist_bhutan}
    \end{subfigure}
    \hfill
    \begin{subfigure}[b]{0.4\textwidth}
        \centering
        \includegraphics[width=\textwidth]{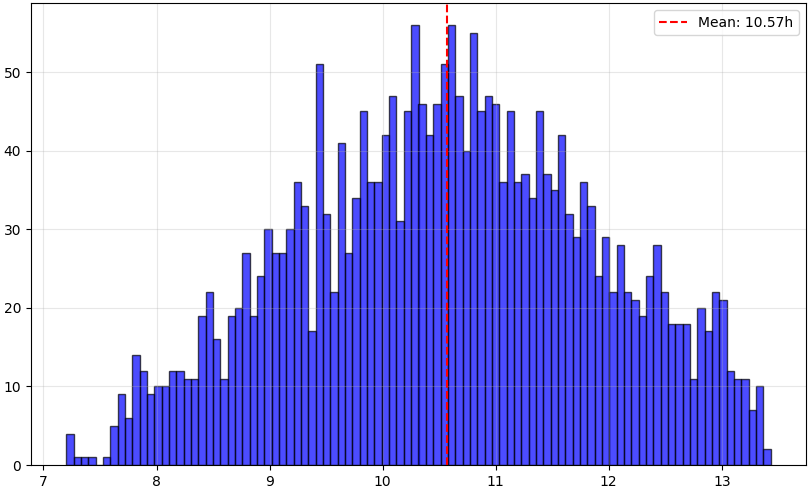}
        \caption{Mongol (1750--1950)}
        \label{fig:hist_mongol}
    \end{subfigure}
    \caption{Temporal offsets ($\Delta t = t_{\text{Tib}} - t_{\text{DE422}}$) in hours for true new moon calculations.}
    \label{fig:offset_histograms}
\end{figure}

Two key features are immediately apparent in these distributions: the spread (variance) and the average location of the peak.

First, the offsets do not form a single sharp peak but rather a spread of values spanning roughly $\pm 2$ hours from the center. This variance is caused by the lunar "equation of center," the correction applied to account for the Moon's elliptical orbit. As illustrated schematically in Figure \ref{fig:anomaly_diagram}, the Moon travels faster near perigee and slower near apogee relative to its mean motion. The spread in the histograms represents the residual error between the Tibetan trigonometric approximation of this anomaly and the actual complex lunar motion.

\begin{figure}[htpb]
    \centering
    \includegraphics[width=0.9\textwidth]{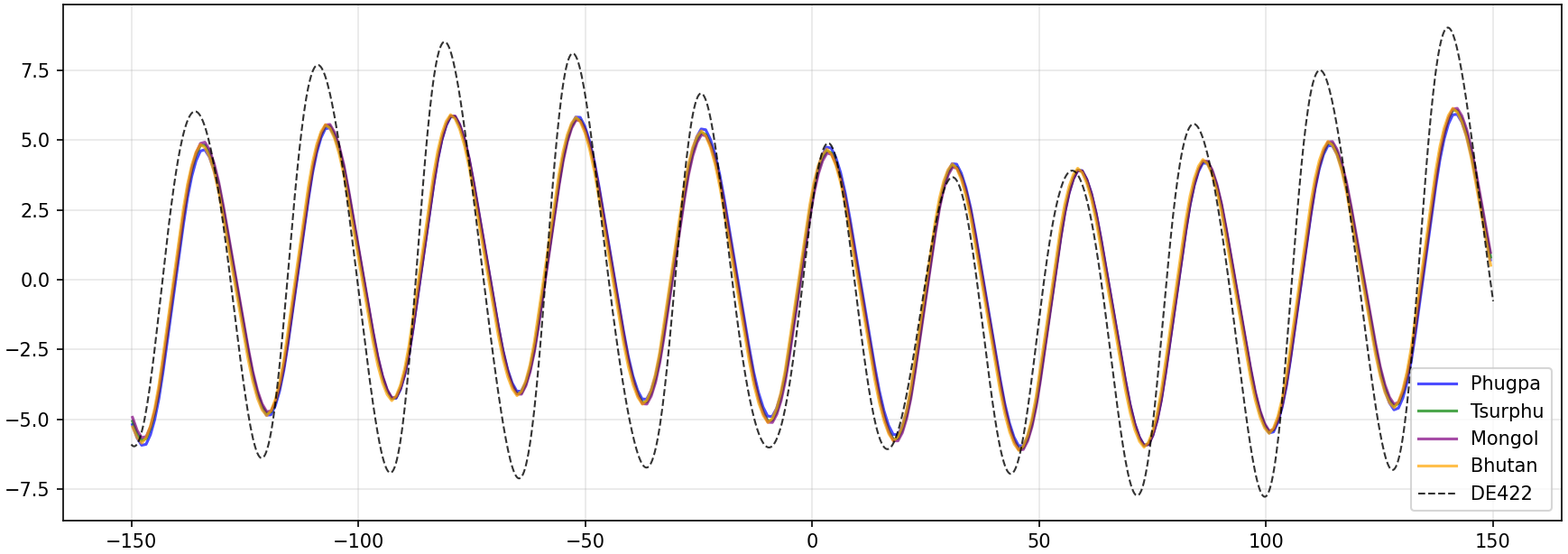} 
    \caption{Lunar anomaly (in degrees) vs. days relative to Feb~1, 2026.}
    \label{fig:anomaly_diagram}
\end{figure}

Second, the average offsets are significantly positive. This does not imply the calendar is ``wrong'' by this amount; rather, it reveals the implicit zero-point of the Tibetan internal clock relative to the civil day, which begins at sunrise. 

Consider the Phugpa tradition (Figure \ref{fig:hist_phugpa}) with its average offset of $+9.65$ hours. Imagine a true astronomical new moon occurring exactly at Greenwich Noon (12:00 GMT). In Lhasa, the geographic reference point for the Phugpa tradition ($91.1^\circ$ E), the Local Mean Time (LMT) shift is $+6.07$ hours, meaning this event physically occurs at 18:04 LMT (18.07 in decimal hours). The Phugpa calendar registers this event with a value of $9.65$ hours. This implies that the Phugpa internal clock's zero-point is set at $18.07 \text{ LMT} - 9.65\text{h} = 08.42 \text{ LMT}$ (or 08:25). If we assume an average sunrise of 05:56 LMT (5.93 decimal hours), the calendar's internal ``start of day'' is functionally delayed by $2.49$ hours (approximately 2 hours and 29 minutes) relative to the actual sunrise. Equivalently, one can view this as the calendar computing the new moon approximately 2.5 hours prematurely relative to the actual dawn. Since the Moon's elongation increases by roughly $0.51^\circ$ per hour relative to the Sun, this temporal shortfall implies the calendar's internal model overestimates the lunar elongation by roughly $1.3^\circ$, mathematically corroborating Janson's observation that the mean elongation is structurally ``about $2^\circ$ too large'' \cite[\S12.1]{janson}.\footnote{The variation between our $\approx 1.3^\circ$ estimate (using 1450--1650 data and a 05:56 dawn) and Janson's $\approx 2^\circ$ figure (using 2013 data and a 05:00 dawn) stems from both secular drift and the chosen sunrise baseline. If we evaluate modern data (where the 1900--2100 average offset shifts to $\approx 9.4$h) against an idealized 05:00 LMT dawn, the functional delay increases to $\approx 3.7$ hours. At a rate of $0.51^\circ/\text{h}$, this produces an elongation error of $\approx 1.9^\circ$, seamlessly reconciling with Janson's observation.}
One hypothesis is that this systematic delay acts as a temporal buffer, potentially keeping calculated new moon times within desired civil day boundaries for religious observance, though whether this was an intentional design choice remains uncertain.

Applying this exact geographic and temporal logic to the other principal traditions reveals their specific local calibrations. 
The results are summarized in Table \ref{tab:local_calibrations}. In each case, the local mean time (LMT) shift of the 12:00 GMT new moon is adjusted by the tradition's average histogram offset to determine its internal zero-point, which is then compared against a standard 05:56 dawn.

\begin{table}[htpb]
    \centering
    \begin{tabular}{lccccc}
        \toprule
        Tradition & Longitude & LMT shift (h) & Avg. offset (h) & Zero-point (LMT) & Delay \\
        \midrule
        Phugpa & $91.1^\circ$ E & $+6.07$ & $+9.65$ & 08:25 & 2h 29m \\
        Tsurphu & $90.4^\circ$ E & $+6.03$ & $+10.76$ & 07:16 & 1h 20m \\
        Bhutan & $89.6^\circ$ E & $+5.97$ & $+10.47$ & 07:30 & 1h 34m \\
        Mongol & $106.9^\circ$ E & $+7.13$ & $+10.57$ & 08:34 & 2h 38m \\
        \bottomrule
    \end{tabular}
    \caption{Geographic offsets and internal zero-points for principal traditions.}
    \label{tab:local_calibrations}
\end{table}

These varying functional delays naturally invite a geographic adaptation hypothesis: that by tuning the epoch constants, the founders of these traditions might have intentionally shifted the calendar's effective geographic meridian to suit their local observer locations. Under this conjecture, such localized calibration would ensure that the distribution of computed new moons fell consistently on a pragmatically safe side of dawn for their specific longitudes. However, as we will quantitatively demonstrate in the following subsections and especially in \S\ref{sec:geo-intercalation}--\ref{sec:geo-epoch}, a purely geographic interpretation faces severe structural challenges. When analyzed against the broader artifacts of the calendar, the data strongly suggests that these temporal shifts are not straightforward spatial coordinate corrections, but are rather systemic historical recalibrations and textual anchors that function largely independently of true observer geography.

\subsubsection{Mean temporal drift and the lunar engine}
\label{sec:lunar_engine_drift}

Having established the distribution of offsets in specific eras, we now examine how the
\emph{mean} offset evolves over long time windows. This reveals the long-term behavior of
the underlying arithmetic engines and helps separate purely calendrical effects from secular
features of the Earth--Moon system and of time-scale conversion.

We begin by isolating a single modern tradition and measuring its raw drift against a modern
ephemeris product. Figure~\ref{fig:mongol_raw_fit} shows the mean temporal offset of the
Mongol tradition from 500~AD to 2000~AD, obtained by subtracting the TT-based DE422 new-moon
times \cite{DE422JPL} from the raw calendar output and fitting the result by a quadratic curve.

\begin{figure}[htpb]
    \centering
    \includegraphics[width=0.8\textwidth]{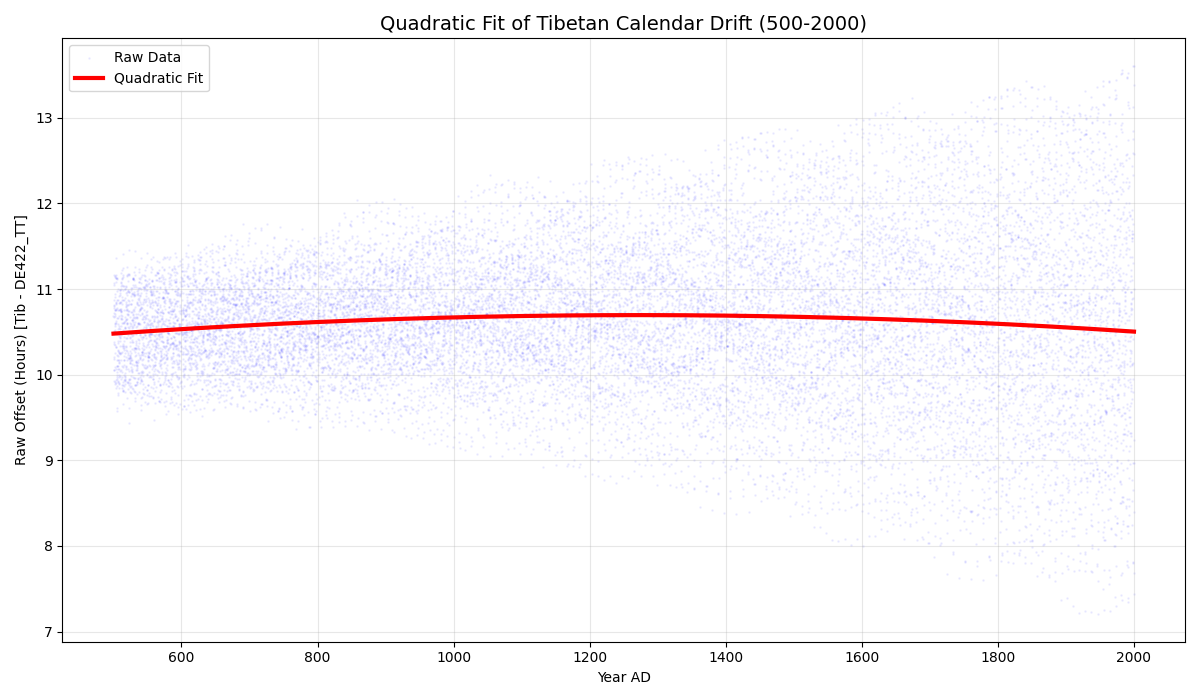}
    \caption{Raw mean temporal offset of the Mongol tradition's new moon calculations compared to DE422 in TT, fitted with a quadratic curve.}
    \label{fig:mongol_raw_fit}
\end{figure}

The fitted curve is visibly parabolic, with empirical quadratic coefficient
$c_2 = 13.07\ \text{s/cy}^2$.
This number should be read as a measured property of the TT-based comparison, not as a
fundamental constant of the Tibetan calendar itself. At the level of mean motions, the
calendar's lunar engine is affine in the month index, with only periodic corrections; it
contains no intrinsic quadratic secular term. The observed curvature therefore arises from
comparing a fixed arithmetic month model, expressed in civil-day units, against a modern
dynamical ephemeris in TT. In particular, it mixes the time-scale discrepancy between TT and
civil time with the genuine slow secular evolution of lunar phases over long intervals.
The fitted vertex of this raw parabola lies near 1240~AD.

The linear term (the ``tilt'' of the parabola) is governed primarily by the calendar's
fundamental synodic-month constant $m_1$. The modern \emph{Grub rtsis} (siddh\=anta-family)
traditions use the highly accurate rational value
\[
m_1=\frac{167025}{5656}\approx 29.530587\ \text{days},
\]
equivalently
\[
m_2=\frac{m_1}{30}=\frac{11135}{11312},
\]
traditionally described as the rule that $11312$ lunar days correspond to $11135$ solar days,
or equivalently that $177$ civil days are omitted per $11312$ lunar days. Any slight mismatch
between this fixed constant and the ephemeris-implied mean synodic month produces a long-run
linear trend in the mean residual.

To remove the time-scale mismatch coming from the use of TT ephemeris values, we next apply
the astronomical $\Delta T$ correction to the DE422 data. Equivalently, this converts the
comparison into one against UT-based DE422.

\begin{figure}[htpb]
    \centering
    \includegraphics[width=0.8\textwidth]{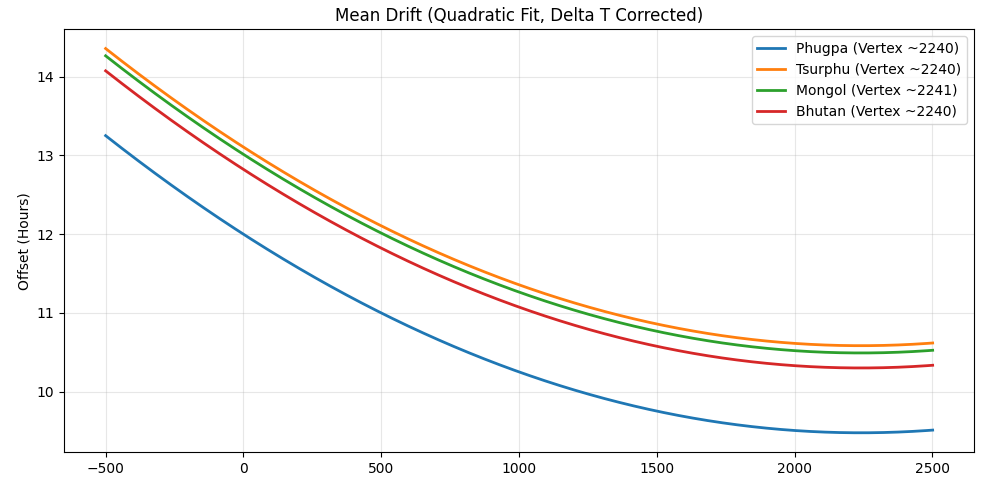}
    \caption{Mean drift after applying the $\Delta T$ correction, equivalently comparing with UT-based DE422.}
    \label{fig:delta_t_corrected}
\end{figure}

As shown in Figure~\ref{fig:delta_t_corrected}, the corrected curve still retains a parabolic
shape, but its geometry changes substantially: the fitted vertex is shifted far into the
future, to about 2240~AD. This shift should not be interpreted as a new physical effect.
Rather, it is the algebraic consequence of removing the TT--civil-time discrepancy while
retaining the small long-run mismatch between the fixed arithmetic constant $m_1$ and the
ephemeris lunar month. After the $\Delta T$ correction, the remaining curvature is therefore
best read as a diagnostic of the interaction between a fixed lunar arithmetic engine and the
slow secular evolution of the true Moon, while the displacement of the vertex is controlled
mainly by the linear tilt.

The significance of the \emph{Grub} value of $m_1$ becomes especially clear when one compares
it with the older K\=alacakra baseline, the \emph{kara\d{n}a} (Tib.\ \emph{byed rtsis}).
Figure~\ref{fig:karana_vs_grub} overlays the corresponding mean drifts.

\begin{figure}[htpb]
    \centering
    \includegraphics[width=0.8\textwidth]{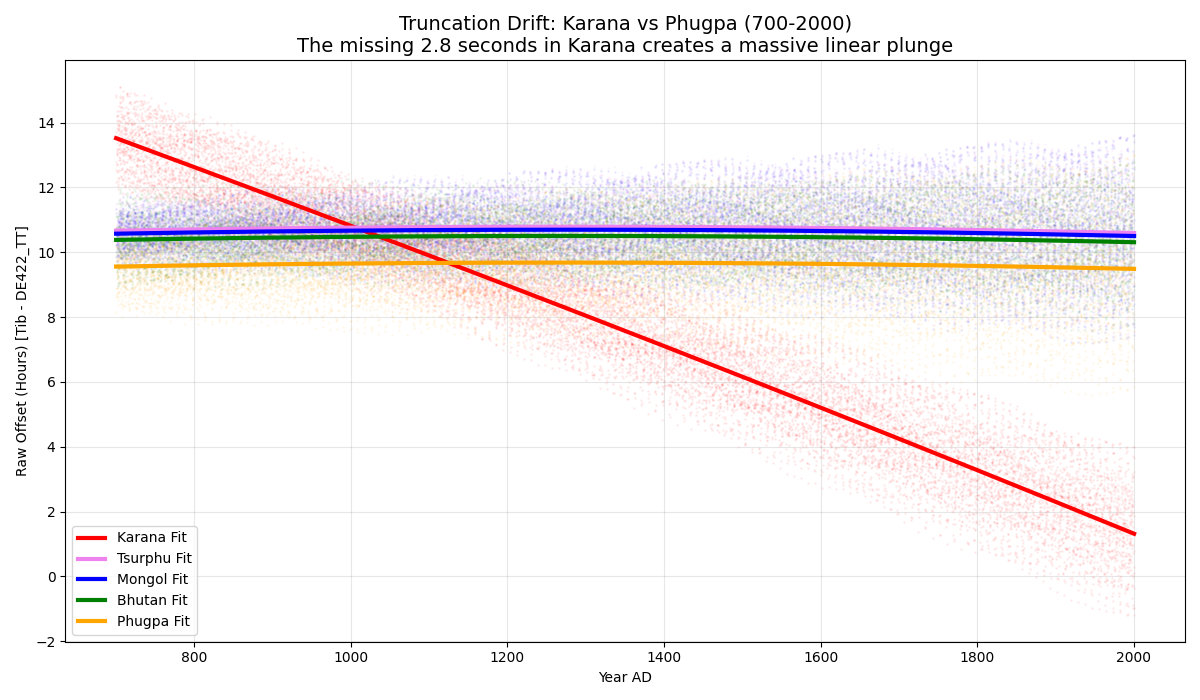}
    \caption{Temporal drift of the \emph{kara\d{n}a} baseline and the modern \emph{Grub} traditions.}
    \label{fig:karana_vs_grub}
\end{figure}

The \emph{kara\d{n}a} framework uses the simplified fraction
\[
m_1^{\mathrm{kar}}=\frac{10631}{360}\approx 29.530555\ \text{days},
\]
which is shorter than the \emph{Grub} value by about $2.7$ seconds per synodic month.
At roughly $1237$ lunations per century, this truncation accumulates to a linear drift of
about $0.93$ hours per century. As Figure~\ref{fig:karana_vs_grub} shows, this much larger
linear defect completely dominates the more delicate quadratic effects visible in
Figures~\ref{fig:mongol_raw_fit}--\ref{fig:delta_t_corrected}. From this perspective, the
historical adoption of
\[
m_1=\frac{167025}{5656}
\]
by the \emph{Grub} reformers should be understood as a highly successful correction of the
dominant linear defect in the older baseline, bringing the lunar engine into a regime where
subtler long-run effects only become visible after comparison with modern ephemerides.

\subsubsection{Anomaly variance and possible phase anchoring}
\label{sec:anomaly_archaeology}

While \(m_1\) dictates the long-term mean drift discussed above, the \emph{spread} or variance of the histograms in Figure~\ref{fig:offset_histograms} is influenced primarily by the anomaly equations. The true new moon timings fluctuate around the mean because of the nonuniform motions of the Sun and Moon, and the calendar models these fluctuations by simplified trigonometric schemes involving apsidal parameters and epoch phase constants. In particular, the quality of fit depends not only on the mean motions but also on how the assumed anomaly phases line up with the actual sky.

Because the calendrical anomaly rates do not perfectly match the physical precession of the lunar and solar apsides, this phase agreement is not stationary over long periods. As the theoretical anomaly model drifts relative to the sky, the spread in the temporal offset changes as well. One therefore expects the variance curve to have a broad minimum in the era where the adopted anomaly phases happen to align comparatively well with the corresponding physical configuration.

\begin{figure}[htpb]
    \centering
    \includegraphics[width=0.8\textwidth]{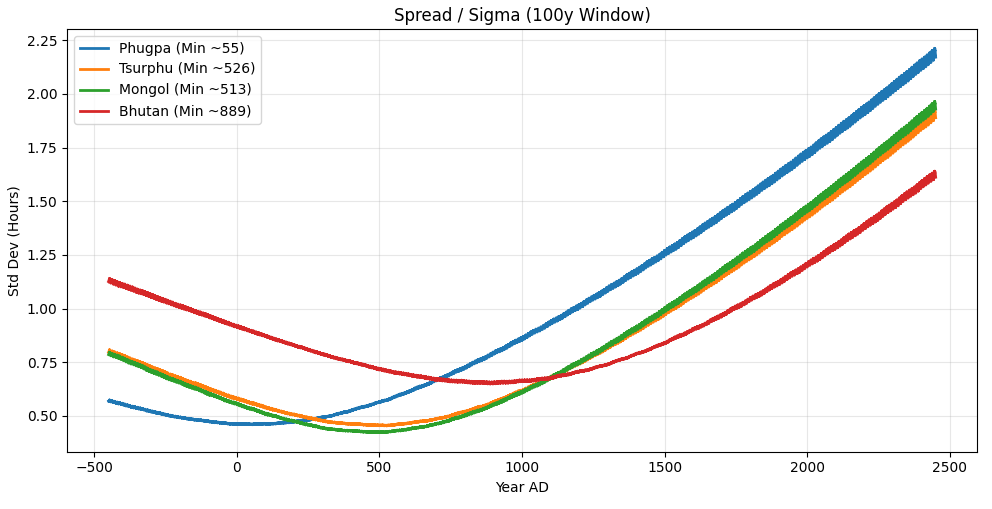}
    \caption{Spread (standard deviation) of the temporal offset in new moon timings.}
    \label{fig:anomaly_spread}
\end{figure}

Figure~\ref{fig:anomaly_spread} shows such broad minima for all four traditions, though at different epochs. This suggests that the variance profile may contain information about the effective phase anchoring of the anomaly model. The point should be treated cautiously: these minima are indirect diagnostics, and their location depends on the adopted astronomical comparison, smoothing window, and the internal form of the calendar model. Still, they offer a useful comparative clue.

For Tsurphu and Mongol, the minima occur in the early sixth century, around 500~AD. This may be consistent with the possibility that these traditions inherited anomaly parameters ultimately tied to an older Indian siddh\=antic layer, or at least to a model whose phase alignment was best suited to that era. Such a reading remains conjectural, but it would help explain why the variance is relatively tight near late antique dates and broader in the periods when these calendars were historically active.

The Bhutanese curve reaches its minimum considerably later, around the late ninth century. This may point to a different effective phase anchoring, perhaps reflecting a later stage in the transmission or adjustment of anomaly constants. It is tempting to connect this with the medieval Indian background of the K\=alacakra material, but the plot by itself does not establish such a link; at most, it suggests that the Bhutanese anomaly model is better aligned with a later sky than the Tsurphu--Mongol branch.

Phugpa is the most difficult case. Its minimum lies much earlier, roughly near the beginning of the Common Era. It would be unwise to interpret this literally as evidence for a direct empirical calibration in that period. More modestly, it indicates that the internal phase structure of the Phugpa anomaly model is such that, when propagated backward, it aligns best with the modern astronomical reference in a much earlier epoch. Whether this reflects deep inheritance, indirect back-projection, or simply the compounded effect of the chosen rates and phase constants is not clear from the variance curve alone.

Accordingly, Figure~\ref{fig:anomaly_spread} is best read as a comparative diagnostic of anomaly-phase behavior rather than as a dating instrument in any strong sense. The broad location of the minima may preserve traces of different underlying phase anchorings, but the historical interpretation should remain tentative.

\subsubsection{The sidereal problem, solar lag, and epoch anchoring}
\label{sec:sidereal_sgang}

Transitioning from temporal offsets to absolute celestial positions exposes the most profound structural vulnerability of the Tibetan calendar: its uncorrected sidereal framework. Because the calendar calculates solar longitude relative to a fixed background of stars rather than the moving tropical Vernal Equinox, it is entirely blind to the precession of the equinoxes.

A natural diagnostic is therefore to evaluate the solar longitude produced by the calendar itself at the instant of the astronomical Vernal Equinox, that is, when tropical longitude is exactly \(0^\circ\). Although this solar longitude is not the primary object used in the operational month-making procedure, it is still a meaningful output of the system, and it provides a direct way to compare the internal solar geometry of different traditions against the seasonal year.

\begin{figure}[htpb]
    \centering
    \includegraphics[width=0.8\textwidth]{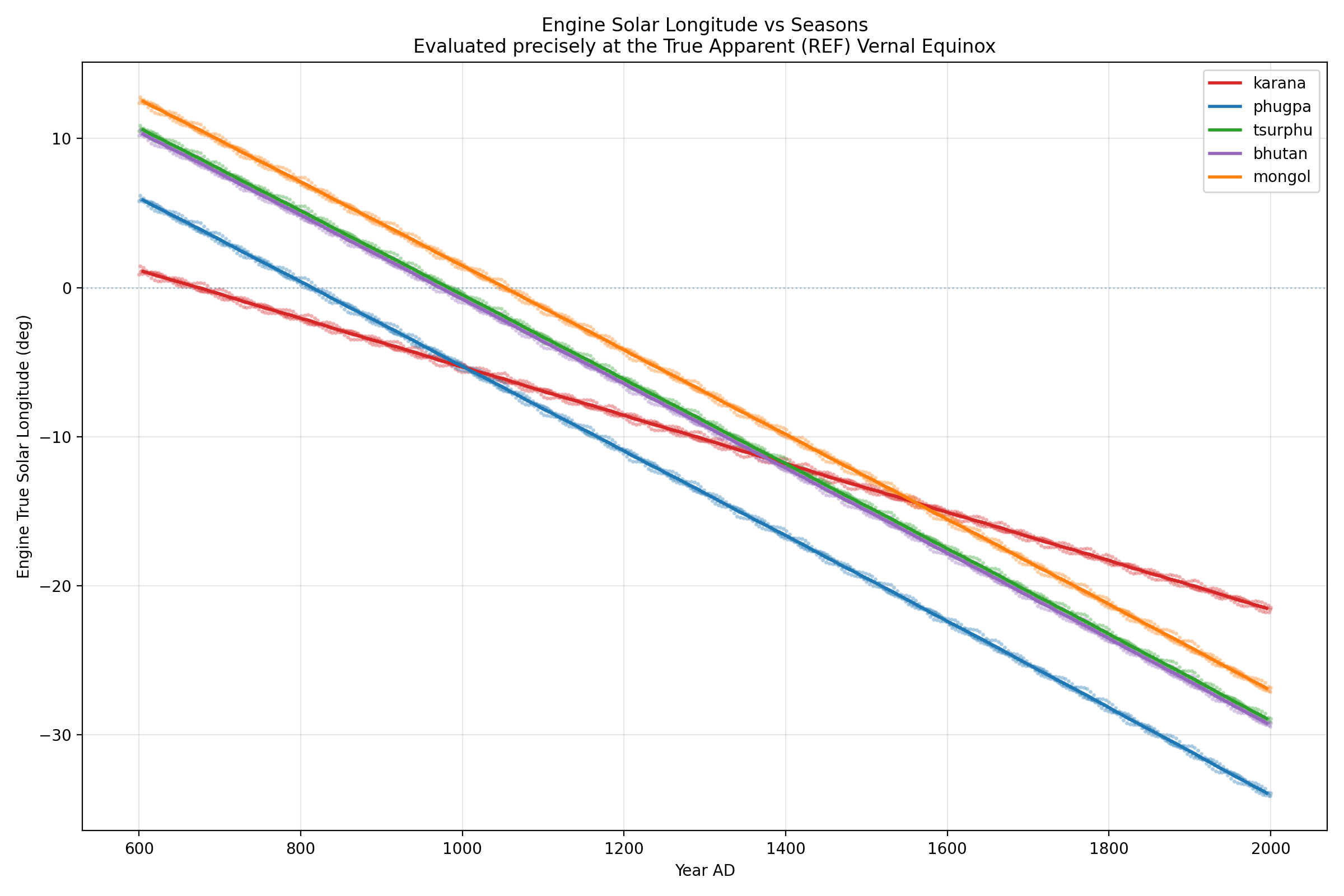}
    \caption{Solar longitude produced by several Tibetan calendar engines, evaluated at the astronomical Vernal Equinox.}
    \label{fig:raw_equinox}
\end{figure}

Figure~\ref{fig:raw_equinox} shows a remarkably stable qualitative picture. All traditions exhibit a nearly linear downward drift, reflecting the fact that a sidereal solar framework, left uncorrected, must gradually slip relative to the tropical equinox. Quantitatively, the common drift is about \(1.4^\circ\) per century, as expected from precession. The exact vertical placement of the lines, however, should not be over-interpreted: it depends on whether one plots mean or true solar longitude, and also on the precise equinox diagnostic. By contrast, the horizontal structure is much more robust. In particular, the relative slopes remain very close, and the main crossing years are largely unchanged under these choices.

One striking feature is that the \emph{karana} line crosses the tropical zero line around the mid-seventh century, roughly near 640~AD. This is at least compatible with the idea that the inherited Indian solar framework was originally calibrated in an earlier epoch when its sidereal zero still lay relatively close to the tropical Vernal Equinox. Whatever the historical mechanism, the plot strongly suggests that the later Tibetan traditions should not be read simply as passive continuations of a single drifting line.

At the same time, the later Tibetan branches are shaped by more than sidereal precession alone. The \emph{Grub} reforms updated the lunar month to the more precise value \(167025/5656\), while retaining the traditional solar ratio \(S_1=65/804\). This combination forces an overlong solar year, about \(365.2706\) days, and hence introduces an additional solar lag independent of precession. The branch geometry in Figure~\ref{fig:raw_equinox} should therefore be read as the combined effect of inherited sidereal phase and reform-era handling of this solar lag.

A useful way to summarize these relative placements is to ask when each later branch would intersect the \emph{karana} baseline under linear extrapolation. These intersection years should not be taken as exact historical dates, but rather as effective mathematical anchor years. In this sense, the Phugpa branch points back to about 1003~AD, Tsurphu to about 1398~AD, Bhutan to about 1358~AD, and Mongol to about 1561~AD. These values are more stable than the absolute vertical offsets, and they help organize the historical interpretation of the figure.

The most striking case is Phugpa. Its intersection with \emph{karana} near the turn of the
first millennium suggests that the Phugpa founders may not have simply inherited the
then-current sidereal drift. One possible reading---which we may call the \emph{backtracking
hypothesis}---is that the Phugpa reform implicitly re-anchored its solar data toward the
older 1024/1027 \emph{karana} horizon rather than passively continuing the contemporary
drift. A related possibility is that some already shifted Indian correction, perhaps connected
with \emph{'khrul sel}, had entered the source material before the Phugpa reform and was
then absorbed into the Tibetan system. The figure does not decide between these scenarios,
but it does suggest that uninterrupted drift alone is unlikely to explain the Phugpa phase
shift.

The Tsurphu and Bhutanese branches suggest a different pattern. Tsurphu lies much closer
to a conservative replacement of \emph{karana}; one may heuristically describe this as a
\emph{hot-swap hypothesis}, meaning a change of engine without a dramatic re-anchoring
to a very early epoch. Bhutan, in turn, tracks Tsurphu so closely that it is natural to regard
it as inheriting essentially the same solar data, or at least the same Kagyu-lineage anchoring,
rather than recalibrating independently to its own contemporary sky.

The Mongol branch stands somewhat apart, beginning higher and crossing later. This pattern
is consistent with a pragmatic late recalibration: neither a return to the old \emph{karana}
horizon nor passive acceptance of the heavily drifted eighteenth-century baseline, but an
intermediate re-anchoring that partly reduced the accumulated mismatch.

Figure~\ref{fig:raw_equinox} is therefore most informative at the structural level, not at the
level of exact intercepts. It shows a common sidereal drift, superimposed solar lag, and
nontrivial phase differences among the traditions. The absolute values on the vertical axis
depend on the diagnostic convention and should not be read as exact recoveries of
historically intended offsets. The drift itself, the crossing pattern, and the effective anchor
years are more robust features, and they likely reflect genuine differences in the solar
baselines transmitted by the traditions. The historical readings suggested above remain
interpretive hypotheses supported by this geometry, not direct historical proofs.

\subsubsection{Intercalation and latitude}
\label{sec:geo-intercalation}

The deliberate mathematical interventions required to manage the solar lag—whether Phugpa's deep historical back-tuning or the Mongol reform's forward drag—demonstrate that these calendar lineages were shaped by active, systemic textual choices rather than passive observation. This structural reality provides the necessary context to evaluate a persistent historical question: to what extent were these differing parameters driven by physical geography? 

As emphasized by Schuh, Henning, and Janson \cite{schuh,henning,janson}, the structural divergences among the major Tibetan calendar lineages are not the result of passive geographical coordinate adjustments, but rather these exact types of active historical recalibrations. While some historical authors invoked local geography to justify their parameters, translating these claims into mathematical reality reveals that the required geographic shifts are physically implausible. 

We can quantitatively formalize this conclusion by first examining the calendar's seasonal tuning, which is often mistakenly attributed to latitudinal climate adaptation. The placement of leap months, and consequently the seasonal position of the New Year (\emph{Losar} or \emph{Tsagaan Sar}), is governed by the alignment of the epoch mean sun ($s_0$) relative to the theoretical intercalation target ($d_0$). As established in Remark~\ref{r:d0-s0}, the effective seasonal tuning is quantified by the phase lag $\delta \approx d_0 - s_0$, which measures how far the calendar's solar cycle is delayed relative to its earliest theoretical limit.

Comparing these phase lags across the principal traditions reveals a pattern that directly contradicts climatic adaptation:
\begin{itemize}
    \item \emph{Mongolia vs.\ Tibet:} Mongolia ($48^\circ$ N) experiences a significantly colder, later spring than Central Tibet. A climatic adaptation would theoretically require a larger $\delta$ to push the Mongolian New Year later into the warming season. However, the Mongol tradition utilizes a minimal lag ($\delta \approx 0.3$ days). Conversely, the Phugpa tradition---native to the milder climate of Lhasa---implements a massive lag ($\delta \approx 6.1$ days), pushing its New Year roughly a week later than the Mongol system on average.
    \item \emph{Bhutan:} Situated south of the Himalayas ($27^\circ$ N), Bhutan experiences the earliest, warmest spring. Yet, its calendar employs a moderate lag ($\delta \approx 3.0$ days), placing its New Year structurally \emph{later} than the colder Mongolian tradition.
    \item \emph{Tsurphu vs.\ Phugpa:} Both traditions originated in Central Tibet and share identical climatic constraints. Despite this, they exhibit the largest divergence in seasonal tuning ($\approx 5.9$ days difference in solar phase lag).
\end{itemize}

As the precise calculations in Remark~\ref{r:d0-s0} demonstrate, these phase parameters fix the macroscopic seasonal behavior of the calendar. The lack of any coherent correlation with latitude or local climate quantitatively confirms the historical consensus \cite{schuh,henning,janson}: these massive, lineage-specific divergences are not geographical adaptations. Instead, they represent deliberate textual recalibrations of the solar epoch, driven by competing historical interpretations of solstice boundaries and seasonal definitions rather than the physical environment of the observer.

\subsubsection{Epoch constants and longitude}
\label{sec:geo-epoch}

A second geographical hypothesis posits that the differences in foundational epoch constants ($m_0, s_0, a_0$) represent longitude corrections (time zone shifts). To test this quantitatively, we analyze the mean epoch constants normalized to a common historical moment. Janson \cite[Table~6]{janson} projects the parameters of the principal traditions back to the standard epoch of 23 March 806 AD (JD 2015531), which we summarize in Table \ref{tab:epoch-constants-806} for convenience.

\begin{table}[ht]
\centering
\begin{tabular}{l c c c}
\toprule
{Tradition} & $m_0$ & $s_0$ & $a_0$ \\
\midrule
Phugpa     & $2.376238$ & $0.004975$ & $0.206349$ \\
Tsurphu    & $2.422338$ & $0.018261$ & $0.210317$ \\
Bhutan     & $2.410537$ & $0.017413$ & $0.220522$ \\
Mongol     & $2.418494$ & $0.023632$ & $0.207200$ \\ \bottomrule
\end{tabular}
\vspace{1mm}
\caption{Epoch constants at JD 2015531 (23 March 806), cf. \cite{janson}.}
\label{tab:epoch-constants-806}
\end{table}

Recall that $m_0$ represents a time offset. Specifically, it is the interval in fractional civil days from the epoch's local dawn to the mean conjunction. The absolute astronomical time of the conjunction is a globally fixed event. However, changing terrestrial longitude inherently shifts the time of local dawn. Because of this, a geographic shift produces a predictable, proportional change in the duration of $m_0$.

Ulaanbaatar ($107^\circ$ E) is roughly $16^\circ$ east of Lhasa ($91^\circ$ E), meaning local dawn occurs approximately 1.07 hours (or $0.044$ civil days) earlier. Because this local dawn occurs earlier, the elapsed time from dawn to the epoch conjunction must increase by this exact amount. Examining Table \ref{tab:epoch-constants-806}, the mean epoch constant ($m_0$) for the Mongol tradition ($2.418$) differs from the Phugpa tradition ($2.376$) by exactly $0.042$ civil days. This 0.042-day difference (roughly 1.01 hours) directly mirrors the rotational longitude difference, making the time-zone adaptation hypothesis initially appear extremely compelling.

However, several structural and historical factors make this purely geographic interpretation highly unlikely. First, re-integrating the older Kagyu lineages complicates the narrative. The Tsurphu tradition's value ($2.422$) is nearly identical to the Mongol value. Because the Tsurphu monastery is co-located with Phugpa in Central Tibet, its $+0.046$ shift relative to Phugpa cannot be a spatial correction. While it remains theoretically possible that Sumpa Khenpo independently calculated a precise geographic longitude correction for the Mongolian steppe that coincidentally mirrored the centuries-old Central Tibetan Tsurphu parameter, it is significantly more likely that the Mongol tradition simply inherited or adapted existing mathematical baselines \cite{janson}.

Second, analyzing the calendar's internal clock (established in Section \ref{sec:offset_distribution}) reveals that the required precision for a time-zone shift is swamped by both the system's structural temporal buffers and its inherent mathematical variance. The implicit zero-points for these calendars occur well after local dawn: roughly 08:25 LMT for Phugpa, 07:16 for Tsurphu, 07:30 for Bhutan, and 08:34 for Mongol. While the shift in $m_0$ does cause the Mongol calendar's internal zero-point (08:34 LMT) to align closely with Phugpa's (08:25 LMT) in their respective local times, both traditions maintain a massive functional delay of over two and a half hours relative to actual sunrise. Crucially, this structural delay operates in tandem with the traditional lunar model's intrinsic error spread, which exhibits a variance of roughly three to four hours for true new moon times. Because a one-hour geographic shift is entirely eclipsed by this massive baseline variance and built-in safety buffer, attempting precise time-zone adjustments within the classical framework is practically redundant.

Finally, a genuine geographic coordinate transformation demands that the epoch constants shift consistently: if $m_0$ is adjusted for a new local dawn, $s_0$ must shift proportionally. Evaluating the mean sun ($s_0$) across traditions completely severs this linkage. A one-hour time zone difference should shift the mean sun's absolute position by only $0.04^\circ$ (about $0.0001$ fractional turns). Instead, the discrepancy between the Phugpa ($s_0 \approx 0.005$) and Mongol ($s_0 \approx 0.024$) epochs is roughly $0.019$ fractional turns. Since a $0.019$ turn of the Sun takes approximately $0.019 \times 365.25 \approx 6.9$ days, this represents nearly a full week of solar motion. This massive displacement confirms that reformers utilized $s_0$ to perform entirely different structural jobs---such as anchoring to earlier standard epochs or tuning solstice definitions (as seen in Section~\ref{sec:geo-intercalation})---fundamentally disconnecting the solar anomaly from pure geographic time shifts.

In summary, while we cannot strictly rule out all geographic motivations, the quantitative signature of these traditions makes passive geographic coordinate transformations highly improbable. The massive solar displacements, the structural temporal buffers, and the shared Kagyu baselines strongly support the established historical consensus \cite{schuh,henning,janson}: these divergences primarily represent active, systemic retunings meant to counteract precessional drift or satisfy specific textual definitions, effectively eclipsing any strict observer-location adjustments.

\section{Astronomical Inaccuracies and Technical Design Space}
\label{s:inaccuracies}

With the arithmetic skeleton and its intended celestial meaning now explicit, we can
separate and quantify the different kinds of drift that accumulate from mismatches between the traditional constants
and modern mean quantities.
While the internal logic of the Tibetan calendar is perfectly consistent, ensuring rigid predictability for civil and religious planning,
its fixed arithmetic foundation creates
an inherent disconnect from the dynamic physical realities of the celestial sphere. 
This section provides an exhaustive analysis of the astronomical inaccuracies embedded within the traditional model, ranging from secular seasonal drift to the nuances of anomalistic phase slippage. Furthermore, it defines the technical design space available for addressing these errors, evaluating the trade-offs between retaining traditional algorithmic forms and adopting modern astronomical precision.
By mapping each inaccuracy to a spectrum of potential numerical and kinematic fixes, we establish the foundation for the concrete reform layers presented 
in Section~\ref{s:reforms}.

\subsection{Seasonal drift and the intercalation design}
\label{ss:seasonal-drift}

At the heart of the calendar's long-term instability is the foundational arithmetic axiom that governs intercalation: the assertion that 67 mean lunar months are exactly equal to 65 mean solar months. In the traditional siddhānta systems, this relation is treated not as an approximation but as a definition. It enforces a rigid intercalation cycle where exactly two leap months are inserted every 65 solar months, creating a repeatable 65-year cycle containing 804 lunations.

While computationally convenient, this ratio implies a mean solar year length that diverges significantly from the tropical year. 
The mean synodic month constant used in the principal traditions is 
$$
m_1 = \frac{167025}{5656} \approx 29.530587 \text{ days}.
$$
This value is remarkably precise, differing from the modern mean synodic month (approximately 29.5305889 days) by only about 0.16 seconds. 
Consequently, the internal drift of the moon's mean elongation is negligible for practical purposes, 
accumulating to only a few minutes over centuries. 

However, the derived mean solar year is 
$$
Y_{\mathrm{model}} = \frac{804}{65} \times m_1 \approx 365.270645 \text{ days},
$$
which is significantly longer than the true tropical year ($\approx 365.24219$ days). The discrepancy is $0.028455$ days per year, which accumulates to approximately {\em 2.85 days per century}, or nearly a full month per millennium. As discussed previously, it is precisely this secular drift that explains the gradual postponement of fixed seasonal festivals, such as Losar and Tsagaan Sar, which move progressively later into the spring (cf. Figure~\ref{fig:losar}).

To contextualize this error, it is instructive to compare the Tibetan rule with other historical intercalation cycles. The Metonic cycle, used in other lunisolar systems, equates 235 lunations to 19 years, implying a year length that drifts by only +0.0868 days per 19 years, or roughly 0.46 days per century—significantly more stable than the Tibetan parameter. Another theoretical alternative, the "168/163 rule," forces 5 leap months every 168 lunations, resulting in a drift of -0.81 days per 163 years, cf. Example~\ref{rem:leap-month-rules-drift}. The Tibetan choice of the 65-year cycle, while culturally entrenched, represents a suboptimal approximation of the tropical year, prioritizing arithmetic simplicity over long-term seasonal fidelity.

For addressing seasonal drift, we consider two complementary design directions.
One is \emph{improved arithmetic intercalation}: keep the traditional idea of a fixed periodic leap-month skeleton,
but replace the $67/65$ commensurability by a more accurate rational approximation to the lunations-per-year ratio.
The other is \emph{dynamic intercalation}: abandon a fixed cycle and determine month structure directly from modeled
solar sign (or definition-point) crossings, as in a fully astronomical scheme.

\subsubsection{Arithmetic intercalation cycles: rational design space}
\label{ss:arith-intercalation}

The arithmetic approach keeps a fixed periodic leap-month skeleton, but replaces the effective
lunations-per-year ratio induced by the traditional $67/65$ scheme by a more accurate rational fit.
Concretely, the classical $67/65$ commensurability is a relation at the month scale, 
and it implies the year-scale ratio
\[
R_{\mathrm{trad}}=\frac{12\cdot 67}{65}=\frac{804}{65}.
\]
In a reform we instead choose a rational year-scale constant
\[
R=\frac{p}{q}\;\approx\;\frac{Y_{\mathrm{trop}}}{S_{\mathrm{syn}}}
\qquad\text{(lunations per tropical year),}
\]
so that over $q$ years there are exactly $p$ lunations and hence $p-12q$ leap months.

If we assume that the approximation $m_1\approx S_{\mathrm{syn}}$ is already sufficiently accurate, and ask only how the choice of $R$ affects seasonal drift, 
then the modeled year length is
\[
Y_{\mathrm{model}}\approx RS_{\mathrm{syn}},
\qquad\text{so}\qquad
Y_{\mathrm{model}}-Y_{\mathrm{trop}}
\approx \Bigl(\frac{p}{q}\Bigr)S_{\mathrm{syn}}-Y_{\mathrm{trop}}.
\]
We record four convenient rational approximants to the benchmark ratio $Y_{\mathrm{trop}}/S_{\mathrm{syn}}$,
with $S_{\mathrm{syn}}\approx 29.53058885$\,d and $Y_{\mathrm{trop}}\approx 365.2421897$\,d, 
chosen from the continued-fraction approximation ladder (including intermediate approximants).

\begin{enumerate}\itemsep2pt
\item \emph{334-year cycle.}
A close rational alternative to the traditional $67/65$ scheme is ${1377}/{1336}$, corresponding to
\[
R=\frac{12\cdot 1377}{1336}=\frac{4131}{334}.
\]
This is a $334$-year cycle with $4131$ lunations, hence $4131-12\cdot 334=123$ leap months.
With the benchmark $S_{\mathrm{syn}}$,
\[
\frac{4131}{334}\cdot S_{\mathrm{syn}}\approx 365.242104\ \text{d},
\qquad
Y_{\mathrm{model}}-Y_{\mathrm{trop}}\approx -7.4\ \text{s/yr},
\]
i.e.\ about $-2.06$ hours per millennium.
\item \emph{353-year cycle.}
The rational
\[
R=\frac{4366}{353}
\]
gives a $353$-year cycle with $4366-12\cdot 353=130$ leap months.
Numerically,
\[
\frac{4366}{353}\cdot S_{\mathrm{syn}}\approx 365.242354\ \text{d},
\qquad
Y_{\mathrm{model}}-Y_{\mathrm{trop}}\approx +14.2\ \text{s/yr},
\]
i.e.\ about $+3.94$ hours per millennium.
\item \emph{687-year cycle.}
The rational
\[
R=\frac{8497}{687}
\]
gives a $687$-year cycle with $8497-12\cdot 687=253$ leap months.
Numerically,
\[
\frac{8497}{687}\cdot S_{\mathrm{syn}}\approx 365.242233\ \text{d},
\qquad
Y_{\mathrm{model}}-Y_{\mathrm{trop}}\approx +3.7\ \text{s/yr},
\]
i.e.\ about $+1.03$ hour per millennium.
\item \emph{1021-year cycle.}
The rational
\[
R=\frac{12628}{1021}
\]
gives a $1021$-year cycle with $12628-12\cdot1021=376$ leap months.
Numerically,
\[
\frac{12628}{1021}\cdot S_{\mathrm{syn}}\approx 365.242190\ \text{d},
\qquad
Y_{\mathrm{model}}-Y_{\mathrm{trop}}\approx +0.03\ \text{s/yr},
\]
i.e.\ about $+26$ seconds per millennium.
\end{enumerate}

\noindent
The benchmark constants used above are not literally constant: the tropical year decreases slowly in time (of order a few seconds per millennium), while the mean synodic month varies more weakly on the same horizons.  Consequently, once an arithmetic cycle is tuned to produce drift at the level of only a few tens of seconds per millennium, the ranking between such cycles depends mildly on the reference epoch and on the precise convention used for the benchmark mean values.  
For this reason, the $1021$-year cycle is best interpreted as ``near-zero drift'' rather than as a qualitatively new regime.  From a design standpoint, the shorter $334$-year cycle already captures most of the gain: 
it captures most of the achievable reduction in drift while keeping the cycle length and bookkeeping overhead moderate; the longer cycles primarily refine the residual drift within a range where secular variability starts to matter.

\subsubsection{Dynamic intercalation}
\label{sss:dynamic-intercal}

A conceptually clean way to eliminate seasonal drift is to abandon a fixed arithmetic intercalation
cycle and instead \emph{determine month structure from solar motion itself}.  In fact, the traditional
``\emph{sgang} rule'' is naturally dynamical in form: fix twelve definition points $\sgang_M$ on the ecliptic,
and label each lunation by the definition point crossed by the Sun during that lunation; if no definition
point is crossed, the lunation is intercalary.  In the language of \S\ref{subsec:dual}, this is a
\emph{containment} rule: the foreground intervals are lunations (new moon to new moon), the background
posts are the solar crossings of the $\sgang_M$, and the month label is determined by which crossing(s) occur
during the lunation \cite{janson,henning,schuh}.

The essential difference from \S\ref{ss:arith-intercalation} is what is treated as fixed \emph{data}.
In a purely arithmetic calendar one hard-codes a commensurability (e.g.\ the $67/65$ scheme together with a
trigger test), so the leap-month skeleton is determined without ever consulting a solar longitude.  In a
dynamical reform, one instead fixes (i) a lunar model (how ``new moon'' is computed), (ii) a solar model
(mean Sun or true Sun), (iii) a choice of definition-point longitudes $\sgang_M$, and (iv) a precise convention
for boundary/degenerate cases; the intercalation pattern is then \emph{recomputed} from these ingredients.

To make the phrase ``crosses a definition point between new moons'' operational, one should treat a crossing
as an \emph{interval event}, not as a point-sampling convention.  Let $t_k$ be successive new-moon instants in
the chosen lunar model, and let $\lambda_\sun(t)$ be the chosen solar longitude.  Define
\[
L_{\mathrm{con}}(k)
:=\Bigl\{\,M:\ \exists\,t\in(t_k,t_{k+1}]\ \text{such that}\ \lambda_\sun(t)\ \text{crosses}\ \sgang_M\,\Bigr\}.
\]
We use a right-closed interval so that a crossing exactly at $t_{k+1}$ is assigned to the lunation that ends
at $t_{k+1}$; any consistent tie-breaking convention would do, but it must be stated.

The regular ``leap-month'' phenomenon corresponds to the empty-set case \(L_{\mathrm{con}}(k)=\varnothing\): no definition point is crossed during the lunation, so the month is unambiguously the extra one.  What is \emph{not} automatic, however, is its name: a convention is still required to decide whether the repeated label is inherited from the preceding or the following definition point, cf.\ Remark~\ref{rem:leap-naming} and \S\ref{ss:sgang}.  In the classical mean-Sun setting, with uniform solar motion and a uniform mean-lunation skeleton, the containment set is constrained by construction and typically satisfies
\(|L_{\mathrm{con}}(k)|\in\{0,1\}\).
Thus skipped labels do not occur, and the leap pattern becomes perfectly regular; indeed, the intercalation index \(\ix\) may be regarded as a compact arithmetic encoding of this crossing test.

A genuinely new regime begins when one replaces the mean Sun by a \emph{true} (non-uniform) solar longitude and/or adopts a lunar model with variable lunation length. Then \(|L_{\mathrm{con}}(k)|\) can vary with \(k\), and in rare boundary situations one may encounter \(|L_{\mathrm{con}}(k)|=2\), meaning that two consecutive definition points fall within a single lunation. 
This is the containment version of a skipped-month phenomenon: one interval is forced to carry only one of the two available labels, while the other label disappears from the month sequence. Accordingly, a fully specified reform must supplement the astronomical model with an explicit naming rule for this case as well, namely a deterministic convention deciding which of the two labels is retained by the lunation and which is treated as the skipped one, cf.\ \S\ref{ss:sgang}. The key point is that improved seasonal fidelity is obtained only by replacing the short periodic arithmetic skeleton with a dynamical one tied to a chosen astronomical model, together with explicit conventions for handling the resulting edge cases.

A dynamical formulation also clarifies what it means to ``restore the traditional seasonal anchors.''  Once the month
structure is defined by actual solar crossings, the longitudes $\sgang_M$ are no longer merely an interpretive gloss on a
congruence table: they become parameters of the calendar.  
In particular, the freedom to choose an overall longitude offset is not unconstrained: preserving a given tradition’s intercalation behavior restricts 
$\sgang_0$ to an admissible interval. We record these admissible $\sgang_0$-ranges for the principal traditions in Remark~\ref{r:sgang-trad}; see also \cite[Appendix C]{janson}.

Finally, dynamic intercalation makes the \emph{choice of coordinate frame} operational rather than interpretive.
The $\sgang_M$ must be interpreted as longitudes in some frame: if taken sidereally (star-fixed), their seasonal meaning
drifts under precession, whereas if taken tropically (equinox-fixed), they remain season-anchored by construction.
Thus, as soon as intercalation is defined by actual sign/definition-point crossings rather than by a frozen congruence
table, ``sidereal vs.\ tropical'' becomes part of the calendar’s mathematical specification; we turn to this design
choice next.

\subsection{Coordinate frame as a design choice}
\label{ss:frame-choice}

The phrase ``sidereal Tibetan calendar'' can mean two different things, depending on which parts of the
pipeline are treated as fixed data.  In the classical arithmetic calendar, the leap-month pattern is
hard-coded by a periodic congruence (the $67/65$ count with a trigger test), and day numbering is governed
by elongation computed in a common frame.  In that regime, changing the longitude origin is essentially a
relabeling: it does not alter the intercalation pattern, and it cancels out of elongation-based day models.

By contrast, once month structure is defined dynamically by \emph{solar crossings of definition points}
(\S\ref{sss:dynamic-intercal}), the coordinate frame becomes part of the specification: the definition points
must be interpreted as longitudes in some frame.  If they are taken sidereally (star-fixed), their seasonal
meaning drifts under precession; if they are taken tropically (equinox-fixed), the month markers remain
season-anchored by construction.

The traditional Tibetan presentation follows its Indian progenitors in using a \emph{sidereal} zodiac, anchored
to the fixed stars rather than to the equinoxes.  The seasons, however, are governed by the \emph{tropical} year,
defined by the Sun's return to the vernal equinox.  Because of axial precession, the equinox drifts westward
through the sidereal zodiac at roughly $1^\circ$ per $72$ years, so a sidereally anchored set of month markers
slowly loses its seasonal meaning.  At present the accumulated precessional offset is about $24^\circ$.  Janson
further observes that at the astronomical vernal equinox (tropical longitude $0^\circ$) the calendar's computed
mean solar longitude is roughly $36^\circ$ smaller \cite{janson}; one may view this as the precessional drift
($\approx 24^\circ$) together with an additional internal lag ($\approx 12^\circ$) coming from the mean-Sun model.

For the \emph{frozen arithmetic} calendar this mismatch is largely interpretive rather than operational.
The leap-month pattern is fixed \emph{once and for all} by the congruence rule and is not recomputed from a solar
longitude tied to the moving equinox.  Likewise, insofar as day calculation depends only on elongation
\[
D(t)=\lambda_{\moon}(t)-\lambda_{\sun}(t),
\]
a global change of longitude origin cancels out of $D(t)$ provided both $\lambda_{\sun}$ and $\lambda_{\moon}$
are computed consistently in the same frame.

Nevertheless, even in an elongation-based pipeline it is good practice to specify (and, in a reform, to update)
\emph{both} $\lambda_{\sun}$ and $\lambda_{\moon}$ to the chosen frame, rather than treating the frame choice as
purely rhetorical.  This avoids hidden mixing of conventions and makes later extensions straightforward; in
particular, any step that uses $\lambda_{\sun}$ or $\lambda_{\moon}$ separately (e.g.\ sunrise geometry or a
true-Sun month rule) immediately forces an explicit and internally consistent frame choice.

The sidereal/tropical choice becomes genuinely substantive as soon as one changes the \emph{logical basis} of month
structure.  If one abandons a fixed intercalation table and instead determines months from (mean or true) solar motion
relative to explicit $30^\circ$ boundaries, then one must decide what those boundaries mean: sidereal boundaries follow
the stars, while tropical boundaries follow the seasons.  In such a dynamical setting, precession is no longer inert;
it enters the definition of the month structure itself.

If one aims for long-term seasonal alignment, it is therefore natural to adopt a \emph{tropical} framework, re-anchoring the definition points (\emph{sgang}) to seasonally meaningful longitudes. Because this transition requires mapping a historically sidereal framework onto a tropical one, there are several justifiable choices for the new anchor point, depending on which traditional characteristic a reform prioritizes. For instance, to strictly preserve the historical ``offset-from-sign-boundaries'' structure, one might place the sequence $8^\circ$ into each tropical sign (e.g., anchoring the first point at $308^\circ$). Alternatively, to better synchronize the calendar's mathematical start with specific seasonal boundaries---as proposed in the specific reform modules later in this paper---one could anchor the sequence at $336^\circ$ (i.e., $6^\circ$ into Pisces). By treating this base longitude as a configurable parameter, a reform can maintain the classical incidence rules of the $\textit{sgang}$-sequence while permanently arresting the slow precessional loss of seasonal meaning.

\subsection{Mean motion constants}
\label{ss:mean-motion}

In the siddh\=anta pipeline the \emph{mean} lunar-day boundary times and the mean solar longitude are affine
functions of the lunar day $d\in\{0,\dots,30\}$ and the true-month index $n\in\mathbb Z$:
\begin{equation}\label{eq:mean-affine}
\md(d,n)=m_0+n\,m_1+d\,m_2,
\qquad
\ms(d,n)\equiv s_0+n\,s_1+d\,s_2 \pmod 1,
\end{equation}
where mean date $\md$ is measured in Julian days, while the mean sun $\ms$ is measured in revolutions.
Traditionally one imposes the internal-consistency constraints
\[
m_2=\frac{m_1}{30},\qquad s_2=\frac{s_1}{30},
\]
encoding the fact that a lunation is divided into $30$ equal lunar-day-steps, and that the mean Sun advances
uniformly across those $30$ steps.

As Janson noted in his foundational analysis of the modern calendar \cite{janson}, the calculated new moons consistently arrive several hours earlier than the true astronomical new moons. At a lunar separation rate of $\approx 0.5^\circ$ per hour, this temporal offset translates directly into a systematic angular error in the calculated elongation. As our macroscopic evaluation demonstrates (\S\ref{sec:offset_distribution} and \S\ref{sec:lunar_engine_drift}), this is not merely a modern artifact, but a persistent, long-term secular drift. Because the calendar's equations of anomaly only oscillate around a central baseline, these macroscopic errors cannot be attributed to periodic trigonometric simplifications (such as the omission of evection). Instead, they are the direct mathematical consequence of this affine framework: rooted entirely in the inherited inaccuracies of the foundational epoch constants ($m_0, s_0$) and the steady accumulation of error from the linear mean motion rates ($m_1, s_1$).

\subsubsection{Lunar mean motion}

For the benchmark mean synodic month $S_{\mathrm{syn}}$, we refer to the modern values detailed in Appendix~\ref{app:modern-periods}. While the traditional Tibetan constant $m_1 = \frac{167025}{5656}$ provides excellent precision (sub-second per lunation), it is not derived from the continued-fraction expansion of the modern benchmark. 

If we treat $m_1$ as a tunable rational parameter, we can derive optimal approximants from the continued-fraction convergents of $S_{\mathrm{syn}}$. The sequence begins:
\[
\frac{59}{2},\ \frac{443}{15},\ \frac{502}{17},\ \frac{1447}{49},\ \frac{25101}{850},\ \frac{51649}{1749},\ \frac{283346}{9595},\ \dots
\]
These provide the ``best rational approximants'' at their respective denominator scales. For the purposes of reform, we compare the traditional constant against three progressively more accurate modern candidates (all assuming $m_2 = m_1/30$).

\begin{center}
\renewcommand{\arraystretch}{1.3}
\begin{tabular}{@{}l l l l@{}}
\toprule
{Parameter} ${m_1}$ & {Description} & {Error} ($\Delta t/\text{lun}$) & {Drift} ($/1000\text{y}$) \\
\midrule
$\frac{167025}{5656}$ & Traditional & $-0.161$ s & $-33.2$ min \\
$\frac{51649}{1749}$ & Low-denom CF & $+0.005$ s & $+0.98$ min \\
$\frac{283346}{9595}$ & High-accuracy CF & $-0.0004$ s & $-5.0$ s \\
$\frac{2756710}{93351}$ & Near-convergent & $-0.00001$ s & $-0.19$ s \\
\bottomrule
\end{tabular}
\end{center}

The analysis shows that there is nothing intrinsically wrong with retaining the traditional $m_1$ for an \emph{intercalation-only} reform; a drift of 30 minutes over a millennium is manageable for approximate civil dating. However, if the goal is to render the residual mean-elongation drift negligible while maintaining a fully rational structure, the convergent $\frac{283346}{9595}$ offers a particularly clean upgrade, providing ``small denominator'' efficiency with ``dynamical timescale'' accuracy.

\subsubsection{Solar mean motion and intercalation consistency}

To ensure the kinematic model remains internally consistent with the chosen arithmetic intercalation scheme, the mean solar motion must be derived directly from the intercalation ratio $R$. Defining $R$ as the number of mean lunations per tropical year:
\begin{equation}
R = \frac{p}{q} \approx \frac{Y_{\mathrm{trop}}}{S_{\mathrm{syn}}} \qquad \text{(lunations per year)},
\end{equation}
the mean Sun must advance by exactly $1/R$ revolutions per lunation. Consequently, the solar rate constants $s_1$ (per lunation) and $s_2$ (per lunar-day) are fixed by the relations:
\begin{equation}
s_1 = \frac{1}{R} = \frac{q}{p}, \qquad s_2 = \frac{s_1}{30} = \frac{q}{30p}.
\end{equation}
To implement the 334-year arithmetic reform proposed in Section \ref{ss:arith-intercalation}, these solar parameters must be updated from the classical baseline. In the traditional 65-year cycle, the intercalation ratio $R = 804/65$ yields a monthly solar advance of $s_1 = 65/804$ revolutions per lunation and a daily advance of $s_2 = 13/4824$ revolutions per day. 

In contrast, the highly accurate 334-year cycle operates on a ratio of $R = 4131/334$. Applying the relations above, this requires the mean Sun to advance by $s_1 = 334/4131$ revolutions per lunation. Dividing this by 30 (and reducing the fraction by the greatest common divisor) gives a daily solar advance of $s_2 = 167/61965$ revolutions per day. Because this 334-year cycle offers a massive improvement in seasonal drift stability while maintaining solar denominators that remain computationally manageable, it serves as the arithmetic foundation for the rational reform modules presented in Section 5.

\subsubsection{Epoch offsets}

Once $(m_1,m_2,s_1,s_2)$ are fixed, the offsets $(m_0,s_0)$ are simply initial conditions specifying the absolute placement of the model on the Julian-day line and the initial solar phase. Concretely, if one declares that the reference month has true-month index $n=0$ and uses $d=0$ at the chosen reference boundary, then \eqref{eq:mean-affine} forces
\[
m_0=\md(0,0),\qquad s_0\equiv \ms(0,0)\pmod 1.
\]
For the epoch used throughout our reform proposals, we select E1987, as it marks the start of the current 60-year cycle (Rabjung 17). At this epoch, the classical Phugpa parameters are (as per Table~\ref{tab:epoch-constants}):
\[
m_0=2446914+\frac{135}{707},\qquad s_0\equiv 0\pmod 1.
\]
In a reform where one updates the mean rates $m_1$ and/or $s_1$, one may choose to preserve these traditional $(m_0,s_0)$ constants to maintain the same classical absolute anchor instant. However, if a different absolute tuning is desired to align with true physical events, one must re-fit $(m_0,s_0)$ to a precise astronomical baseline. 

For the E1987 epoch, an ``astronomically pure'' calibration (derived directly from high-precision modern ephemerides, such as Meeus) yields the precise fractional values:
\[
m_0 = \frac{244691379521131}{100000000} = 2446913.79521131, \qquad s_0 = \frac{128634}{1296000}.
\]
Here, $m_0$ is shown alongside its exact decimal Julian date, while the fractional denominator of $s_0$ corresponds to the $1,296,000$ arcseconds in a full $360^\circ$ circle (meaning the initial solar phase is placed at exactly $128,634$ arcseconds). 

For computational efficiency, these massive exact fractions can be replaced by closely ``optimized'' values possessing much smaller prime denominators, without sacrificing practical precision:
\[
m_0 = \frac{160957989449}{65780}, \qquad s_0 = \frac{3609}{36361}.
\]
By substituting these optimized constants into the E1987 anchor point, a reform achieves strict astronomical alignment at the start of the contemporary calendar cycle while keeping individual mathematical denominators visually and computationally minimal.

However, in certain programming contexts, the ``pure'' unreduced fractions may actually be preferable. If a software implementation relies on exact rational arithmetic and combines multiple constants through successive iterations (such as within an iterative solver), mixing disparate prime denominators will cause the least common multiple (LCM) of the system to explode rapidly. In these scenarios, adhering to standardized, highly composite denominators (such as $10^8$ and $1,296,000$) prevents this LCM explosion and ensures that integer boundaries are not breached during complex mathematical operations.

\subsection{Mean anomalistic periods and phases}
\label{ss:mean-anomaly}

Mean-motion constants control the \emph{average} passage of lunar and solar longitudes, but civil-day labels ultimately depend on \emph{event times} (such as lunar day boundaries and new moons), which are highly sensitive to the non-uniform orbital speeds of the Sun and Moon. In the traditional Tibetan siddh\=anta pipeline, this non-uniformity enters through table-driven ``equations'' evaluated at specific anomaly arguments. For reform design, it is essential to separate two independent computational choices:
\begin{enumerate}\itemsep2pt
    \item the \emph{phase evolution} of the anomaly arguments (perihelion for the Sun; perigee for the Moon), representing their rates relative to time; and
    \item the \emph{shape and content} of the structural corrections (for instance, utilizing a first-harmonic ``equation of center'' versus expanding to a short harmonic series that includes dominant lunar perturbations).
\end{enumerate}

The classical first-anomaly scheme effectively locks the \emph{solar} anomaly phase to the calendar's mean Sun. Modern astronomy, by contrast, physically distinguishes the tropical (seasonal) year from the anomalistic (perihelion-to-perihelion) year, and likewise distinguishes the synodic month (new-moon to new-moon) from the anomalistic month (perigee to perigee). It is important to emphasize that the traditional Tibetan system is \emph{not} oblivious to the latter distinction: it explicitly introduces an independent lunar anomaly phase with its own rate constant \(a_1\), thereby implicitly separating the perigee cycle from the mean lunation count. 

The structural issues lie elsewhere. First, on the solar side, the anomaly argument is hard-wired to the mean Sun, forcing the anomalistic year to strictly follow the calendar's mean seasonal year. Second, on the lunar side, small but nonzero rate errors persist alongside the outright omission of the largest non-Keplerian \(D\)-coupled perturbations.%
\footnote{Here \(D\) denotes the mean elongation \(D=\lambda_{\moon}-\lambda_{\sun}\).} 
For the principal traditions, the adopted value of \(a_1\) implies an anomalistic period that differs from the modern mean anomalistic month by about \(3.6\) seconds. This rate error accumulates into a phase discrepancy of roughly \(0.75^\circ\) per century, contributing about 1 to 2 hours of timing error along the anomalistic cycle over a century \cite{janson,henning}. On the solar side, the anomalistic--tropical mismatch also produces a slow phase slip, though in practice this contribution is largely swamped by the dominant seasonal drift generated by the mean-sun/intercalation layer (as discussed in \S\ref{ss:seasonal-drift}).

However, the most severe timing discrepancies stem from the structural limitations of the classical first-anomaly lunar scheme itself. By neglecting evection—a solar-driven modulation of the Moon's elliptic anomaly with a longitude amplitude of up to \(\pm 1.274^\circ\)—the model permits positional errors that translate into timing shifts of approximately 2.6 hours near syzygies, given that relative elongation changes at about \(0.5^\circ\) per hour. Similarly, the omission of variation—a \(\sin(2D)\)-type term tied to elongation with an amplitude of up to \(\pm 0.66^\circ\)—adds a further 1.4 hours of potential error \cite{Meeus}. Consequently, a reform targeting robust, few-minute accuracy over centuries must inevitably decouple the anomaly phases from the mean Sun and introduce at least these dominant lunar \(D\)-coupled perturbations.

\subsubsection{Solar anomaly}
\label{sss:solar-anomaly}

Instead of defining the solar anomaly argument as a fixed shift of the mean Sun, a modern reform introduces an independent phase:
\[
  A_\sun(d,n)\equiv \san_0+n\san_1+d\san_2 \pmod 1,\qquad
  \san_2=\frac{\san_1}{30},
\]
where \(\san_1\) is chosen so that the implied anomalistic year
\[
  Y_{\rm anom}=\frac{m_1}{\san_1}
\]
matches a modern anomalistic-year benchmark at the chosen reference epoch, \emph{without} forcing \(Y_{\rm anom}\) to equal the calendar's tropical-year model.

A convenient, \emph{compact} rational choice is
\[
\san_1=\frac{122}{1509},\qquad
\san_2=\frac{61}{22635},
\]
while a higher-accuracy, yet still moderate, option is
\[
\san_1=\frac{1689}{20891},\qquad
\san_2=\frac{563}{626730}.
\]

We can evaluate these approximants against the modern benchmark for the solar anomaly rate (roughly \(0.0808482\) turns per lunation). For comparison, the traditional models effectively use the mean solar motion \(s_1 = 65/804\) for the anomaly, which ignores the slow drift of the perihelion. 

\begin{center}
\renewcommand{\arraystretch}{1.3}
\begin{tabular}{llll}
\toprule
Parameter \(\san_1\) & Description & Error (\(\Delta\theta\)/lun) & Drift (/1000y) \\
\midrule
\(\frac{65}{804}\) & Traditional (\(s_1\)) & \(-3.12''\) & \(-10.7^\circ\) \\
\(\frac{122}{1509}\) & Low-denom CF & \(+0.07''\) & \(+0.26^\circ\) \\
\(\frac{1689}{20891}\) & High-accuracy CF & \(+0.03''\) & \(+0.12^\circ\) \\
\bottomrule
\end{tabular}
\end{center}

The epoch phase \(\san_0\) is then fixed by one explicit condition, such as requiring \(A_\sun=0\) at a chosen perihelion reference time, or matching a tabulated perihelion longitude at the epoch. For the E1987 epoch, an ``astronomically pure'' calibration derived from modern ephemerides places this initial anomaly phase at exactly $406,845$ arcseconds:
\[
\san_0 = \frac{406845}{1296000}.
\]
Unlike the massive fractions required for the mean solar epoch, this exact phase constant simplifies cleanly to a manageable rational fraction ($9041/28800$). Therefore, computational implementations can safely utilize this exact value without requiring further rational approximation to prevent denominator overflow.

\subsubsection{Lunar anomaly}
\label{sss:lunar-anomaly}

Similarly, the independent lunar anomaly phase is defined as:
\[
A_\moon(d,n)\equiv a_0+n a_1+d a_2\pmod 1.
\]
A convenient and internally consistent parameterization enforces
\[
a_2=\frac{1+a_1}{30},
\]
so that over one lunation, the anomaly advances by \(1+a_1\) turns. The mathematical design target is therefore:
\[
1+a_1 \approx \frac{m_1}{S_{\rm anom}},
\qquad\text{equivalently}\qquad
a_1 \approx \frac{m_1}{S_{\rm anom}}-1,
\]
where \(S_{\rm anom}\) is the physically accurate anomalistic month. Two practical rational choices include a very small denominator version:
\[
a_1=\frac{18}{251},\qquad
a_2=\frac{269}{7530},
\]
and a high-accuracy version with still-manageable sizes:
\[
a_1=\frac{503}{7014},\qquad
a_2=\frac{7517}{210420}.
\]
An optimal rational ``upgrade'' that balances a modest denominator with high precision is:
\[
  a_1=\frac{4583}{63907},\qquad a_2=\frac{2283}{63907}.
\]

The traditional Tibetan parameter is \(a_1 = 253/3528\). We compare its performance against the modern rational candidates, targeting the benchmark lunar anomaly rate of approximately \(0.0717136\) turns per lunation:

\begin{center}
\renewcommand{\arraystretch}{1.3}
\begin{tabular}{llll}
\toprule
Parameter \(a_1\) & Description & Error (\(\Delta\theta\)/lun) & Drift (/1000y) \\
\midrule
\(\frac{253}{3528}\) & Traditional & \(-2.03''\) & \(-7.0^\circ\) \\
\(\frac{18}{251}\) & Low-denom CF & \(-0.57''\) & \(-2.0^\circ\) \\
\(\frac{503}{7014}\) & High-accuracy CF & \(+0.16''\) & \(+0.57^\circ\) \\
\(\frac{4583}{63907}\) & Optimal upgrade & \(-0.01''\) & \(-0.02^\circ\) \\
\bottomrule
\end{tabular}
\end{center}

The implied anomalistic-month constant is thus:
\[
S_{\rm anom}^{\rm(model)}=\frac{m_2}{a_2}=\frac{m_1}{1+a_1}.
\]
As with \(\san_0\), the epoch phase \(a_0\) is fixed by a single observational condition, such as matching a reference perigee time or longitude at the chosen epoch. For the E1987 epoch, an ``astronomically pure'' calibration derived from modern lunar ephemerides places this initial lunar anomaly phase at exactly $389,900$ arcseconds:
\[
a_0 = \frac{389900}{1296000}.
\]
Just like the solar anomaly phase, this exact fraction simplifies cleanly to a manageable rational number ($3899/12960$). Therefore, computational implementations can utilize this exact astronomical value directly, without requiring a secondary rational approximation to keep denominators small.

\subsection{Beyond the first-anomaly model}
\label{ss:beyond-mean-anomaly}

Once a calendar design moves beyond few-term anomaly models, the reform problem naturally splits into two distinct components: a kinematic model for the true longitudes \(\lambda_\odot(t)\) and \(\lambda_M(t)\), and an inversion method to solve for event times, such as lunar day boundaries.

It is conceptually useful to separate each longitude into an ``elliptic core'' plus additional periodic perturbations:
\[
\lambda(t)=\lambda_{\rm mean}(t) + \Delta\lambda_{\rm ell}(t) + \Delta\lambda_{\rm pert}(t).
\]
Here, \(\Delta\lambda_{\rm ell}\) represents the primary correction obtained by converting mean anomaly to true anomaly (the Keplerian eccentricity effect), while \(\Delta\lambda_{\rm pert}\) collects the dominant additional lunar terms driven by the Sun, alongside any further harmonics included by the reform.

In principle, one can evaluate \(\Delta\lambda_{\rm ell}\) with arbitrary precision by solving Kepler's equation \(E-e\sin E=M\), which yields the mean-to-true anomaly map \(M \mapsto \nu\). However, computational effort here yields diminishing returns, as the perturbation term \(\Delta\lambda_{\rm pert}\) remains an approximation. Furthermore, the overall pipeline ultimately requires solving an equation of the form \(D(t)=12^\circ k\) to determine lunar day boundaries. It is therefore highly efficient to approximate the combined term \(\Delta\lambda_{\rm ell}+\Delta\lambda_{\rm pert}\) directly via a truncated series, and then solve \(D(t)=12^\circ k\) globally using a few iterations of a fixed-point method. This methodology is a direct, natural extension of the classical Tibetan calendar, which classically evaluates a single term for \(\Delta\lambda_{\rm ell}\) (assuming \(\Delta\lambda_{\rm pert}=0\)) and relies on a single iteration of a basic fixed-point method to solve for the boundaries.

A clean strategy to construct the total perturbation model \(\Delta\lambda_{\rm ell}+\Delta\lambda_{\rm pert}\) is to adopt a published, standardized truncated series and state the truncation rule explicitly. Practical solar and lunar longitude algorithms take the form of finite trigonometric sums, as detailed in Appendices~\ref{app:solar-series} and \ref{app:lunar-series} (see also \cite{Meeus}). The amplitude cutoff for these series can be rigorously determined by translating angular residuals into timing residuals. If the tolerated error budget for a lunar day boundary is \(\tau\) hours and the elongation speed is \(\omega\approx 0.5^\circ/{\rm hr}\), then it suffices to bound the residual elongation error by \(\omega\tau\) degrees. For instance, a strict tolerance of \(\tau=5\) minutes requires an angular precision of \(\omega\tau\approx 0.042^\circ\), which immediately dictates the necessary number of harmonic terms to retain.

Finally, to ensure absolute cross-platform reproducibility, the reform must explicitly specify the numerical approximations used to evaluate the transcendental functions (such as sine) appearing in these series, as well as the exact mechanics of the iterative equation solving. These foundational numerical protocols are addressed in detail in \S\ref{ss:numerical} and Appendix~\ref{app:numerical}. The following subsections detail the physical significance and mathematical formulation of the five largest structural corrections required for a modern lunar series.

\subsubsection{Evection and variation}
\label{sss:evection-variation}

The traditional first-anomaly model successfully captures the primary elliptic effect, known as the equation of center. However, it systematically omits the dominant \(D\)-coupled perturbations driven by the solar gravitational field, most notably evection and variation \cite{Meeus}. A minimal harmonic upgrade incorporates these effects by applying a targeted correction to the lunar longitude. Specifically, this correction takes the form:
\begin{equation}
    \Delta\lambda_M = E\sin(2D-M') + V\sin(2D),
\end{equation}
where \(D = \lambda_\moon - \lambda_\sun\) represents the mean elongation and \(M'\) is the lunar mean anomaly. Using standard peak amplitudes of \(E \approx 1.274^\circ\) and \(V \approx 0.66^\circ\), these corrections can be securely translated into explicit rationals in \emph{turns} to support purely arithmetic fixed-point implementations. 
We refer to Table~\ref{tab:lunar-primary} in Appendix~\ref{app:lunar-series} for more accurate constants.

\subsubsection{Secondary inequalities and geometric corrections}
\label{sss:secondary-inequalities}

Beyond the primary solar perturbations of evection and variation, a robust calendrical model must address the next tier of inequalities to suppress residual errors below a ten-minute threshold. The largest of these is the annual equation, an effect driven by the seasonal variation in the Earth-Sun distance, which modulates the solar gravitational pull on the lunar orbit. Peaking at an amplitude of approximately \(0.186^\circ \sin(M)\), its omission introduces a timing shift of roughly 22 minutes into the calendar. Additionally, the second elliptic term, arising from the higher-order expansion of the equation of center with an amplitude of roughly \(0.214^\circ \sin(2M')\), contributes a periodic timing error of about 25 minutes if neglected. 

Finally, a precision calendar must account for the reduction to the ecliptic. Unlike the previous dynamical perturbations, this is a purely geometric projection of the Moon's motion from its tilted orbital plane onto the ecliptic. Modeled generally as \(-k \sin(2(L-F))\), where \(F\) is the argument of latitude, it possesses a peak amplitude of \(0.114^\circ\) and generates a rhythmic error of roughly 13 minutes. 

To accurately compute this reduction to the ecliptic, an independent phase for the argument of latitude must be tracked:
\[
  A_F(d,n)\equiv f_0+n f_1+d f_2\pmod 1.
\]
Consistent with the lunar anomaly, the daily advance is constrained so that the argument completes \(1+f_1\) turns per lunation:
\[
  f_2=\frac{1+f_1}{30}.
\]
The target rate is derived from the modern draconic month \(S_{\rm drac}\) via \(f_1 \approx \frac{m_1}{S_{\rm drac}}-1\). A convenient, low-denominator choice for this fractional phase advance is:
\[
  f_1=\frac{61}{716},\qquad f_2=\frac{259}{7160}.
\]
An optimal rational upgrade offering high precision with a moderate denominator is:
\[
  f_1=\frac{324}{3803},\qquad f_2=\frac{4127}{114090}.
\]

Because traditional day-counting models omit the reduction to the ecliptic, they do not track \(F\) in their primary longitude algorithms. We evaluate the new rational candidates directly against the modern benchmark rate for the argument of latitude (approximately \(1.0851958\) turns per lunation):

\begin{center}
\renewcommand{\arraystretch}{1.3}
\begin{tabular}{llll}
\toprule
Parameter \(f_1\) & Description & Error (\(\Delta\theta\)/lun) & Drift (/1000y) \\
\midrule
\(\frac{61}{716}\) & Low-denom CF & \(-0.41''\) & \(-1.4^\circ\) \\
\(\frac{324}{3803}\) & Optimal upgrade & \(+0.07''\) & \(+0.24^\circ\) \\
\bottomrule
\end{tabular}
\end{center}

As in previous models, the epoch phase \(f_0\) is fixed by a single observational condition, such as matching a reference nodal longitude at the chosen epoch. For the E1987 epoch, an ``astronomically pure'' calibration derived from modern ephemerides places this initial latitude argument phase at exactly $91,591$ arcseconds:
\[
f_0 = \frac{91591}{1296000}.
\]
For computational environments where minimizing prime denominators is prioritized over maintaining standardized rational scales, this exact constant can be replaced by a closely ``optimized'' fraction:
\[
f_0 = \frac{4596}{65033}.
\]
Substituting this optimized fraction into the E1987 anchor allows the calendar to maintain precise nodal alignment while ensuring the mathematical denominators remain tightly bounded.

\begin{remark}
While the traditional calendars ignore these secondary terms entirely, their cumulative neglect can lead to a convergence of errors exceeding four hours. By explicitly incorporating these five primary and secondary inequalities—evection, variation, the annual equation, the second elliptic term, and the reduction to the ecliptic—the residual longitudinal error is suppressed to strictly less than \(0.1^\circ\). This level of precision ensures that syzygy and lunar day boundary timings remain robust to within approximately 12 minutes \cite{Meeus}.
\end{remark}

\subsection{The necessity of geographic specification}
\label{ss:locale}

A defining characteristic of traditional Tibetan systems, such as those derived from the Kālacakra Tantra, is their \emph{location independence}. These calendars operate on abstract arithmetic mean motions that are agnostic to any specific geographical meridian, treating "dawn" merely as a procedural marker to sample the output rather than as a computational input. Historically, the underlying parameters were likely derived from a synthesis of observations across the Indian subcontinent, 
and later, potentially Central Asia, creating a "composite" reference that smoothed over local differences. This approach effectively tuned the system to a virtual average meridian rather than a single specific observatory, occurring under a framework where the strict separation of local versus mean time was not a primary architectural constraint.

However, the moment a reform seeks to interpret "dawn" as a precise physical event rather than a symbolic proxy, geographic dependence becomes a mathematical necessity. Because the Earth rotates at $15^\circ$ of longitude per hour, an observer's local time shifts by four minutes for every degree of deviation from a standard reference. Across the vast longitudinal breadth of the Tibetan cultural sphere and Mongolia (a span of roughly $40^\circ$), this geometric discrepancy exceeds two hours. As established in \S\ref{sec:geo-epoch}, the traditional model successfully absorbed this regional variance for two reasons. First, it possesses a massive built-in temporal safety buffer. Second, the intrinsic error spread of the traditional new moon calculations spans roughly three to four hours. This inherent baseline variance effectively swamps the two-hour geographic discrepancy, rendering localized adjustments practically moot within the classical framework. However, if a modern reform eliminates these classical inaccuracies to track true astronomical dawn with high-precision lunar models, it must actively incorporate geographic longitude. Without a localized longitude correction, the unmasked geographic drift will inevitably cause the systematic misidentification of the lunar day.

To bridge the gap between traditional abstraction and physical reality, a reform must treat longitude $\lambda$ and latitude $\varphi$ as fundamental instance parameters. We distinguish three tiers of implementation, which correspond to the models detailed in Appendix~\ref{app:sunrise}:

\begin{itemize}
    \item \emph{Fixed local proxy:} The day start is defined by a constant time, such as \emph{05:56}. This prioritizes convention while acknowledging the $\sim 4$-minute anticipation of apparent sunrise at the equator due to refraction. It effectively uses a "mean local meridian" but ignores seasonal shifts.
    \item \emph{Geometric local sunrise:} The day start is tied to the \emph{local mean sunrise}, accounting for the observer's latitude $\varphi$ and seasonal solar declination $\delta$. This solves the large errors caused by latitude and seasonality but neglects the $\pm 16$-minute drift of the Sun relative to a uniform clock.
    \item \emph{Astronomical local sunrise:} The day is anchored to the \emph{local true sunrise} using a dynamic ephemeris. This model incorporates the \emph{equation of time} to bridge the gap between apparent solar time and civil mean time, achieving robustness to within seconds.
\end{itemize}

Adopting any of these methods represents a fundamental shift from traditional arithmetic to spatiotemporal specificity. By explicitly choosing a geographic anchor, the reform ensures that the calendar is no longer a detached abstraction, but a model synchronized with the physical horizon of the observer.

\begin{remark}
Beyond the geometric constraints of latitude and axial tilt, the timing of the \emph{visible} sunrise is subject to the chaotic dynamics of the Earth's atmosphere. While standard astronomical models assume a constant refraction lift of $34'$ to account for the Sun appearing above the horizon while it is still geometrically depressed, real-world fluctuations in air temperature, barometric pressure, and humidity cause the actual refraction to vary significantly from day to day.
These stochastic shifts create an inherent ``noise'' in observed timings—often on the order of several minutes—that cannot be eliminated by orbital theory. Consequently, for calendars targeting high precision, the definition of the day must be explicitly decoupled from the \emph{visible} event. Instead, one must adhere strictly to the \emph{standardized} sunrise (calculated with the fixed $34'$ constant), accepting that the calculated civil day will occasionally drift from the observed solar disc by a few minutes due to unmodelable local weather.
\end{remark}

\subsection{Earth’s rotation as a hard limit on precision}
\label{ss:delta-t-barrier}

A fundamental barrier to achieving arbitrary calendrical accuracy over millennial scales is the non-uniformity of the Earth's rotation. While orbital events (such as syzygies) occur in uniform \emph{Terrestrial Time} ($\TT$), horizon-based phenomena (such as sunrises) are anchored to the Earth's physical orientation, measured as \emph{Universal Time} ($\UTone$). The discrepancy between the uniform dynamical time of celestial mechanics and the erratic rotating reference frame of the observer is captured by the parameter $\DeltaT = \TT - \UTone$.

As detailed in Appendix~\ref{app:tt}, the inherent degradation of temporal alignment as we move away from the current epoch is driven by two distinct modes of uncertainty:
\begin{itemize}
    \item \emph{Short-term stochastic noise:} Chaotic fluctuations in the Earth's rotation are driven by geophysical processes, including angular momentum exchange between the core and mantle, seasonal atmospheric mass redistributions, and post-glacial isostatic adjustments. Because these processes are non-deterministic, $\DeltaT$ cannot be calculated purely from first principles. 
    This necessitates the periodic insertion of \emph{leap seconds} to keep \emph{Coordinated Universal Time} ($\UTC$)—the atomic standard used for civil timekeeping—within 0.9 seconds of $\UTone$. 
    Since these are based on observations and announced only months in advance, a guaranteed accuracy of 1 second is physically impossible even for the near future.
    \item \emph{Long-term secular drift:} The Earth's rotation is gradually slowing, primarily due to tidal braking. 
    While this is approximated by a fitted quadratic model \eqref{e:delta-t-quad} from Appendix~\ref{app:tt} (Figure~\ref{fig:Delta-T}), 
    the coefficients themselves are empirical averages with inherent uncertainty. Because this model represents an integrated effect, even a minute uncertainty in the deceleration rate scales with the square of time. This creates a widening ``parabola of uncertainty'': a small precision error in the tidal coefficient today can translate into a divergence of tens of minutes when projected back to the era of the \emph{Kālacakra} or forward into the distant future.
\end{itemize}

\begin{figure}[ht]
\centering
\includegraphics[width=.7\textwidth]{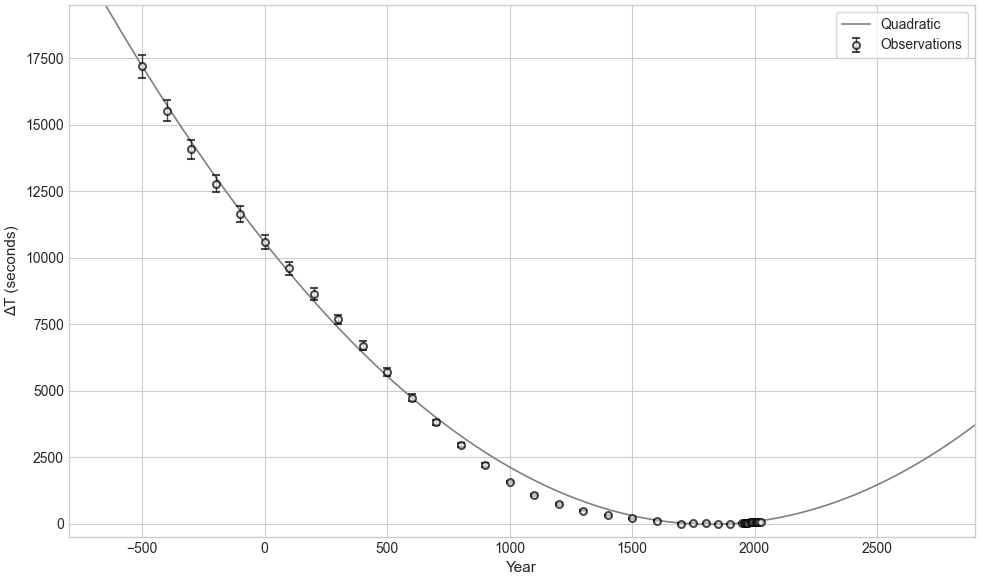}
\caption{Comparison of the quadratic model \eqref{e:delta-t-quad} against historical eclipse records and modern atomic data \cite{Espenak2006,IERS2026}.}
\label{fig:Delta-T}
\end{figure}

Crucially, the uncertainty in $\DeltaT$ does not imply that celestial mechanics themselves are imprecise; in the uniform domain of Terrestrial Time ($\TT$), planetary positions can be calculated with extreme accuracy over vast spans. The fundamental barrier lies strictly in the translation of these precise dynamical events into the observer's erratic local timeline ($\UTone$ or $\UTC$). Any error in estimating the Earth's braking manifests directly as a timing error for every sun-synchronous event. For the purposes of reform, we must acknowledge this ``uncertainty horizon'': while modern models maintain sub-minute precision for the current era, any calendar calculation spanning several centuries remains a high-fidelity estimation rather than a mathematical absolute. One might speculate that in a distant future where humanity expands beyond Earth, $\TT$ could eventually supersede the solar day as the civil standard, rendering the vagaries of Earth's rotation a mere historical curiosity; until then, however, our calendars remain tethered to the unpredictable spin of our planet.

\subsection{Numerical approximation and portability}
\label{ss:numerical}

The choice of numerical technology is as critical to a successful reform as the astronomical
models themselves. In this framework, approximation methods are treated as formal parts of
the specification rather than as mere implementation details. This is essential for true
reproducibility: the same calendric rule should evaluate identically on a low-resource
embedded device, a desktop implementation, and a cloud service. In particular, a modern
calendar standard must specify not only \emph{what} functions are to be evaluated, but also
\emph{how} they are evaluated---whether by exact rational arithmetic, tabular interpolation,
prescribed floating-point approximants, or a fixed-iteration inverse solver.

For rational-style architectures, the natural starting point is an \emph{exact rational} regime built from
piecewise-linear tables and reverse interpolation. This approach is exceptionally easy to
audit, aligns well with traditional computational styles, and completely avoids floating-point
ambiguity. It is especially well suited to models whose accuracy target is already limited by
simplified geometry or by a coarse civil-day trigger. At the same time, its limitations are
clear: linear interpolation introduces visible slope discontinuities, and once one aims at
genuinely few-minute or sub-minute decisions, these discretization effects become part of
the error budget rather than negligible background noise.

A second design constraint is that exact arithmetic by itself is not enough: one must also
control \emph{denominator growth}.
In deep-time evaluations, naive rational secular
terms can cause exact-fraction arithmetic to develop enormous least common multiples,
making continuous evaluation impractical. 
Rational architectures therefore benefit from
synchronizing secular coefficients with the existing harmonic denominator structure of the
calendar, so that exact evaluation remains stable over long horizons rather than exploding
combinatorially. 
This denominator-growth constraint is part of the rationale for the highly
structured rational parameter choices adopted in the exact rational designs considered here.

To reach higher precision while preserving deterministic cross-platform behavior, one naturally
moves from tables to \emph{prescribed minimax polynomials} with stored coefficients.
Angular arguments are reduced by symmetry to a primary interval, and fixed-degree
polynomials are then used for sine- and arctangent-type kernels. This yields smoother
behavior and a nearly uniform approximation error while keeping storage modest. 
For full portability, the coefficients should be frozen in hexadecimal floating-point form, so
that their intended binary64 values become part of the standard itself rather than being left
to host-language decimal parsing.
The concrete approximants and error bounds are recorded in
Appendix~\ref{app:numerical}. The same philosophy extends to auxiliary kernels whenever
bit-level reproducibility matters: if a square root, inverse tangent, or similar subroutine is
part of the calendric logic, its evaluation method should also be standardized rather than
delegated silently to a platform library.

Inverse problems require the same level of standardization as forward evaluation. Sunrise
times and lunar-day boundaries are not obtained by a closed formula, but by solving for the
instant at which a modeled phase quantity reaches a prescribed threshold. For such tasks,
\emph{fixed-iteration solvers} are preferable to tolerance-based loops:
stopping when an error estimate falls below a machine-dependent threshold invites
platform-dependent branch behavior, whereas a prescribed iteration count guarantees that
every conforming implementation executes the same arithmetic path. In the rational and
mid-tier engines, a short Picard iteration is the natural standardized choice; in higher tiers,
one may use a fixed-count Newton or Steffensen refinement, but the count itself must be
part of the specification. The key point is that determinism comes not from the abstract
method name, but from freezing the exact iteration protocol.

Finally, portability requires us to standardize the numbers themselves and the exact rules of
their manipulation. Decimal literals are not enough: even when two systems nominally use
the same constant, the conversion from decimal source text to internal binary format may
depend subtly on compiler, runtime, or host language. 
For floating-point architectures, fundamental constants should therefore preferably be
published in \emph{hexadecimal floating-point notation}, so that their intended bit patterns are
fixed in advance. In principle, decimal literals can also be acceptable when there is only a
single, correctly rounded conversion to IEEE-754 binary64, but hexadecimal notation avoids
placing any burden on decimal-to-binary parsing and makes the intended representation
completely explicit.
Just as importantly, the evaluation order
must be prescribed. Floating-point arithmetic is sensitive to grouping, contraction, and
library-specific optimizations such as fused multiply-add. A portable calendar standard
must therefore fix not only the constants, but also the sequencing of arithmetic operations
and any required rounding conventions, so that a date computed on one IEEE-754 binary64
platform remains identical on another.

In short, the numerical layer is itself part of the calendric contract: some architectures are
best standardized through exact rational semantics, while others require explicitly prescribed
floating-point semantics.

\section{Concrete Reform Proposals}
\label{s:reforms}

We propose a stratified ladder of reform standards (L1--L6), ranging from conservative rational repairs of the traditional system to fully dynamical astronomical realizations. 
Each level embodies a distinct compromise among historical continuity, astronomical fidelity, numerical transparency, and implementation burden.

\subsection{Generalities}
\label{ss:general}

Any serious lunisolar calendar reform must be specified at two levels at once: the level of
\emph{calendric logic}, which determines what the calendar means, and the level of
\emph{numerical semantics}, which determines how that meaning is realized reproducibly in practice.
Without both, one has at most a suggestive model, not a standard that can be implemented,
audited, and published without ambiguity.

At the calendric level, every reform must specify four things.

\begin{enumerate}\itemsep2pt
\item \emph{Intercalation logic.}
One must specify the rule that determines the month skeleton: which solar model is used, how month
labels are assigned, and under what condition an extra lunation is inserted.
This is the level at which one decides, for example, whether the calendar is governed by a fixed
arithmetic congruence, a mean-Sun crossing rule, or a more fully astronomical month definition.

\item \emph{Lunar-day delimitation.}
One must specify how the boundaries of lunar days are computed, i.e.\ how the elongation
\(\lambda_{\moon}-\lambda_{\sun}\) is modeled and how the instants at which it reaches multiples of
\(12^\circ\) are located.
This includes not only the celestial model but also the inversion procedure used to recover the event times.

\item \emph{Civil-day synchronization.}
One must specify how the continuous sequence of lunar-day boundaries is sampled by civil days.
In practice this means choosing the civil-day trigger: a fixed surrogate such as mean sunrise,
or a genuinely geometric local sunrise model tied to a specified location and atmospheric convention.

\item \emph{Temporal reference.}
One must specify the time scale in which the celestial model is formulated, and the rule by which it is
related to civil time.
In particular, any reform that aims at long-term physical fidelity must state how it treats the distinction
between uniform dynamical time and irregular Earth-rotation time.
\end{enumerate}

These four pillars determine the mathematical content of the calendar, but not yet a
\emph{reproducible} calendar. A modern standard must also fix the numerical contract by which
the formulas are evaluated. Expressions such as \(\sin x\), \(\arctan x\), or ``solve for the root''
are not, by themselves, portable specifications: different implementations may diverge because of
library choices, hidden rounding, stopping rules, tie-breaking conventions, or evaluation order.
A reform intended for actual use must therefore prescribe its numerical semantics explicitly.
Depending on the level of the reform, this may mean exact rational arithmetic with tabular
trigonometry, or a fully specified floating-point protocol with frozen constants, prescribed
polynomial kernels, and fixed iteration counts.

Reproducibility also requires explicit verification. A reform proposal should come with benchmark
dates, invariants, and diagnostic checks. Benchmark dates include, for instance, New Year dates
(Losar, Tsagaan Sar), representative skipped and repeated days, and historically published almanac
entries. Invariants include monotonicity of the underlying time model, consistency of month/day
labeling, and robust behavior in near-tie configurations where lunar-day boundaries lie close to the
civil-day trigger. A reform that cannot be tested systematically cannot be trusted as an operational
standard.

Finally, calendric reform is also a question of transparency. A good proposal should separate as
clearly as possible three kinds of change: \emph{conventions} (for example, the reference meridian
or the definition of sunrise), \emph{scientific updates} (for example, improved mean motions or
anomaly terms), and \emph{numerical choices} (for example, exact rational tables versus prescribed
polynomial evaluation). This separation makes it easier to see what is being preserved, what is being
changed, and what belongs merely to implementation discipline rather than calendric doctrine. In
practice, it also means that an adopted standard should be published in versioned form, together
with fixed benchmark outputs sufficient to test conformance across software environments.

\subsection{Design philosophy and error budget}
\label{ss:philosophy}

Our reform program is guided by a simple but far-reaching asymmetry: the Sun is a slow, stable background clock,
whereas the Moon is a fast, sensitive foreground clock.
This suggests that a lunisolar calendar should not be treated as a monolithic object, but designed in layers, with a relatively stable \emph{month skeleton} at the macro level and a more delicate \emph{day texture} at the micro level.
In particular, the model that governs intercalation need not be identical in complexity to the model that governs
lunar-day boundaries and civil-day labels.

This layered viewpoint is not alien to the Tibetan tradition; it is already implicit in the classical system.
The classical system uses a comparatively rigid mean-motion structure to determine leap months, while reserving
anomaly corrections for the day calculation.
We make this implicit architecture explicit and elevate it to a design principle: \emph{normative stability}.
A calendar intended for communal use should not allow its month skeleton to jitter in response to microscopic
physical refinements.
Leap-month placement is a discrete classification problem, and once a community standardizes that classification,
it is usually desirable that tiny perturbations in constants or negligible higher-order terms do not cause the
sequence of months to twitch.
For this reason, many of our proposals intentionally define the month layer using a smoothed or low-order lunisolar model, while allowing a more refined lunar and solar geometry to govern the day layer.

The error budget supports this stratification.
Three characteristic angular rates dominate the problem:
Earth's rotation, about \(15^\circ/\mathrm{h}\);
the synodic elongation, about \(0.5^\circ/\mathrm{h}\);
and the Sun's longitude, about \(0.04^\circ/\mathrm{h}\).
The lesson is immediate.
The timing of conjunctions and lunar-day boundaries is overwhelmingly controlled by the lunar engine, because the Moon
moves much faster than the Sun.
A moderate refinement of the solar model often changes these timings far less than a comparable refinement of the
lunar model.
By contrast, the civil-day labeling problem is intrinsically high-frequency:
one must track a rapidly moving elongation boundary against the discrete ``sampling strobe'' of sunrise.
Small timing errors can therefore move a lunar-day boundary across the civil-day trigger and change whether a date is
skipped or repeated.
At higher levels of precision, even the sunrise model itself becomes part of the decisive error budget, since errors
of a few minutes in the sunrise time are large enough to affect the day label.

This leads naturally to a two-level design philosophy.
At the macro level, one seeks a \emph{canonical month model}: stable, auditable, and resistant to spurious jitter.
At the micro level, one seeks a \emph{physical day model}: accurate enough to track the true texture of the Moon's
motion and the local civil-day boundary.
Different reform levels will combine these two layers in different ways, but all of them are organized by this same
principle.

The same philosophy extends to arithmetic.
For the lower and middle reform levels, reproducibility is best served by exact rational arithmetic together with
fully prescribed discrete trigonometric machinery.
For the higher levels, greater physical fidelity requires floating-point evaluation, but floating point is acceptable
only if it is itself standardized: constants, approximants, iteration counts, and tie-breaking rules must all be part
of the specification.
In other words, numerical semantics are not an afterthought to the calendar; they are part of the calendar.

The concrete proposals in this section are arranged as a ladder of increasing commitment.
Levels L1--L3 retain a predominantly rational and explicitly auditable architecture, while progressively improving
the astronomical content.
Levels L4--L5 move to standardized floating-point semi-analytic standards, trading some arithmetic austerity for
higher physical accuracy.
Level L6 represents the fully astronomical endpoint, in which the calendar is tied directly to a modern ephemeris-style model.
For contrast, low-commitment alternatives are discussed separately in \S\ref{ss:low_commitment}, while the arithmetic L0 baseline is recorded in Appendix~\ref{app:arith-day}.

A useful way to read the reform ladder is as a family of \emph{standard bundles} built from a
small number of reusable computational modules. At the month level, the serious proposals
draw on two main modules: an \emph{arithmetic month module}, in which intercalation is governed
by a fixed rational cycle and congruence rule, and a \emph{rational month module}, in which month
decisions are derived from a low-order conjunction / solar-crossing model. At the day level,
the serious proposals draw on two main modules: a \emph{rational day module} and a
\emph{floating-point day module}, both of which may be viewed as instances of a common day
algorithm template with different numerical contracts, sunrise components, and anomaly terms
activated. The individual reform levels L1--L5 should therefore be understood not as five
wholly unrelated calendars, but as five named selections from this common menu. In the 
subsections that follow, each level is presented with both its design rationale and an explicit 
technical specification card defining its exact selection of components. 
Appendix~\ref{app:reference_modules} serves as the foundational parts library, recording the 
rigorous internal mathematics and rational constants of the shared modules from which these 
standards are built.

\subsection{L1: Modernized traditional standard} 
\label{ss:L1} 

L1 is the most conservative serious reform. Its purpose is not to redesign the Tibetan calendar from the ground up, but to formalize a cleaned-up traditional standard that preserves the familiar discrete architecture while removing the most conspicuous long-term structural defects discussed in \S\ref{ss:mean-anomaly}. In particular, it keeps the classical style of calculation---mean motions, first-anomaly corrections, and a simple local day trigger---but makes the role of locale explicit, updates the constants, and decouples the solar anomaly phase from the tropical year.

\reformspec
{Minimal corrective intervention. Preserve the traditional style of month and day logic while making locale explicit, updating the constants, and removing the gross solar-anomaly phase defect.}
{$\approx3-5$ hours.}
{Arithmetic month module using the 334-year cycle ($p=1336$, $q=1377$). The epoch phase is deduced dynamically from the sky anchored at $d_{1}=336^\circ$.}
{Rational day module. Uses a 1-term solar anomaly and a 1-term lunar anomaly (major equation). Solved via 1 Picard iteration.}
{Constant sunrise calculated as $1/4$ (6:00 AM local mean time) for the target coordinate.}
{Constant $\Delta T$ evaluated at 69 seconds.}
{Exact rational arithmetic. Evaluates all trigonometric functions using the discrete traditional tables from \S\ref{a-ss:tabs}.}
{Communities that want a conservative reform with maximal continuity of procedure and notation.}

\emph{Month layer.} Intercalation remains purely arithmetic, with the month structure governed by the mean-Sun/mean-Moon rule of \S\ref{ss:arith-intercalation}. 
This is the natural choice for a minimally
disruptive reform: it preserves the traditional arithmetic style and the stable skipped-month-free
topology of months, even though the precise long cycle is updated from the classical baseline.

\emph{Day layer.} Lunar-day boundaries are computed with the classical first-anomaly architecture of \S\ref{ss:day_models}, but with modernized mean and anomalistic constants. 
The astronomical repair at this level is the one identified in \S\ref{sss:solar-anomaly}: the solar anomaly is no longer tied to the tropical year. Just as important operationally, however, the day calculation is now placed in an explicit local framework rather than being left only implicitly tied to a place of publication. Thus L1 preserves the traditional inverse-style day calculation while clarifying its locale dependence and removing the slow phase drift built into the older solar model.

\emph{Civil-day boundary.} At this level we do not yet attempt a geometric sunrise model. The civil-day trigger is still a fixed local surrogate, chosen for simplicity and continuity with almanac-style publication practice; but unlike the traditional systems discussed in \S\ref{ss:locale}, the locale is now an explicit part of the specification. Seasonal and latitudinal variations in actual sunrise are therefore not yet modeled, but the calendar is at least unambiguously tied to a stated longitude (and, if desired, latitude).

\emph{Numerical contract.} L1 should be specified in exact rational arithmetic. The point of this level is not merely low computational cost, but full transparency: every constant, table value, and interpolation rule can be published and audited directly. 

\emph{Accuracy and intended use.} This level removes a genuine structural defect but does not yet resolve the dominant lunar and civil-day errors. The omitted lunar inequalities---especially evection and variation---still move lunar-day boundaries by hours, and the fixed civil-day trigger introduces an additional seasonal uncertainty of comparable scale. For that reason L1 should be understood as a conservative ritual/publishing standard, not as a physically faithful astronomical calendar.

\subsection{L2: Evection standard} 
\label{ss:L2} 

L2 is the first level at which the lunar error budget is taken seriously as a design target. The philosophy is to bring the lunar-day engine up to the natural accuracy ceiling imposed by a fixed local civil-day trigger, without yet paying the conceptual or computational cost of geometric sunrise. In that sense L2 is an \emph{error-balancing} reform: it keeps the stable arithmetic month layer of L1, but repairs the dominant missing lunar physics.

\reformspec
{Repair the dominant missing lunar inequalities while staying within the accuracy ceiling imposed by a fixed local civil-day trigger.}
{$\approx 1$--$2$ hours.} 
{Identical to L1.} 
{Rational day module extending the lunar anomaly to 3 terms (adds evection and variation). Solved via 1 Picard iteration.} 
{Identical to L1.} 
{Identical to L1.} 
{Identical to L1.} 
{A balanced rational standard for civil or communal use when full sunrise geometry is not yet desired.} 

\emph{Month layer.} As in L1, the intercalation rule remains purely arithmetic. This is deliberate: the month skeleton is kept stable and auditable, while the improvements are concentrated in the day layer. 

\emph{Day layer.} The distinguishing feature of L2 is the inclusion of the leading missing \(D\)-coupled lunar terms from \S\ref{sss:evection-variation}. The classical first-anomaly lunar model is not merely slightly inaccurate; it omits the largest perturbations beyond the elliptic core. Once evection and variation are restored, the lunar-day computation is no longer dominated by gross lunar kinematic error, and the principal remaining limitation comes from the civil-day trigger rather than from the elongation model itself. 

\emph{Civil-day boundary.} L2 still uses a fixed local surrogate for sunrise. This is the defining tradeoff of the level: one accepts that the civil-day boundary is only approximate, but ensures that the lunar engine is no worse than that approximation. 

\emph{Numerical contract.} This level can still be kept entirely rational. The only real change from L1 is that the lookup or interpolation layer must now be fine enough that table quantization does not erase the gain achieved by the added lunar terms. 

\emph{Accuracy and intended use.} After the inclusion of evection and variation, the dominant residual error is no longer the Moon but the day trigger. At mid-latitudes, the seasonal displacement between true sunrise and a fixed local surrogate is already of order an hour, so there is little point in pushing the lunar series far beyond this level unless the sunrise model is also improved. L2 is therefore a natural ``balanced rational'' proposal for communities that want a visibly improved calendar while remaining within a simple and fully auditable arithmetic framework.

\subsection{L3: Geometric standard} 
\label{ss:L3} 

L3 is the first fully \emph{geometric} rational reform. It keeps the stable arithmetic month skeleton, but upgrades the civil-day trigger from a fixed surrogate to a latitude-dependent sunrise model. This is the point at which geography enters the specification as an essential datum rather than a hidden assumption; see \S\ref{ss:locale}. L3 is therefore the natural culmination of the rational branch of the reform ladder. 

\reformspec
{Achieve geometric consistency while remaining in a fully rational and reproducible arithmetic framework.} 
{$\approx 15$--$30$ minutes.} 
{Identical to L1.} 
{Rational day module extending the solar anomaly to 1 term with secular drift, and the lunar anomaly to 6 terms (adds annual equation, second elliptic, and reduction to ecliptic). Solved via 3 Picard iterations utilizing the decoupled rational preconditioner ($295306/10000$) to prevent lowest common multiple explosion.} 
{Spherical sunrise model. Uses a geometric depression of $h_0 = -1/432$ turns ($-50'$) and an obliquity of $\varepsilon = 4219/64800$ turns ($\approx 23.44^\circ$), evaluating solar coordinates at a baseline of 5:56 AM local mean time ($89/360$ day fraction).}
{Explicit quadratic $\Delta T$ ($-20 + 32u^2$) and exact factorized rational secular solar drift included.} 
{Exact rational arithmetic. Evaluates geometry using the higher-precision sine table from \S\ref{ss:sine-table}.}
{A high-quality rational standard for long-term publication and use where exact reproducibility is paramount.} 

\emph{Month layer.} L3 still prefers the mean-Sun month skeleton. At this level that is no longer because nothing better is available, but because the month layer is intentionally insulated from small high-frequency perturbations. The principle of normative stability remains in force: month topology should not twitch merely because the day engine has become more realistic. 

\emph{Day layer.} The day engine is now explicitly geometric. Sunrise is computed from latitude and solar declination by a rational spherical model, as discussed in \S\ref{ss:locale}. The lunar engine is also strengthened relative to L2, incorporating the main low-order corrections needed for compatibility with the new sunrise precision scale. In particular, once the civil-day boundary has been refined to the tens-of-minutes regime, omitted lunar terms at the level of only a few minutes become relevant and should no longer be ignored. 

\emph{Civil-day boundary.} The decisive innovation of L3 is that the civil-day trigger is no longer a fixed clock surrogate but an actual geometric sunrise model. We still omit a full equation-of-time correction, so the remaining day-boundary uncertainty is largely set by the difference between apparent and mean solar time. At this tier, explicit secular drift in the solar anomaly and an explicit \(\Delta T\) model are best viewed as optional add-ons rather than defining requirements: they improve long-range consistency, but they do not yet dominate the error budget. This is precisely why L3 improves accuracy by roughly an order of magnitude over L2 without yet requiring floating-point machinery.

\emph{Numerical contract.} L3 must remain rigorously rational if it is to deserve its place as the summit of the rational branch. This includes not only the affine predictors but also the trigonometric and inverse-trigonometric layer. In the reference design, the same discrete grid supports both forward and reverse interpolation, including the sunrise inversion. The prime-factor constraint discussed in \S\ref{ss:numerical} belongs naturally here: it is the mechanism that prevents exact rational evaluation from becoming unmanageable over deep time. 

\emph{Accuracy and intended use.} At this level the dominant residual error is no longer missing geography but the omission of the equation of time and higher-order solar/lunar refinements. L3 is therefore a serious publication-grade rational standard: fully portable, fully reproducible, and accurate to a scale that is already small compared with most historical calendric uncertainties.

\subsection{L4: Primary standard} 
\label{ss:L4} 

L4 is the first modern \emph{reproducible floating-point} standard. It is intended as the primary practical standard for contemporary computational environments: substantially more accurate than the rational tiers, yet still light enough to be fast, transparent, and portable. In the logic of the reform ladder, L4 is the point where one accepts that strictly rational arithmetic is no longer the best vehicle for a civil standard, provided that floating-point semantics are themselves frozen and made reproducible. 

\reformspec
{A practical modern standard: few-minute accuracy with fully specified floating-point semantics.} 
{$\approx 3$--$5$ minutes.} 
{Rational month module. Uses 1-term solar and 1-term lunar anomalies, with secular drift on the solar term, solved via 1 Picard iteration, physically anchored at $\sgang_1=336^\circ$. Evaluated using the sine table from \S\ref{ss:sine-table}.}
{Floating-point semi-analytic module. Uses a 2-term solar series and a 14-term lunar series (truncated from Table~\ref{tab:lunar-primary}). Solved via 3 fixed iterations.} 
{Floating-point spherical sunrise ($h_0 = -50'$, $\varepsilon = 23.44^\circ$).} 
{Explicit floating-point quadratic $\Delta T$ and secular solar drift.}
{Reproducible binary64 arithmetic. Day-layer transcendental functions are evaluated using rigidly prescribed 5th-degree hex-float minimax polynomials from \S\ref{ss:minimax}, supported by the deterministic square root algorithm from \S\ref{app:sqrt}.}
{The default modern civil standard for software, publication, and broad public use.}

\emph{Month layer.} Unlike the day layer, the month layer at L4 remains deliberately low-order and effectively rational in structure. The natural L4 choice is a matched first-anomaly month model: if the Sun is promoted from mean motion to a first-anomaly model in the month logic, then the Moon should be promoted as well, so that the conjunction logic is not destabilized by a mismatched order of approximation. This is the first level at which skipped months become a genuine possibility if the chosen solar transit rule permits them.

\emph{Day layer.} The day engine is now a medium-depth floating-point series model. In the implementation architecture this corresponds to a deliberately trimmed semi-analytic engine: a low-order solar model, a 14-term lunar model, and the omission of terms whose contribution lies below the intended few-minute scale over the relevant historical window. This is not a mere computational compromise; it is a standardization choice based on the error budget. 

\emph{Civil-day boundary.} L4 supplements the spherical sunrise model of L3 by incorporating the equation of time. Once the civil-day trigger is being used at the few-minute scale, the distinction between mean and apparent solar time becomes operational and can no longer be ignored. 
At this point both an explicit secular treatment of the solar anomaly and a prescribed \(\Delta T\) model cease to be optional embellishments and become part of the standard itself.

\emph{Numerical contract.} This is where the reproducible-float philosophy first becomes central. Standard library calls are not sufficient. The specification must prescribe the polynomial approximants for the transcendental layer, the floating constants, the iteration count for the fixed-point solver, and the exact order of evaluation. L4 is therefore not ``float'' in the loose engineering sense; it is a rigorously standardized floating-point protocol. 

\emph{Accuracy and intended use.} The intended error scale is a few minutes. At that level the leading residuals come not from gross truncation error but from the deliberate omissions needed to keep the standard lightweight: very small secondary lunar terms, very small solar higher harmonics, and the uncertainties built into long-term \(\Delta T\) modeling. This makes L4 the natural primary standard for modern public use: accurate, portable, fast, and much easier to standardize socially than a full ephemeris calendar.

\subsection{L5: High-precision standard} 
\label{ss:L5} 

L5 is the flagship semi-analytic standard. It accepts the same reproducible floating-point philosophy as L4, but pushes the astronomical model to the point where the limiting errors are no longer mainly numerical or truncational, but arise from the physical fuzziness of the civil-day trigger itself. In other words, L5 is designed so that the mathematics is no longer the main bottleneck. 

\reformspec
{A high-fidelity modern standard: the strongest semi-analytic proposal short of direct ephemerides.} 
{$\approx 30$ seconds in the astronomical kernel, with practical civil accuracy limited by atmosphere and \(\Delta T\).} 
{Rational month module extending the anomalies to 2 solar terms (still with secular drift) and 6 lunar terms. Solved via 2 Picard iterations utilizing the rational preconditioner ($295306/10000$) and the sine table from \S\ref{ss:sine-table}.}
{Deep semi-analytic floating-point module. Employs a 2-term solar series and a 64-term lunar series (Table~\ref{tab:lunar-primary}-\ref{tab:lunar-supp}). Solved via 3 fixed iterations.}
{Identical to L4.}
{Explicit floating-point quadratic $\Delta T$, secular solar drift, and lunar tidal acceleration incorporated via pre-compiled hexadecimal drift constants.}
{Identical to L4.}
{A flagship modern standard for communities that want sub-minute internal consistency without going fully ephemeris-based.}

\emph{Month layer.} L5 continues to respect normative stability. The month layer remains low-order and effectively rational in structure, even though it is more refined than in L4. It is still not allowed to become a noisy mirror of every microscopic perturbation in the physical sky. The point is to give the calendar a stable and standardizable skeleton while allowing the day layer to carry the finer dynamical texture.

\emph{Day layer.} The day engine is now a deep semi-analytic model. This includes a higher-order solar model, a substantially expanded lunar series, and the selective retention of secular terms that matter at the sub-minute level over long time spans. In particular, lunar tidal acceleration first becomes part of the defining specification at this level. At this tier, pruning becomes scientifically nontrivial: one must decide not only what to include, but what to omit on the grounds that it contributes less than the intended physical noise floor.

\emph{Civil-day boundary.} A decisive formalization at L5 is the treatment of sunrise as a \emph{standardized astronomical event} rather than an actual meteorological observation. Once the internal astronomical kernel reaches the tens-of-seconds regime, local weather variability in atmospheric refraction becomes larger than the algorithmic error. For a civil standard, the only reasonable choice is therefore to standardize the atmosphere and define sunrise theoretically rather than observationally. 

\emph{Numerical contract.} The reproducibility requirements of L4 become even stricter here. Polynomial approximants must be chosen so that numerical noise remains comfortably below the physical truncation error of the retained series, and the fixed-point solver must be iterated to a correspondingly tighter tolerance. The important point is again conceptual: these are not merely implementation details, but part of the standard itself. 

\emph{Accuracy and intended use.} The astronomical kernel can be pushed to the \(\sim 30\)-second scale, but the civil calendar as experienced by human observers is already limited by atmospheric and geophysical uncertainties. L5 is therefore best understood as the highest meaningful semi-analytic civil standard: beyond it, additional precision becomes difficult to interpret as observable calendric improvement rather than as internal numerical refinement.

\subsection{L6: Astronomical standard} 
\label{ss:L6} 

L6 is the fully astronomical endpoint of the reform ladder. In its strongest form, it abandons closed-form and truncated semi-analytic models in favor of a directly ephemeris-driven calendar. This is the level at which the physical sky itself becomes the primary reference, offering absolute astronomical fidelity. However, the same architecture still allows a community to pair an ephemeris-driven day layer with a somewhat smoother month layer if normative stability is preferred. Because L6 relies on external data rather than a self-contained algorithmic module, it is defined by interface rules and architectural commitments rather than a rigid parameter specification.

\emph{Month layer.} In its strongest form, L6 accepts the full dynamical consequences of the astronomical model. If the Sun crosses no definition point during a lunation, the month repeats; if it crosses two, a skipped label occurs. These possibilities are no longer viewed as pathologies to be excluded by smoothing, but as genuine features of the astronomical calendar. At the same time, nothing in the architecture forbids a hybrid astronomical reform in which the day layer is ephemeris-driven while the month layer is kept at a slightly lower order for the sake of normative stability.

\emph{Day layer.} The day engine is no longer a closed semi-analytic approximation but a direct numerical evaluation of solar and lunar positions. In that sense, L6 is less a single algorithm than a commitment to an external astronomical standard together with a prescribed interface to calendric logic. 

\emph{Civil-day boundary.} Even at this highest level, the civil-day boundary cannot be reduced to raw observation, because the physical moment of sunrise is weather-dependent. Thus L6 still inherits the principle established in L5: sunrise is a standardized astronomical event defined theoretically in a prescribed atmosphere, not an observational datum that changes from day to day with meteorological conditions. 

\emph{Numerical contract.} The challenge of L6 is not arithmetic portability in the narrow sense, but interface portability. One must specify the ephemeris source and version, the time scales, the interpolation protocol, and the rule by which continuous events are converted into discrete month and day labels. Because any physical standard is evaluated with finite precision against continuous astronomical data, near-boundary events require an explicit canonical policy. If a new moon, a definition-point crossing, or a sunrise falls within a prescribed tolerance of a discrete boundary, the standard must specify how the tie is resolved. This is not a defect of L6, but part of what it means to turn a continuous astronomical model into a discrete civil standard.

\emph{Accuracy and intended use.} L6 offers the greatest physical fidelity, limited mainly by long-term uncertainty in $\Delta T$ and the conventions used to resolve boundary ties. However, it is not automatically the only acceptable civil standard. Once adopted, any of the lower tiers is a standard in its own right, with its own declared balance between stability, transparency, and astronomical realism. The role of L6 is therefore twofold: it is the astronomical endpoint of the reform program, and it provides a high-accuracy reference tier against which lower tiers may be studied and diagnosed when such comparison is useful.

\subsection{Low-commitment alternatives}
\label{ss:low_commitment}

For completeness, we also record several alternatives that improve some aspects of the
traditional calendar while stopping short of a fully balanced reform.  These options are
useful for delineating the design space, and in some settings they may serve as transitional,
diagnostic, or fallback standards.  They fall into two different structural categories.  The
first two are \emph{month-layer} compromises: they modify how the lunation skeleton is
regulated while leaving the day layer largely untouched.  The third is a \emph{day-layer}
simplification: it replaces the traditional lunar-day machinery by a purely arithmetic rule,
while remaining compatible with any chosen month engine.  Thus these alternatives are not
mutually exclusive components of a single ladder; rather, they are partial pathways that may
be combined in different ways.  Nevertheless, they are not our primary recommendations,
because each of them repairs only part of the pipeline while leaving another part
conceptually or numerically misaligned.

A first class of such alternatives changes the \emph{month skeleton} while leaving the rest of
the calendar as untouched as possible.  The most obvious example is to replace the
traditional \(67/65\) intercalation rule by a more accurate rational cycle, as in
\S\ref{ss:arith-intercalation}, while retaining the old day engine, the old anomaly tables,
and the traditional style of publication.  One might, for instance, adopt a \(334\)-year or
\(687\)-year month cycle in place of the \(65\)-year cycle, but continue to compute lunar days
exactly as before.  This can dramatically improve seasonal stability at very low conceptual
cost, and it preserves the attractive feature that leap months remain governed by a short,
rigid arithmetic rule.  However, by itself this only corrects the macro-level drift.  The day
model still inherits the anomaly-phase defects, the missing lunar inequalities, and the
mismatch between a location-independent dawn surrogate and the physical sunrise of an
actual observer.  A cycle-only reform is therefore best understood as a partial repair of the
input data, not as a complete recalibration of the calendar.
It is, in other words, a modification of the \emph{month layer} only, and can in principle be
paired with either a traditional, rational-geometric, or more modern day rule.

A second class keeps the traditional engine essentially intact, but superimposes an
\emph{occasional corrective intervention} on the published calendar.  The idea is to preserve
the inherited month and day machinery most of the time, while inserting or suppressing a
whole month when the accumulated seasonal error is judged to have become too large.
Such a correction could be prescribed by a fixed long cycle---for example, omitting one leap
month every so many centuries---or by an external trigger, such as declaring a correction
when Losar drifts past a chosen seasonal marker in the Gregorian calendar.  This may look
appealing because it leaves the familiar traditional calendar untouched for long stretches and
acts only rarely.  But mathematically it produces a coarse ``sawtooth'' correction profile:
the error is allowed to accumulate continuously and is then repaired in one jump of roughly a
full lunation.  The result is therefore neither a genuinely stable arithmetic calendar nor a
genuinely astronomical one.  It is a patched hybrid whose visible output is governed by the
traditional engine most of the time, but periodically overridden by an external correction
rule.  This weakens transparency, makes long-term behavior harder to analyze, and obscures
the causal link between the published calendar and the underlying calendric model.
Again, this is a month-layer intervention: the day rule may remain entirely traditional, or
may be combined with any of the more refined day-layer proposals discussed elsewhere.

A third alternative, which we regard as an important \emph{baseline} but not as a primary
proposal, is the fully arithmetic mean-elongation day model developed in
Appendix~\ref{app:arith-day}.  This is our L0 level.  Unlike the first two alternatives, it is
not primarily a month-layer proposal at all.  Rather, it is a \emph{day-layer} simplification
that can be combined with essentially any month engine, traditional or reformed.  The day
layer is replaced by a purely arithmetic rule in which the mean elongation advances at a
constant rational rate and day labels are assigned by a simple carry/no-carry test at dawn.
In the most natural parameter range, this produces
\emph{skipped days only} and \emph{no repeated days}.  The result has genuine merits: it is
exceptionally easy to audit, naturally suited to exact rational arithmetic, and useful for
pedagogy, stress testing, and deep-time diagnostics.  It also serves as a clean control model
against which more structured Tibetan or more astronomical proposals may be compared.
At the same time, it abandons a large part of the specifically Tibetan day-calculation
architecture: the first-anomaly inverse day map disappears, and the resulting
skipped-day-only arithmetic is no longer the same phenomenon as the traditional
siddh\=anta pattern of repeated and skipped days.  For that reason we do not elevate L0 to
the same status as levels L1--L6.

These low-commitment alternatives therefore have real uses, but those uses are limited.
The month-layer compromises may function as migration aids or conservative seasonal
repairs, while the L0 day layer provides a clean control model, a fallback arithmetic mode,
and a useful diagnostic baseline.  What they do \emph{not} provide, whether separately or in
combination, is the main goal of this section: a coherent, end-to-end reform in which the
month rule, the day rule, the civil-day trigger, and the numerical semantics are specified
together and judged by a single transparent error budget.  The serious proposals remain the
layered standards L1--L6, while the alternatives collected here show how one may choose to
intervene in only one layer at a time.

\subsection{Reference implementation and diagnostics infrastructure}
\label{ss:reference_impl}

The reform ladder proposed in this paper has been realized as an executable software framework, together with diagnostic and visualization tools. This implementation did not merely illustrate the mathematics after the fact; it clarified several distinctions that are easy to blur on paper but impossible to ignore in code, including the separation of month and day layers, the distinction between mean and true lunation boundaries, the role of locale in sunrise-based calendars, and the need to treat numerical semantics as part of the calendric specification.

\emph{Architecture.}
The reference implementation is organized around the same decomposition that emerges from the mathematical analysis. A \emph{month engine} determines the lunation skeleton and month labels; a \emph{day engine} computes lunar-day boundaries and their interaction with the civil-day trigger; and an \emph{orchestrating calendar layer} synchronizes the two, including any epoch shift between the month-layer lunation index and the day-layer computational index. This separation is not cosmetic. It reflects the paper's central architectural claim that the calendar has a relatively stable macro-structure and a more delicate micro-structure, and that a reform may alter one without altering the other in the same way.

\emph{Declarative specifications.}
The software distinguishes sharply between declarative data and execution logic. A calendar instance is specified by a small structured record of constants and conventions: month rule, day rule, location, epoch offsets, anomaly tables or series, and month-labeling policy. This makes it possible to express traditional calendars and reform tiers within a common framework, compare them systematically, and subject them to the same diagnostic tooling.

\emph{Numerical realization.}
The implementation mirrors the stratified numerical philosophy of \S\ref{ss:numerical}. The rational engines are built on an exact arithmetic layer using Python \texttt{Fraction} arithmetic, discrete trigonometric tables, and reverse interpolation. The floating-point engines use frozen polynomial kernels and fixed-iteration inverse solvers rather than tolerance-driven black-box routines. In both branches, the numerical layer is explicit and reproducible rather than left implicit in the programming environment.

\emph{Tier realization.}
The reference library implements the principal traditional calendars together with the reform tiers L0--L5. The rational branch culminates in a geometric rational engine with explicit sunrise geometry and exact arithmetic. The floating-point branch provides a practical workhorse tier and a higher-precision flagship tier, differing mainly in the depth of the solar and lunar models and in the treatment of smaller secular terms. The fully ephemeris-driven level L6 plays a different role: it serves primarily as a high-accuracy reference against which the analytic tiers can be compared.

\emph{Diagnostics infrastructure.}
The implementation includes a substantial diagnostics layer. At the calendric level, it can enumerate months, repeated and skipped days, New Year dates, and leap-month placements for both traditional and reformed systems. At the structural level, it can inspect month topology, intercalation spacing, and near-tie behavior. At the astronomical level, it can compare predicted event times against modern reference models, fit long-horizon drift curves, reconstruct anomaly profiles, and compute rolling error statistics. The astronomical diagnostics were cross-checked against external reference standards, including JPL DE422 for solar and lunar event timing, the NREL Solar Position Algorithm for solar-position and sunrise calculations, and standard modern \(\Delta T\) references for time-scale conversion \cite{DE422JPL,NRELSPA,IERS2010,Espenak2006,IERS2026}. These tools are what make the comparative numerical study of \S\ref{ss:comparative_diagnostics} possible.

\emph{Public tools.}
The implementation is exposed through a broader public-facing ecosystem. In addition to the Python library itself, there is a command-line interface for inspection and testing, a web calendar capable of rendering repeated/skipped days and intercalary months, and a web diagnostics tool for exploring drift and reform behavior over long time spans. These tools and the underlying library are publicly accessible at \cite{caltib_web}. They make the reform proposals inspectable by users, auditable by specialists, and directly testable by communities that may wish to compare alternatives before adopting a standard.

\emph{What the implementation clarified.}
The existence of a working implementation sharpened several parts of the paper. It forced explicit treatment of trigger labels versus repeated labels, the precise role of the true-month index in day computation, the handling of Bhutanese leap-month naming by reparametrization, the distinction between mean and true lunation boundaries at the epoch, and the fact that sunrise, \(\Delta T\), and tie-breaking conventions cannot remain implicit if a reform is intended to be portable. It also forced the numerical layer to become more explicit than in the initial design sketches: denominator-growth constraints in exact arithmetic, fixed-iteration inverse solvers, and coefficient-level control of floating-point evaluation all had to be treated as part of the implementation contract. In this sense, the software was not merely an appendix to the theory; it was one of the tools by which the theory itself was clarified.

\subsection{Comparative diagnostics and error suppression}
\label{ss:comparative_diagnostics}

The reform ladder can be evaluated at three distinct levels: seasonal placement, anomaly shape, and residual timing error. The first level is visible in the distribution of Losar dates; the second in the reconstructed angular anomaly; the third in the spread and histogram of conjunction offsets against a modern reference.

We begin with the coarsest seasonal diagnostic. Figure~\ref{fig:losar_scatter} compares the long-range Losar scatter of the Phugpa baseline with the first reform tier. The unreformed system occupies a band that drifts steadily later in the year, eventually pushing New Year deep into spring. By contrast, the L1 reform is visually flat on the historical and practical horizon: the Losar pattern stays confined to a stable early-spring window, with no discernible secular rise across the plotted interval. This reflects the fact that the revised month ratio is accurate enough to push any residual seasonal drift out to timescales far beyond ordinary calendrical use. This does not by itself measure conjunction accuracy, but it shows that the reform removes the calendar's most conspicuous large-scale seasonal defect rather than merely rearranging individual month labels.

\begin{figure}[htpb]
    \centering
    \includegraphics[width=0.78\textwidth]{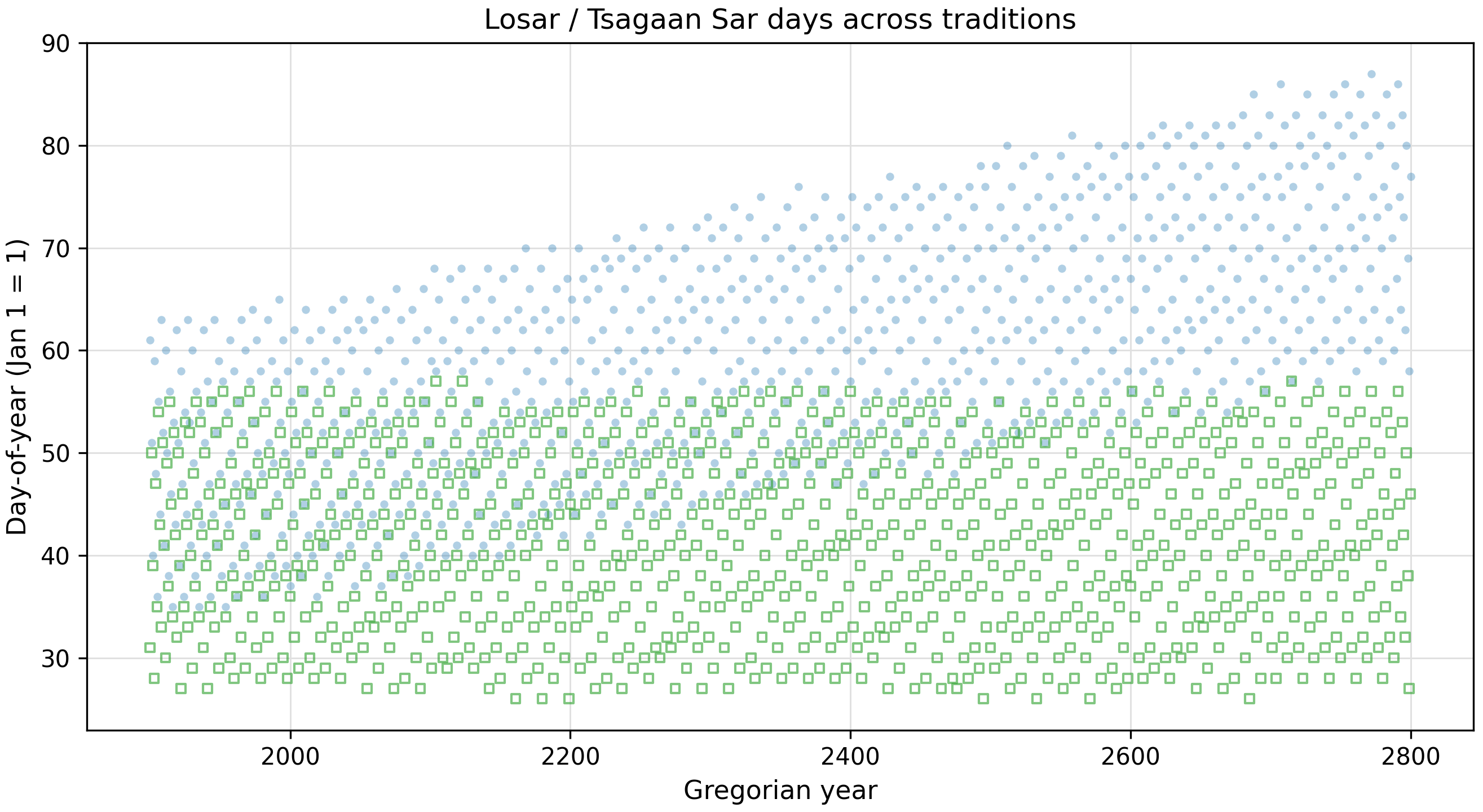}
    \caption{Long-range Losar scatter for the Phugpa baseline and the L1 reform.}
    \label{fig:losar_scatter}
\end{figure}

A more intrinsic test is the anomaly kernel itself. Figure~\ref{fig:anomaly_kernel} compares the reconstructed angular anomaly against the DE422 ephemeris \cite{DE422JPL} on a representative modern window. Here the hierarchy of the ladder is already clear. The traditional branches and L1 retain visible phase and amplitude distortion, with L1 correcting some gross features while still missing the detailed profile. L2 removes the dominant shape error and tracks the reference much more closely. The L3 curve is then visually almost indistinguishable from the reference on the displayed interval. Thus the first major gain of the ladder is not yet extreme numerical precision, but the correction of the anomaly \emph{shape} and \emph{phase}.

\begin{figure}[htpb]
    \centering
    \includegraphics[width=0.86\textwidth]{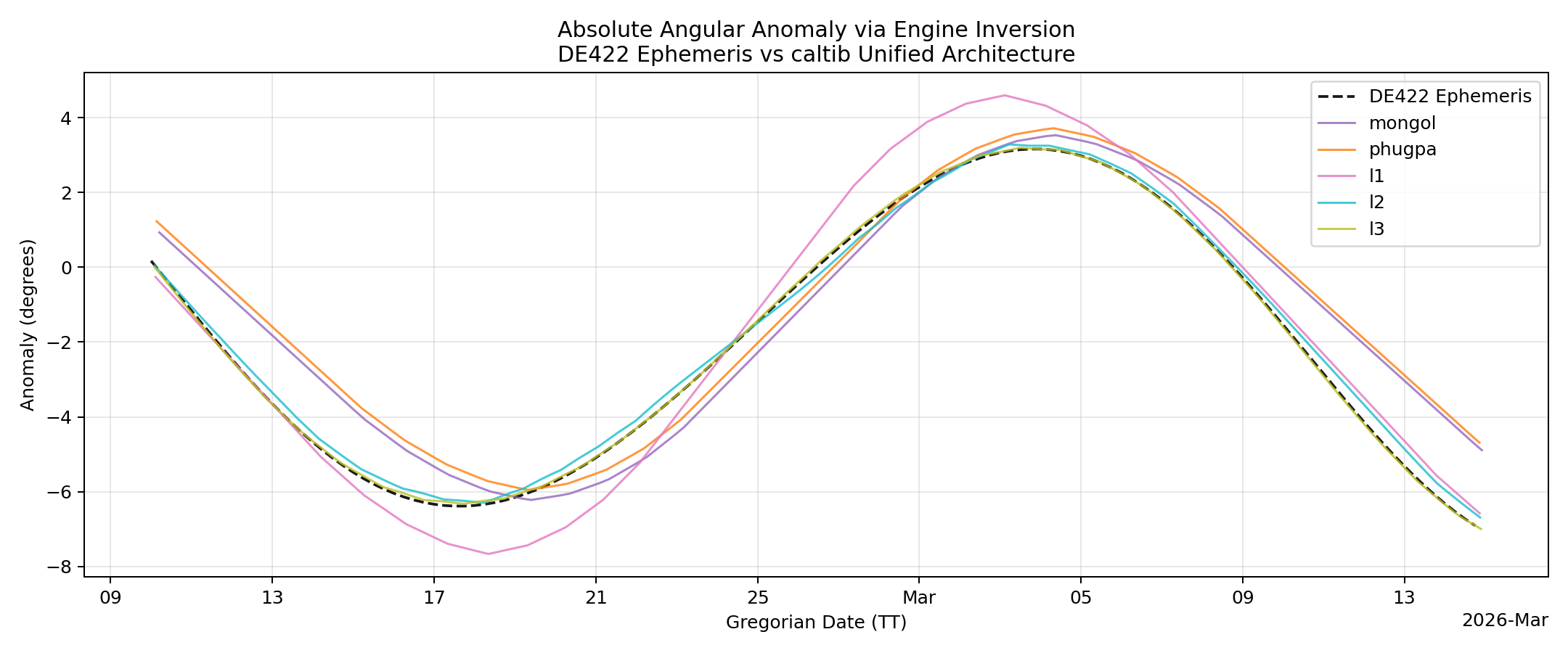}
    \caption{Angular anomaly compared with the DE422 ephemeris.}
    \label{fig:anomaly_kernel}
\end{figure}

The same distinction appears statistically in the rolling standard deviation of conjunction timing error shown in Figure~\ref{fig:rolling_sigma}. Over the post-11th-century window displayed here, all traditional branches already lie on broad rising curves, with standard deviations typically between about \(0.7\) and \(2.5\) hours and worsening steadily into the future. In other words, by the historical era of actual Tibetan and Mongolian usage, the inherited anomaly architecture is already far from its best phase alignment. L1 does not substantially cure this problem: it stabilizes the seasonal skeleton, but the conjunction cloud remains wide, essentially flat, and still at the hour scale. L2 is the first tier that decisively compresses the spread, reducing the rolling standard deviation to roughly the half-hour scale and keeping it nearly constant over long intervals. L3 compresses it further to the \(\sim 0.2\)-hour range. In this sense, the passage from L1 to L2 is the first genuinely dynamical repair of the calendar, while L3 begins to suppress the remaining secondary structure.

\begin{figure}[htpb]
    \centering
    \includegraphics[width=0.82\textwidth]{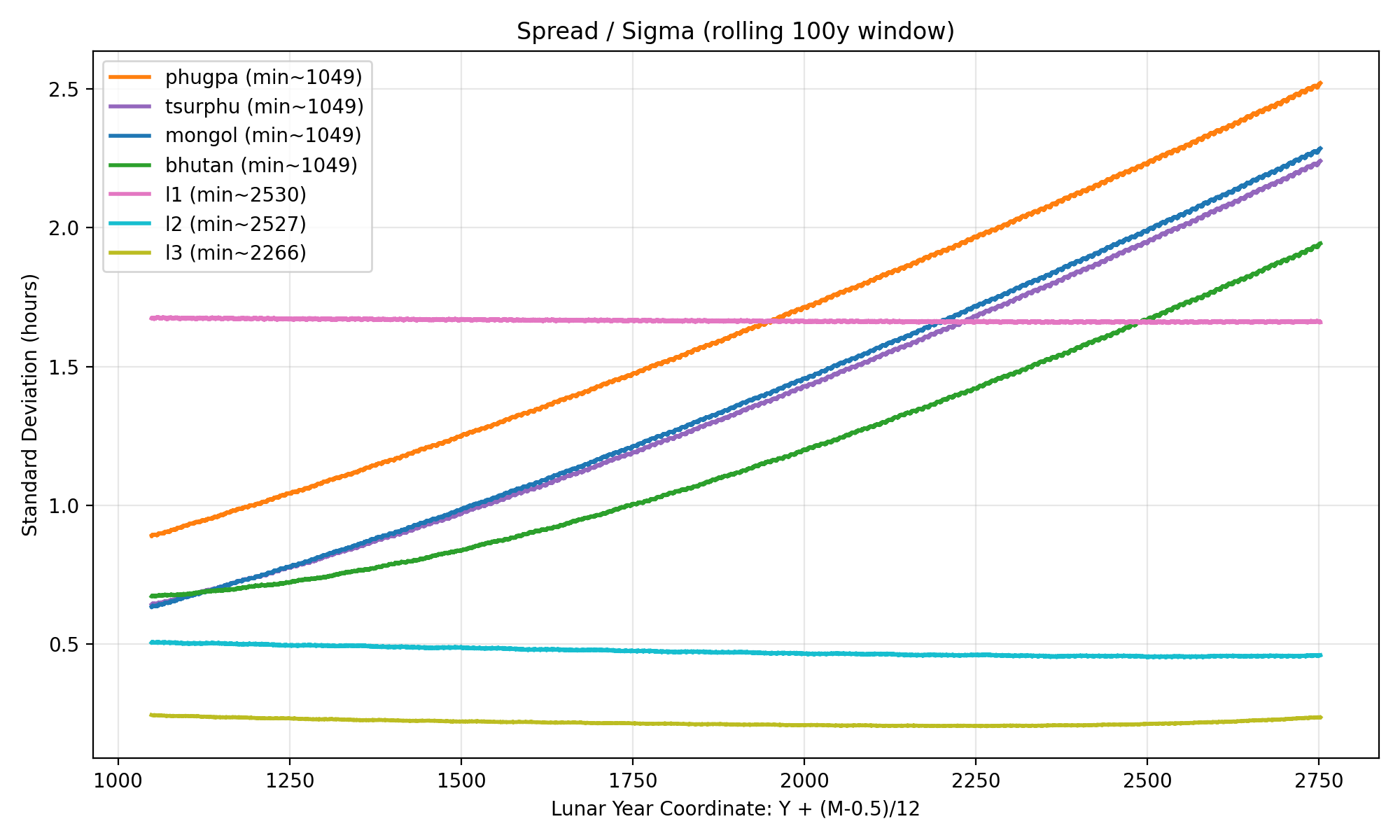}
    \caption{Rolling \(100\)-year standard deviation of conjunction timing error.}
    \label{fig:rolling_sigma}
\end{figure}

The endpoint of this suppression is best seen directly in the offset histograms. Figure~\ref{fig:offset_histograms_reform} shows the physical new-moon timing offsets of L2 and L4 against the DE422 ephemeris. Both distributions are centered essentially at zero mean, so by these tiers the large systematic bias has already been removed. The remaining difference is almost entirely in the width. L2 still has visible hour-scale tails, extending to about \(\pm 1.5\) hours. L4 compresses the same cloud into a narrow band of only a few minutes, with tails of order \(\pm 0.08\) hours. At that stage the dominant classical anomaly error has already been eliminated, and what remains is the residue after higher-order geometric and perturbative corrections.

\begin{figure}[htpb]
    \centering
    \includegraphics[width=0.48\textwidth]{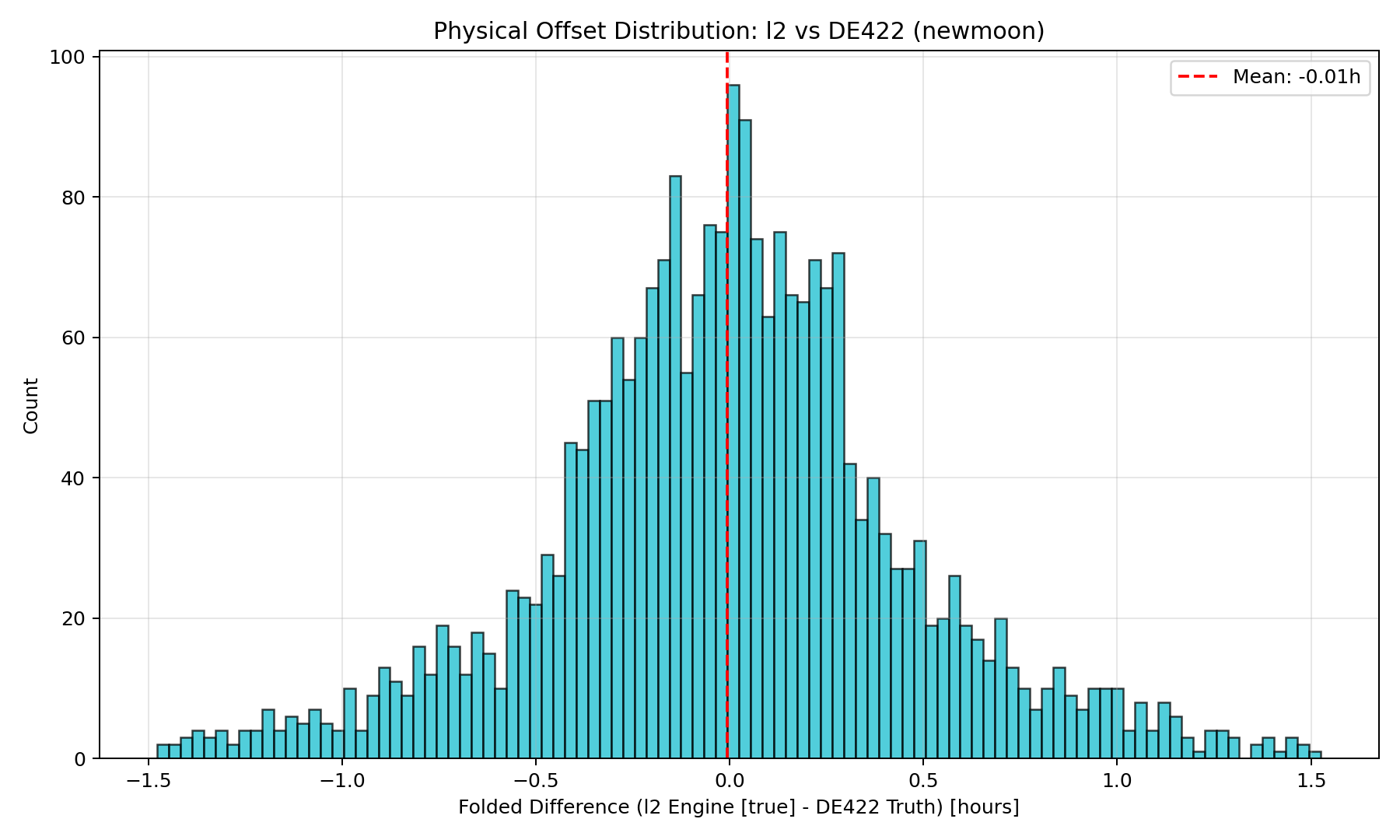}
    \includegraphics[width=0.48\textwidth]{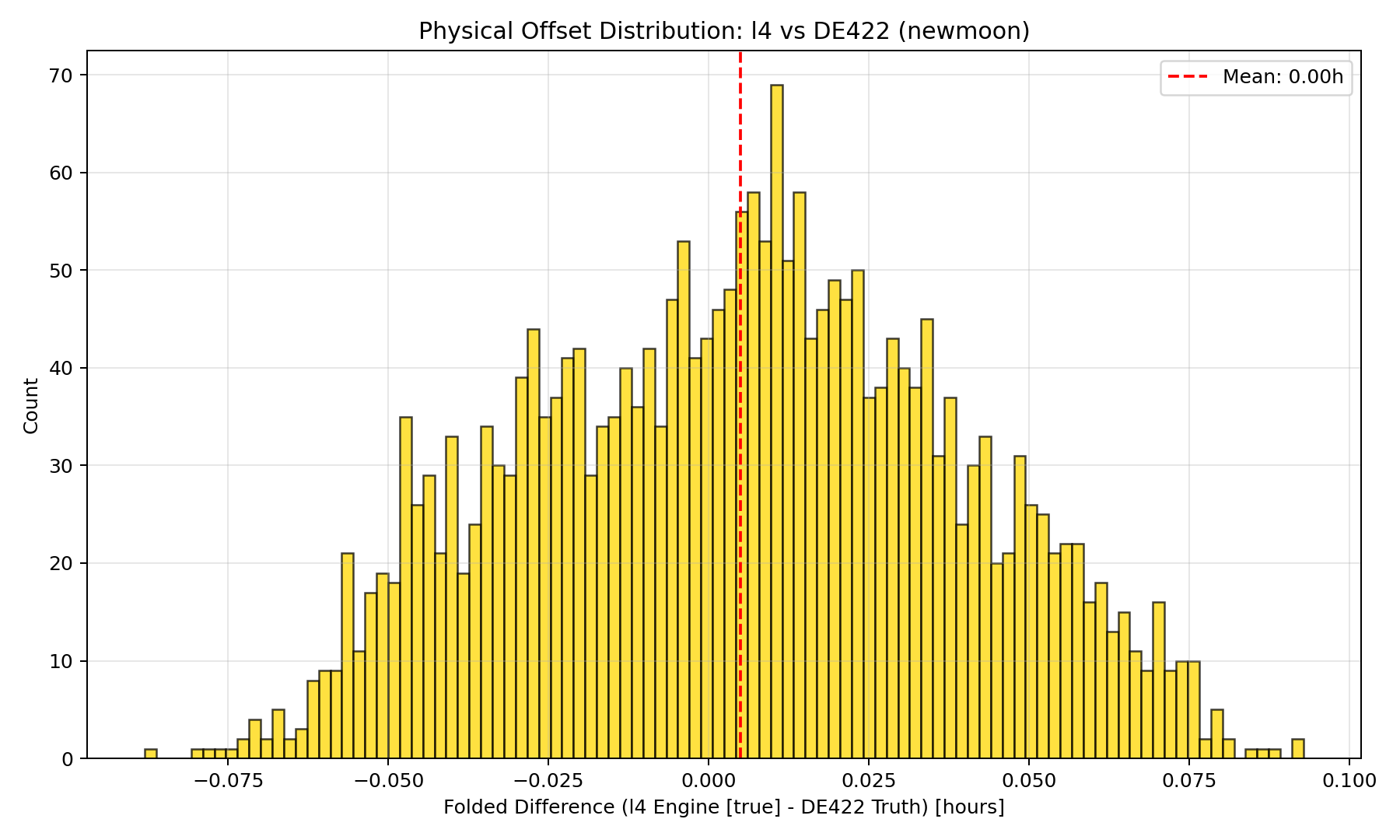}
    \caption{Physical new-moon timing offsets of L2 and L4 against the DE422 ephemeris.}
    \label{fig:offset_histograms_reform}
\end{figure}

Taken together, these diagnostics show that the ladder suppresses error in a definite order. The seasonal envelope is stabilized first; the anomaly kernel is then repaired; the rolling variance collapses; finally the residual timing distribution sharpens from hours to minutes. The comparison also clarifies where the decisive improvements occur. L1 mainly repairs the large-scale seasonal behavior while leaving the anomaly-driven spread mostly intact. L2 provides the first substantial physical repair of the lunar model. L3 and L4 then refine this architecture to the point where the remaining timing error is small on the scale of practical calendrical use. The result is not an all-or-nothing replacement by external ephemerides, but a controlled suppression of distinct error sources, with each reform tier paying only for the fidelity it adds.

\subsection{Enactment and adoption of calendric reforms}

A practical advantage of many lunisolar calendar reforms is that they can often be enacted
cleanly at a calendrical boundary, especially at Losar/Tsagaan Sar. Unlike the Gregorian
reform which required skipping 10 days to realign a solar calendar, such a transition need not require the excision of a block of civil dates. This does
not mean that every reform is socially seamless: month labels, leap-month placement,
publication conventions, and ritual expectations may still shift. But from the narrow
arithmetical point of view, one can often move from one rule set to another without
introducing ``missing'' civil dates in the Gregorian sense.

\section{Conclusion}
\label{s:conclusion}

The Tibetan calendar is often described either as a venerable traditional system or as an
astronomical scheme that has gradually drifted away from the sky. One of the main
conclusions of this paper is that neither description, by itself, is sufficient. What we
actually find is a layered object: a discrete calendric structure, with its own internal logic
of month labels and lunar-day numbering, coupled to a family of celestial models of varying
fidelity. Reform is therefore not a single yes-or-no question, nor a choice between
``keeping tradition'' and ``abandoning tradition.'' It is a structured design problem in which
one must decide which layers of the system are to be preserved, which are to be updated,
and how those decisions are to be specified reproducibly.

\subsection{Structural findings}
\label{ss:conclusion_structure}

A first conclusion is conceptual. The month layer and the day layer obey different
mathematical logics and should not be conflated. Leap months arise from a containment rule
for solar crossings and admit a rigid arithmetic encoding in terms of the intercalation index.
Repeated and skipped day labels, by contrast, arise from sampling a continuous sequence of
lunar-day boundaries by a civil-day trigger. These two phenomena are often discussed
together in traditional presentation because they both produce irregular visible labels, but
their internal mechanisms are fundamentally different. Furthermore, our exhaustive tie-case 
analysis proved that the specific discrete arithmetic governing these traditional day rules 
effectively eliminates boundary ambiguities. Making these structural mechanisms explicit is 
not merely a matter of exposition: it is what allows us to reform one layer while leaving 
another relatively intact.

A second conclusion is historical and geographic. The dominant traditions---Phugpa, Tsurphu, Bhutan,
and Mongol---are not best understood as simple geographical adaptations of a single
underlying calendar to different locations. In fact, our computational analysis demonstrates 
that the historical calendar was remarkably insulated from geographic variance. A combination 
of massive internal temporal buffers and the inherent multi-hour inaccuracy of the classical 
lunar model ensured that longitudinal shifts rarely, if ever, triggered visible calendric changes.
Consequently, the observed differences between traditions are better explained as
lineage-specific recalibrations of epoch constants, intercalation thresholds, and related
computational conventions. Once the problem is written in explicit arithmetic form, it
becomes clear that several of the observed shifts do not have the signature of a longitude or
climate correction, but rather of historical retuning within a sidereal framework that had
already drifted relative to the seasons.

A third conclusion concerns the status of the traditional astronomical model itself. The
classical Tibetan calendar is not ``wrong'' in a single simple way. It contains a hierarchy of
approximations: some are benign and structurally useful, while others become limiting once
one asks for modern reproducibility or stronger seasonal alignment. The \(67/65\) cycle is
one example. The coupling of the solar anomaly phase to the mean Sun is another. The
omission of dominant lunar perturbations is a third. These defects do not all live on the
same scale, and they do not all require the same style of repair. This is why a layered
reform program is more natural than a single once-and-for-all replacement.

\subsection{A ladder of reform standards}
\label{ss:conclusion_ladder}

The central practical result of the paper is that reform is best organized not as a single
proposal but as a \emph{ladder of standards}. The serious proposals L1--L6 do not represent
minor numerical variants of one another; they represent distinct points in the design space,
with different balances of conservatism, transparency, physical fidelity, and implementation
burden.

At one end of the spectrum lie standards that preserve a strongly traditional arithmetic
character. These reforms keep the month layer arithmetic, retain exact rational semantics,
and repair only those aspects of the model whose defects are both large and structurally
unnecessary.
Such standards are attractive when communal continuity, auditability, and
long-term portability are paramount. At the other end lie fully astronomical standards in
which the sky itself becomes the primary reference and the calendar is driven by true or
ephemeris-level solar and lunar motion. These achieve the greatest physical fidelity, but at
the price of heavier specification requirements and a weaker claim to purely arithmetic
transparency.

Between these extremes lie the most interesting possibilities. 
A calendar may preserve a stable arithmetic month rule while making the day layer substantially more physical.
It may adopt
geometric sunrise without yet committing to a fully dynamical month model. It may retain a
low-order, normatively stable month rule while standardizing a deeper floating-point day
engine. The resulting ladder is not a sign of indecision; it is the natural consequence of the
calendar's layered structure. Different communities may rationally prefer different points on
this ladder depending on whether they prioritize traditional style, ease of publication,
astronomical realism, or strict reproducibility.

This also clarifies the role of the low-commitment alternatives discussed in
\S\ref{ss:low_commitment}. Those alternatives are not pointless, but neither are they
complete. Some modify only the month layer, others only the day layer. They are useful as
baselines, migration aids, or conservative stopgaps, but they do not provide the same kind
of coherent end-to-end standard as the main L1--L6 proposals.

\subsection{From reform proposals to executable standards}
\label{ss:conclusion_executable}

A further conclusion, sharpened by the implementation work, is that calendric reform is not
only a matter of formulas and constants. A modern standard must also specify its
\emph{numerical semantics}. It is not enough to say that one should use ``sine,''
``arctangent,'' or ``solve for sunrise.'' One must also state whether these are to be
realized by exact rational arithmetic, prescribed tabular interpolation, fixed floating-point
coefficients, fixed iteration counts, or an external ephemeris interface. Without this layer,
one has at best a family of related implementations, not a reproducible standard.

This is not merely a philosophical point. The reference implementation developed alongside
this paper showed that distinctions easy to blur in prose become decisive in code: the
difference between labeled months and the true-month index, between mean and true
lunation boundaries at the epoch, the role of locale in sunrise-based day computation, the
need for explicit tie-breaking conventions, and the interaction between exact arithmetic and
denominator growth. In this sense, implementation did not merely illustrate the theory; it
helped sharpen it.

The practical consequence is that a calendar standard should now be understood as having at
least four layers:
\begin{enumerate}\itemsep2pt
\item a mathematical specification of month and day logic;
\item an explicit parameter set fixing constants, epochs, and conventions;
\item a numerical contract specifying how the formulas are to be evaluated; and
\item a conformance-tested reference implementation, together with benchmark cases and
diagnostic tools.
\end{enumerate}
This is why the present work formulates its proposals not only in prose, but also as a
technical specification library (Appendix~\ref{app:reference_modules}) supported by a
reproducible software framework. The existence of an executable library, a web calendar,
and diagnostic tools does not decide which reform ought to be adopted. It does, however,
change the nature of the discussion: the options are no longer vague aspirations or isolated
formulas, but concrete, comparable, and reproducible standards.

The broader cultural question remains open, and rightly so. A calendar is not merely a
scientific device; it is also a communal instrument of ritual, civic, and historical continuity.
For that reason, the choice among reform levels cannot be made by astronomical accuracy
alone. But the tradeoffs can now be stated explicitly, and their consequences can be
computed in advance.

We therefore regard the main achievement of this work not as the advocacy of a single final
calendar, but as the establishment of a disciplined reform framework. It isolates the true
structural questions, distinguishes them from merely numerical ones, and provides a ladder
of executable standards ranging from conservative rational repairs to fully astronomical
realizations. That framework, rather than any one particular parameter choice, is what can
support a thoughtful and durable calendric future for Tibetan, Bhutanese, Mongolian, and
related communities.

\section*{Acknowledgements}

I am especially grateful to Svante Janson, whose careful exposition \cite{janson} made the Tibetan calendrical computations accessible to a modern mathematical audience and served as my primary guide while learning the subject. 
I also benefited from the treatments of Schuh, Henning, Reingold–Dershowitz, and Aslaksen \cite{schuh,henning,reingold,aslaksen}, which provide complementary perspectives and documentation.


\appendix
\section{Parameters of the Principal Traditions}
\label{app:constants}

To instantiate a concrete calendar instance in the siddh\=anta-style pipeline, we fix the following inputs.

\medskip\noindent\emph{Tradition-independent (shared by all four traditions):}
\begin{itemize}\itemsep2pt
\item {Mean-motion increments} \(m_1,m_2,s_1,s_2,a_1,a_2\), i.e.\ the universal step sizes in the affine predictors
\eqref{eq:affine-pred-app}, recorded in \S\ref{a-ss:mean-motion}.
\item {True-date lookup tables} \texttt{moon\_tab} and \texttt{sun\_tab} used in the tabular correction step,
recorded in \S\ref{a-ss:tabs}.
\end{itemize}

\noindent\emph{Tradition-dependent (choice of epoch and labeling convention):}
\begin{itemize}\itemsep2pt
\item {Month-label and intercalation data:} the epoch label \((Y_0,M_0)\), the congruence shift \(\beta^*\),
and the trigger set \(T\subset\mathbb Z/65\mathbb Z\), which together determine the \(65\)-periodic leap-month skeleton and
the conversion from a labeled month \((Y,M)\) to the true-month index \(n\) (Remarks~\ref{rem:leap-naming}
and~\ref{rem:true-month-index}).
\item {Epoch offsets:} the absolute-time and phase offsets \((m_0,s_0,a_0)\) that pin the affine predictors to the chosen
epoch, recorded in \S\ref{a-ss:epoch}.
\end{itemize}

\noindent
Given these data, the standard pipeline computes mean and true boundary times for lunar days, maps them to civil days
(to produce repeated/skipped date numbers), and applies the intercalation rule to insert extra lunations and assign month labels.
This appendix records the numerical constants needed to implement, in a uniform way, the four
principal traditions (Phugpa, Tsurphu, Bhutan, Mongol), together with selected karana data included
for completeness and comparison.

\subsection{Leap-month rules}
\label{a-ss:ix}

It is convenient to separate the tradition-dependent inputs into two layers.  The first layer
(Table~\ref{tab:ix-rules-app}) is purely arithmetic and determines \emph{which lunations exist} (i.e.\ where an extra lunation
is inserted) and how they are \emph{labeled} by pairs $(Y,M)$.  The second layer (Table~\ref{tab:epoch-constants}) fixes
\emph{where the resulting calendar sits on the Julian-day line} by specifying the absolute-time and phase offsets used in
the mean/true-date computations of \S\ref{ss:day_calculation}.

\begin{table}[ht]
\centering
\begin{tabular}{lrrc}
\toprule
Tradition & $(Y_0,M_0)$ & $\beta^\ast$ & Repeat-label set\\
\midrule
Karana (E806) & $(806,3)$ & $0$  & $\{63,64\}$\\
Phugpa (E1927) & $(1927,3)$ & $55$  & $\{48,49\}$\\
Phugpa (E1987) & $(1987,3)$ & $0$  & $\{48,49\}$\\
Bhutan (E1754) & $(1754,3)$ & $2$ & $\{57,58\}$\\
Tsurphu (E1732) & $(1732,3)$ & $59$ & $\{0,1\}$\\
Tsurphu (E1852) & $(1852,3)$ & $14$ & $\{0,1\}$\\
Mongol (E1747) & $(1747,3)$ & $10$ & $\{46,47\}$\\
\bottomrule
\end{tabular}
\caption{Leap-month parameters in the reparameterized repeat-label convention.}
\label{tab:ix-rules-app}
\end{table}

Table~\ref{tab:ix-rules-app} specifies the intercalation index in terms of labeled months.  We normalize the labeled-month
count by $M_0=3$ and define
\[
M^*=12(Y-Y_0)+(M-M_0),\qquad \ix \equiv 2M^*+\beta^* \pmod{65}.
\]
In this appendix the final column is written in the reparameterized repeat-label trigger convention:
if $\ix$ lies in the listed set, then the label $(Y,M)$ occurs twice in succession.  The only
remaining tradition-dependent issue is which of the two copies is called \emph{leap}.  For Phugpa,
Tsurphu, and Mongol, the first of the two copies is the leap month; for Bhutan and Karana, the
second copy is called leap.  Thus, for Bhutan and Karana, we are using the reparameterized form of
Remark~\ref{rem:leap-naming}: traditionally the repeated label is the preceding one, so the trigger
set is usually presented with a shift by $2$ modulo $65$, but the location of the inserted lunation
in the sequence of lunations is unchanged.

\subsection{Epoch-dependent offsets}
\label{a-ss:epoch}

For the day-calculation formulas it is most convenient to index lunations by a single running counter
$n\in\Z$ (the \emph{true-month index}).  We fix an epoch lunation boundary and declare it to be $n=0$;
then $n$ increases by $1$ from one lunation to the next, regardless of whether a labeled month is repeated.
The conversion from a labeled month $(Y,M)$ to the corresponding $n$ is a \emph{derived} step, given by
Remark~\ref{rem:true-month-index}.

Published sources often normalize the epoch as ``month $3$'' ({\em nag pa}).  This should be read as a statement
about the intended \emph{labeling convention}, not as a guarantee that the reference lunation $n=0$ necessarily
carries the label $(Y_0,3)$.  For a $65$-periodic arithmetic intercalation rule with trigger pair
$T=\{\tau,\tau+1\}\subset\Z/65\Z$, let
\[
\gamma\in\{0,1,\dots,64\}
\qquad\text{be such that}\qquad
\gamma\equiv -\tau \pmod{65}.
\]
Then, for the epoch label $(Y_0,3)$ (so $M^*=0$), the later lunation carrying that label has index
\[
n_+(Y_0,3)=\Bigl\lfloor \frac{\beta^*+\gamma}{65}\Bigr\rfloor .
\]
Accordingly, if $(Y_0,3)$ is a non-trigger label, then the reference lunation $n=0$ carries the label
$(Y_0,3)$ iff $\beta^*+\gamma<65$; otherwise $n=0$ carries the preceding label $(Y_0,2)$.
If $(Y_0,3)$ is a trigger label, then two consecutive lunations carry the label $(Y_0,3)$, and we take
$n=0$ to be the earlier of the two.

For the epochs listed in Table~\ref{tab:ix_rules}, the label $(Y_0,3)$ is non-trigger, so it occurs uniquely.
Among these examples, the only case in which $(Y_0,3)$ is not the reference lunation is Phugpa~E1927:
there $(Y_0,3)$ has true-month index $n=1$, while the reference lunation $n=0$ carries the preceding label
$(Y_0,2)$.  By contrast, for both Tsurphu epochs one has $\tau=0$ and hence $\gamma=0$, so the published epoch
month $(Y_0,3)$ is indeed the reference lunation $n=0$.

\begin{table}[ht]
\centering
\begin{tabular}{lllcc}
\toprule
Tradition & Epoch (JD) & $m_0$ (JD units) & $s_0$ (mod $1$) & $a_0$ (mod $1$)\\
\midrule
Karana (E806) &
$2015531$ &
$2015531+\frac{1}{2}$ &
$\frac{809}{810}$ &
$\frac{53}{252}$\\[2pt]
Phugpa (E1927) &
$2424972$ &
$2424972+\frac{5457}{5656}$ &
$\frac{749}{804}$ &
$\frac{1741}{3528}$\\[2pt]
Phugpa (E1987) &
$2446914$ &
$2446914+\frac{135}{707}$ &
$0$ &
$\frac{38}{49}$\\[2pt]
Bhutan (E1754) &
$2361807$ &
$2361807+\frac{52}{707}$ &
$\frac{1}{67}$ &
$\frac{17}{147}$\\[2pt]
Tsurphu (E1732) &
$2353745$ &
$2353745+\frac{1795153}{7635600}$ &
$-\frac{5983}{108540}$ &
$\frac{207}{392}$\\[2pt]
Tsurphu (E1852) &
$2397598$ &
$2397598+\frac{1197103}{7635600}$ &
$\frac{23}{27135}$ &
$\frac{1}{49}$\\[2pt]
Mongol (E1747) &
$2359237$ &
$2359237+\frac{2603}{2828}$ &
$\frac{397}{402}$ &
$\frac{1523}{1764}$\\
\bottomrule
\end{tabular}
\caption{Epoch-dependent offsets.}
\label{tab:epoch-constants}
\end{table}

Once the true-month index $n$ is fixed, the affine predictors \eqref{eq:affine-pred-app} require the epoch offsets
$(m_0,s_0,a_0)$, listed in Table~\ref{tab:epoch-constants}.  Since the true-date correction depends on angular
arguments only modulo one full turn, only the congruence classes $s_0\bmod 1$ and $a_0\bmod 1$ are operationally
relevant; this is why Table~\ref{tab:epoch-constants} records only their fractional parts.

The column ``Epoch (JD)'' is the integer part $\lfloor m_0\rfloor$: it is included to make the correspondence with
published epoch dates transparent, but an implementation uses the full real constant $m_0$ (including its fractional
part).  Concretely, $m_0$ is the \emph{mean} lunation-boundary time at $n=d=0$ in the affine predictor
$t_{\rm mean}(d,n)=m_0+n\,m_1+d\,m_2$, i.e.\ the time of the \emph{mean new moon} starting the reference lunation $n=0$
(cf.\ \cite[Remark~16]{janson}).  The corresponding \emph{true} new moon used in the full siddh\=anta pipeline is obtained
only after applying the true-date correction, and may fall on a neighboring civil date even when $m_0$ is fixed.


\subsection{Mean-motion increments}
\label{a-ss:mean-motion}

The four principal traditions considered here share the same mean-motion \emph{increments} in the
siddh\=anta day-calculation scheme.  The basic affine predictors are (cf.\ \S\ref{ss:day_calculation})
\begin{equation}\label{eq:affine-pred-app}
\begin{split}
\md(d,n) &= m_0 + n\,m_1 + d\,m_2, \\
\ms(d,n)  &\equiv s_0 + n\,s_1 + d\,s_2 \pmod 1,\\
A_{\moon}(d,n) &\equiv a_0 + n\,a_1 + d\,a_2 \pmod 1.
\end{split}
\end{equation}
Here $n\in\mathbb Z$ is the \emph{true-month index} (running count of lunations from the chosen epoch),
and $d\in\{1,\dots,30\}$ is the lunar-day number within that lunation; we also allow $d=0$ as a
notational convenience for the lunation boundary at its beginning.

For these four traditions the shared increments are
\begin{align*}
m_1 &= \frac{167025}{5656} \approx 29.53058699, &
m_2 &= \frac{11135}{11312}=\frac{m_1}{30}\approx 0.98435290,\\
s_1 &= \frac{65}{804} \approx 0.080845771, &
s_2 &= \frac{13}{4824}=\frac{s_1}{30}\approx 0.002694859,\\
a_1 &= \frac{253}{3528} \approx 0.071712018, &
a_2 &= \frac{1}{28}\approx 0.035714286.
\end{align*}
An alternative ``exact'' choice proposed by Minling Lochen Dharmashri (1654--1717), and used by
Henning, is
\[
a_2=\frac{1+a_1}{30}=\frac{3781}{105840}
=\frac1{28}+\frac1{105840},
\]
which differs extremely rarely at the level of integer day assignment, cf.\ \cite[Remark~14]{janson}.

For completeness, the karana system uses the same anomaly increments $a_1,a_2$, but replaces the
synodic-month increment by
\[
m_1^{\mathrm{kar}}=\frac{10631}{360}\approx 29.53055556,
\]
with the corresponding day increment $m_2^{\mathrm{kar}}=m_1^{\mathrm{kar}}/30$.  The associated mean-Sun increment is
\[
s_1^{\mathrm{kar}}=\frac{1277}{15795}\approx 0.080848369,
\]
and hence $s_2^{\mathrm{kar}}=s_1^{\mathrm{kar}}/30$. 
Thus the four principal traditions remain the shared siddh\=anta block, while karana differs
essentially in the choice of mean month, with the corresponding solar increment following suit.

\subsection{Lookup tables for the lunar and solar equations}
\label{a-ss:tabs}

The tables are specified by base values and symmetry/periodicity rules, with linear
interpolation between integer arguments.
See Figure~\ref{fig:equ-tables} for illustration.

\begin{itemize}
\item Lunar table $\texttt{moon\_tab}(i)$ for $i=0,\dots,7$:
\[
(0,5,10,15,19,22,24,25),
\]
extended by $\texttt{moon\_tab}(14-i)=\texttt{moon\_tab}(i)$,
$\texttt{moon\_tab}(14+i)=-\texttt{moon\_tab}(i)$, and $\texttt{moon\_tab}(28+i)=\texttt{moon\_tab}(i)$.
\item Solar table $\texttt{sun\_tab}(i)$ for $i=0,\dots,3$:
\[
(0,6,10,11),
\]
extended by $\texttt{sun\_tab}(6-i)=\texttt{sun\_tab}(i)$,
$\texttt{sun\_tab}(6+i)=-\texttt{sun\_tab}(i)$, and $\texttt{sun\_tab}(12+i)=\texttt{sun\_tab}(i)$.
\end{itemize}

\section{Congruence Filters for Tie-Case Search}
\label{app:congruence}

This appendix records the number-theoretic reductions used in the tie-case search of \S\ref{sss:ties}.
A \emph{tie} occurs when the true-date function $t(d,n)$ lands exactly on a civil-day boundary
(normalized to $\Z$).  Since $t(d,n)$ is a finite sum of affine terms and $1$--periodic piecewise-affine
table terms with rational knots, its fractional part is periodic in $(d,n)$ with an explicit period.
Moreover, for suitable primes $p$ one can localize the tie condition to congruence constraints on
distinguished affine phases (notably the mean date $M_d(n)$ and the solar-table argument $v(d,n)$),
yielding strong residue-class filters.  These reductions are abstract and apply to any reform specified
by rational constants, and in \S\ref{sss:ties} they are used to accelerate the computations for the principal
traditions.

If $p$ is a prime, we write
\[
\Z_{(p)}=\left\{\frac ab\in\Q:\gcd(b,p)=1\right\}.
\]
Then for any $x\in \Z_{(p)}+\frac1p\Z$ there is a unique decomposition
\[
x=x_{(p)}+\frac{r_p(x)}{p},
\qquad x_{(p)}\in \Z_{(p)},\ \ r_p(x)\in\{0,1,\dots,p-1\}.
\]
We regard $r_p(x)$ as the $p$--primary fractional residue of $x$; in the tie setting, vanishing of the
integer boundary forces certain $p$--residues to vanish and thus produces congruence constraints.

\begin{proposition}
\label{prop:periodicity_prime_obstruction}
After translating the time origin, assume civil-day boundaries occur at times $\Z$.
Let
\begin{equation}\label{eq:sum_inverse_general}
  t(d,n)=\sum_{j=0}^J F_j\bigl(U_j(d,n)\bigr),
\qquad
U_j(d,n)=u_{j,0}+n u_{j,1}+d u_{j,2},
\end{equation}
where $u_{j,k}\in\Q$ and each $F_j$ is either
\begin{enumerate}[label=\textnormal{(\roman*)}, leftmargin=2.0em, itemsep=2pt]
\item the identity map $F_j(y)=y$, or
\item a $1$--periodic piecewise-affine function on $\R$ obtained by extending a piecewise-affine function
on $[0,1]$ with \emph{rational breakpoints} and \emph{rational vertex values}, evaluated by exact affine
interpolation on each segment.
\end{enumerate}
Then:
\begin{enumerate}[label=\textnormal{(\alph*)}, leftmargin=2.2em, itemsep=2pt]
\item\label{it:periodicity_existsN}
There exists an integer $N\ge 1$ such that $N\,t(d,n)\in\Z$ for all $(d,n)$.
In particular, $\{t(d,n)\}\in \frac1N\Z/\Z$ takes only finitely many values,
and the map $(d,n)\mapsto \{t(d,n)\}$ is periodic in $n$ with period $P_n$ and in $d$ with period $P_d$, where
\begin{equation}\label{eq:periods-general}
  P_n=\mathrm{lcm}_j\bigl(\den(u_{j,1})\bigr),
  \qquad
  P_d=\mathrm{lcm}_j\bigl(\den(u_{j,2})\bigr).
\end{equation}

\item\label{it:rp-general}
Fix $(d,n)$ and assume $t(d,n)\in \Z_{(p)}+\frac1p\Z$ for some prime $p$.
Then a tie $t(d,n)\in\Z$ forces $r_p\bigl(t(d,n)\bigr)=0$.

\item\label{it:single-table-local}
Fix $p$ and decompose
\[
t=\sum_{j=0}^J F_j\bigl(U_j(d,n)\bigr),
\qquad
Q=\sum_{j\neq J} F_j\bigl(U_j(d,n)\bigr),
\qquad\text{so that }t=Q+F_J\bigl(U_J(d,n)\bigr).
\]
Assume $t\in \Z_{(p)}+\frac1p\Z$ and $Q\in \Z_{(p)}$.  Then any tie at $(d,n)$ forces
\begin{equation}\label{eq:rpfj0}
r_p\!\left(F_J\bigl(U_J(d,n)\bigr)\right)=0 .
\end{equation}
Moreover, suppose in addition that there exist $L\ge 1$ and an $L$--periodic function $\Tab:\R\to\Q$ such that
\begin{equation}\label{eq:FJTab}
F_J(y)=\Tab(Ly),
\end{equation}
where $\Tab$ is piecewise-affine on $[0,L]$ with knots at the integers $\{0,1,\dots,L\}$ and is evaluated by exact
linear interpolation between successive knots, and assume that for every $k\in\Z$,
\begin{equation}\label{eq:tab-increments}
\Tab(k)\in \Z_{(p)} \qquad\text{and}\qquad \Tab(k+1)-\Tab(k) \in \Z_{(p)}\setminus p\Z_{(p)} .
\end{equation}
If furthermore
\begin{equation}\label{eq:scaled-arg-local}
L\,U_J(d,n)\in \Z_{(p)}+\frac1p\Z,
\end{equation}
then any tie at $(d,n)$ forces the \emph{scaled argument} to be $p$--local:
\begin{equation}\label{eq:scaled-arg-inZp}
L\,U_J(d,n)\in \Z_{(p)}.
\end{equation}
\end{enumerate}
\end{proposition}

\begin{proof}
\ref{it:periodicity_existsN}
For each $j$, write $F_j$ on $[0,1]$ using rational breakpoints
$0=\xi_{j,0}<\cdots<\xi_{j,M_j}=1$ and rational affine data
\[
F_j(y)=a_{j,m}y+b_{j,m}\qquad (y\in[\xi_{j,m},\xi_{j,m+1}]),
\qquad a_{j,m},b_{j,m}\in\Q.
\]
Choose $N$ to clear all denominators that can appear in the following finite list:
\begin{equation}\label{eq:chooseN}
u_{j,k},\quad b_{j,m},\quad a_{j,m}u_{j,k}
\qquad (j=0,\dots,J,\ k=0,1,2,\ m=0,\dots,M_j-1).
\end{equation}
Such an $N$ exists because the list \eqref{eq:chooseN} is finite and consists of rationals.

Now fix $(d,n)$.  By construction of $N$ we have $U_j(d,n)\in \frac1N\Z$ for every $j$.
On any segment where $F_j(y)=a_{j,m}y+b_{j,m}$ we then have
\[
F_j\bigl(U_j(d,n)\bigr)=a_{j,m}U_j(d,n)+b_{j,m}.
\]
Since $a_{j,m}u_{j,k}\in \frac1N\Z$ for $k=0,1,2$, it follows that
$a_{j,m}U_j(d,n)\in \frac1N\Z$, and since $b_{j,m}\in \frac1N\Z$ we obtain
$F_j(U_j(d,n))\in \frac1N\Z$ for all $j$.  Summing over $j$ yields $t(d,n)\in \frac1N\Z$ or $Nt(d,n)\in\Z$.

The explicit periods \eqref{eq:periods-general} follow from
\[
U_j(d,n+P_n)\equiv U_j(d,n)\pmod{1},
\qquad
U_j(d+P_d,n)\equiv U_j(d,n)\pmod{1},
\]
and $1$--periodicity of all nontrivial $F_j$.

\smallskip\noindent\ref{it:rp-general}
Write $t=t(d,n)$ and decompose $t=t_{(p)}+r_p(t)/p$ as above.
If $t\in\Z$, then $r_p(t)/p=t-t_{(p)}\in \Z-\Z_{(p)}\subset \Z_{(p)}$,
as subtracting an integer preserves membership in $\Z_{(p)}$.
But $r_p(t)/p\in \frac1p\Z$, and the only elements of $\frac1p\Z$ that lie in $\Z_{(p)}$ are the integers;
hence $r_p(t)/p\in\Z$, forcing $r_p(t)=0$.

\smallskip\noindent\ref{it:single-table-local}
Since $Q\in \Z_{(p)}$ we have $r_p(Q)=0$.  As $t=Q+F_J(U_J)$ and $r_p$ is defined on $\Z_{(p)}+\frac1p\Z$,
it follows that
\[
r_p(t)=r_p\!\left(F_J\bigl(U_J(d,n)\bigr)\right).
\]
If a tie occurs, $t\in\Z$, hence \eqref{it:rp-general} gives $r_p(t)=0$, proving \eqref{eq:rpfj0}.
For the final assertion, write $v=L\,U_J(d,n)=k+\theta$ with $k\in\Z$ and $\theta\in[0,1)$.
Exact interpolation on $[k,k+1]$ gives
\[
F_J\bigl(U_J(d,n)\bigr)=\Tab(v)=\Tab(k)+\theta\,\Delta_k,\qquad \Delta_k:=\Tab(k+1)-\Tab(k).
\]
By \eqref{eq:tab-increments}, $\Tab(k)\in \Z_{(p)}$ and $\Delta_k\in \Z_{(p)}\setminus p\Z_{(p)}$, so
\[
r_p\!\left(F_J\bigl(U_J(d,n)\bigr)\right)=r_p(\theta\,\Delta_k).
\]
Moreover, \eqref{eq:scaled-arg-local} implies $\theta=\{v\}\in \frac1p\Z$, so $\theta=r/p$ for some
$r\in\{0,1,\dots,p-1\}$.  If $r\neq 0$, then $(r/p)\Delta_k\notin \Z_{(p)}$ (because $\Delta_k$ has reduced numerator not
divisible by $p$), hence $r_p(\theta\Delta_k)\neq 0$.  Therefore \eqref{eq:rpfj0} forces $r=0$, i.e.\ $\theta=0$, and thus
$v=L\,U_J(d,n)\in\Z\subset \Z_{(p)}$, proving \eqref{eq:scaled-arg-inZp}.
\end{proof}

Let us apply this proposition to the concrete traditions.
We first record the explicit principal periods, then extract the two prime localisations ($p=101$ and $p=67$)
that become the algorithmic prefilters in \S\ref{sss:ties}.

\begin{remark}
\label{rem:princ-period}
In the Tibetan-style true-date setup, take $J=2$ and write
\[
  t(d,n)=F_0\bigl(U_0(d,n)\bigr)+F_1\bigl(U_1(d,n)\bigr)+F_2\bigl(U_2(d,n)\bigr),
\]
with
\[
  F_0(y)=y,\qquad
  F_1(y)=\frac1{60}\,\moonTab(28y),\qquad
  F_2(y)=-\frac1{60}\,\sunTab(12y),
\]
and affine phases
\[
  U_0(d,n)=m_0+nm_1+dm_2,\qquad
  U_1(d,n)=A_{\moon}(d,n),\qquad
  U_2(d,n)=A_{\sun}(d,n),
\]
where $A_{\moon}$ and $A_{\sun}$ are as in \eqref{eq:anom_linear_new}.
The \emph{principal} traditions (Phugpa, Tsurphu, Mongolian, Bhutan) share the parameters
\begin{equation}\label{eq:principal_motions}
m_1=\frac{167025}{5656},\quad
s_1=\frac{65}{804},\quad
a_1=\frac{253}{3528},
\qquad
m_2=\frac{m_1}{30},\quad s_2=\frac{s_1}{30},\quad a_2=\frac{1}{28},
\end{equation}
and differ only in the epoch offsets $(m_0,s_0,a_0)$.

For fixed $d$, the fractional part of the mean date
\[
M_d(n) = U_0(d,n) = m_0+n m_1+d m_2 
\]
has period
\[
P_{\rm md}=\den(m_1)=5656=2^3\cdot 7\cdot 101.
\]
The solar and lunar table arguments advance with step sizes $s_1$ and $a_1$ modulo $1$, hence have
periods
\[
P_\sun=\den(s_1)=804=2^2\cdot 3\cdot 67,
\qquad
P_\moon=\den(a_1)=3528=2^3\cdot 3^2\cdot 7^2.
\]
Therefore the fractional part of the combined correction term
\[
C_d(n)= F_1\bigl(U_1(d,n)\bigr)+F_2\bigl(U_2(d,n)\bigr)
\]
has period
\[
P_{\rm corr}=\lcm(P_\sun,P_\moon)
=\lcm(2^2\cdot 3\cdot 67,\ 2^3\cdot 3^2\cdot 7^2)
=2^3\cdot 3^2\cdot 7^2\cdot 67
=236\,376.
\]
Hence the full true-date fractional part $\{t(d,n)\}=\{M_d(n)+C_d(n)\}$ is periodic in $n$ with period
\[
P_n=\lcm(P_{\rm md},P_{\rm corr}).
\]
Since $\gcd(5656,236\,376)=2^3\cdot 7=56$, we compute explicitly
\[
P_n=\frac{5656\cdot 236\,376}{56}=101\cdot 236\,376=23\,873\,976.
\]
By Proposition~\ref{prop:periodicity_prime_obstruction}\textnormal{(a)}, for each fixed $d$ the tie
condition $t(d,n)\in\Z$ depends only on $n\bmod P_n$ and thus forms a union of residue classes modulo
$P_n=23\,873\,976$.
This corresponds to nearly 2 million years, cf. \cite[\S13]{janson}.
\end{remark}

Hence, for each fixed $d$ the tie problem reduces to a finite search in $n$ modulo an explicit
period, which is the starting point of the meet-in-the-middle/CRT algorithm in \S\ref{sss:ties}.

\begin{remark}
\label{rem:local-101}
Writing $t(d,n)=M_d(n)+C_d(n)$, we have
\[
C_d(n)\in \frac{1}{60\,P_{\rm corr}}\Z.
\]
In the principal case,
\[
60\,P_{\rm corr}=60\cdot 236376
=2^4\cdot 3\cdot (2^3\cdot 3^2\cdot 7^2\cdot 67),
\]
so $101\nmid 60P_{\rm corr}$ and hence $C_d(n)\in\Z_{(101)}$ for all $n$.
Thus any tie $t(d,n)\in\Z$ forces
\[
M_d(n)=t(d,n)-C_d(n) \in\Z_{(101)}.
\]
\end{remark}

\begin{remark}
\label{rem:local-67}
In view of applying Proposition~\ref{prop:periodicity_prime_obstruction}\textnormal{(c)},
we single out the solar term:
\[
  F_2\bigl(U_2(d,n)\bigr)=-\frac1{60}\,\sunTab\bigl(v(d,n)\bigr),
\]
with
\[
v(d,n)=12U_2(d,n)=12\big(\textstyle-\frac14+s_0+n s_1+d s_2\big) ,
\]
and write $t=Q+F_2(U_2(d,n))$.
For the principal case \eqref{eq:principal_motions} we compute
\[
v(d,n)=v_0+\frac{65}{67}\,n+\frac{13}{67\cdot6}\,d .
\]
It is immediate that $v(d,n)\in\Z_{(67)}+\frac1{67}\Z$, and the map $F_2$ does not introduce additional $67$ in the denominator, so $F_2(U_2(d,n))\in\Z_{(67)}+\frac1{67}\Z$.
Moreover, non-solar part $Q$ does not contain $67$ in the denominator, meaning that $Q\in\Z_{(67)}$ and so $t\in\Z_{(67)}+\frac1{67}\Z$.

The table $F_2$ as a function of $v(d,n)$ satisfies the conditions of Proposition~\ref{prop:periodicity_prime_obstruction}\textnormal{(c)}, 
and we conclude that any tie $t(d,n)\in\Z$ forces $v(d,n)\in\Z_{(67)}$.
\end{remark}

The preceding remarks show that a tie $t(d,n)\in\Z$ forces certain distinguished affine quantities
(notably the mean-date phase $M_d(n)$ and the raw solar-table argument $v(d,n)$) to lie in a localisation
ring $\Z_{(p)}$ for suitable primes $p$.  We now turn these membership conditions into explicit residue
class constraints on $n$.  The next proposition is a purely algebraic device: it converts a condition of
the form $X(n)\in\Z_{(p)}$ for an affine function $X(n)=x_0+n x_1$ into a \emph{single linear congruence}
modulo $p$.  This is exactly what is needed to make the $101$-- and $67$--filters algorithmic, i.e.\ to
replace them by explicit residue classes $n\equiv n_d\pmod{101}$ and $n\equiv \nu_d\pmod{67}$ (for each
fixed $d$).

\begin{proposition}[$p$--local affine congruence principle]\label{prop:p-local-affine}
Fix a prime $p$.  Let
\[
X(n)=x_0+n x_1,\qquad x_0,x_1\in\Q,
\]
be an affine function of $n\in\Z$.  Assume that
\begin{equation}\label{eq:x1-onep}
x_1=\frac{A}{pB},\qquad A,B\in\Z,\qquad \gcd(A,p)=\gcd(B,p)=1,
\end{equation}
i.e.\ $x_1$ has \emph{exactly one factor $p$ in the denominator}.  Assume moreover that
\begin{equation}\label{eq:x0-local-plus}
x_0\in \Z_{(p)}+\frac1p\Z .
\end{equation}
Then:
\begin{enumerate}[label=\textnormal{(\alph*)}, leftmargin=2.2em, itemsep=2pt]
\item\label{it:uniqueclass}
There exists a \emph{unique} residue class $\nu\in\{0,1,\dots,p-1\}$ such that
\begin{equation}\label{eq:uniqueclass}
X(n)\in \Z_{(p)} \quad\Longleftrightarrow\quad n\equiv \nu \pmod p.
\end{equation}
Equivalently, the condition $X(n)\in\Z_{(p)}$ is a single linear congruence modulo $p$.

\item\label{it:explicitnu}
More explicitly, write $x_0=x_{(p)}+\frac{r_0}{p}$ with $x_{(p)}\in\Z_{(p)}$ and $r_0\in\{0,1,\dots,p-1\}$,
and write $x_1=\frac{1}{p}\cdot \frac{A}{B}$ with $A,B$ as in \eqref{eq:x1-onep}.  Then the unique class in
\eqref{eq:uniqueclass} is given by
\begin{equation}\label{eq:nu-formula}
\nu \equiv -\,r_0\,(A B^{-1})^{-1} \pmod p,
\end{equation}
where $B^{-1}$ and $(AB^{-1})^{-1}$ denote inverses modulo $p$.
\end{enumerate}
\end{proposition}

\begin{proof}
By \eqref{eq:x0-local-plus} there exist $x_{(p)}\in\Z_{(p)}$ and $r_0\in\{0,\dots,p-1\}$ such that
$x_0=x_{(p)}+\frac{r_0}{p}$.
Also \eqref{eq:x1-onep} implies $x_1=\frac{1}{p}\cdot \frac{A}{B}$ with $\frac{A}{B}\in\Z_{(p)}$, hence
\[
X(n)=x_{(p)}+\frac{1}{p}\Bigl(r_0+n\frac{A}{B}\Bigr).
\]
Since $x_{(p)}\in\Z_{(p)}$, we have $X(n)\in\Z_{(p)}$ if and only if the parenthesis
$r_0+n\frac{A}{B}$ is divisible by $p$ in the sense of $\Z_{(p)}$, i.e.\ iff
\[
r_0+n\frac{A}{B}\in p\,\Z_{(p)}.
\]
Reducing modulo $p$ (and using that $B$ is invertible modulo $p$) this is equivalent to the linear
congruence
\[
r_0+(A B^{-1})\,n\equiv 0\pmod p.
\]
Because $\gcd(A,p)=1$, the coefficient $AB^{-1}$ is invertible modulo $p$, so this congruence has a unique
solution class $n\equiv \nu\pmod p$, proving \eqref{it:uniqueclass}.  Solving explicitly yields
\eqref{eq:nu-formula}.
\end{proof}

Assume the principal siddh\=anta mean motions as in \eqref{eq:principal_motions}, and normalize civil-day boundaries to $\Z$.
Fix $d$.  By Remarks~\ref{rem:local-101}--\ref{rem:local-67}, any tie at $(d,n)$ forces
\[
M_d(n)\in \Z_{(101)}
\qquad\text{and}\qquad
v(d,n)\in \Z_{(67)},
\]
where $M_d(n)=m_0+n m_1+d m_2$ and $v(d,n)=12(-\tfrac14+s_0+n s_1+d s_2)$
is the \emph{raw} solar-table argument.

\begin{remark}[The $101$-filter for $M_d(n)$]
Write
\[
M_d(n)=x_0(d)+n x_1,\qquad x_0(d)=m_0+d m_2,\qquad x_1=m_1.
\]
For principal parameters one has
\[
m_1=\frac{167025}{5656}=\frac{167025}{101\cdot 56},
\]
so Proposition~\ref{prop:p-local-affine} applies with $p=101$, $A=167025$, $B=56$.
Reducing modulo $101$ gives
\[
A\equiv 72 \pmod{101},\qquad 56^{-1}\equiv 92 \pmod{101},
\qquad AB^{-1}\equiv 72\cdot 92 \equiv 59 \pmod{101},
\]
hence
\[
(AB^{-1})^{-1}\equiv 59^{-1}\equiv 12 \pmod{101}.
\]
Moreover $m_2=m_1/30=\frac{167025}{101\cdot 1680}$, so
\[
r_{101}(m_2)\equiv 167025\cdot (1680^{-1}) \equiv 72\cdot 64^{-1}
\equiv 72\cdot 30 \equiv 39 \pmod{101},
\]
using $1680\equiv 64$ and $64^{-1}\equiv 30 \pmod{101}$.
By additivity of $r_{101}$ on $\Z_{(101)}+\frac1{101}\Z$ we obtain
\[
r_{101}\bigl(x_0(d)\bigr)\equiv r_{101}(m_0)+ d\,r_{101}(m_2)\equiv r_{101}(m_0)+39d \pmod{101}.
\]
Therefore Proposition~\ref{prop:p-local-affine} yields a unique residue class
\[
n\equiv n_d \pmod{101},
\qquad
n_d\equiv -\,r_{101}\bigl(x_0(d)\bigr)\,(AB^{-1})^{-1}
\equiv -12\bigl(r_{101}(m_0)+39d\bigr)
\equiv n_0+37d \pmod{101},
\]
where $n_0\equiv -12\,r_{101}(m_0)\pmod{101}$ and $-12\cdot 39\equiv 37\pmod{101}$.
\end{remark}

\begin{remark}[The $67$-filter for $v(d,n)$]
Write
\[
v(d,n)=y_0(d)+n y_1,
\qquad
y_1=\frac{65}{67},\qquad
y_0(d)=12\Bigl(-\frac14+s_0\Bigr)+12d\,s_2.
\]
For principal parameters $s_2=s_1/30=\frac{65}{804\cdot 30}$, and a direct simplification gives
\[
12d\,s_2=\frac{13}{402}\,d=\frac{13}{67\cdot 6}\,d.
\]
Thus Proposition~\ref{prop:p-local-affine} applies with $p=67$, $A=65$, $B=1$, and $y_0(d)\in \Z_{(67)}+\frac1{67}\Z$.
Since $65\equiv -2\pmod{67}$ one has
\[
65^{-1}\equiv (-2)^{-1}\equiv -34\equiv 33 \pmod{67}.
\]
Also
\[
r_{67}\!\left(\frac{13}{67\cdot 6}\right)\equiv 13\cdot 6^{-1}\equiv 13\cdot 56\equiv 58 \pmod{67},
\]
since $6^{-1}\equiv 56\pmod{67}$.
Noting that $12(-\frac14+s_0)\equiv 12s_0 \pmod{1}$, we obtain
\[
r_{67}\bigl(y_0(d)\bigr)\equiv r_{67}(12s_0)+58d \pmod{67},
\]
hence Proposition~\ref{prop:p-local-affine} gives a unique residue class
\[
n\equiv \nu_d \pmod{67},
\qquad
\nu_d\equiv -\,r_{67}\bigl(y_0(d)\bigr)\,65^{-1}
\equiv -33\bigl(r_{67}(12s_0)+58d\bigr)
\equiv \nu_0+29d \pmod{67},
\]
where $\nu_0\equiv -33\,r_{67}(12s_0)\pmod{67}$ and $-33\cdot 58\equiv 29\pmod{67}$.
\end{remark}

\medskip
\noindent
To summarize: for the principal siddh\=anta parameter set, a tie $t(d,n)\in\Z$ forces the two
localisation conditions
\[
M_d(n)\in\Z_{(101)}
\qquad\text{and}\qquad
v(d,n)\in\Z_{(67)},
\]
hence (by Proposition~\ref{prop:p-local-affine}) forces \emph{congruence} conditions on the month index:
for each fixed $d$ there are unique residues $n_d\in\Z/101$ and $\nu_d\in\Z/67$ such that
\[
t(d,n)\in\Z\quad\Longrightarrow\quad n\equiv n_d \pmod{101}
\qquad\text{and}\qquad
n\equiv \nu_d \pmod{67}.
\]
Moreover, the dependence on $d$ is universal across the principal traditions:
\[
n_d\equiv n_0+37d \pmod{101},
\qquad
\nu_d\equiv \nu_0+29d \pmod{67},
\]
with tradition-dependent intercepts $n_0,\nu_0$ determined by the epoch constants $(m_0,s_0)$.
By the Chinese remainder theorem, for each fixed $d$ these two constraints combine into a \emph{single}
residue class modulo $6767=101\cdot 67$.  This is the key arithmetic prefilter used in the fast
meet-in-the-middle/CRT search of \S\ref{sss:ties}.  For convenience, the explicit intercepts for the
principal traditions are listed in Table~\ref{tab:principal-filters}.

\begin{table}[ht]
\centering
\begin{tabular}{lcc}
\toprule
Tradition & $n\equiv n_0+37d\pmod{101}$ & $n\equiv \nu_0+29d\pmod{67}$\\
\midrule
Phugpa & $n_0=25$ & $\nu_0=21$\\
Tsurphu (E1732) & $n_0=67$ & $\nu_0=10$\\
Tsurphu (E1852) & $n_0=97$ & $\nu_0=46$\\
Bhutan & $n_0=84$ & $\nu_0=6$\\
Mongol & $n_0=82$ & $\nu_0=62$\\
\bottomrule
\end{tabular}
\caption{Explicit residue classes forced by the $101$- and $67$-filters for the principal traditions.
For each fixed $d$, a tie can occur only if $n$ lies simultaneously in these two classes.}
\label{tab:principal-filters}
\end{table}

\section{Mean-Elongation Model and Skipped-Day Arithmetic}
\label{app:arith-day}

We record a simplified day rule obtained by pushing the mean-motion viewpoint of
\S\ref{ss:mean_sun_models} one level down, from months to days.  The idea is to replace the full
day-boundary computation of \S\ref{ss:day_models} by the simplest possible celestial input: assume
that the Sun--Moon elongation increases at a constant mean rate from one dawn to the next, and use
its value at dawn to assign the lunar-day label of the civil day that begins there.

If the mean elongation advance per dawn step is chosen to be a rational number, the resulting day
labels are generated by a rigid arithmetic counter, so the skip pattern is periodic and can be
decided by a single congruence test (a day-level analogue of the intercalation index).  When the
mean advance exceeds one lunar day per civil day, repeated dates are impossible.  Pairing this day
rule with the month model of \S\ref{ss:mean_sun_models} yields a low-commitment calendar variant,
which we call the {\em L0} reform, in the sense of \S\ref{ss:low_commitment}.

\subsection{Model set-up}

Let $j\in\mathbb{Z}$ enumerate successive civil days (dawn to dawn).  Let
$\Delta:\mathbb{Z}\to\mathbb{R}$ be a chosen lift of the mean elongation at dawn $j$, measured in
lunar-day units (so $1$ unit $=1/30$ turn).  We assume a linear mean law
\begin{equation}\label{eq:daymodel_linear}
\Delta(j)=\Delta_0 + j\,\frac{U}{V}
\qquad
(\Delta_0\in\R,\ U,V\in\N),
\end{equation}
where $U/V$ (in lowest terms) is the chosen mean elongation advance per civil day, expressed in
lunar days.
To obtain a ``skips-only'' day rule it is enough to assume
\begin{equation}\label{eq:daymodel_skips_only}
\frac{U}{V}>1,
\end{equation}
since then $\Delta(j+1)-\Delta(j)>1$ and the inherited label advances at least by one at each dawn,
so repeated dates are impossible.  In any reasonable calendar one also has $U/V<2$, which rules out
multiple skipped labels within a single dawn-to-dawn step; we will keep this mild restriction in
place for simplicity.  Write
\begin{equation}\label{eq:daymodel_kappa}
\kappa = U-V \qquad (0<\kappa<V),
\end{equation}
so that $U/V=1+\kappa/V$.

The inherited lunar-day index at dawn $j$ is encoded by the floor counter
\begin{equation}\label{eq:daymodel_A}
A_j = \lfloor \Delta(j)\rfloor
      = \left\lfloor \Delta_0 + j\frac{U}{V}\right\rfloor \in \mathbb{Z}.
\end{equation}

\begin{lemma}\label{lem:daymodel_basic}
Assume \eqref{eq:daymodel_linear}--\eqref{eq:daymodel_A} with $1<U/V<2$ and $\kappa=U-V$.
Then:
\begin{enumerate}
\item[(a)] For every $j$,
\begin{equation}\label{eq:daymodel_inc12}
A_{j+1}-A_j\in\{1,2\}.
\end{equation}
Moreover, the civil day starting at dawn $j$ is \emph{regular} iff $A_{j+1}-A_j=1$, and it \emph{skips}
exactly one lunar-day label iff $A_{j+1}-A_j=2$.

\item[(b)] One has the explicit decomposition
\begin{equation}\label{eq:daymodel_split}
A_j = j + \left\lfloor \Delta_0 + j\frac{\kappa}{V}\right\rfloor.
\end{equation}

\item[(c)] For every $j$,
\begin{equation}\label{eq:daymodel_period}
A_{j+V}-A_j = U.
\end{equation}
Equivalently, among the $V$ increments
$\{A_{j+1}-A_j,\dots,A_{j+V}-A_{j+V-1}\}$ there are exactly $\kappa$ twos and $V-\kappa$ ones.
In particular, the L0 rule forces exactly $\kappa$ skipped labels per $V$ civil days.

\item[(d)] Fix an epoch day $j_0$ and define $D^*(j):=A_j-A_{j_0}\in\mathbb{Z}$.
Then $D^*(j)$ tracks the elapsed (mean-model) lunar-day count from the epoch.
\end{enumerate}
\end{lemma}

\begin{proof}
(a) Since $\Delta(j+1)-\Delta(j)=U/V\in(1,2)$, the floor can increase only by $1$ or $2$.

(b) Insert $U=V+\kappa$ into \eqref{eq:daymodel_A}:
$\Delta_0+j\frac{U}{V}=\Delta_0 + j + j\frac{\kappa}{V}$, then take floors.

(c) From \eqref{eq:daymodel_split},
\[
A_{j+V}-A_j
=V+\Bigl\lfloor \Delta_0+(j+V)\frac{\kappa}{V}\Bigr\rfloor
 -\Bigl\lfloor \Delta_0+j\frac{\kappa}{V}\Bigr\rfloor
=V+\kappa=U.
\]
Summing $V$ increments each in $\{1,2\}$ forces exactly $\kappa$ twos.

(d) Immediate.
\end{proof}

\begin{example}\label{ex:UV_three}
To give the skip counts in Lemma~\ref{lem:daymodel_basic} a concrete scale, let us compare three choices.  
In the principal Tibetan siddh\=anta (grub-rtsis) choice one has
\[
\frac{U}{V}=\frac{11312}{11135},\qquad \kappa=177.
\]
The K\=alacakra karana (byed-rtsis) truncation gives a slightly different daily advance,
\[
\frac{U}{V}=\frac{10800}{10631},\qquad \kappa=169,
\]
reflecting the well-known karana practice of truncating a more accurate constant for manual
computation; in the K\=alacakra literature this corresponds to truncating the monthly increment
$1;31,50,0,480\,(707)$ to $1;31,50$ (cf.\ \cite{henning}).  Finally, as a modern high-precision
benchmark we may use
\[
\frac{U}{V}=\frac{143925}{141673},\qquad \kappa=2252.
\]
In all three cases, the typical spacing between skipped labels is about $63$ days.
\end{example}

\subsection{Skip-day index (chad-index) and the inverse map}

Skipped labels are exactly the places where the floor counter \eqref{eq:daymodel_A} jumps by $2$.
As in the leap-month discussion of \S3.2, the jump set admits an equivalent \emph{congruence test}
that decides whether a lunar-day label is absent without scanning the timeline.  We call this the
\emph{skip-day index}; in Tibetan terminology, a skipped day is a \emph{chad} (``cut'') day, hence
\emph{chad-index}.

\begin{definition}\label{def:day_index}
Let $U,V$ be as above and define
\begin{equation}\label{eq:daymodel_delta_def}
\delta=-V\Delta_0,\qquad \delta^*=\lceil \delta\rceil-1\in\mathbb{Z}.
\end{equation}
For $K\in\mathbb{Z}$ define the \emph{skip-day index} (chad-index)
\begin{equation}\label{eq:daymodel_chi}
\chi(K)=(VK+\delta^*)\bmod U \ \in \ \{0,1,\dots,U-1\}.
\end{equation}
\end{definition}

\begin{proposition}\label{prop:daymodel_skip_inverse}
Let $S=\{0,1,\dots,\kappa-1\}\subset\mathbb{Z}/U\mathbb{Z}$ where we recall $\kappa=U-V$.
\begin{enumerate}
\item[(a)] A lunar-day label $K\in\mathbb{Z}$ is skipped (i.e.\ does not occur among the $A_j$) iff
\begin{equation}\label{eq:daymodel_skiptest}
\chi(K)\in S.
\end{equation}
Equivalently, fixing an epoch $j_0$ and writing $\delta_0^*\equiv(VA_{j_0}+\delta^*)\pmod U$,
the elapsed label $K^*=K-A_{j_0}$ is skipped iff
\begin{equation}\label{eq:daymodel_skiptest_elapsed}
(VK^*+\delta_0^*)\bmod U \in S.
\end{equation}

\item[(b)] Define the inverse map (label $\to$ civil day) by
\begin{equation}\label{eq:daymodel_inverse}
J(K)=\left\lceil \frac{V(K+1-\Delta_0)}{U}\right\rceil-1.
\end{equation}
If $K$ occurs, then $A_{J(K)}=K$; if $K$ is skipped, then $A_{J(K)}=K-1$.
\end{enumerate}
\end{proposition}

\begin{proof}
(a) The label $K$ appears at some dawn $j$ iff $A_j=\lfloor \Delta(j)\rfloor=K$, i.e.
\[
K\le \Delta_0+j\frac{U}{V}<K+1.
\]
This is equivalent to the existence of an integer $j$ in the half-open interval
\(I_K=[a,b)\) with
\[
a=\frac{V(K-\Delta_0)}{U}=\frac{VK+\delta}{U},\qquad b=a+\frac{V}{U}.
\]
If $a\in\Z$ then $a\in I_K$ and $K$ is not skipped.  Otherwise $a=m+r$ with $m\in\Z$ and
$r\in(0,1)$, and the unique candidate integer is $m+1=\lceil a\rceil$; it lies in $I_K$ iff
$\lceil a\rceil<b$, i.e.\ iff $r+V/U>1$.  Thus $K$ is skipped iff $r\le 1-V/U=\kappa/U$.

Now write $\delta=\delta^*+\theta$ with $\delta^*=\lceil \delta\rceil-1\in\Z$ and $\theta\in(0,1]$.
Then
\[
a=\frac{VK+\delta^*}{U}+\frac{\theta}{U},
\qquad
r=\{a\}=\frac{(VK+\delta^*)\bmod U+\theta}{U}.
\]
Since $\theta>0$, the inequality $r\le \kappa/U$ is equivalent to
\[
(VK+\delta^*)\bmod U\in\{0,1,\dots,\kappa-1\}=S,
\]
which is \eqref{eq:daymodel_skiptest}.  Shifting by $A_{j_0}$ replaces $K$ by $K^*=K-A_{j_0}$ and
absorbs the constant $VA_{j_0}$ into $\delta_0^*\equiv(VA_{j_0}+\delta^*)\pmod U$, yielding
\eqref{eq:daymodel_skiptest_elapsed}.

\medskip
(b) Let $b=\frac{V(K+1-\Delta_0)}{U}$ be the right endpoint of $I_K$.  Since $I_K$ is open at $b$,
the correct choice is the largest integer strictly less than $b$, namely $\lceil b\rceil-1$, which
is exactly \eqref{eq:daymodel_inverse}.  If $K$ occurs then $I_K$ contains a unique integer, and
necessarily $J(K)$ is that integer, so $A_{J(K)}=K$.  If $K$ is skipped then $I_K$ contains no
integer, hence the largest integer $<b$ is the same as for $K-1$ (the two intervals have the same
right endpoint), so $J(K)=J(K-1)$ and therefore $A_{J(K)}=K-1$.
\end{proof}

\begin{corollary}\label{cor:L0_29_30}
Let $L\ge 1$ and consider any block of $L$ consecutive lunar-day labels
$\{K,K+1,\dots,K+L-1\}$.  In the L0 model, the number of skipped labels in such a block is either
$\lfloor L\kappa/U\rfloor$ or $\lceil L\kappa/U\rceil$, where $\kappa=U-V$.
In particular, if
\(30\kappa < {U}\),
then every block of $30$ consecutive labels contains at most one skipped label.  Consequently,
when month~$n$ is linearized by the $30$ labels $K=30n,30n+1,\dots,30n+29$, the L0 month length is
always $30$ days (no skip) or $29$ days (one skip).
\end{corollary}

\begin{proof}
Let $J(K)$ be the inverse map from Proposition~\ref{prop:daymodel_skip_inverse}(b).  A label $K$ is
present iff $J(K)-J(K-1)=1$ and skipped iff $J(K)-J(K-1)=0$, so the number of present labels in
$\{K,\dots,K+L-1\}$ telescopes to $J(K+L-1)-J(K-1)$.  Using
$J(K)=\left\lceil \frac{V(K+1-\Delta_0)}{U}\right\rceil-1$ gives
\[
J(K+L-1)-J(K-1)
=
\left\lceil \frac{V(K+L-\Delta_0)}{U}\right\rceil
-
\left\lceil \frac{V(K-\Delta_0)}{U}\right\rceil.
\]
Writing $x=\frac{V(K-\Delta_0)}{U}$ and $a=\frac{LV}{U}$, this equals $\lceil x+a\rceil-\lceil x\rceil$,
which is either $\lfloor a\rfloor$ or $\lceil a\rceil$.  Hence the number of present labels is either
$\lfloor LV/U\rfloor$ or $\lceil LV/U\rceil$, and subtracting from $L$ yields the stated alternatives
$\lfloor L\kappa/U\rfloor$ or $\lceil L\kappa/U\rceil$ for skipped labels.  If $\kappa<U/30$ then
$30\kappa/U<1$, hence $\lceil 30\kappa/U\rceil\le 1$, proving the last claim.
\end{proof}

\begin{example}\label{ex:chi_realizations}
Fix a month count $n$ and consider the labels $K=30n,30n+1,\dots,30n+29$.  Since
\[
\chi(K+1)\equiv \chi(K)+V \equiv \chi(K)-\kappa \pmod U,
\]
writing $r_n:=\chi(30n)\in\{0,1,\dots,U-1\}$ gives, for $d=1,\dots,30$,
\begin{equation}\label{eq:chi_within_month}
\chi(30n+d-1)\equiv r_n-(d-1)\kappa \pmod U.
\end{equation}
Under $30\kappa<U$, at most one value of $d$ can place this
residue in the skip window $\{0,1,\dots,\kappa-1\}$.  If $r_n<30\kappa$, the unique skipped label is
\[
d_{\mathrm{skip}}(n)=\Bigl\lfloor \frac{r_n}{\kappa}\Bigr\rfloor+1,
\]
whereas if $r_n\ge 30\kappa$ then month~$n$ has no skipped label and hence has $30$ civil days.
\end{example}

\subsection{Practitioner pipeline: from month count to day labels}

In applications, the month model of \S\ref{ss:mean_sun_models} supplies the true month count
$n\in\Z$ (epoch $n=0$), and one works with absolute lunar-day labels inside the month.  Given a month
count $n$ and a lunar-day label $d\in\{1,\dots,30\}$, we linearize the pair by
\[
K:=30n+(d-1)\in\Z.
\]
The L0 day layer is then used only to decide whether the label $K$ occurs or is skipped, and (if
desired) to map an occurring label back to the civil-day index.

The day model is specified by (i) a rational mean elongation advance $U/V$ per dawn step, and (ii) a
single phase residue $\eta\in\Z/U\Z$ (equivalently $\delta^*\bmod U$).  In our library, astronomical
epochs are expressed in TT; to align the day rule with UT-based civil days one may first convert the
epoch new-moon time by $m_0^{\mathrm{UT}}=m_0^{\mathrm{TT}}-\Delta T/86400$.  Since L0 is purely
arithmetical, this choice affects only the phase residue $\eta$.

Fix a reference longitude $\lambda$ (in turns, positive east) and a constant mean dawn time
$\sigma$ (we take $\sigma=1/4$ for 6:00am local mean time).  Working on the local dawn-shifted clock,
\[
m_0^{\mathrm{loc}}:=m_0^{\mathrm{UT}}+\lambda+\Bigl(\frac12-\sigma\Bigr),
\qquad
f:=\{m_0^{\mathrm{loc}}\}\in[0,1),
\]
so that $f$ is the fraction of a civil day between the epoch dawn $\lfloor m_0^{\mathrm{loc}}\rfloor$
and the epoch new moon.  In the constant-rate mean model this determines the elongation phase at the
epoch dawn,
\[
\Delta_0=-f\,\frac{U}{V},
\]
and hence the stored integer phase and residue
\[
\delta^*=\lceil -V\Delta_0\rceil-1=\lceil fU\rceil-1,
\qquad
\eta\equiv\delta^*\pmod U.
\]
With these choices the skip-day index is
\[
\chi(K)=(VK+\eta)\bmod U,
\]
and $K$ is skipped iff $\chi(K)\in\{0,1,\dots,U-V-1\}$.

\section{Modern Astronomical Constants and Models}
\label{app:astro}

This appendix records modern benchmark constants and standard computational models that are useful for
(i) validating arithmetic calendars, and (ii) designing reforms with explicit accuracy targets.
Unless stated otherwise, all time arguments are understood on the uniform \emph{Terrestrial Time} (TT) scale,
and angles are in degrees (or turns) reduced modulo one revolution.

\subsection{Astronomical and civil time scales}
\label{app:tt}


Civil timekeeping is based on the rotation of the Earth, which is irregular and gradually slowing down. 
In contrast, the laws of physics (and thus planetary motions) operate on a strictly uniform time scale.
Bridging these two domains requires defining specific time standards \cite{IERS2010}.

\begin{itemize}
    \item {Terrestrial Time ($\TT$):} This is the uniform time scale used for ephemerides. It is independent of Earth's rotation and is related to International Atomic Time ($\TAI$) by a fixed offset: $\TT = \TAI + 32.184\,\mathrm{s}$.
    \item {Universal Time ($\UTone$):} This scale is defined by the actual rotation of the Earth relative to the mean sun. Because the Earth's rotation is erratic, $\UTone$ is not a uniform time scale.
    \item {Coordinated Universal Time ($\UTC$):} This is the basis for civil clocks worldwide. It is an atomic time scale that is kept within $\pm 0.9$ seconds of $\UTone$ by the occasional insertion of leap seconds.
\end{itemize}

The difference between uniform time and rotational time is captured by the correction factor $\Delta T$, defined as $\Delta T = \TT - \UTone$. The value of $\Delta T$ changes irregularly; for example, it was approximately $+63.8\,\mathrm{s}$ in 2000 and is currently approximately $+69.2\,\mathrm{s}$ in early 2026.

For applications requiring only minute-level precision, the irregular fluctuations of Earth's rotation can be approximated by a single unified quadratic model \cite{Morrison2004,Espenak2006}. 
This approach prioritizes long-term secular stability over decadal "jitters," representing the mean rate of tidal deceleration across millennia. 
The following formula remains physically consistent with the tidal braking recorded in ancient eclipse catalogs while maintaining a one-minute error budget for the modern era, cf. Figure~\ref{fig:Delta-T}:
\begin{equation}\label{e:delta-t-quad}
\Delta T \approx -20 + 32 u^2
\end{equation}
where $u = (Y - 1820)/100$ is the time in centuries from the anchor epoch of 1820.

\subsection{The Julian Date}
\label{app:jd}

While $\UTC$ provides a convenient civil clock (hours, minutes, seconds), calculating time intervals across months and years is cumbersome due to the variable lengths of months and leap years. 
To simplify these calculations, astronomers use the {\em Julian Date} $\JD$, a continuous count of days elapsed since a distant epoch.

The integer part of the Julian Date is defined as the {\em astronomical} Julian Day Number $\JDN$. 
While $\JDN$ is often used in historical chronology as a meridian-independent integer label for a "calendar day," astronomical calculations define it strictly as the integer count beginning at Greenwich Mean Noon (12:00 UT). 
Under this convention, the standard epoch J2000.0 (January 1, 2000, 12:00 TT) corresponds to exactly $\JD = 2451545.0$.

Strictly speaking, a Julian Date can be defined for any time scale (e.g., $\JD_{\mathrm{TT}}$ or $\JD_{\mathrm{UTC}}$). For the purposes of civil calendar conversion, we assume the input time argument is Coordinated Universal Time ($\UTC$). It is defined as:
\[
\JD_{\UTC} = \JDN + \frac{\text{seconds since 12{:}00}}{86400}
\]
To perform astronomical calculations requiring a uniform time scale, civil $\JD$ must be converted to Terrestrial Time ($\TT$) by applying the $\Delta T$ correction:
\begin{equation}
\JD_{\TT} \approx \JD_{\UTC} + \frac{\Delta T}{86400}.
\end{equation}
Note that this expression is an approximation within $\pm 0.9$ seconds, as $\Delta T$ is defined relative to $\UTone$ rather than $\UTC$.

\begin{example}
The discrepancy between civil and astronomical dating is most apparent when a single local calendar day spans two different Julian Day Numbers. For an observer in Beijing ($\mathrm{UTC+8}$), the astronomical $\JDN$ remains unchanged throughout the morning and afternoon because Greenwich noon has not yet been reached. At 20:00 (8:00 PM) local time, corresponding exactly to 12:00 UTC, the astronomical $\JDN$ increments to its next integer value. Consequently, a single civil ``calendar day'' contains two distinct astronomical day numbers, with the transition occurring in the evening rather than at local midnight or dawn.
\end{example}

\noindent
For practical implementation, the algebraic conversion from a civil Gregorian date $(Y, M, D)$ to its corresponding Julian Date is given by the following standard algorithm \cite{Meeus}.

\begin{algorithm}[H]
\caption{Conversion from Gregorian to Julian Date}
Given a Gregorian date $(Y, M, D)$, where $D$ may include a fractional day (e.g., $D=17.5$ indicates the noon of day 17):
\begin{enumerate}
    \item If $M \le 2$, set $Y \leftarrow Y-1$ and $M \leftarrow M+12$.
    \item Calculate the calendar correction terms:
    \[
    A = \left\lfloor \frac{Y}{100}\right\rfloor, \qquad B = 2 - A + \left\lfloor \frac{A}{4}\right\rfloor.
    \]
    (Note: For the proleptic Julian calendar, use $B=0$.)
    \item The Julian Date is given by:
    \[
    \JD = \left\lfloor 365.25(Y+4716)\right\rfloor + \left\lfloor 30.6001(M+1)\right\rfloor + D + B - 1524.5.
    \]
\end{enumerate}
\end{algorithm}

\subsection{Local time vs. local mean time}
\label{app:local-time}

In civil contexts, time is regulated by political boundaries known as time zones. 
\emph{Local Civil Time} is determined by adding a fixed integer (or half-integer) offset $z$ to UTC:
\[
\text{Local Civil Time} = \UTC + z .
\]
For example, Eastern Standard Time (EST) is defined as $\UTC - 5$, regardless of the observer's precise location within that zone.

However, astronomical calculations often require \emph{Local Mean Solar Time} (LMT), which is a strictly geometric quantity tied to the observer's longitude. It represents the time as defined by the mean sun crossing the local meridian. For an observer at longitude $\lambda$ (measured in degrees, positive to the East), the Local Mean Time is:
\[
\text{LMT} = \UTC + \frac{\lambda}{15}
\]
where $\lambda/15$ converts the longitude into hours.

\begin{example}
Consider an observer in Lhasa ($\lambda = 91.1^\circ$ E). Although Lhasa is geographically situated in a region corresponding to $\mathrm{UTC+6}$, it officially follows China Standard Time ($z = +8$). If the civil clock reads 12:00 (Noon), the corresponding $\mathrm{UTC}$ is 04:00. The Local Mean Time, however, is determined by the geometric offset of $+91.1/15 \approx +6.07$ hours, or roughly $+6\text{h } 04\text{m}$. Thus, when the clock strikes noon, the Local Mean Time is only $04:00 + 6\text{h } 04\text{m} = 10:04$. This discrepancy of nearly two hours illustrates why civil "noon" often deviates significantly from the sun's highest point in the sky.
\end{example}

\begin{example}
Consider an observer in Ulaanbaatar ($\lambda = 106.9^\circ$ E). While the city is geographically positioned within a region corresponding to $\mathrm{UTC+7}$, it officially observes Mongolia Standard Time ($z = +8$). If the civil clock reads 12:00 (Noon), the corresponding $\mathrm{UTC}$ is 04:00. The Local Mean Time is calculated by the geometric offset of $+106.9/15 \approx +7.127$ hours, or approximately $+7\text{h } 08\text{m}$. Consequently, at the stroke of civil noon, the Local Mean Time is $04:00 + 7\text{h } 08\text{m} = 11:08$. This discrepancy of nearly one hour demonstrates that civil time in Mongolia is significantly "advanced" relative to the local solar cycle.
\end{example}

\begin{example}
Consider an observer in Thimphu, Bhutan ($\lambda = 89.6^\circ$ E). The nation observes Bhutan Time ($z = +6$), which aligns closely with its geographic longitude. If the civil clock reads 12:00 (Noon), the corresponding $\mathrm{UTC}$ is 06:00. The Local Mean Time is determined by the geometric offset of $+89.6/15 \approx +5.973$ hours, or approximately $+5\text{h } 58\text{m}$. Consequently, when the clock strikes civil noon, the Local Mean Time is $06:00 + 5\text{h } 58\text{m} = 11:58$. Unlike the significant offsets found in Lhasa or Ulaanbaatar, Thimphu's civil noon is remarkably synchronized with the solar cycle, deviating by only two minutes.
\end{example}

\subsection{Mean periods and secular trends}
\label{app:modern-periods}

For reference at the J2000 epoch, representative mean values are:
\begin{equation}
\begin{split}
Y_{\mathrm{trop}} &\approx 365.2421897\,\mathrm{d},\qquad
Y_{\mathrm{anom}} \approx 365.2596359\,\mathrm{d},\\
S_{\mathrm{syn}} &\approx 29.53058885\,\mathrm{d},\qquad
S_{\mathrm{anom}} \approx 27.55454988\,\mathrm{d}.
\end{split}
\end{equation}
Here $Y_{\mathrm{trop}}$ is the mean tropical year (equinox to equinox), $Y_{\mathrm{anom}}$ the mean anomalistic year (perihelion to perihelion), $S_{\mathrm{syn}}$ the mean synodic month (new moon to new moon), and $S_{\mathrm{anom}}$ the mean anomalistic month (perigee to perigee) \cite{Meeus}.

For secular trends valid on millennial scales around J2000, we use polynomial fits in fixed-length days of 86400 seconds:
\begin{align*}
Y_{\mathrm{trop}}(T) &= 365.24218967 - 6.15\times 10^{-6}T - 7.29\times 10^{-10}T^2,\\
S_{\mathrm{syn}}(T) &= 29.53058885 + 2.16\times 10^{-7}T - 3.64\times 10^{-10}T^2,
\end{align*}
where $T = (JD_{\mathrm{TT}}-2451545.0)/36525$ is the number of Julian centuries from J2000.0. For completeness, the corresponding secular variations for the anomalistic periods are:
\begin{align*}
Y_{\mathrm{anom}}(T) &= 365.25963586 + 3.04 \times 10^{-6}T, \\
S_{\mathrm{anom}}(T) &= 27.55454988 - 1.039 \times 10^{-5}T - 2.0 \times 10^{-9}T^2.
\end{align*}

\subsection{Solar longitude series}
\label{app:solar-series}

Civil-day labels in lunisolar calendars depend primarily on the Sun's apparent ecliptic longitude $\lambda_{\sun}$ and its resulting elongation from the Moon. For calendar-grade precision, we employ standard truncated series expansions \cite{Meeus}.

Let $T = (JD_{\mathrm{TT}}-2451545.0)/36525$ denote the number of Julian centuries since the epoch J2000.0 (TT). We first compute the geometric mean elements:
\begin{align}
    L_0 &= 280.46646^\circ + 36000.76983^\circ T + 0.0003032^\circ T^2 \pmod{360^\circ}, \\
    M   &= 357.52911^\circ + 35999.05029^\circ T - 0.0001537^\circ T^2 \pmod{360^\circ}.
\end{align}
Here $L_0$ is the mean longitude of the Sun and $M$ is the mean anomaly. The true geometric longitude is
\[
\lambda_{\sun,\mathrm{true}} = L_0 + C_\sun ,
\]
where $C_\sun$ is the Equation of Center.
A three-term approximation is given by
\begin{equation}
\label{eq:solar-eqn-center}
\begin{split}
    C_\sun &= \bigl(1.914602 - 0.004817T - 0.000014T^2\bigr)\sin M \\
            &\quad + \bigl(0.019993 - 0.000101T\bigr)\sin(2M) \\
            &\quad + 0.000289 \sin(3M).
\end{split}
\end{equation}
To achieve higher precision (within $\approx 0.01^\circ$), one must account for the \emph{apparent} longitude $\lambda_{\sun,\mathrm{app}}$ by correcting for aberration and nutation \cite{Meeus}:
\begin{align}
    \Omega &= 125.04^\circ - 1934.136^\circ T, \\
    \lambda_{\sun,\mathrm{app}} &= \lambda_{\sun,\mathrm{true}} - 0.00569^\circ - 0.00478^\circ \sin\Omega,
\end{align}
where $-0.00569^\circ$ represents the constant of aberration and the term involving $\Omega$ (the longitude of the Moon's ascending node) approximates the nutation in longitude.

\subsection{Sunrise models}
\label{app:sunrise}

A calendar pipeline that anchors civil days to sunrise requires an explicit model for the event. We distinguish three levels of precision, ranging from static proxies to dynamic astronomical calculations. 

At the most basic level, one may adopt a \emph{constant-time proxy} where the civil-day boundary is fixed to a static local clock time. In the absence of latitude or seasonal data, sunrise is often idealized as \emph{06:00} Local Mean Time. However, a more physically grounded static proxy is \emph{05:56}. This four-minute anticipation accounts for the fact that a ``visible'' sunrise occurs when the Sun's upper limb first clears the horizon, which occurs while the Sun's center is still geometrically depressed by $h_0 \approx -0.833^\circ$. This threshold incorporates both the solar semidiameter ($\sim 16'$) and the lift provided by atmospheric refraction ($\sim 34'$), which allows the Sun to be seen while it is still below the horizon. At the equator, where the Sun rises vertically at a rate of $0.25^\circ$ per minute, traversing this $0.833^\circ$ gap takes approximately $3.3$ minutes. Rounding this to 4 minutes is a standard convention that ensures the proxy remains conservative even under varying atmospheric conditions \cite{Meeus}.

To move beyond these static approximations, we must account for the Earth's geometry and its axial tilt. The following sections detail a hierarchical approach: first, deriving the \emph{spherical-Earth sunrise} from solar declination to account for latitude and seasonality, and second, incorporating the \emph{equation of time} (EOT) to bridge the gap between uniform mean time and true solar time.

\subsubsection{Spherical-Earth sunrise from declination}

To account for latitude $\varphi$ and seasonal solar declination $\delta$, we compute the semi-diurnal arc. We define the standard apparent-altitude threshold as:
\begin{equation}
h_0 = -0.833^\circ,
\end{equation}
which incorporates the average atmospheric refraction ($34'$) and the solar semidiameter ($16'$). The solar declination $\delta$ is obtained from the Sun's ecliptic longitude $\lambda_{\sun}$ and the obliquity of the ecliptic $\varepsilon \approx 23.44^\circ$:
\begin{equation}
\sin \delta = \sin \varepsilon \sin \lambda_{\sun}.
\end{equation}
The hour angle $H_0$ at sunrise is then derived from the spherical law of cosines:
\begin{equation}
\label{eq:sunrise-cos}
\cos H_0 = \frac{\sin h_0 - \sin \varphi \sin \delta}{\cos \varphi \cos \delta}.
\end{equation}
Then in \emph{Local Apparent Solar Time} (where solar noon is fixed at exactly 12:00), the sunrise time is simply:
\begin{equation}
t_{\mathrm{rise,app}} = 12\mathrm{h} - \frac{H_0}{15^\circ/\mathrm{h}}.
\end{equation}

\begin{remark}
Using this apparent time directly as a civil clock introduces significant seasonal error. Because the Earth's orbital speed varies and the ecliptic is tilted, true solar noon can deviate from mean clock noon by up to $\pm 16$ minutes. If the objective is to match the $\sim 15$-minute error margin inherent in ignoring the equation of time (EOT), the requirements for the solar longitude model are surprisingly relaxed.
In this context, a \emph{Mean Sun} model—using only the linear mean longitude $L_0$ and ignoring the Equation of Center—is sufficient. While the Mean Sun can deviate from the True Sun by up to $\pm 1.9^\circ$, this longitudinal error translates to a timing discrepancy of approximately 8 minutes at the equator. Even at mid-latitudes, this error typically stays within the 15-minute budget set by the omission of the EOT. Similarly, for the obliquity $\varepsilon$, a constant value of $23.44^\circ$ is more than adequate, as its secular variation is measured in arcseconds per century and has no perceptible impact on the 15-minute scale. Consequently, for a model that accepts the seasonal "drift" of solar noon, complex truncated series for $\lambda_{\sun}$ are unnecessary; the Mean Sun provides a mathematically consistent level of precision.
\end{remark}

\begin{figure}[ht]
    \centering
    \includegraphics[width=.7\textwidth]{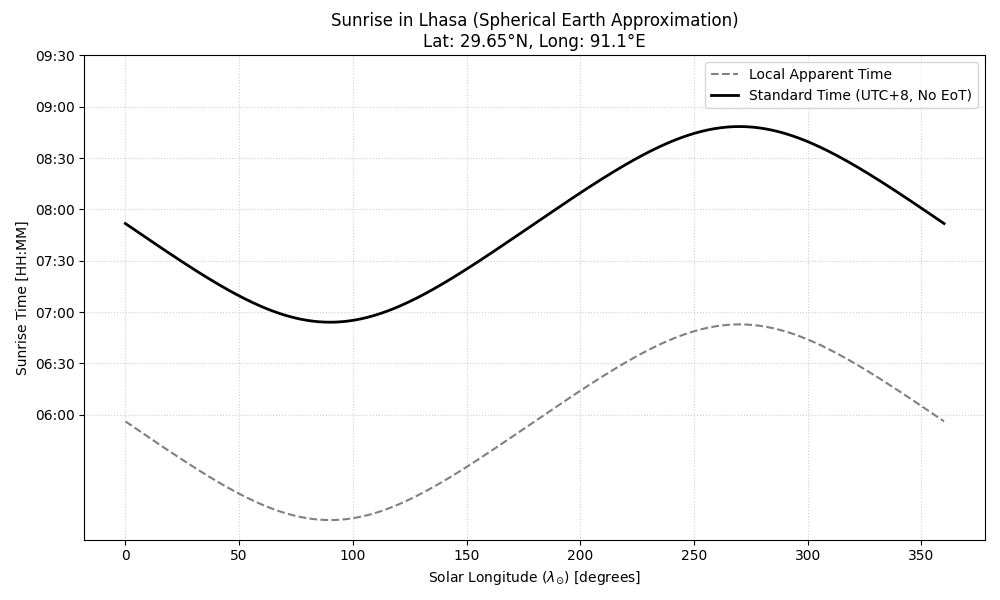}
	\caption{Sunrise in Lhasa ($\varphi = 29.65^\circ$ N) under the spherical-Earth approximation.}
    \label{fig:sunrise-lhasa}
\end{figure}

\begin{example}
Figure~\ref{fig:sunrise-lhasa} illustrates the seasonal variation of sunrise in Lhasa over a full solar year, neglecting the equation of time (EOT). In Local Apparent Time, sunrise ranges from approximately 05:01 at the summer solstice to 06:59 at the winter solstice. However, civil clock time (UTC+8) reflects a significant delay; because Lhasa ($91.1^\circ$ E) sits nearly $29^\circ$ west of its reference meridian ($120^\circ$ E), civil sunrise is shifted roughly $1\text{h } 56\text{m}$ later than Local Mean Time. This model utilizes the mean periods defined in Section~\ref{app:modern-periods} but accepts a seasonal error of up to $\pm 16$ minutes by omitting the EOT, cf. \S\ref{ss:eot}. Furthermore, the calculation assumes a standard depression $h_0 = -0.833^\circ$ to account for atmospheric refraction and solar semidiameter, which effectively advances the visible sunrise by approximately four minutes relative to the geometric horizon \cite{Meeus}.
\end{example}

\begin{remark}
In view of \eqref{eq:sunrise-cos}, the domain of the $\arccos$ function 
naturally identifies the polar phenomena occurring at high latitudes. If the magnitude of the right-hand side in \eqref{eq:sunrise-cos} exceeds unity, the Sun remains strictly above or below the horizon threshold for the entire diurnal cycle:
When $\cos H_0 > 1$, the Sun's maximum altitude never reaches $h_0$, resulting in \emph{Polar Night}.
When $\cos H_0 < -1$, the Sun's minimum altitude never drops to $h_0$, resulting in the \emph{Midnight Sun}.
\end{remark}

\subsubsection{The equation of time}
\label{ss:eot}

To convert the \emph{Local Apparent Time} obtained from the spherical-Earth model into a uniform civil clock (Mean Time), we must apply the {\em equation of time} (EOT). This correction accounts for the combined effects of the Earth's orbital eccentricity and the obliquity of the ecliptic, representing the difference between the Right Ascension of the true Sun ($\alpha_{\sun}$) and the Mean Sun ($L_0$).

Following the definitions in Section~\ref{app:solar-series}, the EoT in minutes is derived as:
\begin{equation}
\Delta t_{\mathrm{EoT}} = 4\left(L_0 - \alpha_{\sun}\right),
\end{equation}
where the factor of 4 converts degrees to minutes of time, reflecting the mean solar rotation rate of $1^\circ$ per $4$ minutes.
The Right Ascension $\alpha_{\sun}$ is computed from the ecliptic longitude $\lambda_{\sun}$ and obliquity $\varepsilon$:
\begin{equation}
\tan \alpha_{\sun} = \cos \varepsilon \tan \lambda_{\sun}.
\end{equation}

We evaluate the cumulative accuracy of the sunrise model when restricted to a \emph{first-order anomaly} (one-term Equation of Center for $\lambda_{\sun}$):

\begin{itemize}
    \item \emph{Solar longitude accuracy:} Utilizing only the first term ($C_{\sun} \approx 1.915^\circ \sin M$) results in a solar longitude error of approximately $\pm 0.034^\circ$, which translates to a negligible $\approx 8$ seconds of time in the solar position itself.
    \item \emph{RA computation:} When this first-order $\lambda_{\sun}$ is projected to find $\alpha_{\sun}$, it captures the vast majority of the EOT's profile. The residual error in Right Ascension relative to higher-order models is less than $0.05^\circ$, maintaining a timing error budget of under 12 seconds.
    \item \emph{Overall sunrise accuracy:} In the spherical-Earth framework, this level of precision for $\lambda_{\sun}$ and $\alpha_{\sun}$ produces a sunrise time that aligns with rigorous astronomical ephemerides to within $\approx 30$ seconds across the year. This error is significantly smaller than the uncertainties introduced by local atmospheric refraction variations, which can fluctuate by several minutes.
    \item \emph{Final clock time:} The integrated formula for civil sunrise is:
    \begin{equation}
    t_{\mathrm{rise,clock}} = 12\mathrm{h} - \frac{H_0}{15^\circ/\mathrm{h}} - \frac{\Delta t_{\mathrm{EoT}}}{60} + \Delta t_{\mathrm{zone}},
    \end{equation}
    where $\Delta t_{\mathrm{zone}}$ is the fixed longitudinal offset defined by the local meridian and the time zone.
\end{itemize}

\begin{remark}
By incorporating the first-order anomaly and the EOT, the model effectively eliminates the $\pm 16$-minute seasonal drift inherent in lower-order apparent-time approximations. The cumulative error in sunrise timing—accounting for the simplified $\varepsilon$, the 1-term $\lambda_{\odot}$, and the spherical-Earth geometry—remains within $\pm 30$ seconds. This level of precision is more than adequate for defining civil day boundaries, as it remains smaller than the typical variance caused by local atmospheric pressure and temperature changes. While Section~\ref{app:modern-periods} provides secular trends for completeness, this first-order approach is sufficient to ensure robust sunrise determinations for calendar-grade applications without the complexity of higher-order series.
\end{remark}

\begin{figure}[ht]
    \centering
    \includegraphics[width=0.85\textwidth]{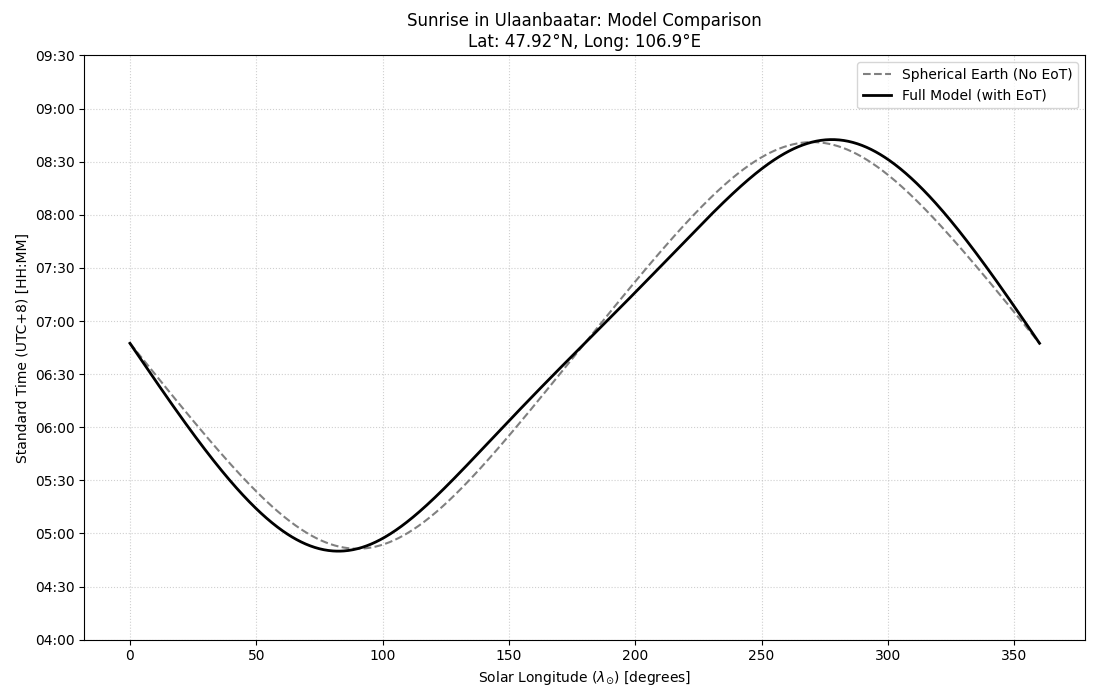}
    \caption{Sunrise in Ulaanbaatar ($\varphi = 47.92^\circ$ N): apparent time vs mean time models.}
    \label{fig:sunrise-ulaanbaatar}
\end{figure}

\begin{example}
Figure~\ref{fig:sunrise-ulaanbaatar} contrasts Ulaanbaatar's seasonal sunrise profile under a spherical-Earth approximation versus a model incorporating the EOT. While the basic model assumes a fixed solar noon, the actual civil sunrise (solid line) deviates by up to $\pm 16$ minutes due to orbital eccentricity and axial tilt \cite{Meeus}. At $47.92^\circ$ N, Ulaanbaatar experiences a pronounced seasonal swing ranging from approximately 05:52 to 08:42. Being situated at $106.9^\circ$ E, its civil offset is roughly 52 minutes relative to the UTC+8 meridian, placing the local cycle closer to the clock than in Lhasa. By utilizing a first-order Equation of Center and Right Ascension projection, this refined model matches rigorous astronomical tables within approximately 30 seconds \cite{Meeus}.
\end{example}

\begin{remark}
While solar longitude serves as a stable geometric coordinate, plotting sunrise against civil calendar dates reveals secular and periodic shifts. Axial precession causes solstices to drift earlier by approximately one day every 71 years, while the forward precession of perihelion slowly shifts the EOT profile relative to the seasons over centuries. On a shorter scale, the misalignment between the tropical and civil years creates a four-year jitter of roughly one minute, periodically reset by leap year intercalations. Furthermore, because longitude is a nonlinear function of time, the Earth's varying orbital speed—governed by the Equation of Center—introduces nonsecular periodic modulations into the sunrise timing that prevent a simple linear mapping between calendar days and solar position.
\end{remark}

\begin{remark}
The precision of the sunrise model is inherently linked to the obliquity of the ecliptic, $\varepsilon$, which undergoes slow secular variations driven by planetary perturbations. 
To discuss this long-term drift, we consider the standard secular polynomial fit \cite{Meeus}:
\begin{equation}
\varepsilon(T) = 23^\circ 26' 21.448'' - 46.8150'' T - 0.00059'' T^2 + 0.001813'' T^3,
\end{equation}
where $T$ represents Julian centuries from J2000.0. The dominant linear term indicates that $\varepsilon$ is currently decreasing by approximately $46.8''$ per century.

For calendar-grade applications, the practical impact of this variation is remarkably constrained. Adopting a constant obliquity of $23.44^\circ$ for the duration of the 21st century introduces a maximum declination error below $0.01^\circ$. When propagated through the spherical sunrise equation, this discrepancy results in a timing shift of only 3--5 seconds at mid-latitudes. Furthermore, while the obliquity-driven component of the EOT varies with $\tan^2(\varepsilon/2)$, the secular decrease in $\varepsilon$ only marginally attenuates the EOT amplitude over millennial scales. Consequently, a constant obliquity remains sufficient to maintain the $\approx 30$-second accuracy target for the standard sunrise implementation.
\end{remark}

\subsection{Lunar longitude series}
\label{app:lunar-series}

We adopt the standard mean arguments from the ELP-2000/82 theory \cite{ELPMPP02}. 
Let $T$ be the number of Julian centuries (36525 days) from the epoch J2000.0 (TT). The fundamental arguments (in degrees) are:
\begin{align}
    L' &= 218.3164477^\circ + 481267.8812342^\circ T - 0.0015786^\circ T^2 \\
    D  &= 297.8501921^\circ + 445267.1114034^\circ T - 0.0018819^\circ T^2 \\
    M  &= 357.5291092^\circ + 35999.0502909^\circ T - 0.0001536^\circ T^2 \\
    M' &= 134.9633964^\circ + 477198.8675055^\circ T + 0.0087414^\circ T^2 \\
    F  &=  93.2720950^\circ + 483202.0175233^\circ T - 0.0036539^\circ T^2
\end{align}
We omitted the higher-order terms in $T^3$ and $T^4$, whose coefficients are generally of the order $10^{-5}$ degrees or smaller.
For instance, the $T^3$ coefficients for $L'$ and $D$ are approximately $1^\circ/540\,000$. 
The only exception is the lunar anomaly $M'$, where the $T^3$ term is larger, roughly equal to $1^\circ/70\,000$, but this remains negligible for our target precision over typical historical intervals.

Physically, these arguments correspond to the primary frequencies of the lunar orbit. The argument $D$ represents the mean elongation of the Moon from the Sun, determining the synodic phase cycle. $M$ is the mean anomaly of the Sun, which governs the Earth's variable orbital speed. Similarly, $M'$ is the mean anomaly of the Moon, measuring the Moon's progress from perigee to apogee. Finally, $F$ denotes the mean argument of latitude, the angular distance of the Moon from the ascending node, which serves as the primary parameter for eclipse prediction.

\begin{remark}
\label{rem:linear-error}
A linear calendar simplifies the model by fixing the mean motions and ignoring the quadratic secular acceleration terms. Over the span of one millennium ($T=10$), this truncation introduces two distinct sources of error. The dominant effect comes from the mean elongation $D$, which possesses a secular term of $-0.00188^\circ T^2$. At $T=10$, this drift accumulates to approximately $-0.188^\circ$. Given the Moon's synodic rate of $\approx 12.19^\circ/\mathrm{day}$, this delay shifts the mean New Moon by roughly 22 minutes. A secondary error arises from the lunar anomaly $M'$, which has a larger secular term of $+0.00874^\circ T^2$. While this does not shift the mean conjunction time directly, it alters the phase of the Major Inequality, inducing a smaller oscillating timing error of roughly $\pm 12$ minutes. Thus, a strictly linear model will inevitably drift by about 20 minutes per millennium.
\end{remark}

\begin{remark}
If the quadratic $T^2$ terms are retained, the remaining errors arise primarily from the omitted cubic terms and the truncation of the periodic series. In the ELP-2000/82 theory, the coefficient for the $T^3$ term in mean elongation is approximately $+0.0066''/\mathrm{cy}^3$. Over a millennium ($T=10$), this contributes an accumulated error of merely $6.6''$, which corresponds to a time drift of approximately 13 seconds. The coefficient for the $T^3$ term in the lunar anomaly $M'$ is significantly larger ($1^\circ/69699 \approx 0.052''/\mathrm{cy}^3$). At $T=10$, this results in an anomaly error of $\Delta M' \approx 52''$. However, this error affects the Moon's true longitude primarily through the Equation of Center, which scales the anomaly by the eccentricity factor $2e \approx 0.11$. Consequently, the effective positional error is reduced to roughly $6''$, resulting in a timing uncertainty of only about 11 seconds. Therefore, once the quadratic secular terms are included, the accuracy of the calendar is limited almost entirely by the number of periodic harmonic terms retained, rather than by the stability of the mean motions.
\end{remark}

The Moon's geocentric longitude $\lambda_{\mathbb{C}}$ is given by the series expansion:
\begin{equation}
\label{eq:lunar_series}
\lambda_{\mathbb{C}} \approx L' + 10^{-6} \sum_{i} C_i \, \sin\Bigl( d_i D + m_i M + m'_i M' + f_i F \Bigr)
\end{equation}
where the coefficients $C_i$ are expressed in {microdegrees} ($10^{-6}$ deg).

Table~\ref{tab:lunar-primary} below lists the 24 dominant terms. This primary set captures the Moon's gross motion with a precision of approximately $\pm 2'$ (roughly $3$--$4$ minutes of time).

\begin{table}[H]
\caption{Primary lunar series (24 terms). Accuracy: $\sim 3$--$4$ min.}
\label{tab:lunar-primary}
\centering
\small
\renewcommand{\arraystretch}{1.1}
\setlength{\tabcolsep}{5pt}
\begin{tabular}{rrrrrl | rrrrr}
\toprule
$d$ & $m$ & $m'$ & $f$ & $C (\mu^\circ)$ & Physical Effect & $d$ & $m$ & $m'$ & $f$ & $C (\mu^\circ)$ \\
\midrule
 0 &  0 &  1 &  0 &  6288774 & Major Inequality &   0 &  1 &  1 &  0 &   -30383 \\
 2 &  0 & -1 &  0 &  1274027 & Evection &   0 &  0 &  0 & -2 &    15327 \\
 2 &  0 &  0 &  0 &   658314 & Variation &   0 &  0 &  1 &  2 &   -12528 \\
 0 &  0 &  2 &  0 &   213618 & 2nd Elliptic &   0 &  0 &  1 & -2 &    10980 \\
 0 &  1 &  0 &  0 &  -185116 & Annual Equation &   4 &  0 & -1 &  0 &    10675 \\
 0 &  0 &  0 &  2 &  -114332 & Red. to Ecliptic &   0 &  0 &  3 &  0 &    10034 \\
 2 &  0 & -2 &  0 &    58793 & Evection Harm. &   4 &  0 & -2 &  0 &     8548 \\
 2 & -1 & -1 &  0 &    57066 & &   2 &  1 & -1 &  0 &    -7888 \\
 2 &  0 &  1 &  0 &    53322 & &   2 &  1 &  0 &  0 &    -6766 \\
 2 & -1 &  0 &  0 &    45758 & &   1 &  0 & -1 &  0 &    -5163 \\
 0 &  1 & -1 &  0 &   -40923 & &   1 &  1 &  0 &  0 &     4987 \\
 1 &  0 &  0 &  0 &   -34720 & Parallactic Ineq. &   2 & -1 &  1 &  0 &     4036 \\
\bottomrule
\end{tabular}
\end{table}

\noindent
To achieve modern ``dynamical'' precision (better than 1 minute of time), the model must include the finer perturbations listed in Table~\ref{tab:lunar-supp}. 
Adding these supplementary terms reduces the maximum periodic error to roughly $\pm 15''$ (roughly $30$ seconds of time).

\begin{table}[H]
\caption{Supplementary series (40 terms). Improves accuracy to $\sim 30$ sec.}
\label{tab:lunar-supp}
\centering
\small
\renewcommand{\arraystretch}{1.1}
\setlength{\tabcolsep}{3pt}
\begin{tabular}{rrrrr | rrrrr | rrrrr | rrrrr}
\toprule
$d$ & $m$ & $m'$ & $f$ & $C$ & $d$ & $m$ & $m'$ & $f$ & $C$ & $d$ & $m$ & $m'$ & $f$ & $C$ & $d$ & $m$ & $m'$ & $f$ & $C$ \\
\midrule
 2 &  0 &  2 &  0 &   3994 &  2 & -2 & -1 &  0 &   2011 &  2 & -1 &  2 &  0 &  -811 &  2 &  2 &  0 &  0 &  -440 \\
 4 &  0 &  0 &  0 &   3861 &  2 &  0 &  1 & -2 &  -1977 &  0 &  0 &  4 &  0 &   769 &  0 &  1 &  3 &  0 &  -425 \\
 2 &  0 & -3 &  0 &   3665 &  4 &  0 & -3 &  0 &  -1736 &  2 &  0 & -2 &  2 &   717 &  4 &  0 &  1 &  0 &  -418 \\
 0 &  1 & -2 &  0 &  -2689 &  4 & -1 & -1 &  0 &  -1671 &  0 &  0 &  2 &  2 &  -712 &  0 &  0 &  2 & -2 &   386 \\
 2 &  0 & -1 &  2 &  -2602 &  2 &  1 &  1 &  0 &  -1557 &  1 &  0 &  2 &  0 &  -663 &  2 &  0 & -5 &  0 &   371 \\
 2 & -1 & -2 &  0 &   2390 &  2 &  1 & -2 &  0 &   1492 &  1 &  1 & -1 &  0 &  -565 &  2 &  2 & -2 &  0 &   362 \\
 1 &  0 &  1 &  0 &  -2348 &  2 &  0 & -4 &  0 &  -1422 &  1 &  0 & -2 &  0 &  -523 &  1 &  1 &  1 &  0 &   317 \\
 2 & -2 &  0 &  0 &   2236 &  4 & -1 & -2 &  0 &  -1205 &  4 &  0 & -4 &  0 &   492 &  2 &  0 & -3 &  2 &  -310 \\
 0 &  1 &  2 &  0 &  -2120 &  2 &  1 &  0 & -2 &  -1111 &  4 & -2 & -1 &  0 &  -488 &  0 &  2 & -1 &  0 &  -307 \\
 0 &  2 &  0 &  0 &  -2069 &  2 & -1 &  1 & -2 &  -1100 &  2 &  2 & -1 &  0 &  -469 &  2 &  0 &  3 &  0 &  -293 \\
\bottomrule
\end{tabular}
\end{table}

\begin{remark}\label{rem:solar-simplicity}
The solar longitude series (Appendix~\ref{app:solar-series}) is intentionally considerably simpler than the lunar series. The timing of a New Moon depends on the relative velocity between the Moon and the Sun, which is approximately $12.2^\circ/\mathrm{d}$ (or roughly $0.508''/\mathrm{s}$). Consequently, a spatial error of $20''$ in \textit{either} the lunar or the solar longitude shifts the calculated moment of conjunction by the exact same amount: approximately 39 seconds. The 3-term solar model maintains an accuracy of roughly $20''$, introducing a maximum timing error of about 40 seconds. This perfectly matches the precision class of the Primary and Supplementary lunar tables, making it sufficient to support the system without becoming a computational bottleneck or the limiting factor in accuracy.
\end{remark}

\subsection{Mean new moon}

While the \emph{true} New Moon fluctuates due to solar perturbations and orbital eccentricity, the \emph{mean} New Moon can be calculated precisely using the formula from Meeus \cite{Meeus}, 
which is derived from ELP-2000/82:
\begin{equation}
\label{eq:meeus_mean}
\JDE_{\mathrm{mean}} = 2451550.09766 + 29.530588861 k + 0.0001337 T^2,
\end{equation}
where $k$ is the integer lunation number relative to the first New Moon of 2000 (Jan 6), and $T \approx k/1236.85$. The result is the Julian Date in Terrestrial Time (TT).

\begin{example}
The deviation between the calculated mean conjunction \eqref{eq:meeus_mean} and the true conjunction can be several hours.
Below we list examples for Spring of 1807, 1927, 1987, and 2026.  
Mean times are in TT, 
while True times are listed in $\UTone$,
as tabulated by the U.S.\ Naval Observatory Moon Phases service \cite{usno-moonphases}.

\begin{center}
\small
\begin{tabular}{l l l l l}
\toprule
{Lunation ($k$)} & {Mean JD} & {Mean Time (TT)} & {True Time (UT1)} & {Deviation} \\
\midrule
\multicolumn{5}{l}{\textit{1. Spring 1807}} \\
$-2385$ & $2381119.64$ & Mar 9,\ \ 03:27 & Mar 9,\ \ 09:04 & Mean is $5.6\,\mathrm{h}$ early \\
$-2384$ & $2381149.17$ & Apr 7,\ \ 16:11 & Apr 8,\ \ 02:11 & Mean is $10.0\,\mathrm{h}$ early \\
\midrule
\multicolumn{5}{l}{\textit{2. Spring 1927}} \\
$-901$ & $2424943.04$ & Mar 3,\ \ 12:54 & Mar 3,\ \ 19:25 & Mean is $6.5\,\mathrm{h}$ early \\
$-900$ & $2424972.57$ & Apr 2,\ \ 01:38 & Apr 2,\ \ 04:24 & Mean is $2.8\,\mathrm{h}$ early \\
\midrule
\multicolumn{5}{l}{\textit{3. Spring 1987}} \\
$-158$ & $2446884.26$ & Mar 29, 18:21 & Mar 29, 12:46 & Mean is $5.6\,\mathrm{h}$ late \\
$-157$ & $2446913.80$ & Apr 28,\ \ 07:05 & Apr 28,\ \ 01:34 & Mean is $5.5\,\mathrm{h}$ late \\
\midrule
\multicolumn{5}{l}{\textit{4. Spring 2026}} \\
$+323$ & $2461088.48$ & Feb 17, 23:31 & Feb 17, 12:01 & Mean is $11.5\,\mathrm{h}$ late \\
$+324$ & $2461118.01$ & Mar 19, 12:14 & Mar 19, 01:23 & Mean is $10.8\,\mathrm{h}$ late \\
\bottomrule
\end{tabular}
\end{center}
In these examples the Mean New Moon may either lag or lead the True New Moon, depending on where the Moon is in its anomalistic cycle (faster-than-average motion near perigee tends to make the true conjunction occur \emph{earlier} than the mean, while slower-than-average motion near apogee tends to make it occur \emph{later}).
\end{example}

\section{Numerical Approximations and Portability}
\label{app:numerical}

For implementation in Level 3 through Level 5 models, where full-precision floating-point libraries may be unavailable or computationally expensive, we provide a hierarchy of numerical approximations. These range from discrete look-up tables for integer-only environments to minimax polynomials for higher-precision pipelines.

\subsection{Sine look-up table}
\label{ss:sine-table}

The table-based profile utilizes a fixed-point sine table with a binary amplitude of $1024$. By utilizing a power of two, normalization operations are reduced to efficient bitwise shifts rather than integer division. Exploiting the symmetry of the sine function ($\sin x = \sin(180^\circ - x) = -\sin(-x)$), we only require the first quadrant, and specifically, only 8 values to define a 28-point full-circle resolution ($360^\circ / 28 \approx 12.86^\circ$ per step).

\begin{table}[ht]
\centering
\begin{tabular}{lrr}
\hline
Step ($k$) & Angle ($\theta$) & Table Value ($1024 \sin \theta$) \\
\hline
0 & $0^\circ$ & 0 \\
1 & $12.86^\circ$ & 228 \\
2 & $25.71^\circ$ & 444 \\
3 & $38.57^\circ$ & 638 \\
4 & $51.43^\circ$ & 801 \\
5 & $64.29^\circ$ & 923 \\
6 & $77.14^\circ$ & 998 \\
7 & $90^\circ$ & 1024 \\
\hline
\end{tabular}
\caption{Primary quadrant values for the 28-point sine table.}
\end{table}

The primary source of inaccuracy in this approach is the linearization error inherent in approximating the sine curve with straight segments between the sample points. For a step size of $\approx 0.22$ radians, the maximum theoretical deviation is approximately $6.5$ units ($0.6\%$ of full scale), occurring most notably near the quadrature ($90^\circ$) where the function's curvature is maximized relative to the chord. While coarse by astronomical standards, this precision is sufficient for simplified solar altitude checks where atmospheric refraction introduces larger uncertainties.

\subsubsection{Inverse look-up ($\arcsin$ and $\arccos$)}

Calculating the exact time of sunrise involves solving for the solar hour angle, which necessitates the use of inverse trigonometric functions, specifically $\arcsin$ and $\arccos$. Rather than maintaining separate algorithms or expanding the memory footprint with new tables, we perform a reverse linear interpolation on the existing sine table. 

Given an input ratio $y$ scaled to our binary amplitude (where $0 \le |y| \le 1024$), the algorithm scans the primary quadrant values to find the index $k$ such that the interval $[y_k, y_{k+1}]$ bounds $|y|$. The fractional step position is then determined via linear interpolation:
\[
x \approx k + \frac{|y| - y_k}{y_{k+1} - y_k}.
\]
This yields the $\arcsin$ value in terms of fractional table steps. The $\arccos$ function is trivially derived using the complementary identity $\arccos(y) = 90^\circ - \arcsin(y)$. 

\begin{remark}
While a dedicated look-up table for the arctangent function ($\arctan$) could be similarly implemented to handle broader astronomical coordinate transformations, the specific dawn and syzygy criteria in our concrete proposals do not require one, allowing us to preserve a strictly minimal memory footprint.
\end{remark}

\subsubsection{Stability against LCM explosion}

When operating within an exact rational arithmetic framework, it is necessary to audit all division operations. If fractions are divided repeatedly inside recursive feedback loops—a challenge we will examine in detail during the discussion of root-finding iterations (Section~\ref{app:fix-pt})—their denominators compound and grow exponentially, causing a lowest common multiple (LCM) explosion that rapidly exhausts hardware integer limits.

For the direct sine look-up, the interpolation is strictly a multiply-and-add operation: $y = y_k + w(y_{k+1} - y_k)$. Because the bounding table values are integers, the output fraction inherits only the denominator of the input fractional weight $w$. For instance, evaluating an input position of $x = 11/7$ yields an interpolated sine value of $2460/7$. The denominator ($7$) passes through completely unaltered. Because no new prime factors are introduced by the trigonometric evaluation itself, the direct sine look-up is mathematically stable and can be safely executed repeatedly \emph{inside} iterative calendrical loops.

The inverse look-up, conversely, requires a division by the interval difference $(y_{k+1} - y_k)$. This operation forces new prime factors into the denominator of the resulting fraction. If this inverse look-up were embedded inside a multi-iteration numerical solver, these factors would compound and rapidly trigger an LCM explosion. 

Fortunately, in our concrete proposals for calculating dawn and sunrise, the inverse trigonometric function is evaluated sequentially. The algorithm uses the initial solar declination to compute the hour angle as a one-time temporal offset. While one could theoretically iterate this process to solve for the ``backreaction'' of the solar longitude (since the Sun's declination shifts slightly between midnight and dawn), the inverse look-up injects small enough prime factors that it would comfortably survive two to three iterations before threatening 64-bit integer limits. In practice, however, such iterative depth is physically unnecessary. Because our rational models specify a spherical sunrise, the sub-second adjustments yielded by further iterations are entirely dwarfed by the broader geometric approximations of the model itself. A single sequential evaluation is therefore both computationally optimal and physically sufficient, fully neutralizing the LCM explosion risk.

\subsection{Minimax polynomials}
\label{ss:minimax}

To minimize computational cost while maintaining strict deterministic reproducibility across platforms, we replace large look-up tables with minimax polynomials. By tuning these polynomials to evaluate directly in \textit{turns} (where 1 turn = $360^\circ = 2\pi$ radians) rather than standard radians, we eliminate the need for costly scaling multiplications inside the execution loops. The general form for the odd functions used here is
$$ 
P(x) = c_1 x + c_3 x^3 + c_5 x^5 + \dots + c_n x^n .
$$
The decimal values provided below are for mathematical reference; the software implementation strictly uses the IEEE 754 hexadecimal floating-point constants to guarantee bitwise reproducibility.

The sine approximation evaluates the function $\sin(2\pi x)$, taking an input phase in turns and outputting the standard amplitude. It is optimized for the primary quadrant interval $x \in [0, 0.25]$. The 5th-degree coefficients provide a maximum absolute error of $\approx 1.1 \times 10^{-4}$ (contributing fewer than 8 seconds of temporal error to the lunar longitude), while the 7th-degree version reduces the error to $\approx 1.1 \times 10^{-6}$ (astronomical precision).

\begin{table}[ht]
\centering
\caption{Coefficients for turn-based sine approximation}
\small
\begin{tabular}{l l l l l}
\hline
\textbf{Coeff} & \textbf{5th-deg} & \textbf{5th-deg (hex)} & \textbf{7th-deg} & \textbf{7th-deg (hex)} \\
\hline
$c_1$ & $+6.281548$  & \texttt{0x1.9204e06298ee5p+2}  & $+6.283171$  & \texttt{0x1.921f7b2a8caaap+2} \\
$c_3$ & $-41.133393$ & \texttt{-0x1.4911303618770p+5} & $-41.337984$ & \texttt{-0x1.4ab430edb388ep+5} \\
$c_5$ & $+74.171088$ & \texttt{0x1.28af31a2c633ep+6}  & $+81.371979$ & \texttt{0x1.457ce7f0c0b57p+6} \\
$c_7$ & ---          & ---                            & $-71.315978$ & \texttt{-0x1.1d438fda450c4p+6} \\
\hline
\end{tabular}
\end{table}

The arctangent approximation takes a raw geometric ratio and directly outputs the evaluated angle in turns. The polynomial is optimized over the continuous interval $x \in [0, 1]$. Inputs larger than 1 are handled using the mandatory geometric reduction identity $\arctan(x) = \frac{1}{4}\text{ turn} - \arctan(1/x)$, requiring a single floating-point division. Avoiding further argument reduction prevents pipeline-stalling branches, allowing the CPU to evaluate the curve using strictly sequenced multiplications and additions to guarantee FMA-free reproducibility. A 5th-degree polynomial yields a maximum error of $\approx 1.5 \times 10^{-4}$, while a 7th-degree extension improves the precision to $\approx 2.2 \times 10^{-5}$.

\begin{table}[ht]
\centering
\caption{Coefficients for turn-based arctangent approximation}
\small
\begin{tabular}{l l l l l}
\hline
\textbf{Coeff} & \textbf{5th-deg} & \textbf{5th-deg (hex)} & \textbf{7th-deg} & \textbf{7th-deg (hex)} \\
\hline
$c_1$ & $+0.158452$  & \texttt{0x1.4482618478638p-3}  & $+0.159056$  & \texttt{0x1.45befa84fe3f6p-3} \\
$c_3$ & $-0.046517$  & \texttt{-0x1.7d119bc1df0c0p-5} & $-0.051293$  & \texttt{-0x1.a431f6c59c177p-5} \\
$c_5$ & $+0.013212$  & \texttt{0x1.b0f17a7c7df9bp-7}  & $+0.023719$  & \texttt{0x1.849c4c475fd16p-6} \\
$c_7$ & ---          & ---                            & $-0.006503$  & \texttt{-0x1.aa34c2cb444dbp-8} \\
\hline
\end{tabular}
\end{table}

\subsection{Square root approximation}
\label{app:sqrt}

The floating-point modules rely on the minimax polynomials for sine and arctangent provided in the previous section. However, the complete physical model required by the civil-day triggers—encompassing both spherical dawn geometry and the exact Equation of Time—demands a full trigonometric suite. To avoid introducing redundant polynomial approximations, the additionally requisite cosine and arccosine components are derived analytically from the primary sine and arctangent basis using fundamental identities, notably computing cosine from sine via $\sqrt{1 - \sin^2(\theta)}$, and mapping arccosine through the arctangent function utilizing $\sqrt{1 - x^2}$. Consequently, a highly reproducible square root operation is structurally required to complete the execution.

While the native IEEE 754 square root operation is strictly standardized and guarantees exact reproducibility across modern hardware, invoking it typically requires linking against an external math library (e.g., \texttt{libm}) or relying on compiler-specific intrinsics. A core design intention of the high-precision floating-point reforms is to evaluate all transcendental functions using rigidly prescribed minimax polynomials. Under this architecture, the square root remains the \emph{single} operation that would otherwise necessitate an external mathematical dependency. To support the creation of completely self-contained reference engines, we record a standardized formulation of the \emph{Fast Inverse Square Root} algorithm below. Ultimately, the communities adopting the calendar will determine the final numerical contract, and they may reasonably choose to mandate the native IEEE 754 square root. However, by specifying this exact sequence of basic arithmetic and bitwise operations, we ensure that a completely independent, bit-exact, and reproducible alternative is available should they desire to avoid external linking requirements entirely.

The algorithm begins by reinterpreting the bits of the floating-point radicand $S$ as a 64-bit integer, $I_S$. We then generate an initial approximation $y_0$ for the inverse square root $1/\sqrt{S}$ by performing the integer subtraction $I_{y_0} = \texttt{0x5fe6eb50c7b537a9} - (I_S \gg 1)$. This operation, which computes a piecewise linear approximation of the logarithm, provides a deterministic starting point with a relative error of approximately $3.2\%$. This estimate is then refined using the division-free Newton-Raphson iteration $y_{k+1} = y_k (1.5 - 0.5 S y_k^2)$. Finally, since the iteration converges to the inverse root, the required square root is obtained by the multiplication $\sqrt{S} \approx S \cdot y_n$.

The choice of iteration count allows the implementation to be tuned to the specific precision requirements of the application profile. A single iteration reduces the relative error to approximately $1.7 \times 10^{-3}$, which is sufficient for table-based lookup methods where linear interpolation errors dominate. A second iteration yields a relative error of $4 \times 10^{-6}$, a necessary threshold for polynomial approximations to preserve the accuracy of high-degree coefficients. Finally, a third iteration achieves a relative error of $\approx 10^{-16}$, effectively reaching the machine epsilon limits of double-precision arithmetic. This three-pass configuration is the recommended standard for the Reference Implementation, providing convergence virtually indistinguishable from hardware intrinsics.

\subsection{Fixed point iterations}
\label{app:fix-pt}

Determining the exact moment of a solar term or lunar phase involves numerically solving a transcendental equation. For a deterministic and reproducible calendar, we specify fixed iteration counts for each method rather than relying on epsilon-convergence checks, which can vary across floating-point implementations.

Let $x(t)$ be the modeled phase quantity (e.g., lunar elongation $\lambda_\moon(t)-\lambda_\sun(t)$) whose threshold defines an event. We seek the solution to $x(t)=x_0$.
Assume $x(t)$ decomposes into a linear mean motion and a bounded periodic correction:
\[
x(t) = A + B t + C(t),
\]
where $B$ is the mean angular velocity and $C(t)$ is the sum of all anomaly terms.

\subsubsection{Picard iteration}
\label{app:picard}

This method generalizes the traditional Tibetan-style update. It approximates the local velocity using the constant mean velocity $B$.
We start with the mean conjunction time
\begin{equation}
t_0=\frac{x_0-A}{B} ,
\end{equation}
as the initial guess.
If $t^*$ is the true solution where $A + B t^* + C(t^*) = x_0$, then $B(t^* - t_0) = -C(t^*)$, and thus the maximum time error is
\[
|t^* - t_0| \approx \frac{\max |C(t)|}{B} .
\]
For the lunar synodic elongation, the combined amplitude of the primary anomalies is approximately $5.1^\circ + 1.3^\circ + 0.7^\circ \approx 7.1^\circ$. 
Given a mean synodic rate $B \approx 12.19^\circ/\mathrm{day}$, the maximum discrepancy is
\[
|t^* - t_0| \approx \frac{7.1^\circ}{12.19^\circ/\mathrm{day}} \approx 0.58 \text{ days} \approx 14 \text{ hours}.
\]
In practice, the actual error rarely reaches this theoretical maximum. 
A value of $10$ hours could serve as a good conservative estimate.

Now, we improve upon $t_0$, by using the recurrence relation
\begin{equation}\label{eq:picard_exact}
t_{n+1} = \frac{x_0 - A - C(t_n)}{B} = t_0 - \frac{C(t_n)}{B}.
\end{equation}
This is a contraction mapping $t_{n+1} = g(t_n)$ with convergence rate determined by the derivative of the iteration function:
\[
|g'(t)| = \left| \frac{d}{dt} \left( t_0 - \frac{C(t)}{B} \right) \right| = \frac{|C'(t)|}{B}.
\]
For the lunar synodic case, the maximum velocity variation $C'(t)$ is roughly $1.5^\circ/\mathrm{day}$. The mean rate is $B \approx 12.2^\circ/\mathrm{day}$.
Thus, the contraction factor is:
\[
k \approx \frac{1.5}{12.2} \approx 0.12.
\]
Since $k < 1$, the method converges linearly, gaining roughly one decimal digit of precision per step.
Starting with a maximum initial error of $\approx 10$ hours, the residuals are:
\begin{itemize}
    \item \emph{Iteration 1:} $10 \text{ h} \times 0.12 \approx 1.2$ hours. (Traditional approximation)
    \item \emph{Iteration 2:} $1.2 \text{ h} \times 0.12 \approx 8$ minutes.
    \item \emph{Iteration 3:} $8 \text{ min} \times 0.12 \approx 1$ minute.
\end{itemize}

\subsubsection{Rational preconditioning}

When implementing Equation~\ref{eq:picard_exact} using exact rational arithmetic, the exact mean velocity $B$ is typically a fraction with massive integer components to accurately capture long-term secular trends. Because the recurrence divides by $B$ (effectively multiplying by $B^{-1}$), the large \emph{numerator} of $B$ is injected into the denominator of the time step. Repeated iteration causes these large prime factors to compound, leading to an exponential growth in the denominators of the time variables—a phenomenon known as lowest common multiple (LCM) explosion—which severely degrades computational performance. 

To guarantee efficiency without losing exactness at the epoch, we introduce a decoupled rational preconditioner $\tilde{B}^{-1} \approx 1/B$. The preconditioner is chosen to have a highly composite, microscopic denominator (e.g., $295306/10000$ for the synodic month) to permanently bound the fraction sizes within fast hardware integer limits. 

The initial anchor $t_0$ is still computed using the exact $B$ to ensure absolute long-term secular stability. However, the iteration loop is decoupled as:
\begin{equation}
t_{n+1} = t_0 - \tilde{B}^{-1} C(t_n).
\end{equation}
Because the Picard iteration is self-correcting, substituting $1/B$ with $\tilde{B}^{-1}$ only slightly alters the slope of the contraction mapping. The sequence converges to a preconditioned fixed point $\tilde{t}^*$ satisfying $\tilde{B}(\tilde{t}^* - t_0) = -C(\tilde{t}^*)$. The structural error introduced by this decoupling is bounded by:
\[
|t^* - \tilde{t}^*| \approx \max|C(t)| \cdot |1 - B\tilde{B}^{-1}|.
\]
By choosing a preconditioner that closely matches the true mean rate, the quantity $|1 - B\tilde{B}^{-1}|$ acts as a severe attenuator. For our chosen constants, the theoretical maximum error introduced by the preconditioner is on the order of fractions of a second—well below any physical relevance—yielding the computational speed of floating-point math with the deterministic safety of exact rationals.

It is worth noting that while the advanced root-finding techniques detailed below offer mathematically superior convergence rates (quadratic rather than linear), practical implementation of the calendrical engines reveals that 3 to 5 iterations of the preconditioned Picard method are entirely sufficient to reach the noise floor of the astronomical models, even in high-precision settings. Consequently, the Newton-Raphson and Steffensen's methods are presented here primarily as high-performance alternatives tailored for floating-point environments. Because standard floating-point arithmetic operates within fixed memory bounds (e.g., IEEE-754 binary64), the rational lowest common multiple (LCM) explosion is structurally impossible in these contexts, rendering the preconditioning techniques discussed previously unnecessary.

\subsubsection{Newton-Raphson method}

When the instantaneous derivative $\dot{x}(t) = B + C'(t)$ is available, we may employ the Newton-Raphson method
\begin{equation}
t_{n+1} = t_n - \frac{x(t_n) - x_0}{\dot{x}(t_n)}.
\end{equation}

A common concern in numerical implementations is the requirement for a \emph{bracketed} fallback, such as bisection, to prevent divergence when the solver encounters local extrema or stationary points. However, employing such a fallback would break the fixed-iteration reproducibility guarantee essential for a standardized calendar. Because bisection converges linearly (halving the error per step), a fixed limit of two or three iterations would leave a significant residual error if the bisection branch were triggered. 

For the specific case of syzygy calculations, such safeguards are functionally {\em unnecessary}. 
Because the Moon's synodic angular velocity is strictly positive and bounded far from zero ($\approx 12.19^\circ/\mathrm{day}$), the elongation function $x(t)$ is strictly monotonic. The absence of local extrema ensures that the Newton step will always move toward the root without the risk of oscillation or divergence.

The efficiency of this approach is confirmed by the quadratic convergence profile, where the error $e_{n+1} \approx K e_n^2$ with $K \approx \frac{x''(t)}{2\dot{x}(t)}$. For lunar motion, $K \approx 0.01 \text{ day}^{-1}$. Starting from an initial error $e_0 \approx 0.4$ days (approx. 10 hours), the residuals are:
\begin{itemize}
    \item \emph{Iteration 1:} $0.01 \times (0.4)^2 = 0.0016$ days $\approx 2.3$ minutes.
    \item \emph{Iterations 2:} $0.01 \times (0.0016)^2 = 2.56 \times 10^{-8}$ days $\approx 0.002$ seconds.
\end{itemize}
While two iterations are more than sufficient to satisfy the physical accuracy requirements of the calendar, a third iteration remains a high-value option for modern implementations. By pushing the mathematical root-finding error below the machine epsilon for 64-bit floating-point numbers ($< 10^{-10}$ seconds), this extra step effectively exhausts the significand. This ensures bit-exact reproducibility, making the final output a deterministic function of the IEEE-754 binary64 environment and the hex-float constants, rather than being subject to the minute numerical jitter of the solver or the initial guess.

\subsubsection{Steffensen's method}

For high-precision lunar calculations, evaluating the analytical derivative of a complex multi-term anomaly series can be architecturally burdensome. \emph{Steffensen's method} provides a derivative-free alternative that preserves quadratic convergence by using the function's own value to approximate the derivative slope. Given the equation $f(t) = x(t) - x_0 = 0$, the iteration is defined as:
\[
t_{n+1} = t_n - \frac{f(t_n)^2}{f(t_n + f(t_n)) - f(t_n)}
\]
This approach avoids the sensitivity issues of fixed-step finite difference approximations by utilizing the residual $f(t_n)$ as the probe increment. 

As established in the Newton-Raphson discussion, the strict monotonicity of the synodic elongation ensures that the denominator never vanishes and the solver is globally convergent within the synodic window. This stability allows for a fixed-iteration execution model.
\begin{itemize}
    \item \emph{Iteration 1:} Reduces the initial $\sim 10$-hour error to approximately $2$--$4$ minutes. This matches the behavior of a single Newton step, where the linear drift is eliminated.
    \item \emph{Iteration 2:} The quadratic property reduces the residual to approximately $10^{-3}$ to $10^{-2}$ seconds (millisecond scale). This is already several orders of magnitude below the physical noise floor of the underlying astronomical series.
\end{itemize}
Following the same logic as the Newton-Raphson case, while two steps of Steffensen’s method achieve sub-second precision, a third iteration may be employed to guarantee bit-exact consistency across different 64-bit computing architectures by eliminating any residual mathematical error before the final rounding.

\subsection{Reproducible floating-point arithmetic}
\label{app:fp-reprod}

One of the central challenges in defining a computationally derived calendar is ensuring long-term stability across different computing epochs. Unlike integer-based civil calendars, which are exact by definition, astronomical calendars rely on floating-point approximations of celestial mechanics. A significant risk to the longevity of such a system is "platform drift," where the same algorithmic specification yields slightly different timestamps depending on the underlying hardware architecture, compiler optimizations, or standard library implementations. To mitigate this, future specifications of the reform tiers should prioritize numerical reproducibility over raw precision or performance.

The most fundamental requirement for a reproducible implementation is the strict definition of numerical constants. Standard decimal representations of irrational numbers (such as $\pi$ or the synodic month) are ambiguous; the conversion from a decimal string to a binary floating-point value depends on the specific parsing algorithm of the host language. A robust specification must therefore bypass this conversion entirely by defining all fundamental epochs and coefficients using hexadecimal floating-point notation (IEEE 754 \texttt{binary64} hex format). This ensures that every compliant machine, regardless of its vintage or manufacturer, initializes its memory with the exact same bit patterns.

Furthermore, because floating-point arithmetic is non-associative, the order of operations dictates the result. The specification cannot simply list a polynomial for the Mean Moon arguments; it must mandate the evaluation sequence. The preferred approach is to standardize on Horner's Method for all polynomial evaluations. This not only improves numerical stability but creates a rigid dependency chain that prevents compilers from reordering instructions for optimization, thereby guaranteeing that the rounding errors accumulate identically on all platforms. Similarly, the summation of the periodic terms in the lunar series must follow a fixed index order—typically summing from the largest coefficient to the smallest—to ensure that the partial sums are rounded consistently.

Finally, implementers must address the variability of hardware instructions, particularly the Fused Multiply-Add (FMA). While modern processors use FMA to compute $a \times b + c$ with a single rounding step, older architectures perform a multiply followed by an add, resulting in two rounding steps. To ensure that a calendar calculation performed on a high-performance cluster matches one performed on a microcontroller or a legacy system, the reference specification should likely prohibit the use of FMA. By mandating distinct multiply and add operations with intermediate rounding, the specification enforces a "lowest common denominator" of arithmetic behavior that remains stable across diverse hardware landscapes.

\section{Technical Specifications for Reform Modules} 
\label{app:reference_modules} 

This appendix records the reusable computational modules that serve as the building blocks for the reform standards L1 through L5. Its primary purpose is to provide the rigorous mathematical definitions for the month, day, and numerical components used throughout this paper. By centralizing these definitions here, we avoid repetitive technical detail in the main text. This allows Section 5 to focus on the logical assembly and rationale of the specific reform levels. While Section 5 specifies which modules are combined to create a given standard, this appendix records the exact internal logic and rational constants of each individual module. 

For the full machine-readable implementation, including high-precision floating-point tables, readers should consult the \texttt{specs.py} module of the accompanying software library. However, the foundational rational constants are recorded here to ensure the mathematical transparency of the reform framework and to elucidate the rationale behind their selection.

\subsection{Fundamental rational constants}
\label{app:shared_fundamentals}

At the heart of all rational reform proposals is the fundamental intercalation ratio $\ru/\rv = 1336/1377$. As derived in \S\ref{ss:arith-intercalation}, this mathematically rigorous fraction replaces the traditional $65/67$ ratio, establishing a precise 334-year leap cycle that restricts the calendar's seasonal drift to a negligible error of approximately $-2$ hours per millennium.

To synchronize the month cycle with the tropical year, the initial definition point is set to $\sgang_1 = 336^\circ$. In the tropical coordinate system, this longitude corresponds to exactly $24^\circ$ before the vernal equinox, or $66^\circ$ past the winter solstice. Anchoring the first month to this specific point guarantees a permanent alignment with the historical early-spring window, ensuring that the civil New Year (Losar / Tsagaan Sar) consistently falls between late January and late February.

\subsubsection{Mean motions and epoch anchors}

The reformed engines rely on a unified set of high-precision mean motions derived from the ELP-2000/82 and VSOP87 theories \cite{ELPMPP02,VSOP87}, anchored to the Julian Day base $\mathrm{JD}=2451545.0$ (J2000.0). To guarantee determinism, these continuous rates are expressed as exact fractions. We provide two interchangeable sets:

\emph{1. Optimized per-lunation rates (default).} 
Derived from highly optimized rational approximations of the exact mean rates, this set locks the continuous day-engine to the exact same $1336/1377$ constants as the discrete month-engine, ensuring perfect internal consistency:
\begin{itemize}\itemsep2pt 
    \item \emph{Mean synodic month:} $m_1 = 283346/9595 \approx 29.5305888$ days.
    \item \emph{Mean solar longitude:} $s_1 = 334/4131$ turns per lunation.
    \item \emph{Mean lunar anomaly:} $a_1 = 4583/63907$ turns per lunation.
    \item \emph{Mean solar anomaly:} $r_1 = 1689/20891$ turns per lunation.
    \item \emph{Mean lunar latitude:} $f_1 = 324/3803$ turns per lunation.
\end{itemize}

\emph{2. Exact J2000 harmonic rates (astronomically pure).}
Evaluated directly from the J2000 baseline, these denominators are strictly factorized into ``time-harmonic'' primes ($487, 2, 3, 5$) to prevent lowest common multiple (LCM) explosions. The base denominator of $1,314,900,000,000$ represents exactly $36525 \text{ days} \times 360^\circ \times 10^5$, factoring to $487 \cdot 2^8 \cdot 3^3 \cdot 5^8$:
\begin{itemize}\itemsep2pt 
    \item \emph{Solar mean motion $L$:} $3600076983 / 1314900000000$ turns per day.
    \item \emph{Lunar elongation rate $D$:} $4452671114 / 131490000000$ turns per day ($487 \cdot 2^7 \cdot 3^3 \cdot 5^7$); the denominator is reduced by a factor of 10 as the raw elongation rate is specified to 4 decimal places.
    \item \emph{Solar anomaly rate $M$:} $3599905029 / 1314900000000$ turns per day.
    \item \emph{Lunar anomaly rate $M'$:} $47719886751 / 1314900000000$ turns per day.
    \item \emph{Lunar latitude rate $F$:} $48320201752 / 1314900000000$ turns per day.
\end{itemize}

\emph{Epoch initialization (E1987):}
The absolute epoch offset is anchored to the true new moon preceding the E1987 calendar epoch. The exact fractional time is 
$$
m_0 = 244691379521131 / 100000000 \qquad \text{(Absolute Julian Date, TT)}. 
$$
The base $100,000,000$ strips wild primes (such as $11, 13, 23$) found in the pure rational evaluation, shifting the physical anchor by a negligible $\sim 43$ microseconds while safeguarding the LCM pool. 
The corresponding orbital phases at $m_0$ are mapped to a harmonic arcsecond grid with a universal denominator of $1,296,000$ ($2^7 \cdot 3^4 \cdot 5^3$):
\begin{itemize}\itemsep2pt 
    \item \emph{Mean solar longitude:} $s_0=128634 / 1296000$ turns.
    \item \emph{Mean lunar anomaly:} $a_0=389900 / 1296000$ turns.
    \item \emph{Mean solar anomaly:} $r_0=406845 / 1296000$ turns.
    \item \emph{Mean lunar latitude:} $f_0=91591 / 1296000$ turns.
\end{itemize}

\subsubsection{Anomaly series and the harmonic arcsecond grid}

To evaluate trigonometric anomaly terms exactly, the rational engines map amplitudes to a {\em harmonized arcsecond grid} with a universal denominator of $1,296,000$ ($360^\circ \times 60' \times 60'' = 2^7 \cdot 3^4 \cdot 5^3$). This means the fractional numerators represent exact integer arcseconds.

For the rational solar series, the 1-term solar equation of center is $6893$ arcseconds. This primary term includes a secular drift of $-1 / (487 \cdot 2^9 \cdot 3^7 \cdot 5)$ turns per day. The 2-term series corrects for elliptical flattening by adding $72$ arcseconds, dropping the solar position residual error to $\sim 26$ seconds. Note that the solar drift and the second elliptical term are generally reserved for the higher-precision reform proposals.

For the rational lunar series, the 1-term series uses the major equation of center ($22640$ arcseconds). The 3-term series adds Evection ($4586$ arcseconds) and Variation ($2370$ arcseconds). The 6-term series completes the high-precision historical emulation by adding the 2nd Elliptic ($769$ arcseconds), the Annual Equation ($-666$ arcseconds), and the Reduction to Ecliptic ($-412$ arcseconds). The Annual Equation incorporates a secular drift of $1 / (2^6 \cdot 3^2 \cdot 5^{11})$ turns per day, but we do not include it in the current proposals.

Finally, regarding the evaluation of long tables, implementations requiring the full 24-term or 64-term lunar expansions tabulate the base amplitudes as integer \emph{microdegrees} (e.g., $6288774$). To safely evaluate these within the exact rational framework without precision loss, they are simply divided by $360,000,000$.

\subsubsection{Secular accelerations and time parameters}

To guarantee millennial accuracy, the engines account for secular drifts and variations in the Earth's rotation. The secular accelerations (quadratic drifts) for the Sun and the Elongation are expressed using time-harmonic primes to prevent LCM bloat during iterations:
\[ \text{Acc}_{\text{Sun}} = \frac{1}{487 \cdot 2^8 \cdot 3^3 \cdot 5^8}, \quad \text{Acc}_{\text{Elong}} = \frac{-1}{487 \cdot 2^9 \cdot 3^7 \cdot 5^4} \]

For the $\Delta T$ models, implementers may choose between a constant offset (e.g., exactly $69$ seconds, reflecting the observed state in early 2026) or a quadratic model evaluated as $-20 + 32u^2$, where $u$ is the Julian century counted from the year 1820.

\subsubsection{Civil-day triggers}

Civil day transitions depend on the chosen definition of dawn. For a constant sunrise, a choice is provided between a standard 6:00 AM mean approximation ($1/4$ day fraction) and a 5:56 AM offset ($89/360$ day fraction). 
The four-minute offset natively accounts for the combined optical effects of standard atmospheric refraction ($\sim 34'$) and the solar semi-diameter ($\sim 16'$), approximating the moment the upper limb breaches the horizon. Alternatively, for true geometric transitions, the spherical sunrise model evaluates this visual depression dynamically at $h_0 = -1/432$ turns ($-50$ arcminutes) using a base-60 harmonious obliquity of $\varepsilon = 4219/64800$ turns ($\approx 23^\circ 26' 20''$). While the rational engines evaluate mean solar time, our exact floating-point engines rigorously apply the Equation of Time to capture the true apparent solar dawn.

Furthermore, the calendar architecture is entirely meridian-agnostic. Rather than being hard-coded to a single longitude, implementations can dynamically pass standard coordinate configurations such as Lhasa ($29.65^\circ$ N, $91.10^\circ$ E), Thimphu ($27.47^\circ$ N, $89.64^\circ$ E), or Ulaanbaatar ($47.92^\circ$ N, $106.92^\circ$ E) to exactly localize the civil triggers.

\subsection{Arithmetic month module} 
\label{app:arith_month_module} 

The arithmetic month module governs intercalation using a fixed rational cycle. Its inputs are a labeled calendar month $(Y,M)$, alongside a defining parameter set: an epoch label $(Y_0,M_0)$, a rational month ratio defined by $P$ lunations per $Q$ solar months, an intercalation shift parameter $\beta^*$, and a cycle placement marker $\tau$. The module outputs the absolute true-month index $n$, and determines whether the calendar label $(Y,M)$ designates a single regular month or an intercalary pair.

\subsubsection*{Forward lookup template (label to lunation)} 
The module evaluates the exact lunation indices for a given calendar label $(Y,M)$ through a deterministic sequence. First, define the cycle difference $\ell = Q - P$ and the internal phase shift $\beta_{\text{int}} = \beta^* + (P - \tau) \pmod P$.
\begin{enumerate}\itemsep2pt 
    \item Linearize the label: Form the elapsed solar-month count from the epoch as $M^* = 12(Y-Y_0) + (M-M_0)$.
    \item Evaluate the index: Compute the internal intercalation index $I_{\text{int}} \equiv \ell M^* + \beta_{\text{int}} \pmod P$.
    \item Identify triggers: The label $(Y,M)$ designates a trigger month (spanning two physical lunations) if and only if $I_{\text{int}} < \ell$.
    \item Compute the right-end index: Apply the exact floor division formula $n_+(M^*) = \lfloor (Q M^* + \beta_{\text{int}}) / P \rfloor$.
    \item Resolve lunations: If the label is non-trigger, it maps to the single physical lunation $n = n_+(M^*)$. If the label is a trigger, it maps to an intercalary pair spanning the chronological lunations $[n_+(M^*) - 1, n_+(M^*)]$. Note that the external human labeling convention dictates which of these two physical copies formally receives the leap designation.
\end{enumerate}

\subsubsection*{Reverse lookup template (lunation to label)}
To invert the mapping and find the human calendar label $(Y,M)$ for any absolute physical lunation $n$:
\begin{enumerate}\itemsep2pt
    \item Compute the right-end solar month: $M^*(n) = \lfloor (P n - \beta_{\text{int}} - 1) / Q \rfloor + 1$.
    \item Find the cumulative month count: Add the epoch offset to find the absolute month count $C = M^*(n) + M_0$.
    \item Extract the month label: The calendar month is $M = ((C - 1) \bmod 12) + 1$.
    \item Extract the year label: The calendar year is $Y = Y_0 + \lfloor (C - M) / 12 \rfloor$.
\end{enumerate}
Because a trigger label spans two physical lunations, mapping $n$ and $n-1$ through this inverse formula will natively yield the exact same $(Y,M)$ label whenever an intercalary pair occurs.

\subsubsection*{Epoch parameter deduction}
To ensure the discrete calendar arithmetic perfectly tracks the underlying astronomical reality, the epoch parameters $M_0$ and $\beta^*$ are systematically deduced by snapping the continuous astronomical phase at the epoch directly to the rational grid. Let $s_0$ be the mean solar longitude (in turns) at the exact moment of the epoch new moon ($n=0$). We compute the continuous solar month offset between the epoch sun and the chosen zodiacal anchor $\sgang_1$ as $x = 12(s_0 - \sgang_1) \pmod{12}$.

The integer part of this offset defines the human-readable epoch calendar month: $M_0 = 1 + \lfloor x \rfloor$. The remaining fractional phase dictates the intercalation shift, scaled by the cycle's solar month length $Q$, yielding exactly $\beta^* = \lfloor Q(x - \lfloor x \rfloor) \rfloor$. This precise extraction maps the continuous floating-point evaluation into strict integer algebra, structurally preventing the off-by-one boundary alignment errors that frequently plague epoch initialization logic.

\subsection{Rational month module} 
\label{app:rational_month_module} 

The rational month module determines the civil month structure astronomically by tracking the true solar transits across a chosen set of zodiacal definition points. Its inputs are an absolute physical lunation index $n$, alongside a defining parameter set: a continuous rational solar anomaly series, a continuous rational lunar anomaly series, a specified anchor longitude $\sgang_1$, and specific naming conventions for handling months with zero or two transits. The module outputs the human calendar label $(Y,M)$ for the lunation and dictates whether it represents a regular, intercalary, or skipped month.

\subsubsection*{Forward lookup template (lunation to label)} 
The module evaluates the civil label for an absolute lunation $n$ using a strictly continuous physical evaluation over the J2000.0 Terrestrial Time ($\mathrm{TT}$) coordinate system.

\begin{enumerate}\itemsep2pt 
    \item Evaluate true date: Calculate the true physical time $t_{\text{TT}}$ of the new moon $n$ using a specified fixed-step Picard iteration. Because this finite iteration serves as the strict mathematical definition of the calendar's astronomical time, it does not merely approximate the root of $E_{\text{true}}(t_{\text{TT}}) = n$. Paradoxically, this means the exact continuous elongation series acts as an approximate inverse to this formally defined time function.
    \item Evaluate true solar phase: At this exact moment $t_{\text{TT}}$, evaluate the true solar longitude $S_{\text{true}}(t_{\text{TT}})$ using the rational solar anomaly series.
    \item Compute the absolute transit index: Subtract the normalized zodiacal anchor offset $\sgang_1$ from the true solar longitude. The floor of this shifted solar phase multiplied by $12$ yields the absolute civil transit name $Z_n$.
    \item Determine transit containment: Compare $Z_n$ to the previous new moon's transit index $Z_{n-1}$. The difference $\Delta Z = Z_n - Z_{n-1}$ indicates exactly how many definition points the mean sun crossed during the chronological time interval between new moon $n-1$ and new moon $n$, which constitutes the physical lunation $n$.
    \item Resolve the civil label: 
    \begin{itemize}
        \item If $\Delta Z = 1$, the lunation contains exactly one transit and is an ordinary month.
        \item If $\Delta Z = 0$, the lunation contains no transits and is designated as a leap month. The civil month number is borrowed from either the previous or following transit according to the chosen naming convention.
        \item If $\Delta Z = 2$, the lunation contains two transits, requiring a month label to be skipped. The civil number is assigned by keeping either the first or second transit name, according to the chosen skipped-month convention.
    \end{itemize}
\end{enumerate}

\subsubsection*{Reverse lookup template (label to lunation)}
Because the rational month engine assigns names dynamically based on physical transits, the mapping from a human label $(Y,M)$ back to a physical lunation $n$ cannot use a closed-form arithmetic floor formula. Instead, it utilizes an informed bounded search.
\begin{enumerate}\itemsep2pt
    \item Compute the target absolute transit index $N_{\text{target}}$ corresponding to the requested year and month, accounting for the epoch anchor $Y_0$.
    \item Generate an initial physical guess $t_{\text{guess}}$ by inverting the linear mean solar motion to find the approximate time of the requested transit.
    \item Map $t_{\text{guess}}$ through the mean elongation series to find an approximate starting lunation index $n_{\text{guess}}$.
    \item Evaluate the true civil label for $n_{\text{guess}}$ and execute a bidirectional monotonic search to locate all contiguous physical lunations $n$ that bear the exact target label $N_{\text{target}}$. This search seamlessly handles cases where a label spans two lunations (leap month) or zero lunations (skipped month).
\end{enumerate}

\subsection{Unified day-module template} 
\label{app:day_module_template} 

For the principal reform proposals, the day layer is best described by a single algorithmic template with several interchangeable components. The traditional rational day rules and the modern floating-point day engines differ less in overall structure than in the physical fidelity and numerical contract of their subcomponents. 

\emph{Inputs.} An absolute lunar-day index $x$ (derived from the active lunation index and the lunar day $d$), a geographic location specification (latitude $\phi$ and longitude $L$), and a parameter bundle consisting of: 
\begin{itemize}\itemsep2pt 
    \item a continuous solar anomaly series, 
    \item a continuous lunar elongation series, 
    \item a $\Delta T$ time-correction package, 
    \item a civil-day trigger (sunrise) package, 
    \item and a numerical backend specification (exact rational or floating-point). 
\end{itemize} 

\subsubsection*{Forward lookup template (lunar-day to civil boundary)}
The engine translates the absolute lunar-day index into a discrete civil Julian Day Number (JDN) through the following sequential mapping:
\begin{enumerate}\itemsep2pt 
    \item Define the target phase: The exact astronomical target for the boundary of lunar-day $x$ corresponds to a continuous lunar elongation of exactly $x/30$ turns.
    \item Evaluate true physical time: Solve the inverse kinematic problem using a fixed-step Picard iteration to find the exact terrestrial time $t_{\text{TT}}$ where the true elongation series satisfies $E_{\text{true}}(t_{\text{TT}}) = x/30$.
    \item Apply secular time corrections: Evaluate the chosen $\Delta T$ model for the computed epoch to dynamically translate the physical terrestrial time $t_{\text{TT}}$ into universal time $t_{\text{UTC}}$.
    \item Evaluate the civil-day trigger: Calculate the true and mean solar longitudes near the target dawn. Pass these continuous coordinates to the sunrise module to evaluate the exact local mean time of dawn for the specified geographic location.
    \item Compute the absolute civil date: Align the UTC lunar-day boundary with the computed local dawn fraction. Applying the mathematical floor to this shifted continuous coordinate yields the discrete absolute civil JDN, denoted $J(x)$, for the lunar-day $x$.
\end{enumerate} 

\subsubsection*{Reverse lookup template (Julian day to calendar label)}
Because the civil anomalies (skipped and repeated days) emerge dynamically from the varying speeds of the Sun and Moon, the mapping from a continuous Julian Day Number back to a calendar label relies on a monotonic boundary search rather than a closed-form division.
\begin{enumerate}\itemsep2pt
    \item Initial guess: Estimate the continuous time from the requested JDN and map it through the mean elongation series to find an approximate absolute lunar-day index $x_{\text{guess}}$.
    \item Monotonic search: Evaluate the exact discrete civil boundary $J(x)$ for the guess. A civil day definitively belongs to lunar-day $x$ if and only if it falls strictly after the boundary of the previous lunar-day and on or before the boundary of the current lunar-day, such that $J(x-1) < \text{JDN} \le J(x)$. Iterate $x$ backward or forward until this enclosure is strictly satisfied.
    \item Label decomposition: Decompose the resolved absolute lunar-day $x$ into the absolute lunation index $n$ and the lunar day $d$, where $d = ((x-1) \bmod 30) + 1$. The lunation $n$ is then passed to the month module to resolve the calendar year and month.
    \item Anomaly extraction: The mathematical boundaries natively encode the calendar anomalies without requiring supplementary rules. If $\text{JDN} - J(x-1) > 1$, the civil day is the second occurrence of a duplicated lunar-day. Similarly, if $J(x-1) = J(x-2)$, the preceding lunar-day was skipped.
\end{enumerate}

\subsubsection*{Numerical backends and subcomponents}
The execution of the template depends strictly on the selected numerical backend and the physical subcomponents provided by the parameter bundle:
\begin{itemize}\itemsep2pt
    \item Numerical backends: The rational engine evaluates all continuous time and trigonometric functions using exact integer arithmetic, a harmonized arcsecond grid, and discrete sine and arctangent lookup tables, structurally guaranteeing identical deterministic execution across all hardware. Conversely, the float engine executes via high-performance 64-bit precision, replacing the discrete lookup tables with continuous quarter-wave polynomial approximations.
    \item Sunrise models: The civil trigger supports three distinct definitions of dawn. The constant model applies a fixed daily fraction (e.g., exactly 6:00 AM LMT). The spherical model evaluates the geometric hour angle dynamically via the spherical law of cosines and a specified visual depression angle. The true physical model goes a step further, perfectly correcting local apparent solar time to local mean solar time by continuously evaluating the exact equation of time. All spherical models natively handle polar day and polar night boundary conditions.
    \item Time correction models: The $\Delta T$ package bridges absolute physical time and the variable rotation of the Earth, utilizing either a constant temporal offset or a dynamic quadratic model mapping historical deceleration.
\end{itemize}

\end{document}